\documentclass[a4paper,11pt,twoside]{report}
\usepackage{graphicx,a4wide,europs,listings}
\usepackage[intlimits]{amsmath}
\usepackage{amsxtra, amssymb,fancyhdr, amsthm,latexsym}
\usepackage[english]{babel}
\usepackage[ansinew]{inputenc}
\usepackage[bookmarks]{}
\usepackage{enumerate}

\usepackage{tikz}
\usepackage[final]{pdfpages}

\numberwithin{equation}{section}
\numberwithin{figure}{chapter}

\title{\textbf{Modelling Crowd Dynamics:\\ a Multiscale, Measure-theoretical Approach}\\\vspace{1 cm}
\large{Master's Thesis Industrial and Applied Mathematics\\
Eindhoven University of Technology, The Netherlands\\\vspace{5 cm}
                    Author:\\
                    Joep Evers\\\vspace{1 cm}
                    Supervisor:\\
                    Dr.~Adrian Muntean}}
\author{}
\date{May 30, 2011}

\newcommand{\setheaders}{
\pagestyle{fancy}
\fancyhf{}
\fancyhead[LE, RO]{Technische Universiteit \textbf{Eindhoven} University of Technology}
\fancyfoot[LE]{\thepage \hspace{0.5 cm} Modelling Crowd Dynamics}
\fancyfoot[RO]{Modelling Crowd Dynamics\hspace{0.5 cm} \thepage }
\renewcommand{\headrulewidth}{0pt}
}
\setheaders{}

\begin{document}
\newtheorem{theorem}{Theorem}[section]
\newtheorem{lemma}[theorem]{Lemma}
\newtheorem{postulate}[theorem]{Postulate}
\newtheorem{definition}[theorem]{Definition}
\newtheorem{assumption}[theorem]{Assumption}
\newtheorem{corollary}[theorem]{Corollary}
\theoremstyle{definition}
\newtheorem{remark}[theorem]{Remark}

\newcommand{\supp}{\operatorname*{supp}}
\newcommand{\argmin}{\operatorname*{arg\,min}}

\pagenumbering{roman}

%
%
%
%
%
%
%
%
%
%
%

\maketitle

\newpage
\abstract{
We present a strategy capable of describing basic features of the dynamics of crowds. The behaviour of the crowd is considered from a twofold perspective: both macroscopically and microscopically. We examine the large scale behaviour of the crowd (considering it as a continuum), simultaneously being able to capture phenomena happening at the individual pedestrian's level. We unify the micro and macro approaches in a single model, by working with general mass measures and their transport.\\
We improve existing modelling by coupling a measure-theoretical framework with basic ideas of mixture theory formulated in terms of measures. This strategy allows us to define several constituents (subpopulations) of the large crowd, each having its own partial velocity. We thus have the possibility to examine the interactive behaviour between subpopulations that have distinct characteristics. We give special features to those pedestrians that are represented by the microscopic (discrete) part. In real life situations they would play the role of firemen, tourist guides, leaders, terrorists, predators etc. Since we are interested in the global behaviour of the rest of the crowd, we model this part as a continuum.\\
By identifying a suitable concept of entropy, we derive an entropy inequality and show that our model agrees with a Clausius-Duhem-like inequality. From this inequality natural restrictions on the proposed velocity fields follow; obeying these restrictions makes our model compatible with thermodynamics.\\
We prove existence and uniqueness of a solution to a time-discrete transport problem for general mass measures. Moreover, we show properties like positivity of the solution and conservation of mass. Although afterwards we opt for a particular form, our results are valid for mass measures in their most general appearance.\\
We give a robust scheme to approximate the solution and illustrate numerically two-scale micro-macro behaviour. We experiment with a number of scenarios, in order to capture the emergent qualitative behaviour.\\
Finally, we formulate open problems and basic research questions, inspired by our modelling, analysis and simulation results.
\textbf{PACS 2010:} 45.50.-j; 47.10.ab; 02.30.Cj; 02.60.Cb; 05.70.-a; 47.51.+a}

\newpage\thispagestyle{empty}\mbox{}
\vspace{5 cm}
\begin{center}
\parbox{11 cm}{\hyphenpenalty=100000 \large{I want to suggest that even with our woeful ignorance of why humans behave the way they do, it is possible to make some predictions about how they behave collectively.}}
\end{center}\hspace{9.8 cm}\begin{footnotesize}Philip Ball, \textit{Critical Mass}\footnote{Quotation taken from \protect\cite{Ball}, p. 6.}\end{footnotesize}

\pagenumbering{arabic}
\pagestyle{empty}
\tableofcontents


\newpage
\setheaders{}
\chapter{Introduction}\label{section introduction}
I would like to use this first section to introduce the central challenge treated in this thesis: the modelling of the behaviour of crowds. This section is also used to place this subject in a broader context.
\section{Process leading to this thesis}
In 2010 I was offered the opportunity to take part in the Honors Program Industrial and Applied Mathematics: a new initiative of the Department of
Mathematics and Computer Science (Eindhoven University of Technology) to challenge "the best students of IAM" by working in an academic, scientific setting. During these months I analyzed and extended a number of currently existing approaches to the modelling of crowd dynamics, under supervision of Dr.~Adrian Muntean and Prof.dr.~Mark Peletier. Since this research project took place in a relatively short period of time, we ultimately had to conclude that there were many more things to be explored in this field. As a natural consequence, we thus decided to dedicate my master's thesis to the same topic.\\
\\
The next part of this introduction is intended to emphasize that we have not been dealing with an irrelevant problem. Indeed, a series of everyday-life events have shown the importance of being able to predict the dynamics of a crowd. Several attempts to capture human behaviour in mathematical models have found their way into useful applications. In the following sections these statements are discussed in more detail.

\section{Societal background}
In his 1895 book \textit{La psychologie des foules}, the French social psychologist Gustave Le Bon wrote: "L'âge où nous entrons sera véritablement \textbf{l'ère des foules}".\footnote{"The age we are about to enter, will truly be the era of crowds". The French quotation is taken from \protect\cite{LeBon}, p.~12, a republication of the original work published in 1895.} Although Le Bon could probably at the time (the \textit{fin de siècle}) not quite foresee what the twentieth century would look like, his words have turned out to be prophetic. In their most literal sense - disregarding any political or metaphorical meaning that Le Bon might have intended to attach to them - they describe what our world has become. Indeed, during the twentieth century the global population has been ever-increasing. Occasions involving large numbers of people in crowded areas have become familiar sights. We experience such situations both under everyday urban conditions, and during large-scale manifestations with huge audiences (which take place occasionally).\\
\\
Most of the "everyday situations" mentioned above happen in a `normal' setting, that is, without the people being in a state of panic. Although this tends to sound reassuring, it does by no means imply that no special attention is to be paid. In order to design safe and comfortable public space, the presence of pedestrians cannot be disregarded. The behaviour of individual pedestrians, possibly clustering together to form larger crowds, is an important factor in urban design and traffic management. Indeed, \textit{un}safe and \textit{un}comfortable situations are often related to congestion and high densities. To assess the quality of the pedestrian environment, the following questions may be taken as a guideline (see \cite{UrbDesComp}):
\begin{itemize}
  \item Are routes direct, leading the people where they actually want to go?
  \item Is it easy to find and follow a certain route?
  \item Are crossings easy to use; how long do pedestrians have to wait before they can cross a road?
  \item Are footways well-lit, and sufficiently wide; what obstructions are there?
  \item If pedestrians, cyclists and vehicles are mixed, does the typical speed of road users allow for this?
\end{itemize}
These questions also indicate which measures can prevent the aforementioned cases of unsafety/discomfort. The necessity of these measures depends on the typical scale of pedestrian traffic (i.e. the `crowdedness') that the area has to deal with. The effectiveness of each specific measure is also determined by its interplay with the others.\\
\\
The comments made here, should be extended to the semi-public domain. To be more concrete, we do not only need to consider the behaviour of people walking in the street, on sidewalks and in parks, but also in railway stations (see Figure \ref{Figure stationEindhoven}), sports stadiums and shopping malls. In fact, the investigation of crowd dynamics is essential in any setting in which a person's desired motion is impeded by the presence of other people, to prevent that inconvenience escalates into danger.\\

\begin{figure}
  \includegraphics[width=\linewidth]{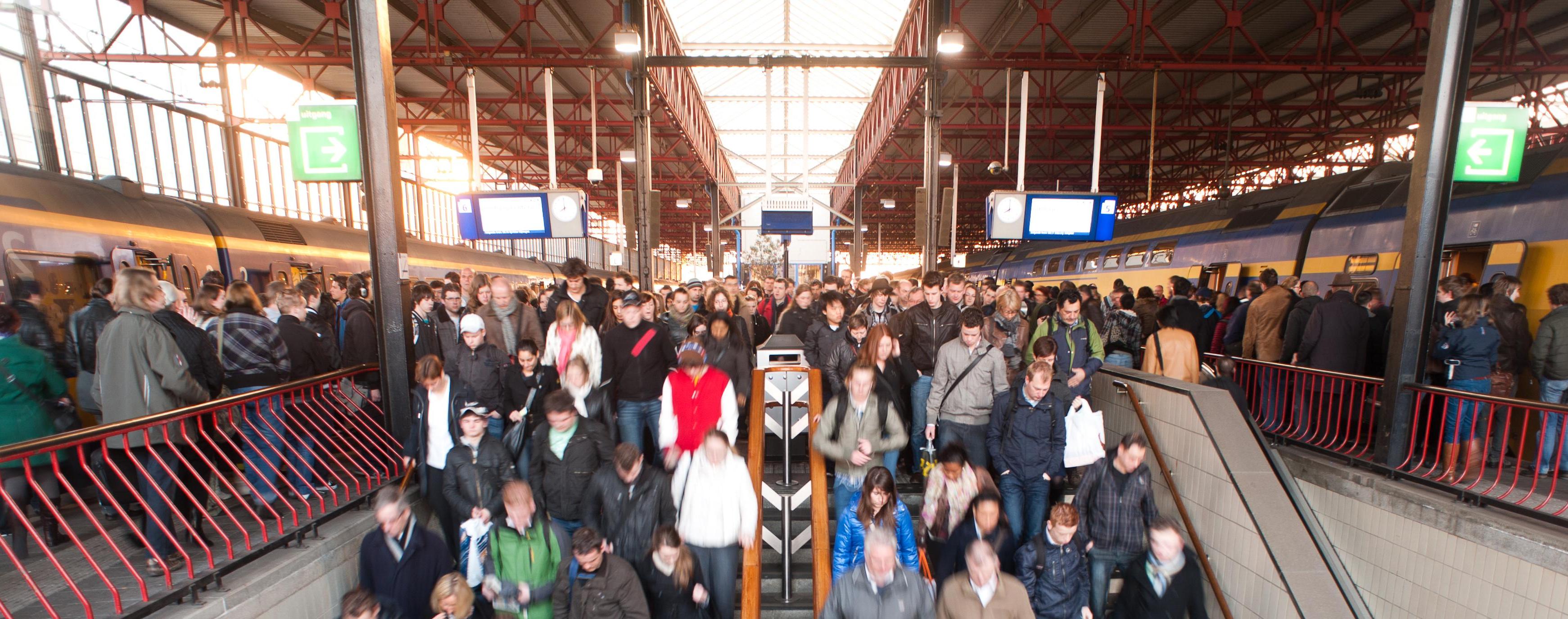}\\
  \caption{Passengers rushing towards the exit of Eindhoven Railway Station, after alighting from two trains that have just arrived. At the transition from the platform to the staircase a bottleneck occurs. \textit{(Photo: Joep Julicher)}}\label{Figure stationEindhoven}
\end{figure}
If proper understanding of, or response to a crowd's behaviour is lacking, events in the (recent) past have shown us what the consequences are. These are dramatic situations during which people's lives are at stake. An example is the pilgrimage (\textit{Hajj}) to Mecca and other holy places in Saudi Arabia, that Muslims perform every year during the \textit{Dhu al-Hijjah} month. Since all adult Muslims are required to perform such Hajj at least once during their lifetimes (unless physical or financial circumstances make this impossible), huge numbers of people (several millions) gather every year during a relatively short period of time. Near the town of Mina, pilgrims as a ritual throw pebbles at three stone pillars (\textit{jamarahs}), representing the devil.\footnote{God/Allah had commanded Abraham/Ibrahim to sacrifice his son Ishmael, but Satan urged Abraham to disobey God. The `stoning of the devil' ritual is to commemorate Abraham rejecting the temptation of the devil.\protect\cite{Hajj}} A special bridge (the Jamarat Bridge) was built is the 1960s to allow pilgrims to take part in the ritual at two floor levels. However, during the 1424 AH (2004) and 1426 AH (2006) Hajjs several hundreds of people were trampled to death due to overcrowding at the bridge. As this urged authorities to take action, crowd experts were asked to investigate the situation, and give recommendations for improvements; see \cite{Still, HelbingJamarat}.\\
The site was completely reconstructed afterwards. The three pillars were replaced by wider and taller oblong concrete structures, to allow more pilgrims at the same time, and prevent them from accidentally throwing pebbles at each other. Crowds flow around these new `pillars' more easily, by which congestion is reduced. A new multi-storey bridge was built with better entrance facilities and more emergency exits. Bottlenecks were removed. Moreover, the protocols for crowd management by stewards and other officials were reconsidered.\cite{HajjBBC}\\
\\
Unfortunately, crowd disasters do not always happen far away from us. Recently, in July 2010, 21 people died at the Loveparade in the German city of Duisburg, at a distance of less than 100 kilometers from Eindhoven. This festival took place in a closed-off area, which could be entered via a small number of tunnels. Each of these ended at a common ramp/staircase, eventually leading to the festival ground. Due to overcrowding at the bottom end of the staircase, a stampede occurred. All lethal victims died because of suffocation.\\
After these disastrous events, there was a lot of criticism on the safety precautions that had, or had not, been taken. The organization of the Loveparade was blamed for having provided too few emergency exits to the Loveparade area. Moreover, the number of security agents was smaller than promised. On the other hand, local authorities are said to have ignored the police's objections, because they were so keen on having this prestigious event in their town. For more information, see e.g. \cite{ZeitDuisburg1, TagesspiegelDuisburg, SueddeutscheDuisburg}.\\
\\
The above emphasizes the importance and urgency of being able to properly describe the dynamics of a crowd. This is a prerequisite for predicting crowd behaviour, which in its turn is needed to anticipate life-threatening situations.

\section{Modelling approaches}\label{section introduction modelling approaches}
Ever since the Renaissance scientists have wondered how human behaviour can be captured in mathematical formalism (cf. e.g. \cite{Ball}). We focus now in particular on the developments of the last decades, during which the description and analysis of a crowd's motion has gained the attention of the scientific world. This is mainly due to the gradual increase of the number of large-audience events being organized, and the accidents happening at such events. An illustration of these developments has been given above. Several models have been developed and explored to catch the crowd-related phenomena we experience in real life, in a scientific (mathematical) framework. These models were treated both analytically and numerically, and were based on a number of distinct perspectives.\\
\\
Focussing on what happens at the level of an individual pedestrian, one obtains discrete, agent-based or microscopic models. In this approach the dynamics of each single person in the crowd are modelled and traced. This perspective is adopted by Dirk Helbing et al. within the framework of a so-called \textit{social force model}, see e.g. \cite{HelbingMolnar}. Roughly speaking, each pedestrian is driven via Newton's Second Law, by means of a \textit{social} force that depends on the presence and social behaviour of other people. Simulation is their main tool for obtaining qualitative and quantitative information.\\
\\
On the other hand, one can also `zoom out', considering a crowd on a more global scale. Working on this macroscopic level, one uses densities, rather than individual pedestrians. Via this approach one ends up in what is essentially a fluid-dynamic setting. A specific property of these models is that we cannot capture local interactions any more. Macroscopic modelling is thus only useful if we are interested in average characteristics of the crowd and its motion. This way of modelling is for instance adopted by Bertrand Maury et al. in \cite{Maury}, where initially a gradient flow structure is proposed. Their results are derived analytically.\\
\\
Up to the initial conditions (which might be generated by means of a random sampling), the models described above are fully deterministic.\\
\\
In this thesis we focus an approach that differs in nature from the aforementioned ones. Our motivation is that we do not want to be forced to choose either of the two perspectives; we aim at `marrying' the two perspectives in one single model. This strategy was described by Benedetto Piccoli, Andrea Tosin et al. in \cite{PiccoliTosinMeasTh, PiccoliTosin, Piccoli2010}. In most of the fluid-dynamic models of pedestrians, non-linear hyperbolic conservation laws appear. Certain analytical and numerical problems inherently arise when treating these PDEs, especially in more than one space dimension. For instance the solution might not be unique, shock waves can occur, it is important whether the corresponding flux is convex or not, boundary conditions are difficult to impose, etc. Numerically, positivity of the density might not always be guaranteed. In order to circumvent these problems, Piccoli et al.'s work takes place in a (time-discrete) measure-theoretical framework. This thesis is mainly built upon the fundaments Piccoli \textit{cum suis} provides.\\
However, we want to be able to cover a much wider range of situations, for instance a crowd consisting of several distinct subpopulations (rather than only one population), each of which is willing to move with its own desired velocity. Moreover, we draw parallels with mixture theory and thermodynamics, and extract inspiration from those fields.\\
\\
The microscopic, macroscopic and micro-macro approaches described above, are not the only existing strategies. See e.g. \cite{Schadschneider2011} (pp.~212--214) for a systematic classification of models. Partially based on \cite{Schadschneider2011}, we describe a number of different perspectives.\\
\\
At an intermediate level between micro and macro are so-called mesoscopic models. These models do not distinguish between individuals and hence cannot trace their trajectories. However, individual \textit{behaviour} can be specified. The description is in terms of probability densities: mostly, the probability to find an individual with specified speed in a certain location at a certain time. To give an example, Dirk Helbing describes vehicular traffic in such framework (also indicated by `Boltzmann-like' or `gas-kinetic' models) in \cite{HelbingPhysA}. Moreover, in \cite{Helbing} he derives hydrodynamic equations for pedestrians, based on the Boltzmann-approach.\\
\\
Up to now, all models have mainly been based on considerations related to physics. This mostly matches our point of view. Other views are (of course) also possible.\\
\\
The fact that we do not fully understand the underlying mechanisms of human behaviour, is most naturally incorporated in a model by adding stochasticity. A small amount of \textit{external noise} can be added to avoid undesired situations. One should think here for example of a situation in which two individuals are positioned `head-on' and both want to move straight ahead. In a model it might happen that a deadlock occurs, although in reality the two most likely just move aside slightly, and get passed one another. A bit of stochastic fluctuation forces a breakthrough if the described deadlock happens. Away from such configurations, the influence of the noise is only small.\\
\textit{Fully} stochastic models are different in the sense that they do not just add random perturbation to intrinsically deterministic dynamics. Here, the underlying decision-making processes are influenced directly by random effects. If you take the deterministic limit in these models (i.e. vanishing stochastic effects) the overall phenomena are completely different from those in the stochastic `normal' case. This makes fully stochastic models intrinsically different from models that contain external noise only.\\
\\
In a \textit{cellular automata} model all variables are discrete. The crowd is described at the individual's level. The spatial domain is subdivided into a number of cells, where each cell possesses a state variable (e.g. $0=$ `empty', $1=$ `occupied'). At specified (discrete) points in time, the model decides which cells are occupied in the subsequent generation; the state variables are then updated accordingly. For reality-mimicking models the total number of occupied cells is constant. Typically, the update is rule-based, i.e. each occupied cell (read: particle/pedestrian) makes a decision where to go, based on the current situation and its goals. These rules are often supported by psychological arguments. Possibly, some stochastic effects are included. The update can either be executed in parallel (for all cells simultaneously), or for randomly selected individuals only.\\
\\
Serge Hoogendoorn and Piet Bovy model individuals as players in a differential game. This approach was first applied to driving behaviour in traffic flows in \cite{HoogendoornGameTh}. The ideas presented therein were subsequently extended to pedestrian dynamics; see \cite{HoogendoornBovyGameTh}. Vehicles and/or pedestrians are assumed to maximize their expected success or profit, or follow a Zipfian principle of \textit{least effort}. In such a game-theoretical approach, individuals are allowed to modify their control decisions based on their observations, and on predictions of the behaviour of other players in their neighbourhood.

\section{A few people in the field}\label{section introduction people in the field}
Throughout 2010 and 2011 Adrian and I have been in contact with many people, whose area of expertise was related to our research. When meeting them, we were mainly interested in what kind of questions these experts would like to be answered by crowd models. Moreover, we obtained insight in the wide range of possible applications.\\
\\
The person who provided us with the actual idea for studying the topic of crowd dynamics is Prof.~Chris Budd (University of Bath, UK). Mark Peletier invited me to meet Prof.~Budd during his stay in Eindhoven, Spring 2010. At that time, I needed to decide what I wanted to do in my Honors Program, and these two people suggested the idea (that has eventually led to this thesis). In January 2011, Adrian Muntean and I visited Mark Peletier during his sabbatical stay in Bath, and we met Prof.~Budd again. He was the one drawing our attention to two-scale phenomena in nature (birds and fish), by showing us a couple of fascinating BBC movies.\footnote{These phenomena will be addressed in Section \ref{section two-scales no sing cont} of this thesis.}\\
\\
Prof.dr.ir.~Serge Hoogendoorn, Dr.ir.~Winnie Daamen and Mario Campanella~MSc (Delft University of Technology) work on modelling/simulation of pedestrian/traffic dynamics, and on the calibration of their models by real-life experiments.\footnote{See \texttt{www.pedestrians.tudelft.nl} for more information on either of these two areas of expertise.} Winnie Daamen gave a talk in Eindhoven (during the CASA colloquium) in September 2010. Adrian Muntean and I paid a visit to their Transport \& Planning Department in January 2011.\\
They have developed a software package called \textsc{Nomad}, which is a microscopic simulation tool that can e.g. be used to assess the geometry of infrastructure. The underlying model is the one described in \cite{HoogendoornBovyGameTh}. The behaviour of the individuals is comprised in their acceleration, which consists of an uncontrollable and a controllable part. The uncontrollable part is due to physical interactions with their surroundings, that is, with other individuals and obstacles (part of the geometry of the domain). The controllable part contains the tendency to minimize the walking cost (cf. Section \ref{section introduction modelling approaches}).\\
Calibration of the simulation is done by observing real-life traffic and pedestrian flows (as far as privacy regulations permit to do so) and laboratory experiments; see e.g. \cite{CampanellaIFAC09, Hoogendoorn}. These laboratory experiments take place in standardized environments, and participants' positions are obtained by video image analysis. An example of such standard setting (which is very well-known in the field) is the \textit{counterflow} or \textit{bidirectional flow}: a narrow corridor in which two groups of people want to move in opposite directions. In this experiment, one expects to observe a specific phenomenon of \textit{self-organization}: lane-formation. That is, people tend to align in such a way that they just follow someone going in the same direction.\footnote{A video of a bidirectional flow experiment is available at \texttt{www.youtube.com/watch?v=J4J\_\_lOOV2E}. Lanes are formed without the participants being instructed to do so (note that these lanes are unsteady). The experiment was performed in connection with the German Hermes project; in Delft similar experiments are done.} In Figure \ref{figure counterflow} a schematic impression of a bidirectional flow is given. Another standard scenario is a bottleneck. It occurs for instance if a corridor suddenly gets narrower. At this transition, individuals get clogged up in circular structures (\textit{arches}) which are hard to break; this scenario is indicated in Figure \ref{figure arches}. See also \cite{Schadschneider2011}, pp.~418--419, for a description of both settings and of the emerging phenomena.
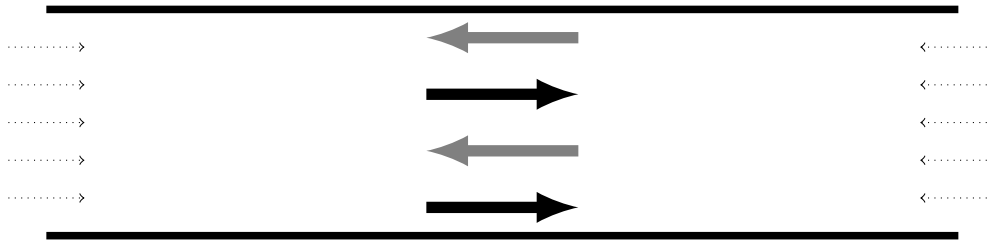
\begin{figure}[ht]
\centering
\begin{tikzpicture}
    \draw[line width = 1 mm] (0,0)--(12,0);
    \draw[line width = 1 mm] (0,3)--(12,3);

    \draw[->, dotted] (-0.5,0.5)--(0.5,0.5);
    \draw[->, dotted] (-0.5,1)--(0.5,1);
    \draw[->, dotted] (-0.5,1.5)--(0.5,1.5);
    \draw[->, dotted] (-0.5,2)--(0.5,2);
    \draw[->, dotted] (-0.5,2.5)--(0.5,2.5);

    \draw[<-, dotted] (11.5,0.5)--(12.5,0.5);
    \draw[<-, dotted] (11.5,1)--(12.5,1);
    \draw[<-, dotted] (11.5,1.5)--(12.5,1.5);
    \draw[<-, dotted] (11.5,2)--(12.5,2);
    \draw[<-, dotted] (11.5,2.5)--(12.5,2.5);

    \draw[arrows={-latex}, line width = 1.5 mm] (5,0.375)--(7,0.375);
    \draw[arrows={latex-}, line width = 1.5 mm, gray] (5,1.125)--(7,1.125);
    \draw[arrows={-latex}, line width = 1.5 mm] (5,1.875)--(7,1.875);
    \draw[arrows={latex-}, line width = 1.5 mm, gray] (5,2.625)--(7,2.625);
\end{tikzpicture}
\caption{Schematic drawing of a counterflow scenario with lane-formation. At both ends of the corridor individuals are supplied; they intend to walk to the other side. Halfway, four lanes emerge.}\label{figure counterflow}
\end{figure}

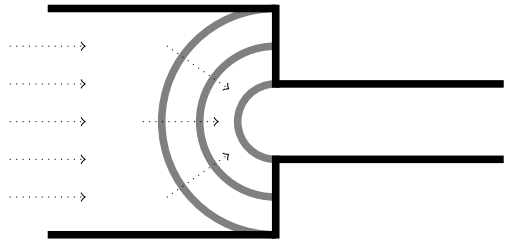
\begin{figure}[ht]
\centering
\begin{tikzpicture}
    \draw[line width = 1 mm, gray] (3,3) arc(90:270:1.5cm);
    \draw[line width = 1 mm, gray] (3,2.5) arc(90:270:1cm);
    \draw[line width = 1 mm, gray] (3,2) arc(90:270:0.5cm);

    \draw[line width = 1 mm] (0,0)--(3,0);
    \draw[line width = 1 mm] (0,3)--(3,3);
    \draw[line width = 1 mm] (3,1)--(6,1);
    \draw[line width = 1 mm] (3,2)--(6,2);
    \draw[line width = 1 mm] (3,3.05)--(3,1.95);
    \draw[line width = 1 mm] (3,1.05)--(3,-0.05);

    \draw[->, dotted] (-0.5,0.5)--(0.5,0.5);
    \draw[->, dotted] (-0.5,1)--(0.5,1);
    \draw[->, dotted] (-0.5,1.5)--(0.5,1.5);
    \draw[->, dotted] (-0.5,2)--(0.5,2);
    \draw[->, dotted] (-0.5,2.5)--(0.5,2.5);

    \draw[->, dotted] (1.25,1.5)--(2.25,1.5);
    \draw[->, dotted, rotate around={-35:(3,1.5)}] (1.25,1.5)--(2.25,1.5);
    \draw[->, dotted, rotate around={35:(3,1.5)}] (1.25,1.5)--(2.25,1.5);

\end{tikzpicture}
\caption{Schematic drawing of arches at a bottleneck. the flow is obstructed there.}\label{figure arches}
\end{figure}

Ing.~Gilian Brouwer works for Adviesburo Nieman: a company specialized in consultancy on quality, safety and building physics.\footnote{\texttt{www.nieman.nl}} He is specialized in fire safety engineering. At Nieman, several software tools are applied to simulate the evacuation of people in case of fire in a building. In case a building does not comply with Dutch regulations (`Bouwbesluit'), simulations are needed to make feasible that nevertheless an equivalent level of safety is provided. In \cite{Zeeburg} a comparison is made between two of the software packages that are used by Nieman (\textit{FDS + Evac} and \textit{Simulex}). Both packages are based on a social force-like model, where FDS + Evac also includes the effect of smoke. An individual's desired route is based on the geometry of the domain and the configuration of other pedestrians. The route is chosen such that it is expected to take a minimal amount of time to reach the destination. In Simulex however, in each spatial position the shortest route to the exit is precalculated (considering the distance). The evacuation is assessed by recording the times at which pedestrians leave the building, and consequently calculating exit flow rates or average exit times.\\
\\
Prof.dr.ir.~Bauke de Vries is the chairman of the Urban Management \& Design Systems group at the Department of Architecture, Building and Planning (Eindhoven University of Technology). We met him in February 2011 and spoke about his working interests and area of expertise. He mainly focuses on interior design of buildings, where for instance routing and the positions of obstacles are taken into account. In the past, he also performed real-life experiments in his department, to obtain data about the walking behaviour of his coworkers within the office space. During our conversation, he also explained the work of Prof.dr.~Harry Timmermans, who is a member of the same group. He is more specialized in the design of exterior urban space, such as shopping malls. An example can be found in \cite{DijkstraDeVries}, in which the two experts cooperated. Their simulation models work with agents (read: pedestrians) that have an agenda and an updating mechanism for their `to do'-list; the individuals behave and move accordingly.\\
\\
We also talked with people about crowds in a broader context. An example is Dr.~Petru Curseu who is an expert in group processes and decision-making at Tilburg University. To get an impression of his area of interest, see \cite{Curseu}. When he visited us in Eindhoven (November 2010), he showed his interest in microscopic pedestrian models. He proposed to use them as a metaphor for more abstract societal phenomena. For instance, if there is strongly repulsive interaction between a large group and a small group, the question arises whether this enables the large group to leave a room relatively faster than the small group. According to Dr.~Curseu, a parallel can be drawn between this pedestrian simulation and a majority in a country that discriminates a minority. The majority group turns out to have better access to resources.\\
I took part in the Study Group Mathematics with Industry 2011 (organized at the Vrije Universiteit Amsterdam, 24--28 January 2011). The problem was provided by \textsc{Chess}.\footnote{\texttt{www.chess.nl}} An \textit{ad hoc} wireless network is considered, in which each node can broadcast and receive messages. However, it is uncertain whether sent messages reach a receiver, and if yes, how many/which nodes are reached. The aim was to let each node send, receive and process `intelligent' messages such that it can estimate the number of nodes in its neighbourhood and in the total network. The application possibilities of such unreliable wireless networks are closely linked to the subject of this thesis. This was also pointed out to us at an earlier stage by Prof.dr.~Fabian Wirth (University of Würzburg, Germany), an expert in logistic networks. Typically, each node is a simple microprocessor, with limited computational power, that measures a property of its environment, like temperature or position. A large number of such sensors can be used to detect forest fires, but also to obtain information about e.g. bird flocking, social behaviour and crowd dynamics. The proceedings of the study group are to be expected.\\
In March 2011 Adrian Muntean and I had a conversation with Dr.~David Fedson, a medical expert in the area of epidemics. We talked about Malcolm Gladwell's \textit{The Tipping Point}, which treats phase transitions that take place everywhere around us. Certain events suddenly take off in social behaviour (fashion), epidemics of infectious diseases, and in animal groups. For the latter see \cite{SwarmScienceNews}. Another topic he brought to our attention is synchronization. A connection between synchronous behaviour in nature and in the financial world is made in \cite{Saavedra}. Prof.~Steven Strogatz also explains this phenomenon in a fascinating way.\footnote{\texttt{www.youtube.com/watch?v=aSNrKS-sCE0}}

\section{Content and structure of this thesis}
In this thesis we focus on the mathematical background of crowd dynamics models. This means that the emphasis will not be on direct applications. We are mainly interested in identifying underlying structure and mechanisms, and in understanding the coupling between the micro- and macro-scale. We work in a fully deterministic setting. For more information in the direction of stochastic modelling approaches (such as stochastically interacting particles), the reader is referred e.g. to \cite{Schadschneider2011}.\\
\\
We first describe the theoretical fundaments on which we build our model. In Section \ref{section basics meas theory} basic measure-theoretical concepts are introduced. The most important parts are the (refined) Lebesgue decomposition and the Radon-Nikodym Theorem.\\
Section \ref{section mixture theory and thermodynamics} is dedicated to mixture theory, a branch of continuum mechanics that turns out to be suitable for describing a crowd consisting of multiple subpopulations. The corresponding theory is presented in measure-theoretical language. It provides us with conservation of mass equations, that are the basis for our model. In Section \ref{section mixture theory and thermodynamics} we also cover aspects of thermodynamics that are used later to derive an entropy inequality for a particular crowds setting.\\
\\
In Section \ref{section Application crowd dynamics}, we apply the modelling ideas of the preceding sections to obtain a model for the dynamical behaviour of a crowd. After a number of (mainly formal) calculations we first present a continuous-in-time model and the accompanying entropy inequality. We also propose a velocity field, governing the dynamics. From this model, we derive a discrete-in-time version. All our mathematical analysis concerns this time-discrete model. The main result is a proof of global existence and uniqueness of a time-discrete solution. We also derive a discrete-in-time equivalent of the entropy inequality.\\
\\
Section \ref{section numerical scheme} contains the full description of the numerical scheme and parameter setting, which we use to simulate our two-scale crowd setting. The aim of the simulation part is to investigate the basic patterns produced by the interplay between one or two individuals and a macroscopic crowd. Section \ref{section numerical illustration} contains our simulation results.\\
\\
In Section \ref{section further work} we review the work done and the obtained results. Moreover, we give suggestions for future work and pose a few basic questions that are still open.\\
\\
Our first attempt \cite{EversMuntean} of modelling pedestrians using the approach described in this thesis, was published in \textit{Nonlinear Phenomena in Complex Systems}. Its content can be found in Appendix \ref{appendix paper}.

\newpage
\chapter{Measure theory}\label{section basics meas theory}

\section{Measures and their Lebesgue decomposition}\label{section Lebesgue decomposition}
Let $\bigl(\Omega,\mathcal{B}(\Omega)\bigr)$ be a measurable space, where $\emptyset\neq\Omega\subset\mathbb{R}^d$. Here, $\mathcal{B}(\Omega)$ denotes the $\sigma$-algebra of Borel subsets of $\Omega$. Suppose that $\mu$ and $\lambda$ are positive, finite measures defined on $\mathcal{B}(\Omega)$.

\begin{definition}[Absolutely continuous measures]\label{def abs cont measure}
The measure $\mu$ is said to be \textit{absolutely continuous} with respect to the measure $\lambda$ if for any $\Omega'\in \mathcal{B}(\Omega)$, $\lambda(\Omega')=0$ implies $\mu(\Omega')=0$. Notation:
\begin{equation*}
\mu\ll\lambda.
\end{equation*}
\end{definition}

\begin{definition}[Singular measures]\label{def singular measure}
The measures $\mu$ and $\lambda$ are said to be \textit{mutually singular} if there exists a $B\in\mathcal{B}(\Omega)$ such that
\begin{equation*}
\mu(\Omega\setminus B)= \lambda(B)=0.
\end{equation*}
Notation:
\begin{equation*}
\mu\perp\lambda.
\end{equation*}
Although the relation $\mu\perp\lambda$ is symmetric, it is also often said that $\mu$ is singular with respect to $\lambda$.
\end{definition}
Definition \ref{def singular measure} is taken from \cite{EvansGariepy}, p.~40. An alternative definition is given by \cite{Rudin}, p.~120, which we will now show to be equivalent.

\begin{lemma}\label{alternative def singular measures}
The following two statements are equivalent:
\begin{enumerate}[(i)]
  \item There exists a $B\in\mathcal{B}(\Omega)$ such that $\mu(\Omega\setminus B)= \lambda(B)=0$.\label{alt def sing meas EvansGariepy}
  \item There are disjoint $A_1, A_2\in\mathcal{B}(\Omega)$ such that $\mu(\Omega')=\mu(\Omega'\cap A_1)$ and $\lambda(\Omega')=\lambda(\Omega'\cap A_2)$ for all $\Omega'\in\mathcal{B}(\Omega)$.\label{alt def sing meas Rudin}
\end{enumerate}
\end{lemma}
\begin{proof}
\begin{enumerate}
  \item Assume that (\ref{alt def sing meas EvansGariepy}) holds. Define $A_1:=B$ and $A_2:=\Omega\setminus B$ (these sets are obviously disjoint). For all $\Omega'\in\mathcal{B}(\Omega)$ we have (using property (\ref{alt def sing meas EvansGariepy}))
    \begin{equation*}
    \mu(\Omega'\setminus A_1)\leqslant \mu(\Omega \setminus A_1)= \mu(\Omega\setminus B)=0,
    \end{equation*}
    and
    \begin{equation*}
    \lambda(\Omega'\setminus A_2)\leqslant \lambda(\Omega \setminus A_2)= \lambda\bigl(\Omega\setminus(\Omega\setminus B)\bigr)=\lambda(B)=0.
    \end{equation*}
    It follows that
    \begin{equation*}
    \mu(\Omega')=\mu(\Omega'\cap A_1)+\mu(\Omega'\setminus A_1)\leqslant \mu(\Omega' \cap A_1),
    \end{equation*}
    and
    \begin{equation*}
    \lambda(\Omega')= \lambda(\Omega' \cap A_2)+ \lambda(\Omega'\setminus A_2)\leqslant \lambda(\Omega'\cap A_2).
    \end{equation*}
    Since $\Omega'\cap A_1$ and $\Omega'\cap A_2$ are subsets of $\Omega'$, it is clear that
    \begin{equation*}
    \mu(\Omega')\geqslant \mu(\Omega' \cap A_1),
    \end{equation*}
    and
    \begin{equation*}
    \lambda(\Omega')\geqslant \lambda(\Omega'\cap A_2),
    \end{equation*}
    thus we have that
    \begin{equation*}
    \mu(\Omega')= \mu(\Omega' \cap A_1),
    \end{equation*}
    and
    \begin{equation*}
    \lambda(\Omega')= \lambda(\Omega'\cap A_2).
    \end{equation*}
  \item Assume that (\ref{alt def sing meas Rudin}) holds. Define $B:=A_1$, then
    \begin{equation*}
    \nonumber \mu(\Omega\setminus B)=\mu\bigl((\Omega\setminus B) \cap A_1\bigr)=\mu\bigl((\Omega\setminus A_1) \cap A_1\bigr)=\mu(\emptyset)=0,
    \end{equation*}
    and
    \begin{equation*}
    \nonumber \lambda(B)=\lambda(B\cap A_2)=\lambda(A_1\cap A_2)=\lambda(\emptyset)=0.
    \end{equation*}
\end{enumerate}
\end{proof}
In the following theorem we relate any positive, finite measure to absolutely continuous and singular measures.
\begin{theorem}[Lebesgue decomposition]\label{Thm Lebesgue decomp}
If $\mu$ and $\lambda$ are positive, finite measures defined on $\mathcal{B}(\Omega)$, then there exists a unique pair of positive, finite measures $\mu_{\text{ac}}$ and $\mu_{\text{s}}$ defined on $\mathcal{B}(\Omega)$, such that:
\begin{enumerate}[(i)]
  \item $\mu=\mu_{\text{ac}}+\mu_{\text{s}}$,
  \item $\mu_{\text{ac}}\ll\lambda$,
  \item $\mu_{\text{s}}\perp\lambda$.
\end{enumerate}
We call $\mu_{\text{ac}}$ the \textit{absolutely continuous} part and $\mu_{\text{s}}$ the \textit{singular} part of $\mu$ w.r.t. $\lambda$. The pair $(\mu_{\text{ac}}, \mu_{\text{s}})$ is called the \textit{Lebesgue decomposition} of $\mu$ w.r.t. $\lambda$.
\end{theorem}
\begin{proof}
The proof of Theorem \ref{Thm Lebesgue decomp} is given in Appendix \ref{Appendix Proof Lebesgue decomp}.
\end{proof}

\begin{definition}[Discrete measures]\label{def discrete measure}
The measure $\mu$ is said to be a \textit{discrete measure} with respect to the measure $\lambda$ if there exists a countable set $A:=\{x_1,x_2,\ldots\}\subset\Omega$ such that
\begin{equation*}
\mu(\Omega\setminus A)= \lambda(A)=0.
\end{equation*}
\end{definition}

\begin{lemma}\label{lemma characterization discrete measure}
Let $\mu$ be a positive, finite measure on $\mathcal{B}(\Omega)$. Then the following two statements are equivalent:
\begin{enumerate}[(i)]
  \item $\mu$ is a discrete measure with respect to the Lebesgue measure $\lambda^d$.\label{statement discrete wrt Lebesgue}
  \item There is a countable set $\{x_1, x_2,\ldots\}\subset \Omega$ and a set of corresponding nonnegative coefficients $\{\alpha_1,\alpha_2,\ldots\}\subset \mathbb{R}$, such that $\mu = \sum_{i=1}^{\infty}\alpha_i \delta_{x_{i}}$. Here $\delta_{x_{i}}$ is the Dirac measure centered at $x_i$.\label{statement sum of diracs}
\end{enumerate}
\end{lemma}
\begin{proof}
\begin{enumerate}
  \item Assume that (\ref{statement discrete wrt Lebesgue}) holds. Let $A:=\{y_1,y_2,\ldots\}\subset\Omega$ be the collection of points such that $\mu(\Omega\setminus A)= \lambda(A)=0$. Without loss of generality, assume that all elements of $A$ are distinct. (In case not all elements of $A$ are distinct, just delete $y_j$ from $A$ if there is a $y_i\in A$ satisfying $i<j$ and $y_i=y_j$.)\\
      Since $\mu(\Omega\setminus A)=0$, we have for any $\Omega'\in\mathcal{B}(\Omega)$:
      \begin{equation*}
      \mu(\Omega'\cap A) \leqslant \mu(\Omega') = \mu(\Omega'\cap A)+ \mu(\Omega'\setminus A) \leqslant \mu(\Omega'\cap A)+ \mu(\Omega\setminus A)= \mu(\Omega'\cap A),
      \end{equation*}
      thus $\mu(\Omega') = \mu(\Omega'\cap A)$. For fixed $\Omega'$, let $\mathcal{J}$ be the index set, such that $i\in\mathcal{J}$ implies $y_i\in \Omega'\cap A$. Since $\Omega'\cap A = \bigcup_{i\in \mathcal{J}}\{y_i\}$ is a disjoint union, it follows that
      \begin{equation*}
      \mu(\Omega')=\mu(\Omega'\cap A)=\sum_{i\in\mathcal{J}}\mu(y_i).
      \end{equation*}
      If we define $\alpha_i:=\mu(y_i)\geqslant0$, and write $\mathbf{1}$ for the indicator function, then the above is equivalent to
      \begin{equation*}
      \mu(\Omega')=\sum_{i=1}^{\infty}\alpha_i\mathbf{1}_{y_i\in\Omega'}.
      \end{equation*}
      This is exactly the definition of the Dirac measure, so we can also write
      \begin{equation*}
      \mu(\Omega')=\sum_{i=1}^{\infty}\alpha_i\delta_{y_i}.
      \end{equation*}
  \item Assume that (\ref{statement sum of diracs}) holds, and let the set $\{x_1, x_2,\ldots\}$ be called $B$. It follows by definition of the Dirac measure that $\mu(\Omega\setminus B)=0$. $B$ is a countable collection of points, which (obviously) all have Lebesgue measure zero. Thus $\lambda^d(B)=\sum_{i=1}^{\infty}\lambda^d(x_i)=0$. We conclude that $\mu$ is discrete with respect to $\lambda^d$.
\end{enumerate}
\end{proof}

\begin{definition}[Singular continuous measures]\label{def singular cont measure}
The measure $\mu$ is said to be \textit{singular continuous} with respect to the measure $\lambda$ if $\mu(x)=0$ for all $x\in\Omega$, and there is a $B\in\mathcal{B}(\Omega)$ such that
\begin{equation*}
\mu(\Omega\setminus B)= \lambda(B)=0.
\end{equation*}
\end{definition}

\begin{remark}
We only consider measures defined on $\mathcal{B}(\Omega)$, i.e. we can only apply them to sets. If we write $\mu(x)$, we therefore actually mean $\mu(\{x\})$.
\end{remark}

\begin{theorem}[Decomposition of singular measures]\label{Thm decomp sing}
If $\mu_{\text{s}}$ is a positive, finite measure defined on $\mathcal{B}(\Omega)$, which is singular w.r.t. $\lambda$, then there exists a unique pair of positive, finite measures $\mu_{\text{d}}$ and $\mu_{\text{sc}}$ defined on $\mathcal{B}(\Omega)$, such that:
\begin{enumerate}[(i)]
  \item $\mu_{\text{s}}=\mu_{\text{d}}+\mu_{\text{sc}}$,
  \item $\mu_{\text{d}}$ is a discrete measure w.r.t. $\lambda$,
  \item $\mu_{\text{sc}}$ is singular continuous w.r.t. $\lambda$.
\end{enumerate}
We call $\mu_{\text{d}}$ the \textit{discrete} part and $\mu_{\text{sc}}$ the \textit{singular continuous} part of $\mu_{\text{s}}$ w.r.t. $\lambda$. The singular continuous part is also called the \textit{Cantor part} of the measure.
\end{theorem}
\begin{proof}
The proof of Theorem \ref{Thm decomp sing} is given in Appendix \ref{Appendix Proof decomp sing}.
\end{proof}

\begin{corollary}[Refined Lebesgue decomposition]\label{Corollary refined Lebesgue decomposition}
Assume the hypotheses of Theorems \ref{Thm Lebesgue decomp} and \ref{Thm decomp sing}. Then for any positive, finite measure $\mu$, there are unique measures $\mu_{\text{ac}}$, $\mu_{\text{d}}$ and $\mu_{\text{sc}}$ such that
\begin{equation*}
\mu=\mu_{\text{ac}}+\mu_{\text{d}}+\mu_{\text{sc}},
\end{equation*}
where $\mu_{\text{ac}}\ll\mu$, $\mu_{\text{d}}$ is a discrete measure, and $\mu_{\text{sc}}$ is singular continuous w.r.t. $\lambda$.
\end{corollary}

\begin{remark}
The statement $\mu(\Omega\setminus B)=\lambda(B)=0$ (cf. Definitions \ref{def singular measure}, \ref{def discrete measure} and \ref{def singular cont measure}), is equivalent to: $\lambda(B)=0$ and $\mu(B)=\mu(\Omega)$. This is due to the identity $\mu(\Omega)=\mu(\Omega\cap B)+\mu(\Omega\setminus B)=\mu(B)+\mu(\Omega\setminus B)$. If $B$ is such that $\mu(B)=\mu(\Omega)$, we call $B$ a set of \textit{full measure}.
\end{remark}

\section{Radon-Nikodym Theorem}
The Radon-Nikodym Theorem is of vital importance in this context, especially for the applicability of measure theory to mixture theory as arising in real-life situations mimicking the dynamics of crowds. The following theorem gives sufficient conditions for a measure to be expressed in terms of a density, with respect to another measure:

\begin{theorem}[Radon-Nikodym for finite measures]\label{Radon-Nikodym} Suppose $\mu$ and $\lambda$ are positive measures on a measurable space $\big(\Omega,\mathcal{B}(\Omega)\big)$ such that $0<\mu(\Omega)<\infty$, $0<\lambda(\Omega)<\infty$, and let $\mu$ be absolutely continuous with respect to $\lambda$. Then there exists a real, nonnegative, $\mathcal{B}(\Omega)$-measurable function $h$ on $\Omega$ such that
\begin{equation*}
\mu(\Omega')=\int_{\Omega'} h d\lambda, \hspace{1 cm} \text{for all }\Omega'\in \mathcal{B}(\Omega).
\end{equation*}
\end{theorem}
\begin{proof}
A comprehensible, well-structured proof of Theorem \ref{Radon-Nikodym} can be found in \cite{Bradley}.
\end{proof}

\begin{remark}
The density $h$ is often called \textit{Radon-Nikodym derivative} and is denoted by
\begin{equation}
h:=\frac{d\mu}{d\lambda}.
\end{equation}
\end{remark}

\begin{lemma}
Assume the hypothesis of Theorem \ref{Radon-Nikodym}. Then the Radon-Nikodym derivative is unique in the following sense: if both $h_1$ and $h_2$ satisfy $\mu(\Omega')=\int_{\Omega'} h_k d\lambda$, for all $\Omega'\in \mathcal{B}(\Omega)$, where $k\in\{1,2\}$, then $h_1=h_2$ almost everywhere (w.r.t. $\lambda$) in $\Omega$.
\end{lemma}
\begin{proof}
This statement can easily be verified. For any $\Omega'\in \mathcal{B}(\Omega)$ we have that
\begin{equation*}
\int_{\Omega'} (h_1-h_2) d\lambda=\int_{\Omega'} h_1 d\lambda-\int_{\Omega'} h_2 d\lambda=\mu(\Omega')-\mu(\Omega')=0.
\end{equation*}
Thus, $h_1-h_2=0$ almost everywhere w.r.t. $\lambda$ in $\Omega$, i.e. $h_1=h_2$ $\lambda$-a.e. in $\Omega$.
\end{proof}
\begin{remark}
Theorem \ref{Radon-Nikodym} also holds for more general situations, for instance, including the case of signed $\sigma$-finite\footnote{The measure $\lambda$ is called $\sigma$-finite if there exists a sequence of sets $\{E_k\}_{k=1}^{\infty}\subset\mathcal{B}(\Omega)$ with $\Omega=\bigcup_{k=1}^{\infty}E_k$ and $|\lambda(E_k)|<\infty$ (for all $k\in\mathbb{N}$); see \protect\cite{Elstrodt}, p.~269.} measures (see \cite{Elstrodt}, e.g.).
\end{remark}
\section{Properties of Radon-Nikodym derivatives}
In the following lemma we state a number of useful properties of Radon-Nikodym derivatives:
\begin{lemma}[Basic properties]\label{Properties RN derivatives}
Assume the hypothesis of Theorem \ref{Radon-Nikodym} on the measures $\nu$, $\mu$, $\lambda$, $\nu_1$, $\nu_2$, $\mu_1$ and $\mu_2$. The Radon-Nikodym derivatives satisfy the following general properties:
\begin{enumerate}
  \item If $\nu\ll\mu\ll\lambda$, then $\dfrac{d\nu}{d\lambda}=\dfrac{d\nu}{d\mu}\dfrac{d\mu}{d\lambda}$.\label{Property three measures}\\
  \item If $\mu\ll\lambda$ and $g\in L^1_{\mu}(\Omega)$, then $\int_{\Omega'}g d\mu = \int_{\Omega'}g \dfrac{d \mu}{d\lambda}d\lambda$ for all $\Omega'\in\mathcal{B}(\Omega)$.\label{Property integration of g}\\
  \item If $\mu\ll\lambda$ and $\lambda\ll\mu$, then $\dfrac{d\mu}{d\lambda}=\Bigl(\dfrac{d\lambda}{d\mu}\Bigr)^{-1}$.\label{Property inverse}\\
  \item If $(\Omega_i,\mathcal{B}(\Omega_i), \mu_i)$, $i\in \{1,2\}$, are two measure spaces with $\nu_i\ll\mu_i$, $i\in \{1,2\}$, then $\nu_1\otimes\nu_2\ll\mu_1\otimes\mu_2$ and $\dfrac{d(\nu_1\otimes\nu_2)}{d(\mu_1\otimes\mu_2)}=\dfrac{d\nu_1}{d\mu_1}\dfrac{d\nu_2}{d\mu_2}$.\label{Property Cartesian product}
\end{enumerate}
\end{lemma}
\begin{proof}The proof of Lemma \ref{Properties RN derivatives} is given in Appendix \ref{Appendix Proof properpties RN deriv}.
\end{proof}

\section{Mass measure $\mu_m$}\label{section Mass measure}
Here, we already give an indication of the particular measures we intend to use in the rest of this thesis. Let $\Omega\subset\mathbb{R}^d$ be a domain (read: object, body) with mass. For physically relevant situations, we consider $d\in\{1,2,3\}$. Let $\mu_m(\Omega')$ be defined as the \textit{mass} contained in $\Omega'\subset\Omega$. Note that we use the concept of a mass measure very much in the spirit of \cite{Boehm}.

\begin{remark}
As a rule, whenever we write $\Omega'\subset\Omega$, we actually mean that $\Omega'$ is such that $\Omega'\in\mathcal{B}(\Omega)$. We assume $\mu_m$ to be defined on all elements of $\mathcal{B}(\Omega)$.
\end{remark}
In Sections \ref{section micro mass measure} and \ref{section macro mass measure} we consider two specific interpretations of this mass measure depending on the localization of the information we wish to capture.
\subsection{Microscopic mass measure $m_m$}\label{section micro mass measure}
Suppose that $\Omega$ contains a collection of $N$ point masses (each of them of mass 1), and denote their positions by $\{p_k\}_{k=1}^N\subset\Omega$, for $N\in\mathbb{N}$. We want the mass measure $m_m$ to be a counting measure with respect to these point masses, i.e. for all $\Omega'\in\mathcal{B}(\Omega)$
\begin{equation}\label{m counting measure}
m_m(\Omega')=\#\{p_k\in \Omega'\}.
\end{equation}
By (\ref{m counting measure}) we mean that $m_m$ counts the number of individuals located in $\Omega'$ (with corresponding, known positions $p_k\in\Omega'$). This can be achieved by representing $m_m$ as the sum of Dirac masses, with their singularities located at the points $p_k$, $k\in\{1,2,\ldots,N\}$:
\begin{equation}\label{m sum of dirac}
m_m=\sum_{k=1}^{N} \delta_{p_k}.
\end{equation}
If $m_m$ is defined as in (\ref{m counting measure}) or (\ref{m sum of dirac}), we call it a \textit{microscopic} mass measure. Note that such measure satisfies the conditions in Definition \ref{def discrete measure} and is thus a discrete measure. The characteristics of such measure are illustrated by Figure \ref{figure micro mass measure}.

\begin{figure}[ht]
\centering
\begin{tikzpicture}
    \draw[thick] (0,0) rectangle (8,6);
    \node at (4,6.4){\Large{$\Omega$}};
    \draw[thick] (3.5,3) ellipse (40pt and 25pt);
    \node at (3.5,4.2){\Large{$\Omega'$}};
    \filldraw [gray] (6.5,1) circle (2pt) node[black,below] {\small{$p_1$}};
    \filldraw [gray] (3,3) circle (2pt) node[black,below] {\small{$p_2$}};
    \filldraw [gray] (6,4) circle (2pt) node[black,below] {\small{$p_3$}};
    \filldraw [gray] (2,5) circle (2pt) node[black,below] {\small{$p_4$}};
    \filldraw [gray] (1,2.5) circle (2pt) node[black,below] {\small{$p_5$}};
    \filldraw [gray] (4,1.5) circle (2pt) node[black,below] {\small{$p_6$}};
    \filldraw [gray] (4,5.5) circle (2pt) node[black,below] {\small{$p_7$}};
    \filldraw [gray] (4,2.8) circle (2pt) node[black,below] {\small{$p_8$}};
\end{tikzpicture}
\caption{A microscopic mass measure counts the number of points $p_k$ that are contained in $\Omega'$, a subset of $\Omega$. In this case $m_m(\Omega')=2$ since $p_2$ and $p_8$ lie in $\Omega'$.}\label{figure micro mass measure}
\end{figure}
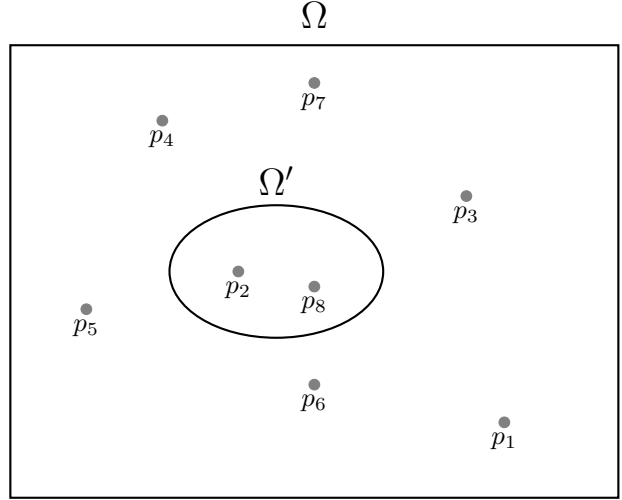

\subsection{Macroscopic mass measure $M_m$}\label{section macro mass measure}
Let us now regard a different mass measure $M_m$. Assume that the following postulate applies to $M_m$:
\begin{postulate}[Properties of $M_m$]\label{Postulate M_m}
\begin{enumerate}
  \item $M_m\geqslant 0$.\label{Postulate M_m 1}\\
  \item $M_m$ is $\sigma$-additive.\label{Postulate M_m 2}\\
  \item $M_m$ is finite.\label{Postulate M_m 3}\\
  \item $M_m\ll\lambda^d$ (with $\lambda^d$ the Lebesgue measure in $\mathbb{R}^d$).\label{Postulate M_m 4}
\end{enumerate}
\end{postulate}
By Parts \ref{Postulate M_m 1}, \ref{Postulate M_m 2} and \ref{Postulate M_m 3} of Postulate \ref{Postulate M_m}, we have that $M_m$ is a positive, finite measure on $\Omega$, whereas Part \ref{Postulate M_m 4} implies that there is no mass present in a set that has no volume (w.r.t. $\lambda^d$). We refer to a mass measure satisfying Postulate \ref{Postulate M_m} as a \textit{macroscopic} mass measure.
Now Theorem \ref{Radon-Nikodym} guarantees the existence of a real, nonnegative density $\rho\in L_{\lambda^d}^1(\Omega)$ such that
\begin{equation*}
M_m(\Omega')=\int_{\Omega'} \rho(x)d\lambda^d(x), \hspace{1 cm} \text{for all } \Omega'\in\mathcal{B}(\Omega).
\end{equation*}
An illustration of such measure is given in Figure \ref{figure macro mass measure}.
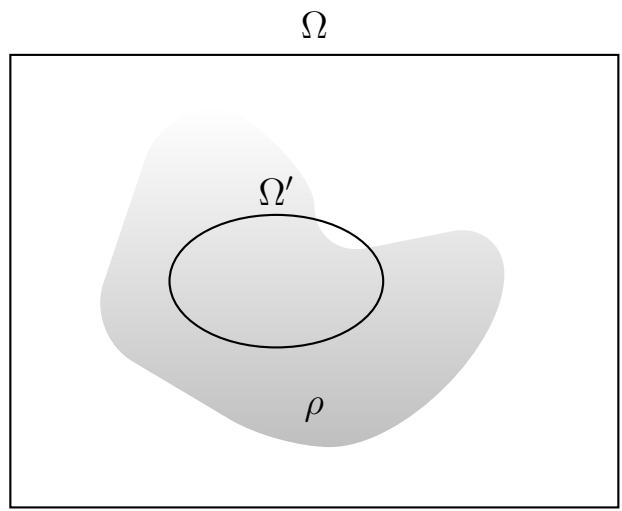
\begin{figure}[ht]
\centering
\begin{tikzpicture}
    \shade[top color=white, bottom color=gray!50] (1,2.3) [rounded corners=20pt] -- (2,5.3) arc(90:0:2cm)-- (6.5,3.8) arc(0:-90:3cm) -- cycle;
    \node at (4,1.3){\Large{$\rho$}};
    \draw[thick] (0,0) rectangle (8,6);
    \node at (4,6.4){\Large{$\Omega$}};
    \draw[thick] (3.5,3) ellipse (40pt and 25pt);
    \node at (3.5,4.2){\Large{$\Omega'$}};
\end{tikzpicture}
\caption{A macroscopic mass measure is used to obtain the mass in $\Omega'\subset\Omega$. $M_m(\Omega')$ depends on the density $\rho$ of the area in $\Omega'$. The density is here indicated by grey shading.}\label{figure macro mass measure}
\end{figure}

\newpage
\chapter{Mixture theory and thermodynamics}\label{section mixture theory and thermodynamics}
In this section, we present some ideas from \cite{Bowen, Muller, FonsvdVen} regarding the theory of mixtures. We cast their description in a measure-theoretical framework. Unlike \cite{Bowen, Muller}, here we mainly consider the Eulerian point of view. In Sections \ref{section description mixture}--\ref{section mixture entropy} we describe the mixture as a continuum; this is the classical and most common way of doing so (cf. \cite{Bowen, Muller} e.g.). In Section \ref{section generalization mixture}, we extend the ideas of the preceding sections to more general mass measures.

\section{Description of a mixture}\label{section description mixture}
At this stage, a mixture is defined as a continuum consisting of a certain number of constituents (also called: components). Constitutive relations typically differ from component to component.\\
\\
Let $\Omega\subset\mathbb{R}^d$ be the domain in which the mixture is located, and consider the measurable space $\bigl(\Omega,\mathcal{B}(\Omega)\bigr)$, where $\mathcal{B}(\Omega)$ is the $\sigma$-algebra of Borel subsets of $\Omega$. Suppose the mixture consists of $\nu$ constituents, with index $\alpha\in\{1,2,\ldots,\nu\}$. For each constituent $\alpha$ we define its volume and mass present in $\Omega'\in\mathcal{B}(\Omega)$ at time $t\in(0,T)$ by $\tau^{\alpha}(t, \Omega')$ and $\mu^{\alpha}(t, \Omega')$, respectively. Note that in fact $t$ is the time variable and can be understood as a parameter. The objects $\tau^{\alpha}(t, \cdot)$ and $\mu^{\alpha}(t, \cdot)$ on $\mathcal{B}(\Omega)$ are for each $t\in(0,T)$ fixed measures.\\
We assume that the following postulate applies:
\begin{postulate}[Properties of $\tau^{\alpha}$ and $\mu^{\alpha}$]\label{Postulate tau^alpha mu^alpha}
For all $\alpha\in\{1,2,\ldots,\nu\}$ and all $t\in(0,T)$ we postulate:
\begin{enumerate}
  \item $\tau^{\alpha}\geqslant 0$.\label{Postulate tau^alpha mu^alpha 1}\\
  \item $\tau^{\alpha}$ is $\sigma$-additive.\\
  \item $\tau^{\alpha}$ is finite.\\
  \item $\tau^{\alpha}\ll\lambda^d$ (with $\lambda^d$ the Lebesgue measure in $\mathbb{R}^d$).\label{Postulate tau^alpha mu^alpha 2}\\
  \item $\mu^{\alpha}\geqslant 0$.\label{Postulate tau^alpha mu^alpha 3}\\
  \item $\mu^{\alpha}$ is $\sigma$-additive.\\
  \item $\mu^{\alpha}$ is finite.\\
  \item $\mu^{\alpha}\ll\tau^{\alpha}$.\label{Postulate tau^alpha mu^alpha 4}
\end{enumerate}
\end{postulate}
Note that the concepts of Section \ref{section macro mass measure} are incorporated in this postulate.
\begin{remark}
It follows in a straightforward way from Parts \ref{Postulate tau^alpha mu^alpha 2} and \ref{Postulate tau^alpha mu^alpha 4} of Postulate \ref{Postulate tau^alpha mu^alpha} that $\mu^{\alpha}\ll\lambda^d$ for all $\alpha\in\{1,2,\ldots,\nu\}$ and all $t\in(0,T)$. The absolute continuity of both $\tau^{\alpha}$ and $\mu^{\alpha}$ complies with the fact that we define mixtures at the continuum level.
\end{remark}

\begin{remark}
We characterize the presence of constituent $\alpha$ at a certain $x$ by the presence of mass \textit{and} volume here. Mathematically, this means that if at time $t\in(0,T)$ a subset $\Omega'\in\mathcal{B}(\Omega)$ contains some fraction of constituent $\alpha$, then both $\mu^{\alpha}(t,\Omega')>0$ and $\tau^{\alpha}(t,\Omega')>0$ must hold. On the other hand, if a constituent $\alpha$ is not present in $\Omega'$, then $\mu^{\alpha}(t,\Omega')$ and $\tau^{\alpha}(t,\Omega')$ are both zero. Following this intuitive characterization, we assume $\tau^{\alpha}\ll\mu^{\alpha}$ for all $\alpha\in\{1,2,\ldots,\nu\}$ and for all $t\in(0,T)$.
\end{remark}
Let $\tau$ and $\mu$ be the (time-dependent) total volume and mass measures, given by
\begin{eqnarray}
\tau(t, \Omega')&:=&\sum_{\alpha=1}^{\nu}\tau^{\alpha}(t, \Omega'),\label{def total volume mixture}\\
\mu(t, \Omega')&:=&\sum_{\alpha=1}^{\nu}\mu^{\alpha}(t, \Omega'),\label{def total mass mixture}
\end{eqnarray}
for all $\Omega'\in\mathcal{B}(\Omega)$ and all $t\in(0,T)$.\\
\\
By the definitions of $\tau$ and $\mu$ given in (\ref{def total volume mixture}) and (\ref{def total mass mixture}) and as a consequence of Postulate \ref{Postulate tau^alpha mu^alpha}, we have $\mu\ll\tau\ll\lambda^d$. Note that since we also assumed $\tau^{\alpha}\ll\mu^{\alpha}$ for each $\alpha$, it follows by definition that also $\tau\ll\mu$ holds.\\
\\
For any $t\in(0,T)$, the Radon-Nikodym Theorem provides the existence of the unique nonnegative densities (Radon-Nikodym derivatives) $\Theta(t,\cdot)$, $\Theta^{\alpha}(t,\cdot)$, $\theta^{\alpha}(t,\cdot)\in L_{\lambda^d}^1(\Omega)$ (for each $\alpha\in\{1,2,\ldots,\nu\}$):
\begin{eqnarray}
\Theta&:=&\dfrac{d\tau}{d\lambda^d},\label{Theta}\\
\Theta^{\alpha}&:=&\dfrac{d\tau^{\alpha}}{d\lambda^d},\\
\theta^{\alpha}&:=&\dfrac{d\tau^{\alpha}}{d\tau}.\label{theta alpha}
\end{eqnarray}
The Radon-Nikodym derivative $\theta^{\alpha}$ arising in (\ref{theta alpha}) is called the \textit{volume fraction} of component $\alpha$ at time $t$. As a result of (\ref{def total volume mixture}):
\begin{equation*}
\sum_{\alpha=1}^{\nu}\theta^{\alpha}=1,\hspace{1 cm}\text{a.e. w.r.t. }\tau,\text{ and for all }t\in(0,T).
\end{equation*}
Note also that due to Part \ref{Property three measures} of Lemma \ref{Properties RN derivatives}, we have that for all $\alpha\in\{1,2,\ldots,\nu\}$, and for all $t\in(0,T)$
\begin{equation*}
\Theta^{\alpha}=\theta^{\alpha}\Theta,\hspace{1 cm}\text{a.e. w.r.t. }\lambda^d.
\end{equation*}
Also for the mass measures unique Radon-Nikodym derivatives exist. For all $\alpha\in\{1,2,\ldots,\nu\}$ and for all $t\in(0,T)$, we define
\begin{eqnarray}
\check{\rho} &:=& \dfrac{d\mu}{d\tau},\label{rho check}\\
\check{\rho}^{\alpha} &:=& \dfrac{d\mu^{\alpha}}{d\tau},\label{rho check alpha}\\
\hat{\rho}^{\alpha} &:=& \dfrac{d\mu^{\alpha}}{d\tau^{\alpha}}.\label{rho hat alpha}
\end{eqnarray}
We call $\check{\rho}$ the density of the mixture, $\check{\rho}^{\alpha}$ the \textit{partial} density of component $\alpha$, and $\hat{\rho}^{\alpha}$ the \textit{intrinsic} density of that component. By (\ref{def total mass mixture}) we have that
\begin{equation}\label{density is sum partial densities w.r.t. tau}
\check{\rho}=\sum_{\alpha=1}^{\nu}\check{\rho}^{\alpha},\hspace{1 cm}\text{a.e. w.r.t. }\tau,\text{ and for all }t\in(0,T).
\end{equation}
From (\ref{Theta})-(\ref{theta alpha}) and (\ref{rho check})-(\ref{rho hat alpha}) a number of identities can be derived (again via Lemma \ref{Properties RN derivatives}, Part \ref{Property three measures}). For all $t\in(0,T)$ we have
\begin{eqnarray}
\check{\rho}^{\alpha} &=& \hat{\rho}^{\alpha}\theta^{\alpha},\hspace{1 cm}\text{a.e. w.r.t. }\tau,\text{ for all } \alpha\in\{1,2,\ldots,\nu\},\\
\check{\rho}^{\alpha}\Theta &=& \hat{\rho}^{\alpha}\Theta^{\alpha},\hspace{1 cm}\text{a.e. w.r.t. }\lambda^d,\text{ for all } \alpha\in\{1,2,\ldots,\nu\}.
\end{eqnarray}
\begin{remark}
Note that the latter $\check{\rho}^{\alpha}\Theta$ equals the Radon-Nikodym derivative $d\mu^{\alpha}/d\lambda^d$. In practice it will be much more convenient to work with $\lambda^d$ than to work with $\tau$. Therefore we define
\begin{equation}\label{rho}
\rho^{\alpha}:=\dfrac{d\mu^{\alpha}}{d\lambda^d}=\check{\rho}^{\alpha}\Theta,
\end{equation}
and henceforth refer to $\rho^{\alpha}$ (rather than to $\check{\rho}^{\alpha}$) as the partial density.\\
Furthermore $\check{\rho}\Theta$ is the Radon-Nikodym derivative of the total mass measure $\mu$ with respect to $\lambda^d$. Analogously to $\rho^{\alpha}$, we define
\begin{equation*}
\rho:=\dfrac{d\mu}{d\lambda^d}=\check{\rho}\Theta,
\end{equation*}
and refer to this $\rho$ as the density, from now on.
\end{remark}

Now, (\ref{density is sum partial densities w.r.t. tau}) implies
\begin{equation}\label{density is sum partial densities w.r.t. lambda}
\rho=\sum_{\alpha=1}^{\nu}\rho^{\alpha},\hspace{1 cm}\text{a.e. w.r.t. }\lambda^d,\text{ and for all }t\in(0,T).
\end{equation}
The \textit{mass concentration} of constituent $\alpha$ at time $t$ is defined as
\begin{equation}\label{mass concentration densities}
c^{\alpha}:=\dfrac{\rho^{\alpha}}{\rho},
\end{equation}
and it follows naturally from (\ref{density is sum partial densities w.r.t. lambda}) and (\ref{mass concentration densities}) that
\begin{equation*}
\sum_{\alpha=1}^{\nu}c^{\alpha}=1,\hspace{1 cm}\text{a.e. w.r.t. }\mu,\text{ and for all }t\in(0,T).
\end{equation*}
\begin{remark}\label{remark where c^alpha defined}
The mass concentration $c^{\alpha}$ is only defined if $\rho>0$. For any $t\in(0,T)$, the subset
\begin{equation*}
\Omega_t^0:=\{x\in\Omega \,\big|\, \rho(t,x)=0\}\subset\Omega
\end{equation*}
is a null set w.r.t. the measure $\mu$. This can easily be seen:
\begin{equation*}
\mu(t, \Omega_t^0)=\int_{\Omega_t^0}\rho(t,x)d\lambda^d(x)=\int_{\Omega_t^0}0\,d\lambda^d(x)=0.
\end{equation*}
This implies that $c^{\alpha}$ is defined almost everywhere w.r.t. $\mu$. We will see later, that defining $c^{\alpha}$ actually only is useful in regions where $\rho>0$ is satisfied (i.e. where mass is present), and that there is thus no serious problem here.
\end{remark}

\begin{remark}
It is obvious that $\dfrac{\rho^{\alpha}}{\rho}=\dfrac{\check{\rho}^{\alpha}\Theta}{\check{\rho}\Theta}=\dfrac{\check{\rho}^{\alpha}}{\check{\rho}}$ holds $\mu$-almost everywhere. We thus could have defined $c^{\alpha}$ as $\check{\rho}^{\alpha}/\check{\rho}$, and would have obtained the same properties almost everywhere w.r.t. $\mu$.
\end{remark}

\section{Kinematics}\label{section kinematics mixture}
The kinematics of component $\alpha$ of the mixture is dictated by a motion mapping $\chi^{\alpha}$, such that
\begin{equation*}
x=\chi^{\alpha}(t,X^{\alpha}),
\end{equation*}
denotes the position at time $t$ of an $\alpha$-particle, which was initially (i.e. at $t=0$) situated in $X^{\alpha}$.\\
By assumption $\chi^{\alpha}$ is smooth, and the inverse mapping $\Xi^{\alpha}$ exists, such that
\begin{equation*}
X^{\alpha}=\Xi^{\alpha}(t,x).
\end{equation*}
Also by assumption, $\Xi^{\alpha}$ is smooth. Such smooth motion mapping $\chi^{\alpha}$, with smooth inverse, is called a \textit{diffeomorphism}.\\
\\
For each component $\alpha$ its velocity at position $x\in\Omega$ at time $t\in(0,T)$ is denoted by $v^{\alpha}(t,x)$ and can be derived from the motion mapping $\chi^{\alpha}$:
\begin{equation}\label{velocity component alpha}
v^{\alpha}(t,x):=\dfrac{D^{(\alpha)}x}{Dt}:=\dfrac{\partial}{\partial t}\chi^{\alpha}(t,X^{\alpha}).
\end{equation}
\begin{remark}
The right-hand side in (\ref{velocity component alpha}) is indeed a function of $(t,x)$, since we have to read $X^{\alpha}$ as $X^{\alpha}=\Xi^{\alpha}(t,x)$. The derivative $\partial/\partial t$ is taken however only with respect to the first variable of $\chi^{\alpha}$.
\end{remark}

\begin{definition}[Time derivatives]\label{def time derivatives}
Let $g$ denote some (physical) scalar quantity associated with component $\alpha$ and depending on space and time. We write $g=\tilde{g}(t,X^{\alpha})$ if we consider the quantity from a Langrangian point of view, and $g=\bar{g}(t,x)$ from a Eulerian point of view. The \textit{time derivative} of $g$ is now defined as
\begin{equation*}
\dfrac{D^{(\alpha)}g}{Dt}:=\dfrac{\partial\tilde{g}(t,X^{\alpha})}{\partial t}=\dfrac{\partial \bar{g}(t,x)}{\partial t}+ \nabla \bar{g}(t,x)\cdot v^{\alpha}(t,x).
\end{equation*}
\end{definition}
In the above definition, note that $\tilde{g}(t,X^{\alpha})=\bar{g}\bigl(t,\chi^{\alpha}(t,X^{\alpha})\bigr)$.\\
\\
Now, we define the \textit{barycentric} velocity (i.e. the velocity of the centre of mass of a volume element) of the mixture by
\begin{equation}\label{def barycentric velocity}
v:=\dfrac{1}{\rho}\sum_{\alpha=1}^{\nu}\rho^{\alpha}v^{\alpha}=\sum_{\alpha=1}^{\nu}c^{\alpha}v^{\alpha}.
\end{equation}
The barycentric velocity is also called \textit{mean velocity}, or just \textit{velocity of the mixture}.
\begin{remark}
A similar expression as in Definition \ref{def time derivatives} holds for time derivatives of scalar functions $g$ associated to the total mixture:
\begin{equation*}
\dfrac{Dg}{Dt}:=\dfrac{\partial\tilde{g}(t,X)}{\partial t}=\dfrac{\partial \bar{g}(t,x)}{\partial t}+ \nabla \bar{g}(t,x)\cdot v(t,x).
\end{equation*}
\end{remark}

\section{Balance of mass}\label{section mass balance mixture}
We assume that there is no mass exchange between the components of the mixture. The conservation of mass implies that for any $\Omega'\in\mathcal{B}(\Omega)$, for all $t\in(0,T)$ and for all $\alpha\in\{1,2,\ldots,\nu\}$
\begin{equation}\label{mass conservation}
\dfrac{d}{dt}\mu^{\alpha}\bigl(t, \chi^{\alpha}(t,\Omega')\bigr)=0.
\end{equation}
Note that $\chi^{\alpha}(t,\Omega')$ is the configuration at time $t$ of all particles of the component $\alpha$ that were initially located in $\Omega'$. For $\mu^{\alpha}\bigl(t, \chi^{\alpha}(t,\Omega')\bigr)$ in (\ref{mass conservation}) to be well-defined, we need that $\chi^{\alpha}(t,\Omega')$ is an element of $\mathcal{B}(\Omega)$ for each $t\in(0,T)$ and for all $\Omega'\in\mathcal{B}(\Omega)$.
\begin{lemma}\label{Diffeomorphism maps Borel to Borel}
If for each $t\in(0,T)$, the mapping $\chi^{\alpha}(t,\cdot):\Omega\rightarrow\Omega$ is a diffeomorphism, then $\chi^{\alpha}(t,\Omega')\in\mathcal{B}(\Omega)$, for each $t\in(0,T)$ and for all $\Omega'\in\mathcal{B}(\Omega)$.
\end{lemma}
\begin{proof}
Given that for any $t\in(0,T)$ fixed, $\chi^{\alpha}(t,\cdot)$ is a diffeomorphism from $\Omega$ to $\Omega$, we have in particular that its inverse $\Xi^{\alpha}(t,\cdot)$ is a continuous mapping from $\Omega$ to $\Omega$. By definition of continuous functions (cf. \cite{Rudin}, p.~8), $(\Xi^{\alpha})^{-1}(t,\Omega')$ is an open set, for any $\Omega'\subset\Omega$ open.\\
The Borel $\sigma$-algebra is defined as the smallest $\sigma$-algebra containing all open subsets of $\Omega$. It follows that we have that $(\Xi^{\alpha})^{-1}(t,\Omega')\in\mathcal{B}(\Omega)$ (i.e. it is a measurable set). By definition of measurable functions (cf. \cite{Rudin}, p.~8), we now have that $\Xi^{\alpha}(t,\cdot)$ is measurable.\\
Theorem 1.12 (b) (\cite{Rudin}, p.~13) provides that $(\Xi^{\alpha})^{-1}(t,\Omega')\in\mathcal{B}(\Omega)$ not only for all $\Omega'\subset\Omega$ open, but also for all $\Omega'\in\mathcal{B}(\Omega)$. That is, $\Xi^{\alpha}$ is a Borel function (cf. \cite{Rudin}, p.~12).\\
$\Xi^{\alpha}(t,\cdot)$ is the inverse of $\chi^{\alpha}(t,\cdot)$, so the above statement also states that for all $t\in(0,T)$ fixed $\chi^{\alpha}(t,\Omega')\in\mathcal{B}(\Omega)$ for all $\Omega'\in\mathcal{B}(\Omega)$.
\end{proof}

We now present a few calculations involving Reynolds' transport theorem, that are based on \cite{TemamMiranville} (pp.~18--21). We use the notation $\det(\nabla \chi^{\alpha})$ for the determinant of the Jacobian matrix of the motion mapping. The following identity is derived e.g. in \cite{TemamMiranville} (p.~20):
\begin{equation*}
\dfrac{d}{dt}\det(\nabla \chi^{\alpha}) = \nabla\cdot v^{\alpha}\det(\nabla \chi^{\alpha}).
\end{equation*}
From the conservation of mass statement, formulated in (\ref{mass conservation}), we derive
\begin{eqnarray}
\nonumber 0=\dfrac{d}{dt}\mu^{\alpha}\Bigl(t,\chi^{\alpha}(t,\Omega')\Bigr) &=& \dfrac{d}{dt} \int_{\chi^{\alpha}(t,\Omega')}\rho^{\alpha}(t,x)d\lambda^d(x)\\
\nonumber &=& \dfrac{d}{dt} \int_{\Omega'}\rho^{\alpha}\bigl(t, \chi^{\alpha}(t,X)\bigr)\det(\nabla \chi^{\alpha})(t,X) d\lambda^d(X)\\
\nonumber &=& \int_{\Omega'}\dfrac{\partial\rho^{\alpha}}{\partial t}\det(\nabla \chi^{\alpha}) d\lambda^d\\
\nonumber &&+ \int_{\Omega'}\nabla\rho^{\alpha}\cdot \dfrac{\partial\chi^{\alpha}}{\partial t}\det(\nabla \chi^{\alpha}) d\lambda^d\\
\nonumber &&+ \int_{\Omega'}\rho^{\alpha}\dfrac{d}{dt}\det(\nabla \chi^{\alpha}) d\lambda^d\\
\nonumber &=& \int_{\Omega'}\Bigl(\dfrac{\partial\rho^{\alpha}}{\partial t}+\nabla\rho^{\alpha}\cdot v^{\alpha}+\rho^{\alpha}\nabla\cdot v^{\alpha} \Bigr)\det(\nabla \chi^{\alpha}) d\lambda^d\\
\nonumber &=& \int_{\Omega'}\Bigl(\dfrac{\partial\rho^{\alpha}}{\partial t}+\nabla\cdot\bigl(\rho^{\alpha} v^{\alpha}\bigr)\Bigr)\det(\nabla \chi^{\alpha}) d\lambda^d\\
&=& \int_{\chi^{\alpha}(t,\Omega')}\Bigl(\dfrac{\partial\rho^{\alpha}}{\partial t}+\nabla\cdot\bigl(\rho^{\alpha} v^{\alpha}\bigr)\Bigr)d\lambda^d.\label{derivation global mass balance}
\end{eqnarray}
Since $t$ and $\Omega'$ can be chosen arbitrarily, we conclude from the expression in the seventh line of (\ref{derivation global mass balance}) that for all $\alpha\in\{1,2,\ldots,\nu\}$ and for all $t\in(0,T)$
\begin{equation}\label{mass balance local including Jacobian}
\Bigl(\dfrac{\partial\rho^{\alpha}}{\partial t}+\nabla\cdot\bigl(\rho^{\alpha} v^{\alpha}\bigr)\Bigr)\det(\nabla \chi^{\alpha})=0,\hspace{1 cm} \lambda^d\text{-a.e. in }\Omega.
\end{equation}
Since the inverse of $\chi^{\alpha}$ is assumed to be differentiable, we know that $\det(\nabla \chi^{\alpha})\neq 0$ (cf. remark on p.~4 of \cite{TemamMiranville}). Thus, we deduce from (\ref{mass balance local including Jacobian}) that
\begin{equation}\label{mass balance local constituent}
\dfrac{\partial\rho^{\alpha}}{\partial t}+\nabla\cdot\bigl(\rho^{\alpha} v^{\alpha}\bigr)=0,\hspace{1 cm}\lambda^d\text{-a.e. in }\Omega.
\end{equation}
The weak formulation of (\ref{mass balance local constituent}) leads to
\begin{equation}\label{weak formulation mass balance constituent}
\dfrac{d}{dt}\int_{\Omega}\psi^{\alpha}\rho^{\alpha} d\lambda^d=\int_{\Omega}v^{\alpha}\cdot\nabla\psi^{\alpha}\rho^{\alpha} d\lambda^d,
\end{equation}
for all $\alpha\in\{1,2,\ldots,\nu\}$, for all $t\in(0,T)$, and for all $\psi^{\alpha}\in C^1_0(\bar{\Omega})$.\\
\\
\begin{remark}
From (\ref{mass balance local constituent}) we can derive the local mass balance for the mixture as a whole. To achieve this, we sum (\ref{mass balance local constituent}) over all $\alpha$ and take (\ref{density is sum partial densities w.r.t. lambda}) and (\ref{def barycentric velocity}) into consideration. This yields that for all $t\in(0,T)$
\begin{equation}\label{mass balance local mixture}
\dfrac{\partial\rho}{\partial t}+\nabla\cdot\bigl(\rho v\bigr)=0,\hspace{1 cm}\lambda^d\text{-a.e. in }\Omega.
\end{equation}
The weak formulation of (\ref{mass balance local mixture}) can be either deduced directly from (\ref{mass balance local mixture}), or derived from (\ref{weak formulation mass balance constituent}). For the latter way, we take the same test function $\psi^{\alpha}=\psi\in C^1_0(\bar{\Omega})$ for each $\alpha$, again sum over all $\alpha$ and use (\ref{density is sum partial densities w.r.t. lambda}) and (\ref{def barycentric velocity}). This procedure results in
\begin{equation}\label{weak formulation mass balance mixture}
\dfrac{d}{dt}\int_{\Omega}\psi\rho d\lambda^d=\int_{\Omega}v\cdot\nabla\psi\rho d\lambda^d.
\end{equation}
In Remark \ref{remark where c^alpha defined} we indicated that defining $c^{\alpha}$ is actually only useful in regions where $\rho>0$. This is made clear by (\ref{mass balance local mixture}) and (\ref{weak formulation mass balance mixture}). The mass concentrations $c^{\alpha}$ are incorporated in the definition of $v$. However $v$ only appears in combination with $\rho$ as the product $\rho v$. If $\rho=0$ the product is zero any way, and thus (at least physically) it does not matter how (or whether) $c^{\alpha}$ is defined.
\end{remark}

\section{Generalization}\label{section generalization mixture}
We would like to extend the ideas presented in Sections \ref{section description mixture}--\ref{section mass balance mixture}, such that they hold not only for those measures that are absolutely continuous w.r.t. $\lambda^d$. Indeed we have seen in Section \ref{section Lebesgue decomposition} that in general, a measure might also contain a singular part. The mixture-theoretical concepts presented so far do however not hold for singular measures.\\
\\
Let $\mu^{\alpha}(t,\Omega')$, like before, denote the mass of constituent $\alpha$ present in $\Omega'\in\mathcal{B}(\Omega)$ at time $t\in(0,T)$. To obtain a more general framework, we need to revoke the assumption that $\mu^{\alpha}\ll \tau^{\alpha}$ for all $\alpha\in\{1,2,\ldots,\nu\}$. In this setting we thus postulate the following:
\begin{postulate}[Properties of $\mu^{\alpha}$]\label{Postulate mu^alpha generalized}
For all $\alpha\in\{1,2,\ldots,\nu\}$ and all $t\in(0,T)$ we postulate:
\begin{enumerate}
  \item $\mu^{\alpha}\geqslant 0$.\label{Postulate mu^alpha generalized 1}\\
  \item $\mu^{\alpha}$ is $\sigma$-additive.\\
\end{enumerate}
\end{postulate}
Furthermore, as in (\ref{def total mass mixture}), we define the total mass by
\begin{equation}
\mu(t, \Omega'):=\sum_{\alpha=1}^{\nu}\mu^{\alpha}(t, \Omega'),\label{def total mass mixture generalized}
\end{equation}
for all $\Omega'\in\mathcal{B}(\Omega)$ and all $t\in(0,T)$.\\
\\
This new setting, without the absolute continuity demand, implies that in most cases of Section \ref{section description mixture} we cannot apply the Radon-Nikodym Theorem anymore. However, note that due to (\ref{def total mass mixture generalized}) and Part \ref{Postulate mu^alpha generalized 1} of Postulate \ref{Postulate mu^alpha generalized} we have that $\mu^{\alpha}\ll\mu$ for all $\alpha\in\{1,\ldots,\nu\}$. Thus for each $\mu^{\alpha}$ a unique Radon-Nikodym derivative exists with respect to $\mu$. Let us denote
\begin{equation}\label{mass concentration RN-deriv generalization}
c^{\alpha}:=\dfrac{d\mu^{\alpha}}{d\mu},
\end{equation}
which is also called the \textit{mass concentration} of constituent $\alpha$ at time $t$.
\begin{remark}
If we still would have been in the absolute continuous case of Section \ref{section description mixture}, both $\mu\ll\tau$ and $\tau\ll\mu$ were true. Due to Part \ref{Property inverse} of Lemma \ref{Properties RN derivatives}, this implies that $1/\check{\rho}=d\tau/d\mu$, which is defined $\mu$-almost everywhere. Part \ref{Property three measures} of the same lemma then implies that
\begin{equation*}
c^{\alpha}=\dfrac{d\mu^{\alpha}}{d\mu}=\dfrac{d\mu^{\alpha}}{d\tau} \dfrac{d\tau}{d\mu} =\dfrac{\check{\rho}^{\alpha}}{\check{\rho}}.
\end{equation*}
Also:
\begin{equation*}
\dfrac{\check{\rho}^{\alpha}}{\check{\rho}}=\dfrac{\check{\rho}^{\alpha}\Theta}{\check{\rho}\Theta}=\dfrac{\rho^{\alpha}}{\rho},
\end{equation*}
and since all expressions above are defined $\mu$-a.e., we have that
\begin{equation*}
c^{\alpha}=\dfrac{\rho^{\alpha}}{\rho},\hspace{1 cm}\mu\text{-almost everywhere.}
\end{equation*}
Hence, in the absolutely continuous case (\ref{mass concentration RN-deriv generalization}) is equivalent to the definition of $c^{\alpha}$ given in Section \ref{section description mixture}.
\end{remark}

It follows from (\ref{def total mass mixture generalized}) that:
\begin{equation}\label{sum of c^alpha = 1 generalization}
\sum_{\alpha=1}^{\nu}c^{\alpha}=1,\hspace{1 cm}\text{a.e. w.r.t. }\mu,\text{ and for all }t\in(0,T).
\end{equation}
Up to the definition of the barycentric velocity, in Section \ref{section kinematics mixture} we have not used the absolute continuity of any of the mass measures. The ideas presented there are also applicable in the generalized setting. We only have to redefine (in a very natural way) the barycentric velocity of the mixture by
\begin{equation}\label{definition barycentric vel generalization}
v:=\sum_{\alpha=1}^{\nu}c^{\alpha}v^{\alpha}.
\end{equation}
\\
From a notational point of view it is not difficult to generalize the weak formulation of the balance of mass concept, as given in (\ref{weak formulation mass balance constituent}). Since we know that $\rho^{\alpha}=d\mu^{\alpha}/d\lambda^d$, the weak form can be written as
\begin{equation}\label{weak formulation mass balance constituent generalized}
\dfrac{d}{dt}\int_{\Omega}\psi^{\alpha}d\mu^{\alpha}=\int_{\Omega}v^{\alpha}\cdot\nabla\psi^{\alpha}d\mu^{\alpha},
\end{equation}
for all $\alpha\in\{1,2,\ldots,\nu\}$, for all $t\in(0,T)$, and for all $\psi^{\alpha}\in C^1_0(\bar{\Omega})$. This `trick' is merely a matter of notation, and the above is equivalent to (\ref{weak formulation mass balance constituent}) for absolutely continuous mass measures. However, we still need to make sure that this expression also makes sense for more general mass measures.\\
\\
We are running ahead by mentioning this now, but we will see later (see e.g. Section \ref{section two-scales no sing cont}) that we are mainly interested in measures that are either absolutely continuous, or discrete (i.e. sum of Dirac distributions), or possibly a combination of the two. This means that we explicitly exclude the singular continuous part from the measure, and only show that (\ref{weak formulation mass balance constituent generalized}) makes sense for discrete measures. Note that the measures we do allow, are (combinations of) exactly those measures introduced in Section \ref{section Mass measure}.\\
\\
We consider $\mu^{\alpha}$ to be a single Dirac measure centered at a time-dependent position $x(t)\in\Omega$ for each $t\in(0,T)$, and assume its motion to be described by $dx(t)/dt=v^{\alpha}(t,x(t))$. We choose an arbitrary function $\psi^{\alpha}\in C^1_0(\bar{\Omega})$ and take the inner product with this function (evaluated in $x(t)$) on both sides of the equality. In the resulting left-hand side, we recognize the chain rule (recall that $\psi^{\alpha}$ is once continuously differentiable), so
\begin{equation*}
\nabla\psi^{\alpha}(x(t))\cdot\dfrac{d}{dt}x(t)= \dfrac{d}{dt}\psi^{\alpha}(x(t)).
\end{equation*}
Now, we have obtained
\begin{equation}\label{time der psi equals grad in velocity}
\dfrac{d}{dt}\psi^{\alpha}(x(t))=\nabla\psi^{\alpha}(x(t))\cdot v^{\alpha}(t,x(t)).
\end{equation}
We apply the definition of the Dirac measure, and since $\mu^{\alpha}=\delta_{x(t)}$, we find that (\ref{time der psi equals grad in velocity}) can also be written as
\begin{equation*}
\dfrac{d}{dt}\int_{\Omega}\psi^{\alpha}d\mu^{\alpha}=\int_{\Omega}v^{\alpha}\cdot\nabla\psi^{\alpha}d\mu^{\alpha}.
\end{equation*}
To extend this approach, now let $\mu^{\alpha}$ be a linear combination of Dirac deltas centered at $x_i(t)$, and let each of these move according to
\begin{equation}\label{dx_i/dt = v^alpha}
\dfrac{d}{dt}x_i(t)=v^{\alpha}(t,x_i(t)).
\end{equation}
Here the index $i$ is taken from a countable, possibly infinite, index set $\mathcal{J}$. If $\mathcal{J}$ is infinite, the coefficients of the Dirac deltas must have finite sum. Similar arguments as before yield that, for each $i\in\mathcal{J}$, we have
\begin{equation*}
\dfrac{d}{dt}\psi^{\alpha}(x_i(t))=\nabla\psi^{\alpha}(x_i(t))\cdot v^{\alpha}(t,x_i(t)).
\end{equation*}
If $\mu^{\alpha}=\sum_{i\in\mathcal{J}}\alpha_i\delta_{x_i(t)}$ with nonnegative coefficients such that $\sum_{i\in\mathcal{J}}\alpha_i<\infty$, then we have
\begin{equation*}
\sum_{i\in\mathcal{J}}\alpha_i\dfrac{d}{dt}\psi^{\alpha}(x_i(t))=\sum_{i\in\mathcal{J}}\alpha_i\nabla\psi^{\alpha}(x_i(t))\cdot v^{\alpha}(t,x_i(t)),
\end{equation*}
which is, by the definition of the Dirac measure, again equal to (\ref{weak formulation mass balance constituent generalized}).\\
\\
Note that the weak formulation (\ref{weak formulation mass balance constituent generalized}) also makes sense for linear combinations of absolutely continuous and discrete measures. For the absolutely continuous part we have the balance of mass-interpretation; for the discrete part we interpret it via (\ref{dx_i/dt = v^alpha}).

\begin{remark}\label{remark sum of diracs preserves constant coeff}
We have seen in Lemma \ref{lemma characterization discrete measure} that any discrete measure can be written as a linear combination of Dirac distributions. However, if (for every $t$) $\mu^{\alpha}(t,\cdot)$ is discrete, not only the centres of the Dirac deltas, but also the coefficients might be time-dependent. It will be shown later that the time-dependent discrete measures we work with have constant coefficients (see Corollary \ref{corollary sum of diracs preserves constant coeff}). This means that the time-dependence is only present in the location of the Dirac deltas. It will also be shown that these positions $x_i(t)$ satisfy (\ref{dx_i/dt = v^alpha}); cf. Lemma \ref{lemma sum of diracs obey dx/dt=v}.\\
In this spirit, (\ref{weak formulation mass balance constituent generalized}) makes sense for the types of measures that are relevant for us.
\end{remark}

\begin{remark}
Analogously to Section \ref{section mass balance mixture}, we can deduce from (\ref{weak formulation mass balance constituent generalized}) the structure of a weak formulation that applies to the total mass measure $\mu$. We take the same test function $\psi^{\alpha}=\psi\in C^1_0(\bar{\Omega})$ for each $\alpha$. Each integral w.r.t. $\mu^{\alpha}$ is transformed into an integral w.r.t. $\mu$ using (\ref{mass concentration RN-deriv generalization}). Consequently, summing over all $\alpha$, and using (\ref{sum of c^alpha = 1 generalization}) and (\ref{definition barycentric vel generalization}), yields
\begin{equation}\label{weak formulation mass balance mixture generalization}
\dfrac{d}{dt}\int_{\Omega}\psi d\mu=\int_{\Omega}v\cdot\nabla\psi d\mu.
\end{equation}
\end{remark}

\section{Entropy inequality}\label{section mixture entropy}
In this section, we place the mixture concepts discussed in Sections \ref{section description mixture}--\ref{section generalization mixture} into a thermodynamical context. Our aim is to derive an \textit{entropy inequality}. This inequality is strongly related to concepts like the second axiom of thermodynamics, or the Clausius-Duhem Inequality.\\
For more information on thermodynamics and on the role of entropy, the reader is referred to e.g. \cite{Zemansky} or \cite{Hawking}.\\
\\
To each constituent of the mixture we assign a function $\eta^{\alpha}:(0,T)\times\Omega\rightarrow\mathbb{R}$, which is called the \textit{entropy density} (per unit mass) of component $\alpha$. The entropy density for the total mixture is defined as
\begin{equation}\label{def total entropy}
\eta := \sum_{\alpha=1}^{\nu}c^{\alpha}\eta^{\alpha}.
\end{equation}
Note that this definition has the same structure as the definition of the barycentric velocity in (\ref{def barycentric velocity}). The entropy $S_{\Omega'}(t)$ assigned to $\Omega'\subset\Omega$ at time $t$ is defined by integration of $\eta$ with respect to the density $\rho$, that is
\begin{equation}\label{def S total entropy integral entr dens}
S_{\Omega'}(t):=\int_{\Omega'}\eta\rho d\lambda^d.
\end{equation}
To each component we also assign a temperature $T^{\alpha}:(0,T)\times\Omega\rightarrow\mathbb{R}^+$.\\
\\
Following the ideas of \cite{Muller} (p.~7), \cite{DunwoodyMueller} (p.~347), and \cite{Bowen} (p.~28), we now postulate the following:
\begin{postulate}[Local entropy inequality]\label{postulate local entropy inequality}
Locally the following inequality holds:
\begin{equation}\label{local entropy inequality}
\rho \dfrac{D\eta}{Dt}+\nabla\cdot j_{\eta}-\sum_{\alpha=1}^{\nu}\rho^{\alpha}\dfrac{r^{\alpha}}{T^{\alpha}}\geqslant0.
\end{equation}
Here $j_{\eta}$ is the \textit{entropy flux vector}, and $r^{\alpha}$ is the volume supply of external heat to constituent $\alpha$.
\end{postulate}

\begin{remark}
In the term $\sum_{\alpha=1}^{\nu}\rho^{\alpha}r^{\alpha}/T^{\alpha}$, which is the volume supply of external entropy to constituent $\alpha$, we recognize a similar structure as in (\ref{def barycentric velocity}) and (\ref{def total entropy}). It can therefore be considered as some "barycentric quantity".
\end{remark}

\begin{remark}
Sometimes (\ref{local entropy inequality}) is written as
\begin{equation*}
\dfrac{\partial}{\partial t}(\rho\eta)+\nabla\cdot(\rho\eta v+j_{\eta})-\sum_{\alpha=1}^{\nu}\rho^{\alpha}\dfrac{r^{\alpha}}{T^{\alpha}}\geqslant0,
\end{equation*}
cf. e.g. \cite{Muller}, p.~7. Using the local balance of mass (\ref{mass balance local mixture}) and Definition \ref{def time derivatives}, we can indeed show that these two formulations are equivalent:
\begin{eqnarray}
\nonumber \dfrac{\partial}{\partial t}(\rho\eta)+\nabla\cdot(\rho\eta v) &=& \dfrac{\partial}{\partial t}(\rho\eta)+\nabla(\rho\eta)\cdot v+\rho\eta \nabla\cdot v \\
\nonumber &=& \dfrac{D}{Dt}(\rho\eta)+\rho\eta \nabla\cdot v \\
\nonumber &=& \dfrac{D\rho}{Dt}\eta+ \rho\dfrac{D\eta}{Dt}+\eta\bigl( \nabla\cdot(\rho v)-\nabla\rho\cdot v\bigr)\\
\nonumber &=& \rho\dfrac{D\eta}{Dt}+ \eta\bigl(\dfrac{D\rho}{Dt}-\nabla\rho\cdot v\bigr)+ \eta\nabla\cdot(\rho v)\\
\nonumber &=& \rho\dfrac{D\eta}{Dt}+ \eta\dfrac{\partial\rho}{\partial t}+ \eta\nabla\cdot(\rho v)\\
\nonumber &=& \rho\dfrac{D\eta}{Dt}+ \eta\bigl(\dfrac{\partial\rho}{\partial t}+ \nabla\cdot(\rho v)\bigr)\\
 &=& \rho\dfrac{D\eta}{Dt}.
\end{eqnarray}
\end{remark}

For an arbitrarily fixed $\Omega'\subset\Omega$ (such that $\partial\Omega'$ is sufficiently regular) we now calculate the time derivative (at fixed time $t\in(0,T)$) of the total entropy contained in $\Omega'$:
\begin{eqnarray}
\nonumber \dfrac{\partial}{\partial t}\int_{\Omega'}\eta\rho d\lambda^d &=& \int_{\Omega'}\dfrac{\partial}{\partial t}(\eta\rho) d\lambda^d\\
\nonumber &\geqslant& \int_{\Omega'}\sum_{\alpha=1}^{\nu}\rho^{\alpha}\dfrac{r^{\alpha}}{T^{\alpha}}- \nabla\cdot(\rho\eta v+j_{\eta})d\lambda^d\\
&=& \int_{\Omega'}\sum_{\alpha=1}^{\nu}\rho^{\alpha}\dfrac{r^{\alpha}}{T^{\alpha}}d\lambda^d - \int_{\partial\Omega'}(\rho\eta v+j_{\eta})\cdot n d\lambda^{d-1},\label{global entropy inequality}
\end{eqnarray}
where we have used the local entropy inequality in the second line of (\ref{global entropy inequality}).\\
We call (\ref{global entropy inequality}) the \textit{global entropy inequality}.

\begin{remark}\label{remark how to derive entropy ineq whole mixture}
It would have been possible to derive the entropy inequality for the whole mixture from the partial entropy inequalities formulated individually, per constituent; cf. e.g. \cite{GreenNaghdi} (pp.~475--476). However, nowadays the general consensus is that this approach gives a too restrictive result: the motion of the mixture is overconstraint. For more details on this fundamental issue, the reader is referred to \cite{Holmes} (pp.~II.12--13), and \cite{BedfordDrumheller} (p.~865).
\end{remark}

\begin{remark}\label{remark physical relevance generalization entropy inequality}
In general we are also interested in mass measures that are not exclusively absolutely continuous. This naturally leads to the question whether we can generalize the concept of entropy and the accompanying entropy inequality. However, at this point, we should ask ourselves what the physical meaning of such generalization is. More understanding is needed in order to judge the physical relevance of an extension to a broader class of measures. There are also practical objections. For example, caution is needed when generalizing boundary terms. It is all but obvious in what way the integral over $\partial \Omega'$ in (\ref{global entropy inequality}) should be defined for a general mass measure, that may also contain a discrete part.\\
In Section \ref{section entropy inequality cont-in-time} we introduce an explicit entropy density, and derive the corresponding entropy inequality. Despite of the aforementioned difficulties, we show afterwards (see Section \ref{section generalization entropy inequality to discrete measures}) that structurally the same inequality can be obtained for discrete measures (compare the statement of Theorem \ref{theorem entropy ineq one component} with (\ref{entropy ineq discrete measure})).
\end{remark}

\newpage
\chapter{Application to crowd dynamics}\label{section Application crowd dynamics}
In this section we explain and extend the measure-theoretical approach, developed by Bene\-detto Piccoli et al. (cf. e.g. \cite{PiccoliTosin, PiccoliTosinMeasTh, Piccoli2010}). We intend to fit some of their ideas to our framework presented in Section \ref{section mixture theory and thermodynamics}.\\
\\
We consider a population located in a given domain $\Omega\subset\mathbb{R}^d$. To capture physically realistic situations we take $d\in\{1,2,3\}$. Let $T\in(0,\infty)$ be the fixed final time of the process. We define a measure $\mu:(0,T)\times\mathcal{B}(\Omega)\rightarrow\mathbb{R}^+$, such that $\mu(t,\Omega')$ represents the mass of the part of the population present in a region $\Omega'\in\mathcal{B}(\Omega)$ at time $t\in(0,T)$.\\
We assume that the population consists of a fixed number of subpopulations (these were called \textit{constituents} in Section \ref{section mixture theory and thermodynamics}), indexed $\alpha\in\{1,2,\ldots,\nu\}$. The mass of each subpopulation present in $\Omega'\in\mathcal{B}(\Omega)$ at time $t\in(0,T)$ is given by the time-dependent mass measure $\mu^{\alpha}(t,\Omega')$.\\
The mass measures $\mu$ and $\mu^{\alpha}$ are related via $\mu=\sum_{\alpha=1}^{\nu}\mu^{\alpha}$. We explicitly assume that $\mu^{\alpha}(t,\cdot)$ is a finite measure for all $t\in(0,T)$ and $\alpha\in\{1,2,\ldots,\nu\}$. As a result, $\mu(t,\cdot)$ becomes a finite measure for all $t\in(0,T)$.\\

\section{Weak formulation}\label{section weak formulation}
Each subpopulation moves according to its own motion mapping $\chi^{\alpha}$, from which a velocity field $v^{\alpha}(t,x)$ follows. Note that for each $\alpha$, the dependence of $v^{\alpha}$ on $t$ indicates the functional dependence on all time-dependent measures $\mu^{\alpha}$, as we will see in Section \ref{section specification velocity}.
\begin{remark}
We henceforth disregard the assumptions with respect to the motion mappings $\chi^{\alpha}$ and velocity fields $v^{\alpha}$, which were done in the derivation of the balance of mass equations presented in Section \ref{section mixture theory and thermodynamics}. The result, Equation (\ref{weak formulation mass balance constituent generalized}), is taken here as a starting point for further modelling, without considering the underlying assumptions.
\end{remark}
Recapitulating: the fact that here we deal with $\nu$ time-dependent mass measures $\mu^{\alpha}$, transported with corresponding velocities $v^{\alpha}$, translates into
\begin{equation}\label{Cons of mass}
\dfrac{\partial \mu^{\alpha}}{\partial t}+\nabla \cdot (\mu^{\alpha} v^{\alpha})=0,\hspace{1 cm}\text{for all }\alpha\in\{1,2,\ldots,\nu\}.
\end{equation}
Equations (\ref{Cons of mass}) are accompanied by the following initial conditions:
\begin{equation}\label{initial condition constituent}
\mu^{\alpha}(0,\cdot)=\mu_0^{\alpha}, \hspace{1 cm}\text{for each }\alpha\in\{1,2,\dots,\nu\},
\end{equation}
for given $\mu_0^{\alpha}$.\\
\\
This partial differential equation in terms of measures is a shorthand notation for the weak formulation presented in Sections \ref{section mass balance mixture} and \ref{section generalization mixture}. Namely, for all test functions $\psi^{\alpha}\in C_0^1(\bar{\Omega})$, where $\alpha\in\{1,2,\ldots,\nu\}$ and for almost every $t\in(0,T)$ the following holds:
\begin{equation}\label{Weak Form}
\dfrac{d}{dt}\int_{\Omega}\psi^{\alpha}(x)d\mu^{\alpha}(t,x)=\int_{\Omega}v^{\alpha}(t,x)\cdot\nabla\psi^{\alpha}(x)d\mu^{\alpha}(t,x), \hspace{1 cm} \alpha\in\{1,2,\ldots,\nu\}.
\end{equation}

\begin{remark}
If, for some $\alpha\in\{1,2,\ldots,\nu\}$, the mapping $t\mapsto \int_{\Omega}\psi^{\alpha}(x)d\mu^{\alpha}(t,x)$ is absolutely continuous on $(0,T)$, for any choice of $\psi^{\alpha}\in C_0^1(\bar{\Omega})$, then it is differentiable at almost every $t\in(0,T)$. See e.g. Theorem 7.20 in \cite{Rudin}, p.~148. This means that the left-hand side of (\ref{Weak Form}) exists for almost every $t$.
\end{remark}

\begin{remark}
If we demand that $v^{\alpha}\in L^1\bigl((0,T);L_{\mu^{\alpha}}^1(\Omega)\bigr)$ for all $\alpha\in\{1,2,\ldots,\nu\}$, then the right-hand side of (\ref{Weak Form}) is well-defined. Indeed, since $\psi^{\alpha}\in C_0^1(\bar{\Omega})$, we have that $\|\nabla \psi^{\alpha}\|_{L^{\infty}(\Omega)}$ is finite for each $\alpha\in\{1,2,\ldots,\nu\}$. We thus have for each $\alpha\in\{1,2,\ldots,\nu\}$ that
\begin{eqnarray}
\nonumber \Bigl|\int_0^T\int_{\Omega}v^{\alpha}(t,x)\cdot\nabla\psi^{\alpha}(x)d\mu^{\alpha}(t,x)dt\Bigr| &\leqslant& \int_0^T\Bigl|\int_{\Omega}v^{\alpha}(t,x)\cdot\nabla\psi^{\alpha}(x)d\mu^{\alpha}(t,x)\Bigr|dt\\
\nonumber &\leqslant& \int_0^T\|\nabla \psi^{\alpha}\|_{L^{\infty}(\Omega)}\int_{\Omega}|v^{\alpha}(t,x)|d\mu^{\alpha}(t,x)dt\\
\nonumber &=& \int_0^T\|\nabla \psi^{\alpha}\|_{L^{\infty}(\Omega)}\|v^{\alpha}(t,\cdot)\|_{L_{\mu^{\alpha}}^1(\Omega)}dt\\
\nonumber &=& \|\nabla \psi^{\alpha}\|_{L^{\infty}(\Omega)}\int_0^T\|v^{\alpha}(t,\cdot)\|_{L_{\mu^{\alpha}}^1(\Omega)}dt\\
\nonumber &=& \|\nabla \psi^{\alpha}\|_{L^{\infty}(\Omega)}\|v^{\alpha}\|_{L^1\bigl([0,T];L_{\mu^{\alpha}}^1(\Omega)\bigr)}\\
&<& \infty.
\end{eqnarray}
In particular, it follows that $\int_{\Omega}v^{\alpha}(t,x)\cdot\nabla\psi^{\alpha}(x)d\mu^{\alpha}(t,x)$ is finite for almost every $t\in(0,T)$ and for all $\alpha\in\{1,2,\ldots,\nu\}$, and thus the right-hand side of (\ref{Weak Form}) is well-defined.
\end{remark}

\begin{definition}[Weak solution of (\ref{Cons of mass})]\label{def weak solution}
The vector of time-dependent measures
\begin{equation*}
\Bigl((\mu^1)_{t\geqslant 0}, (\mu^2)_{t\geqslant 0},\ldots, (\mu^{\nu})_{t\geqslant 0}\Bigr),
\end{equation*}
where each $\mu^{\alpha}$ satisfies the initial condition (\ref{initial condition constituent}), is called a weak solution of (\ref{Cons of mass}), if for all $\alpha\in\{1,2,\ldots,\nu\}$ the following properties hold:
\begin{enumerate}[(i)]
\item for all $\psi^{\alpha}\in C_0^1(\bar{\Omega})$, the mappings $t\mapsto \int_{\Omega}\psi^{\alpha}(x)d\mu^{\alpha}(t,x)$ are absolutely continuous,
\item $v^{\alpha}\in L^1\bigl((0,T);L_{\mu^{\alpha}}^1(\Omega)\bigr)$,
\item (\ref{Weak Form}) is satisfied.
\end{enumerate}
\end{definition}
We deduced a weak formulation with respect to the total mass measure $\mu$ in Section \ref{section generalization mixture}. We use the shorthand notation
\begin{equation}\label{Cons of mass total mass measure}
\dfrac{\partial \mu}{\partial t}+\nabla \cdot (\mu v)=0,
\end{equation}
accompanied, for given $\mu_0$, by the initial condition
\begin{equation}\label{initial condition total}
\mu(0,\cdot)=\mu_0.
\end{equation}
This shorthand notation should, again, be interpreted as (cf. (\ref{weak formulation mass balance mixture generalization}))
\begin{equation}\label{Weak Form system}
\dfrac{d}{dt}\int_{\Omega}\psi(x) d\mu(t,x)=\int_{\Omega}v(t,x)\cdot\nabla\psi(x) d\mu(t,x),
\end{equation}
for all $\psi\in C_0^1(\bar{\Omega})$ and almost every $t\in(0,T)$.\\
\\
A weak solution to (\ref{Cons of mass total mass measure}) will be understood in this section always in the sense of Definition \ref{def weak solution system}:
\begin{definition}[Weak solution of (\ref{Cons of mass total mass measure})]\label{def weak solution system}
The time-dependent measure $(\mu)_{t\geqslant 0}$, satisfying the initial condition (\ref{initial condition total}), is called a weak solution of (\ref{Cons of mass total mass measure}), if the following statements hold:
\begin{enumerate}[(i)]
\item for all $\psi\in C_0^1(\bar{\Omega})$, the mapping $t\mapsto \int_{\Omega}\psi(x)d\mu(t,x)$ is absolutely continuous,
\item $v\in L^1\bigl((0,T);L_{\mu}^1(\Omega)\bigr)$,
\item (\ref{Weak Form system}) is fulfilled.
\end{enumerate}
\end{definition}

The main difference between our approach here and the one presented in \cite{Piccoli2010}, is that we take the freedom to allow each subpopulation to have its own velocity field $v^{\alpha}$. Note that the crucial modelling steps take place when we decide what $v^{\alpha}$ actually looks like. We then characterize the motion of the corresponding subpopulation. In particular, we define the way in which this motion is influenced by the crowd surrounding it (belonging to any of the components).\\
\\
In \cite{Piccoli2010} merely the setting of transporting the total mass measure is considered (see Definition \ref{def weak solution system}). This means that only the  barycentric velocity $v$ of the crowd seen as a whole can be prescribed. However, we thus lose the ability of modelling the interaction between subpopulations of different types. Distinct behaviour of subpopulations is namely driven by underlying partial velocities of distinct nature.\\
\\
To be more specific, \cite{Piccoli2010} considers a combination of an absolutely continuous measure (a `cloud' of people, in which individuals are indistinguishable) and a sum of Dirac measures (point masses). Only $v$ is prescribed in the framework of that paper, partial velocities being not included. This implies that on the macroscopic scale (the crowd) the absolutely continuous and discrete parts essentially move according to the same velocity field. Correspondingly, the Dirac masses cannot evolve independently from the cloud; they are in fact part of the cloud with some "pointer" attached to them.\\
\\
Note that this situation is incorporated in our model as a special case, where we take one $\mu^{\alpha}$ to be a combination of Dirac measures and an absolutely continuous part. However, we feel that it is of no use drawing attention to these point masses, if they cannot behave independently.

\section{Further specification of the velocity fields}\label{section specification velocity}
Until now, we have not explicitly defined the velocity fields $v^{\alpha}$ ($\alpha\in\{1,2,\ldots,\nu\}$). In Section \ref{section weak formulation} we have already remarked that the crucial modelling step takes place at the moment when we decide on the explicit form of $v^{\alpha}$. By the choices we make here, we define the characteristics of a subpopulation. That is, we assign, so to say, a personality to the members of that subpopulation.\\
In the definition of $v^{\alpha}$ we can incorporate the way in which an individual is influenced by the people around him. This individual might be shy, wanting to keep distance from others. The individual's character might also be quite the opposite, driving him to come close to other people. We even have the freedom to distinguish between the way the individual responds to the presence of others, based on the subpopulation that other individual belongs to. For example, one subpopulation might consist of acquaintances (triggering attractive behaviour), while a second subpopulation consists of `enemies' (from which the individual is repelled).\\
\\
In this section, we show the way we want to model the velocity fields. Very much inspired by the \textit{social force model} by Dirk Helbing et al. \cite{HelbingMolnar}, the velocity of a pedestrian is modelled as a \textit{desired velocity} $v_{\text{des}}^{\alpha}$ perturbed by a component $v_{\text{soc}}^{\alpha}$. The latter component, which is called \textit{social velocity}, is due to the presence of other individuals, both from the pedestrian's own subpopulation and from the other subpopulation. This effect is modelled by functional dependence of $v_{\text{soc}}^{\alpha}$ on the set of measures $\bigl\{\mu^1, \mu^2,\ldots, \mu^{\nu}\bigr\}$. The desired velocity is independent of these mass measures, and represents the velocity a pedestrian would have had in absence of other people; this implies that $v_{\text{des}}^{\alpha}$ is independent of $t$, provided the environment of an individual does not change in time. Cf. Remark \ref{remark vdes independent of time}.\\
\\
The velocity $v^{\alpha}$ is thus for $\alpha\in\{1,2,\ldots,\nu\}\}$ defined by superposing $v_{\text{des}}^{\alpha}$ and its perturbation $v_{\text{soc}}^{\alpha}$:
\begin{equation}\label{Def v decomposition}
v^{\alpha}(t,x):= v_{\text{des}}^{\alpha}(x)+v_{\text{soc}}^{\alpha}(t,x), \hspace{1 cm} \text{for all }t\in(0,T) \text{ and }x\in\Omega.
\end{equation}
The component $v_{\text{soc}}^{\alpha}$ models the effect of interactions with other pedestrians on the current velocity. This is essentially a nonlocal contribution. Since the effects (in general) differ from one subpopulation to the other, we assume that $v_{\text{soc}}^{\alpha}$ has the form:
\begin{equation}\label{Def vsoc sum of integrals}
v_{\text{soc}}^{\alpha}(t,x):= \sum_{\beta=1}^{\nu}\int_{\Omega \setminus \{x\}} f^{\alpha}_{\beta}(|y-x|)g(\theta_{xy}^{\alpha})\dfrac{y-x}{|y-x|}d\mu^{\beta}(t,y),\hspace{1 cm}\text{for all }\alpha\in\{1,2,\ldots,\nu\}.
\end{equation}
\begin{remark}\label{remark vdes independent of time}
Note thus that the $t$-dependence of $v^{\alpha}$ and $v^{\alpha}_{\text{soc}}$ is not an explicit time-dependence. As (\ref{Def vsoc sum of integrals}) shows, the social velocity depends functionally on the time-dependent mass measures $\mu^{\beta}$. Although we have not done so, we are also allowed to choose $v_{\text{des}}=v_{\text{des}}(t,x)$. This ought to be an explicit dependence on $t$ due to the interpretation given to $v_{\text{des}}$: the velocity a pedestrian would have if he was the only person in the room.\\
Obviously, this term should depend only on the geometry of the domain (and not on the crowd it contains). If $v_{\text{des}}$ depends explicitly on the time, this reflects that the domain is non-static. For example, such situation occurs when certain areas are not always accessible (opening hours), or are temporarily obstructed. For the moment this is too farfetched.
\end{remark}
In (\ref{Def vsoc sum of integrals}), we have used the following modelling ansätze:
\begin{itemize}
  \item $f^{\alpha}_{\beta}$ is a function from $\mathbb{R}_+$ to $\mathbb{R}$, describing the effect of the mutual distance between individuals on their interaction. The subscript and superscript denote that this specific function incorporates the influence of the members of subpopulation $\beta$ on subpopulation $\alpha$. In principle we distinguish between two kinds of effects:
      \begin{itemize}
        \item \textit{Attraction-repulsion}: this is the interaction between acquaintances (and similar kinds of interaction). In this case, $f^{\alpha}_{\beta}$ is a composition of two effects: at short distance individuals are repelled, since they want to avoid collisions and congestion, but if their mutual distance increases they are attracted to other group mates, in order not to get separated from the group.
        \item \textit{Repulsion}: i.e. `shy' behaviour. Here, $f^{\alpha}_{\beta}$ contains only the repulsive part, since the individual in subpopulation $\alpha$ does not like to come close to subpopulation $\beta$.
      \end{itemize}
      Graphical representations of these two choices of $f^{\alpha}_{\beta}$ are depicted in Figure \ref{Figure graphs fAR fR}.
  \item $\theta_{xy}^{\alpha}$ denotes the angle between $y-x$ and $v_{\text{des}}^{\alpha}(x)$: the angle under which $x$ sees $y$ if it were moving in the direction of $v_{\text{des}}^{\alpha}(x)$.
  \item $g$ is a function from $[-\pi,\pi]$ to $[0,1]$ that encodes the fact that an individual's perception is not equal in all directions. We choose:
  \begin{equation}\label{g(alpha)}
  g(\theta):= \sigma +(1-\sigma)\dfrac{1+\cos(\theta)}{2}, \hspace{1 cm}\text{for }\theta\in[-\pi,\pi].
  \end{equation}
  This definition ensures that an individual experiences the strongest influence from someone straight ahead, since $g(0)=1$ for any $\sigma\in[0,1]$. The constant $\sigma$ is a tuning parameter, called \textit{potential of anisotropy}. It determines how strongly a pedestrian is focussed on what happens in front of him, and how large the influence is of people at his sides or behind him. The effect of the choice of this parameter can be seen in Figure \ref{Figure g}.\\
  Note that $\cos(\theta_{xy}^{\alpha})$ can be calculated in a convenient way:
  \begin{equation*}
  \cos(\theta_{xy}^{\alpha})= \dfrac{(y-x)\cdot v_{\text{des}}^{\alpha}(x)}{|y-x|\,|v_{\text{des}}^{\alpha}(x)|}.
  \end{equation*}
\end{itemize}

\begin{figure}[h]
\vspace{0 cm}
\begin{tabular}{rr}
\hspace{0.5 cm}\includegraphics[width=0.45\linewidth]{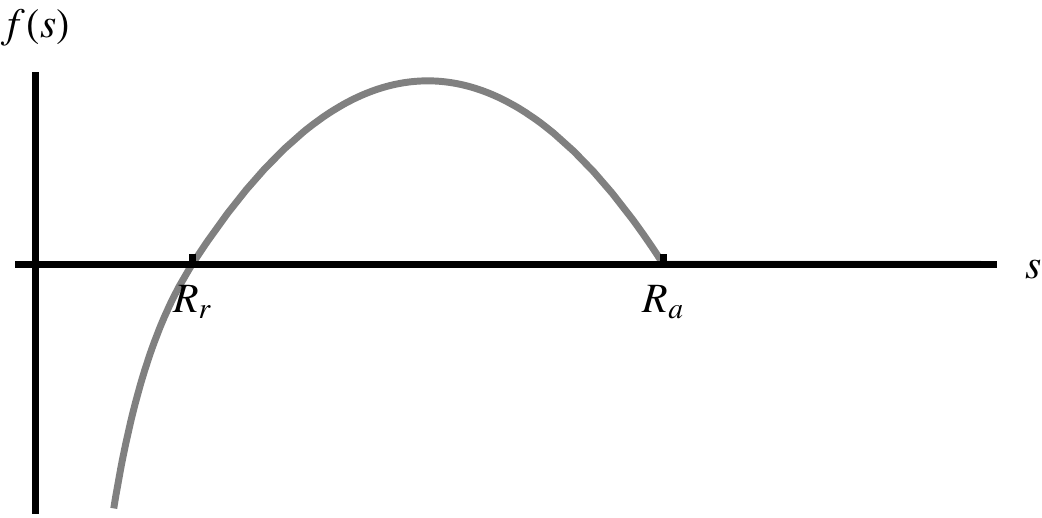}
&
\hspace{0 cm}\includegraphics[width=0.45\linewidth]{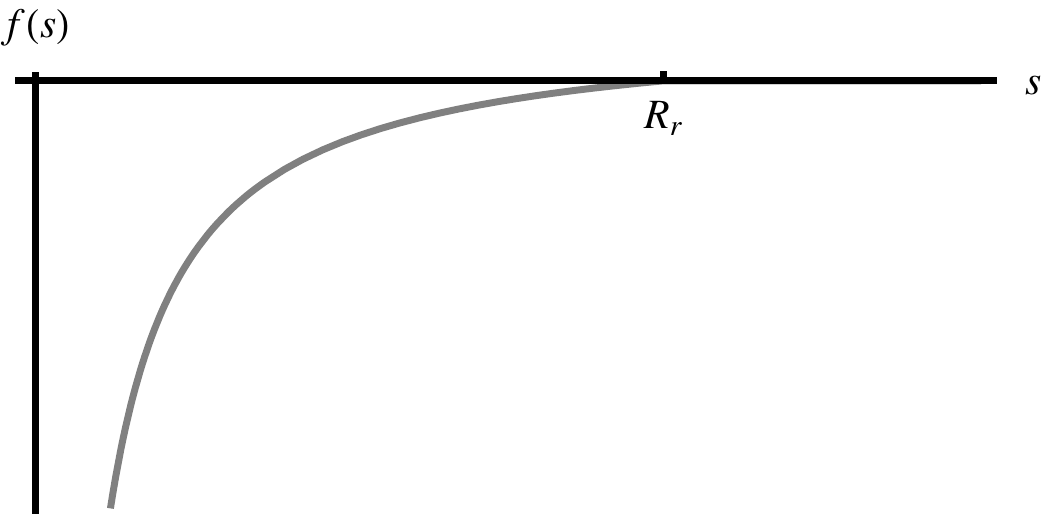} \\
\end{tabular}
\vspace{0 cm}
\caption{Graphical representation of typical examples of the function $f^{\alpha}_{\beta}$: attraction-repulsion (left) and repulsion only (right). Here, $R_r$ is the \textit{radius of repulsion}: an individual is repelled if its distance to someone else is smaller than $R_r$, while $R_a$ is the \textit{radius of attraction}: an individual is attracted to an `acquaintance' if their mutual distance is between $R_r$ and $R_a$. Note that $R_r$ might be chosen differently in either of the two choices. Typically $R_r$ will be larger in the repulsion case than in the attraction-repulsion case.}\label{Figure graphs fAR fR}
\end{figure}

\begin{figure}[h]
\centering
\vspace{0 cm}
\includegraphics[width=0.45\linewidth]{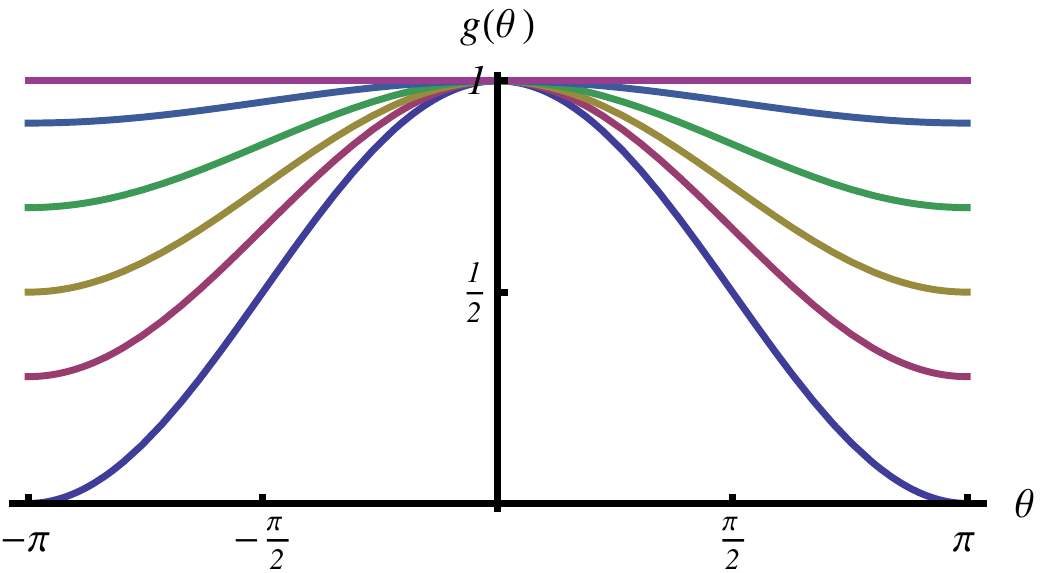}\\
\vspace{0 cm}
\caption{Graphical representation of the function $g$. We have $g(0)=1$ always, and $g(-\pi)=g(\pi)=\sigma$. Furthermore $g$ is symmetric around $\theta=0$ and increasing on $(-\pi,0)$ (thus decreasing on $(0,\pi)$). The plot has been made for $\sigma=0$, $0.3$, $0.5$, $0.7$, $0.9$, $1$ (bottom to top).}\label{Figure g}
\end{figure}

\begin{remark}
Other choices for $v_{\text{soc}}^{\alpha}$ are also possible. Based on the current velocities we can make an individual anticipate the distance he expects to be from another pedestrian in the (near) future. Although this is probably more realistic, it increases the complexity of the model dramatically. If these changes are made, $v^{\alpha}$ namely depends on the velocities $v^{\beta}$ ($\beta\in\{1,2,\ldots,\nu\}$); in particular on $v^{\alpha}$ itself. The definition of $v^{\alpha}$ becomes thus implicit, and is much harder to work with.
\end{remark}

\section{Derivation of an entropy inequality}\label{section entropy inequality cont-in-time}
In this section, we derive an entropy inequality concept in the spirit of (\ref{global entropy inequality}) in Section \ref{section mixture entropy}. Here, we also specify under which conditions this inequality holds. For simplicity and clarity, calculations are performed for a crowd without distinct subpopulations first; see Section \ref{section entropy inequality cont-in-time one population}. A multi-component crowd is considered afterwards in Section \ref{section entropy inequality cont-in-time multi-component}. The inspiration for the one-population case was provided by \cite{CarrilloMoll} (but we slightly deviate from their definitions).\\
\\
We work with a mass measure that is absolutely continuous w.r.t. the Lebesgue measure, that is: we do have a mass density.\\
\\
If one wants to deviate from the setting of \cite{CarrilloMoll}, an alternative approach is suggested in Appendix \ref{Appendix entropy ideal gas}. There, the entropy density follows, if the so-called \textit{free energy} $F$ is explicitly provided. The crucial decision is thus the choice of this free energy. In Appendix \ref{Appendix entropy ideal gas}, this is done for the example of the ideal gas. If one would find a good choice of $F$, the same arguments could be used for crowds.

\subsection{One population}\label{section entropy inequality cont-in-time one population}
We consider the presence of a single population in the domain $\Omega\subset \mathbb{R}^d$. The time-dependent density is $\rho:(0,T)\times\Omega\rightarrow \mathbb{R}^+$. We recall that it satisfies the balance of mass equation (\ref{mass balance local mixture}):
\begin{equation}\label{balance of mass recalled}
\dfrac{\partial\rho}{\partial t}+\nabla\cdot\bigl(\rho v\bigr)=0,\hspace{1 cm}\text{a.e. in }\Omega.
\end{equation}

\begin{assumption}\label{assumption gradient structure velocity}
We assume that the velocity field $v:(0,T)\times\Omega\rightarrow \mathbb{R}^d$ is of the form
\begin{equation}\label{velocity gradient structure}
v:= \nabla V + \nabla W \star \rho,
\end{equation}
where $V:\Omega\rightarrow \mathbb{R}$ is called \textit{confinement potential}, and $W:\mathbb{R}^d\rightarrow \mathbb{R}$ is called \textit{interaction potential}. We also assume that $W$ is symmetric, that is
\begin{equation*}
W(\xi)=W(-\xi),\hspace{1 cm}\text{for all }\xi\in\mathbb{R}^d.
\end{equation*}
\end{assumption}
Note that in \cite{CarrilloMoll} a minus-sign is added to the definition in (\ref{velocity gradient structure}).\\
The convolution $\nabla W \star \rho$ is defined as
\begin{equation}\label{def convolution with density}
(\nabla W \star \rho)(t,x):=\int_{\Omega}\nabla W(x-y)\rho(t,y)d\lambda^d(y),\hspace{1 cm}\text{for all }t\in(0,T),\text{ and }x\in\Omega,
\end{equation}
and thus this part is time-dependent via $\rho$. For brevity we often write $\rho(x)$ instead of $\rho(t,x)$ in the sequel. The time-dependency is to be understood implicitly. Analogously, if we write $(\nabla W \star \rho)(x)$, we actually mean $(\nabla W \star \rho)(t,x)$.

\begin{remark}
The component $\nabla V$ is the \textit{desired velocity} $v_{\text{des}}$ of Section \ref{section specification velocity}. Similarly, the component $\nabla W \star \rho$ is the \textit{social velocity} $v_{\text{soc}}$.
\end{remark}
We already proposed a velocity field of the following form (cf. Section \ref{section specification velocity}):
\begin{equation}\label{velocity decomposition}
v(t,x):= v_{\text{des}}(x)+\int_{\Omega} f(|y-x|)g(\theta_{xy})\dfrac{y-x}{|y-x|}\rho(t,y)d\lambda^d(y).
\end{equation}
We can take $\Omega$ as our domain of integration here, instead of $\Omega\setminus\{x\}$ (cf. (\ref{Def vsoc sum of integrals})). This is because $\{x\}$ is a $\lambda^d$-negligible set; the integral has the same value with or without the exclusion of this set from the domain. This holds for absolutely continuous measures, but not for any measure in general.\\
\\
The velocity as defined in (\ref{velocity decomposition}) fits to the structure of Assumption \ref{assumption gradient structure velocity}, if:
\begin{itemize}
  \item $v_{\text{des}}$ can be written as $\nabla V$ for some potential $V$. Note that in \cite{PiccoliTosin} a comparable situation is covered. The potential is found by solving the Laplace equation $\Delta V =0$ in the domain $\Omega$. For us this would mean that $v_{\text{des}}$ is divergence-free. However, the difference with our situation is that \cite{PiccoliTosin} normalizes and rescales $\nabla V$ afterwards.
  \item we disregard here the factor $g(\theta_{xy})$, that is, we set $g\equiv 1$ (or $\sigma=1$). This is necessary at this stage, because the inclusion of $g$ makes it impossible to find a \textit{symmetric} interaction potential $W$.\footnote{This is an inevitable, but disappointing choice. We thus recommend further research in this direction.}
  \item there is a function $F:\mathbb{R}_+\rightarrow\mathbb{R}$, such that $f=-F'$. Then:
      \begin{equation}
      f(|y-x|)\dfrac{y-x}{|y-x|}=-f(|x-y|)\dfrac{x-y}{|x-y|}=F'(|x-y|)\dfrac{x-y}{|x-y|}=\nabla F(|x-y|),
      \end{equation}
      where the last step of the calculations is due to the chain rule, and the fact that $\nabla|x-y|=\dfrac{x-y}{|x-y|}$. Now write $W(x-y):=F(|x-y|)$; note that this indeed implies that $W(\xi)=W(-\xi)$ for all $\xi\in\mathbb{R}^d$.
\end{itemize}
It is, for example, possible to have $v_{\text{des}}\equiv a \in \mathbb{R}^d$; that is, the desired velocity has constant magnitude and direction. To comply with the assumption that $v_{\text{des}}$ can be written as $\nabla V$, we have to take $V(x)=a\cdot x$.\\
\\
It follows from the aforementioned assumptions on $g$ and $f$ (and the subsequent calculations) that $v_{\text{soc}}=\nabla W\star \rho$. In Figure \ref{Figure graphs FAR FR} an impression is given of the functions $F$ that correspond (via the relation $f=-F'$) to those functions $f$ plotted in Figure \ref{Figure graphs FAR FR}. These functions are unique up to additional constants.
\begin{figure}[h]
\vspace{0 cm}
\begin{tabular}{rr}
\hspace{0.5 cm}\includegraphics[width=0.45\linewidth]{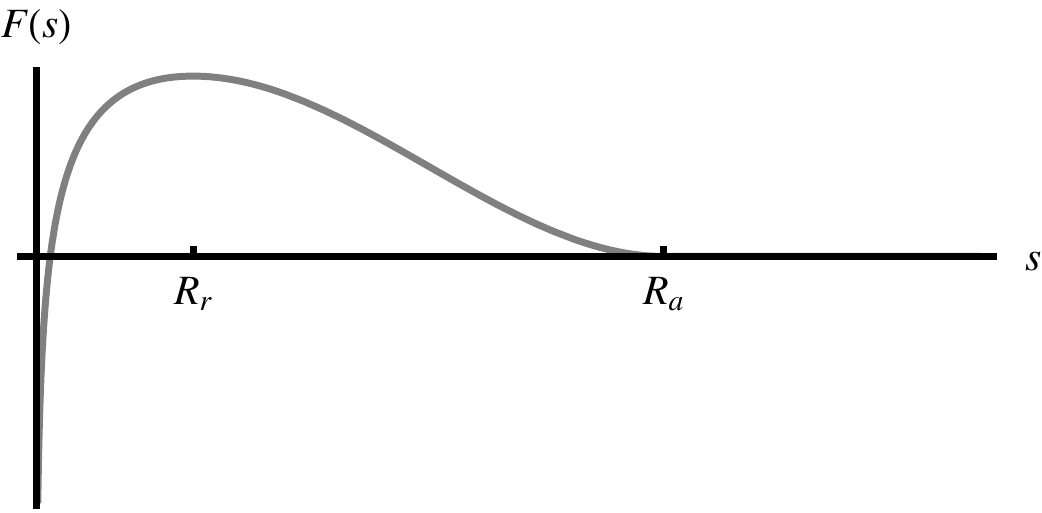}
&
\hspace{0 cm}\includegraphics[width=0.45\linewidth]{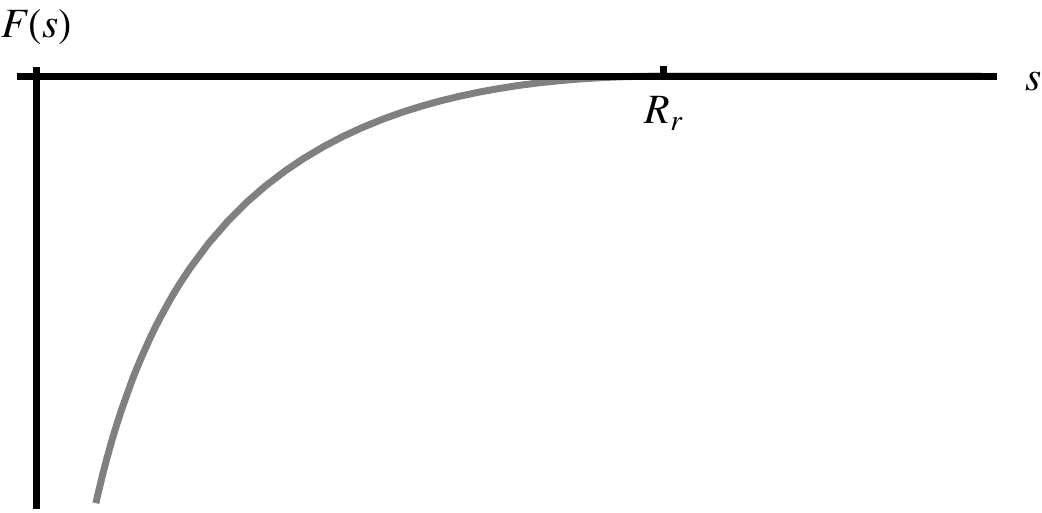} \\
\end{tabular}
\vspace{0 cm}
\caption{Graphical impression of the function $F$ corresponding, via the relation $f=-F'$, to the functions plotted in Figure \ref{Figure graphs fAR fR}: attraction-repulsion (left) and repulsion only (right).}\label{Figure graphs FAR FR}
\end{figure}

As in \cite{CarrilloMoll}, we define the entropy density $\eta$ as
\begin{equation}\label{entropy density single component}
\eta(t,x):= V(x)+\dfrac12 (W\star\rho)(t,x).
\end{equation}
The corresponding entropy of the system in $\Omega$ is then, according to (\ref{def S total entropy integral entr dens}), given by:
\begin{equation}\label{entropy integral single component}
S(t)=\int_{\Omega}\eta(t,x)\rho(t,x)d\lambda^d(x).
\end{equation}
The central question here is: Does the choice (\ref{entropy density single component}) -- (\ref{entropy integral single component}) satisfy the Clausius-Duhem Inequality? We answer this question in Theorem \ref{theorem entropy ineq one component}.

\begin{theorem}[Entropy inequality]\label{theorem entropy ineq one component}
Assume that $v$ satisfies Assumption \ref{assumption gradient structure velocity}, that the entropy $S$ is given by (\ref{entropy density single component}) -- (\ref{entropy integral single component}), and also that the system is isolated.\footnote{The statement that the system is isolated means that there is no flux through the boundary of $\Omega$, i.e. $\rho v\cdot n=0$ at $\partial \Omega$.}\\
Then the following inequality holds:
\begin{equation*}
\dfrac{dS}{dt}\geqslant 0.
\end{equation*}
\end{theorem}
\begin{proof}
We investigate directly the time-derivative of the entropy $S$, i.e. we have:
\begin{align}\label{dS/dt first step}
\notag \dfrac{dS}{dt}&= \int_{\Omega}V(x)\dfrac{\partial \rho}{\partial t}(x)d\lambda^d(x)+\dfrac12\int_{\Omega}\dfrac{\partial}{\partial t}(W\star\rho)(x)\rho(x)d\lambda^d(x)+ \dfrac12\int_{\Omega}(W\star\rho)(x)\dfrac{\partial \rho}{\partial t}(x)d\lambda^d(x)\\
&= \int_{\Omega}V(x)\dfrac{\partial \rho}{\partial t}(x)d\lambda^d(x)+ \dfrac12\int_{\Omega}(W\star\dfrac{\partial\rho}{\partial t})(x)\rho(x)d\lambda^d(x)+ \dfrac12\int_{\Omega}(W\star\rho)(x)\dfrac{\partial \rho}{\partial t}(x)d\lambda^d(x).
\end{align}
Note that
\begin{align}\label{W star d rho/dt}
\notag \int_{\Omega}(W\star\dfrac{\partial\rho}{\partial t})(x)\rho(x)d\lambda^d(x) &= \int_{\Omega}\int_{\Omega}W(x-y)\dfrac{\partial\rho}{\partial t}(y)d\lambda^d(y)\rho(x)d\lambda^d(x)\\
\notag &= \int_{\Omega}\int_{\Omega}W(x-y)\dfrac{\partial\rho}{\partial t}(y)\rho(x)d\lambda^d(y)d\lambda^d(x)\\
\notag &= \int_{\Omega}\int_{\Omega}W(y-x)\dfrac{\partial\rho}{\partial t}(x)\rho(y)d\lambda^d(x)d\lambda^d(y)\\
\notag &= \int_{\Omega}\int_{\Omega}W(x-y)\rho(y)d\lambda^d(y)\dfrac{\partial\rho}{\partial t}(x)d\lambda^d(x)\\
&= \int_{\Omega}(W\star\rho)(x)\dfrac{\partial\rho}{\partial t}(x)d\lambda^d(x).
\end{align}
We have replaced $x$ by $y$, and \textit{vice versa}, to obtain the third equality (this is solely a matter of notation). To obtain the fourth one, we used that for all $x,y\in\Omega$, $W(y-x)=W(x-y)$ and interchanged the order of integration.\\
\\
From (\ref{W star d rho/dt}), we conclude that (\ref{dS/dt first step}) can be written as
\begin{align}\label{dS/dt second step}
\notag \dfrac{dS}{dt}&= \int_{\Omega}V(x)\dfrac{\partial \rho}{\partial t}(x)d\lambda^d(x)+ \int_{\Omega}(W\star\rho)(x)\dfrac{\partial \rho}{\partial t}(x)d\lambda^d(x)\\
&= \int_{\Omega}\bigl(V(x)+(W\star\rho)(x)\bigr)\dfrac{\partial \rho}{\partial t}(x)d\lambda^d(x).
\end{align}
Substitution of the balance of mass (\ref{balance of mass recalled}) in (\ref{dS/dt second step}) yields
\begin{align}\label{dS/dt third step}
\notag \dfrac{dS}{dt}=& - \int_{\Omega}\bigl(V(x)+(W\star\rho)(x)\bigr)\nabla\cdot\bigl(\rho(x)v(x)\bigr)d\lambda^d(x)\\
\notag =& - \int_{\partial\Omega}\bigl(V(x)+(W\star\rho)(x)\bigr)\rho(x)v(x)\cdot nd\lambda^{d-1}(x)\\
& + \int_{\Omega}\nabla\bigl(V(x)+(W\star\rho)(x)\bigr)\cdot\rho(x)v(x)d\lambda^d(x).
\end{align}
By the hypothesis that $\rho v\cdot n=0$ at the boundary $\partial \Omega$, the boundary term in (\ref{dS/dt third step}) vanishes. We proceed as follows:
\begin{align}\label{dS/dt fourth step}
\notag \dfrac{dS}{dt}&= \int_{\Omega}\nabla\bigl(V(x)+(W\star\rho)(x)\bigr)\cdot\rho(x)v(x)d\lambda^d(x)\\
\notag &=  \int_{\Omega}\bigl(\nabla V(x)+\nabla(W\star\rho)(x)\bigr)\cdot\rho(x)v(x)d\lambda^d(x)\\
\notag &=  \int_{\Omega}\bigl(\nabla V(x)+(\nabla W\star\rho)(x)\bigr)\cdot\rho(x)v(x)d\lambda^d(x)\\
\notag &=  \int_{\Omega}v(x)\cdot\rho(x)v(x)d\lambda^d(x)\\
\notag &=  \int_{\Omega}|v(x)|^2\rho(x)d\lambda^d(x)\\
    &\geqslant 0.
\end{align}
Note that the relation
\begin{align}
\notag\nabla(W\star\rho)(x)=& \nabla_x\int_{\Omega}W(x-y)\rho(y)d\lambda^{d}(y)\\
\notag =& \int_{\Omega}\nabla_x W(x-y)\rho(y)d\lambda^{d}(y)\\
\notag =& \int_{\Omega}(\nabla W)(x-y)\rho(y)d\lambda^{d}(y)\\
\notag=&(\nabla W\star\rho)(x),
\end{align}
was used in the third step.\\
\\
In (\ref{dS/dt fourth step}), we recognize the entropy inequality we were looking for:
\begin{equation}\label{dS/dt positive}
\dfrac{dS}{dt}\geqslant 0.
\end{equation}
\end{proof}
\begin{remark}
In Theorem \ref{theorem entropy ineq one component} and in its proof we have implicitly assumed a sufficient amount of regularity of the boundary $\partial\Omega$. This is also important for the steps we are about to take.
\end{remark}
Let us comment on the situation in which $\rho v\cdot n=0$ does not necessarily hold at $\partial \Omega$. This means that we allow mass (and consequently also entropy) to escape from or enter the domain of our focus. Instead of (\ref{dS/dt fourth step}), this would yield
\begin{align}\label{dS/dt including boundary term}
\notag \dfrac{dS}{dt} &\geqslant  - \int_{\partial\Omega}\bigl(V(x)+(W\star\rho)(x)\bigr)\rho(x)v(x)\cdot nd\lambda^{d-1}(x)\\
&= - \int_{\partial\Omega}\eta(x)\rho(x)v(x)\cdot nd\lambda^{d-1}(x)- \int_{\partial\Omega}\dfrac12(W\star\rho)(x)\rho(x)v(x)\cdot nd\lambda^{d-1}(x).
\end{align}
After defining the \textit{entropy flux} $j_{\eta}$ as
\begin{equation*}
j_{\eta}:= \dfrac12(W\star\rho)\rho v,
\end{equation*}
we are thus in the setting of (\ref{global entropy inequality}) if there is no external volume supply of heat:
\begin{equation*}
\dfrac{dS}{dt} \geqslant  - \int_{\partial\Omega}\bigl(\eta\rho v+j_{\eta}\bigr)\cdot nd\lambda^{d-1}.
\end{equation*}

\subsection{Multi-component crowd}\label{section entropy inequality cont-in-time multi-component}
In this section, we extend the results of Section \ref{section entropy inequality cont-in-time one population} to the situation in which the crowd consists of a given number of subpopulations. Let the crowd in the domain $\Omega\subset \mathbb{R}^d$ consist of $\nu$ subpopulations, indexed by $\alpha\in\{1,2,\ldots,\nu\}$. The time-dependent density of component $\alpha$ is denoted by $\rho^{\alpha}:(0,T)\times\Omega\rightarrow \mathbb{R}^+$. Cf. the definition (\ref{rho}) of $\rho^{\alpha}$ as a Radon-Nikodym derivative in Section \ref{section description mixture}.

\begin{assumption}\label{assumption gradient structure velocity multi-comp}
\begin{enumerate}
  \item The velocity field of component $\alpha$, denoted by $v^{\alpha}:(0,T)\times\Omega\rightarrow \mathbb{R}^d$ is assumed to be of the form
\begin{equation}
v^{\alpha}:= \nabla V^{\alpha} + \sum_{\beta=1}^{\nu}\nabla W_{\beta}^{\alpha} \star \rho^{\beta},\hspace{1 cm}\text{for all }\alpha\in\{1,2,\ldots,\nu\},
\end{equation}
where $V^{\alpha}:\Omega\rightarrow \mathbb{R}$ is called the \textit{confinement potential} of component $\alpha$, and $W_{\beta}^{\alpha}:\mathbb{R}^d\rightarrow \mathbb{R}$ is called the \textit{interaction potential} of component $\beta$ (affecting $\alpha$).
  \item For each $\alpha, \beta\in\{1,2,\ldots,\nu\}$ we assume that $W_{\beta}^{\alpha}(\xi)=W_{\beta}^{\alpha}(-\xi)$ holds for all $\xi\in\mathbb{R}^d$.
  \item We assume \textit{symmetric interactions}, that is
\begin{equation}
W_{\beta}^{\alpha}\equiv W_{\alpha}^{\beta}, \hspace{1 cm}\text{for all }\alpha,\beta\in\{1,2,\ldots,\nu\}.
\end{equation}\label{part symm interactions assumption grad stract mult-comp}
\end{enumerate}
\end{assumption}

\begin{remark}\label{remark predator-prey}
Part \ref{part symm interactions assumption grad stract mult-comp} of Assumption \ref{assumption gradient structure velocity multi-comp} implies for example that we do not allow a "predator-prey relation" between two subpopulations. In such relation a predator should be attracted to the prey-population, but a prey should be repelled from the predators. These interactions are \textit{asymmetric}.
\end{remark}

The following balance of mass equation is satisfied for each $\alpha\in\{1,2,\ldots,\nu\}$:
\begin{equation}\label{balance of mass constituent}
\dfrac{\partial\rho^{\alpha}}{\partial t}+\nabla\cdot\bigl(\rho^{\alpha} v^{\alpha}\bigr)=0,\hspace{1 cm}\text{a.e. in }\Omega.
\end{equation}
We recall from Sections \ref{section description mixture}--\ref{section kinematics mixture} the definitions of the total density $\rho$ and barycentric velocity $v$:
\begin{align}
\notag \rho:=& \sum_{\alpha=1}^{\nu}\rho^{\alpha},\\
\notag v:=& \sum_{\alpha=1}^{\nu}\dfrac{\rho^{\alpha}}{\rho}v^{\alpha}.
\end{align}
On the macroscopic scale, the following balance of mass equation is satisfied:
\begin{equation}\label{balance of mass total}
\dfrac{\partial\rho}{\partial t}+\nabla\cdot\bigl(\rho v\bigr)=0,\hspace{1 cm}\text{a.e. in }\Omega.
\end{equation}
For all $t\in(0,T)$ and $x\in\Omega$ we define the partial entropy density of component $\alpha$ as
\begin{equation}\label{partial entropy density}
\eta^{\alpha}(t,x):= V^{\alpha}(x)+\dfrac12\sum_{\beta=1}^{\nu}(W_{\beta}^{\alpha}\star\rho^{\beta})(t,x).
\end{equation}
The entropy density of the whole crowd is defined as
\begin{equation*}
\eta:=\sum_{\alpha=1}^{\nu}\dfrac{\rho^{\alpha}}{\rho}\eta^{\alpha}.
\end{equation*}
Consequently, the entropy of the system at time $t$ is given by
\begin{align}
\notag S(t)=& \int_{\Omega}\eta(t,x)\rho(t,x)d\lambda^d(x)\\
=& \sum_{\alpha=1}^{\nu}\int_{\Omega}\Bigl(V^{\alpha}(x)+\dfrac12\sum_{\beta=1}^{\nu}(W_{\beta}^{\alpha}\star\rho^{\beta})(t,x) \Bigr)\rho^{\alpha}(t,x)d\lambda^d(x).\label{entropy integral multi-component}
\end{align}
In the spirit of Theorem \ref{theorem entropy ineq one component} we can formulate the following theorem:
\begin{theorem}[Entropy inequality]\label{theorem entropy ineq multi-component}
Assume that for each $\alpha$ we have that $v^{\alpha}$ satisfies Assumption \ref{assumption gradient structure velocity multi-comp}. Moreover, assume that the entropy is given by (\ref{entropy integral multi-component}) and that the system is isolated.\\
Then the following inequality holds:
\begin{equation*}
\dfrac{dS}{dt}\geqslant 0.
\end{equation*}
\end{theorem}

\begin{proof}
The approach here is very much in the spirit of the proof of Theorem \ref{theorem entropy ineq one component}. We thus consider the time-derivative of $S$:
\begin{align}\label{dS/dt first step total}
\notag \dfrac{dS}{dt}=& \sum_{\alpha=1}^{\nu}\int_{\Omega}V^{\alpha}(x)\dfrac{\partial \rho^{\alpha}}{\partial t}(x)d\lambda^d(x)\\
\notag &+\dfrac12\sum_{\alpha=1}^{\nu}\int_{\Omega}\sum_{\beta=1}^{\nu}\dfrac{\partial}{\partial t}(W_{\beta}^{\alpha}\star\rho^{\beta})(x)\rho^{\alpha}(x)d\lambda^d(x)\\
\notag &+ \dfrac12\sum_{\alpha=1}^{\nu}\int_{\Omega}\sum_{\beta=1}^{\nu}(W_{\beta}^{\alpha}\star\rho^{\beta})(x)\dfrac{\partial \rho^{\alpha}}{\partial t}(x)d\lambda^d(x)\\
\notag =& \sum_{\alpha=1}^{\nu}\int_{\Omega}V^{\alpha}(x)\dfrac{\partial \rho^{\alpha}}{\partial t}(x)d\lambda^d(x)\\
\notag &+\dfrac12\sum_{\alpha=1}^{\nu}\sum_{\beta=1}^{\nu}\int_{\Omega}(W_{\beta}^{\alpha}\star\dfrac{\partial\rho^{\beta}}{\partial t})(x)\rho^{\alpha}(x)d\lambda^d(x)\\
&+ \dfrac12\sum_{\alpha=1}^{\nu}\sum_{\beta=1}^{\nu}\int_{\Omega}(W_{\beta}^{\alpha}\star\rho^{\beta})(x)\dfrac{\partial \rho^{\alpha}}{\partial t}(x)d\lambda^d(x)
\end{align}
Following the idea of (\ref{W star d rho/dt}), we derive
\begin{align}\label{W star d rho/dt total}
\notag \sum_{\alpha=1}^{\nu}\sum_{\beta=1}^{\nu}\int_{\Omega}(W_{\beta}^{\alpha}\star\dfrac{\partial\rho^{\beta}}{\partial t})(x)\rho^{\alpha}(x)d\lambda^d(x) &= \sum_{\alpha=1}^{\nu}\sum_{\beta=1}^{\nu}\int_{\Omega}\int_{\Omega}W_{\beta}^{\alpha}(x-y)\dfrac{\partial\rho^{\beta}}{\partial t}(y)d\lambda^d(y)\rho^{\alpha}(x)d\lambda^d(x)\\
\notag &= \sum_{\alpha=1}^{\nu}\sum_{\beta=1}^{\nu}\int_{\Omega}\int_{\Omega}W_{\beta}^{\alpha}(x-y)\dfrac{\partial\rho^{\beta}}{\partial t}(y)\rho^{\alpha}(x)d\lambda^d(y)d\lambda^d(x)\\
\notag &= \sum_{\beta=1}^{\nu}\sum_{\alpha=1}^{\nu}\int_{\Omega}\int_{\Omega}W_{\alpha}^{\beta}(y-x)\dfrac{\partial\rho^{\alpha}}{\partial t}(x)\rho^{\beta}(y)d\lambda^d(x)d\lambda^d(y)\\
\notag &= \sum_{\alpha=1}^{\nu}\sum_{\beta=1}^{\nu}\int_{\Omega}\int_{\Omega}W_{\alpha}^{\beta}(x-y)\rho^{\beta}(y)d\lambda^d(y)\dfrac{\partial\rho^{\alpha}}{\partial t}(x)d\lambda^d(x)\\
&= \sum_{\alpha=1}^{\nu}\sum_{\beta=1}^{\nu}\int_{\Omega}(W_{\alpha}^{\beta}\star\rho^{\beta})(x)\dfrac{\partial\rho^{\alpha}}{\partial t}(x)d\lambda^d(x).
\end{align}
We have replaced $x$ by $y$, $\alpha$ by $\beta$, and \textit{vice versa}, to obtain the third equality (this is again solely a matter of notation). To obtain the fourth one, we used the fact that $W_{\alpha}^{\beta}(y-x)=W_{\alpha}^{\beta}(x-y)$ for all $\alpha,\beta\in\{1,2,\ldots\}$. Moreover, we interchanged the order of summation and integration.\\
\\
We now combine (\ref{dS/dt first step total}) and (\ref{W star d rho/dt total}), and conclude that
\begin{align}\label{dS/dt second step total}
\dfrac{dS}{dt}=& \sum_{\alpha=1}^{\nu}\int_{\Omega}\Bigl[V^{\alpha}(x)+\sum_{\beta=1}^{\nu}\bigl(\dfrac12(W_{\beta}^{\alpha}+W_{\alpha}^{\beta})\star \rho^{\beta}\bigr)(x) \Bigr]\dfrac{\partial \rho^{\alpha}}{\partial t}(x)d\lambda^d(x).
\end{align}
We substitute (\ref{balance of mass constituent}), the balance of mass per component, in (\ref{dS/dt second step total}), by which we obtain
\begin{align}\label{dS/dt third step total}
\notag \dfrac{dS}{dt}=& - \sum_{\alpha=1}^{\nu} \int_{\Omega}\Bigl[V^{\alpha}(x)+\sum_{\beta=1}^{\nu}\bigl(\dfrac12(W_{\beta}^{\alpha}+W_{\alpha}^{\beta})\star \rho^{\beta}\bigr)(x) \Bigr]\nabla\cdot\bigl(\rho^{\alpha}(x)v^{\alpha}(x)\bigr)d\lambda^d(x)\\
\notag =& - \sum_{\alpha=1}^{\nu} \int_{\partial\Omega}\Bigl[V^{\alpha}(x)+\sum_{\beta=1}^{\nu}\bigl(\dfrac12(W_{\beta}^{\alpha}+W_{\alpha}^{\beta})\star \rho^{\beta}\bigr)(x) \Bigr]\rho^{\alpha}(x)v^{\alpha}(x)\cdot nd\lambda^{d-1}(x)\\
& + \sum_{\alpha=1}^{\nu}\int_{\Omega}\nabla\Bigl[V^{\alpha}(x)+\sum_{\beta=1}^{\nu}\bigl(\dfrac12(W_{\beta}^{\alpha}+W_{\alpha}^{\beta})\star \rho^{\beta}\bigr)(x) \Bigr]\cdot\rho^{\alpha}(x)v^{\alpha}(x)d\lambda^d(x).
\end{align}
Under the assumption that there is no flux of mass through the boundary of $\Omega$ for any of the constituents $\alpha$, we have $\rho^{\alpha}v^{\alpha}\cdot n=0$ at $\partial \Omega$ for all $\alpha\in\{1,2,\ldots,\nu\}$. This means that the boundary term in (\ref{dS/dt third step total}) vanishes. Taking also Part \ref{part symm interactions assumption grad stract mult-comp} of Assumption \ref{assumption gradient structure velocity multi-comp} (i.e. symmetric interactions) into consideration, (\ref{dS/dt third step total}) reads
\begin{align}\label{dS/dt fourth step total}
\notag \dfrac{dS}{dt}=& \sum_{\alpha=1}^{\nu}\int_{\Omega}\nabla\Bigl[V^{\alpha}(x)+\sum_{\beta=1}^{\nu}(W_{\beta}^{\alpha}\star \rho^{\beta})(x) \Bigr]\cdot\rho^{\alpha}(x)v^{\alpha}(x)d\lambda^d(x)\\
\notag =& \sum_{\alpha=1}^{\nu}\int_{\Omega}\Bigl[\nabla V^{\alpha}(x)+\sum_{\beta=1}^{\nu}(\nabla W_{\beta}^{\alpha}\star \rho^{\beta})(x) \Bigr]\cdot\rho^{\alpha}(x)v^{\alpha}(x)d\lambda^d(x)\\
\notag =& \sum_{\alpha=1}^{\nu}\int_{\Omega}v^{\alpha}(x)\cdot\rho^{\alpha}(x)v^{\alpha}(x)d\lambda^d(x)\\
\notag =& \sum_{\alpha=1}^{\nu}\int_{\Omega}|v^{\alpha}(x)|^2\rho^{\alpha}(x)d\lambda^d(x)\\
\geqslant& 0.
\end{align}
By (\ref{dS/dt fourth step total}) we thus have the following entropy inequality:
\begin{equation}\label{dS/dt positive total}
\dfrac{dS}{dt}\geqslant 0.
\end{equation}
\end{proof}

If we allow mass (and consequently also entropy) to escape or enter through the boundary of $\Omega$, we should consider the boundary term in (\ref{dS/dt third step total}). Let us define the entropy flux $j_{\eta}$ as
\begin{equation*}
j_{\eta}:= \dfrac12\sum_{\alpha=1}^{\nu}\sum_{\beta=1}^{\nu}\bigl(W_{\alpha}^{\beta}\star \rho^{\beta}\bigr)\rho^{\alpha}v^{\alpha},
\end{equation*}
Still assuming symmetric interactions, and using (\ref{partial entropy density}), (\ref{dS/dt third step total}) would now read
\begin{align}\label{dS/dt including boundary term}
\notag \dfrac{dS}{dt} &\geqslant  - \sum_{\alpha=1}^{\nu} \int_{\partial\Omega}\Bigl[V^{\alpha}(x)+\sum_{\beta=1}^{\nu}\bigl(\dfrac12(W_{\beta}^{\alpha}+W_{\alpha}^{\beta})\star \rho^{\beta}\bigr)(x) \Bigr]\rho^{\alpha}(x)v^{\alpha}(x)\cdot nd\lambda^{d-1}(x)\\
&= -\int_{\partial\Omega}\sum_{\alpha=1}^{\nu}\eta^{\alpha}(x)\rho^{\alpha}(x)v^{\alpha}(x)\cdot nd\lambda^{d-1}(x)- \int_{\partial\Omega}j_{\eta}(x)\cdot nd\lambda^{d-1}(x).
\end{align}
\begin{remark}
Note that this does \textbf{not} bring us to the setting of (\ref{global entropy inequality}) if there is no external volume supply of heat. This is because, in general
\begin{equation*}
\sum_{\alpha=1}^{\nu}\eta^{\alpha}\rho^{\alpha}v^{\alpha}\neq \eta\rho v.
\end{equation*}
Compare this to Remark \ref{remark how to derive entropy ineq whole mixture}. Using $\sum_{\alpha=1}^{\nu}\eta^{\alpha}\rho^{\alpha}v^{\alpha}$ in the entropy inequality fits the approach of \cite{GreenNaghdi}. This is however rejected by \cite{BedfordDrumheller, Holmes}, who go for using $\eta\rho v$, as it is mentioned in (\ref{global entropy inequality}).
\end{remark}

\begin{remark}
The symmetry of the interactions, as imposed by Part \ref{part symm interactions assumption grad stract mult-comp} of Assumption \ref{assumption gradient structure velocity multi-comp}, is crucial in the proof of Theorem \ref{theorem entropy ineq multi-component}. However, it is a quite restrictive assumption, that disallows many interesting settings and thus deserves further research. At a later stage we hope to include a drift in the interactions, which is such that we can still formulate an entropy inequality.
\end{remark}

\subsection{Generalization to discrete measures}\label{section generalization entropy inequality to discrete measures}
In this section, we make feasible that it is possible to derive (at least from a mathematical point of view) an analogon for discrete measures of the entropy inequalities mentioned in Sections \ref{section entropy inequality cont-in-time one population} and \ref{section entropy inequality cont-in-time multi-component}. We already announced that we would do so in Remark \ref{remark physical relevance generalization entropy inequality}.\\
\\
We consider a single population with corresponding discrete mass measure of the form
\begin{equation*}
\mu := \sum_{i\in\mathcal{J}} \delta_{x_i(t)},
\end{equation*}
where $\mathcal{J}\subset\mathbb{N}$ is some index set. The evolution in time of the centres $x_i$ is governed by the velocity field $v$ via
\begin{equation*}
\dfrac{d}{dt}x_i(t)=v(t,x_i(t)).
\end{equation*}
This velocity field is assumed to be of the form (cf. (\ref{velocity gradient structure}))
\begin{equation}\label{velocity gradient structure convolution with measure}
v:= \nabla V + \nabla W \star \mu,
\end{equation}
The convolution $\nabla W \star \mu$ is a generalized form of (\ref{def convolution with density}), defined as
\begin{equation}\label{def convolution measure}
(\nabla W \star \mu)(t,x):=\int_{\Omega\setminus\{x\}}\nabla W(x-y)d\mu(t,y),\hspace{1 cm}\text{for all }t\in(0,T),\text{ and }x\in\Omega.
\end{equation}
The point $x$ itself has been excluded from the domain of integration to avoid interaction of a point mass/pedestrian with itself. We assume again that $W$ is symmetric:
\begin{equation*}
W(\xi)=W(-\xi),\hspace{1 cm}\text{for all }\xi\in\mathbb{R}^d.
\end{equation*}
In the spirit of (\ref{entropy density single component}) we define the entropy density $\eta$ as
\begin{equation*}
\eta(t,x):= V(x)+\dfrac12 (W\star\mu)(t,x).
\end{equation*}
The corresponding entropy of the system in $\Omega$ is
\begin{equation*}
S(t)=\int_{\Omega}\eta(t,x)d\mu(t,x).
\end{equation*}
We explicitly restrict ourselves to the situation that all point masses remain in $\Omega$.\footnote{Allowing point masses to leave (or enter) the domain, means that a boundary measure, or a trace of the mass measure on the boundary, needs to be defined properly. It is all but trivial to do so for general measures.} This restriction implies that an integral with respect to the measure $\mu$ can be represented by a sum, in which all centres $x_i$ contribute. We have that
\begin{equation*}
S(t)=\sum_{i\in\mathcal{J}}\Big\{V(x_i)+\dfrac12 \sum_{\substack{j\in\mathcal{J}\\x_j\neq x_i}}W(x_i-x_j)\mathbf{1}_{x_j\in\Omega}  \Big\}\mathbf{1}_{x_i\in\Omega},
\end{equation*}
where all positions $x_i$ are time-dependent. The indicator function $\mathbf{1}_{x_i\in\Omega}$ is $1$ if $x_i\in\Omega$ is true, and $0$ otherwise.\\
The requirement that all centres $x_i$ remain in $\Omega$, makes that $\mathbf{1}_{x_i\in\Omega}$ never becomes $0$. For all $i\in\mathcal{J}$ thus $\mathbf{1}_{x_i\in\Omega}\equiv 1$ holds, for all $t\in[0,T]$. As a result, we lose time-dependence in these indicator functions, and we can just write
\begin{equation*}
S(t)=\sum_{i\in\mathcal{J}}\Big\{V(x_i)+\dfrac12 \sum_{\substack{j\in\mathcal{J}\\x_j\neq x_i}}W(x_i-x_j)\Big\}.
\end{equation*}
\begin{remark}
To derive an entropy inequality we examine the time derivative of $S$. Consider the `forbidden' situation that $x_i(t)$ leaves (or enters) the domain, say at time $t=t^*$. Then $\mathbf{1}_{x_i(t)\in\Omega}$ is discontinuous in $t=t^*$. We disallow this situation, because the time derivative of $S$ does not exist in such point (not even in a weak sense).
\end{remark}
We take the time derivative of $S(t)$:
\begin{eqnarray}
\nonumber \dfrac{dS}{dt} &=& \sum_{i\in\mathcal{J}}\Big\{\nabla V(x_i)\cdot\dfrac{dx_i}{dt}+\dfrac12 \sum_{\substack{j\in\mathcal{J}\\x_j\neq x_i}}\nabla_{x_i}W(x_i-x_j)\cdot\dfrac{dx_i}{dt}+\dfrac12 \sum_{\substack{j\in\mathcal{J}\\x_j\neq x_i}}\nabla_{x_j}W(x_i-x_j)\cdot\dfrac{dx_j}{dt}\Big\}\\
\nonumber &=& \sum_{i\in\mathcal{J}}\Big\{\nabla V(x_i)\cdot\dfrac{dx_i}{dt}+\dfrac12 \sum_{\substack{j\in\mathcal{J}\\x_j\neq x_i}}\nabla W(x_i-x_j)\cdot\dfrac{dx_i}{dt}-\dfrac12 \sum_{\substack{j\in\mathcal{J}\\x_j\neq x_i}}\nabla W(x_i-x_j)\cdot\dfrac{dx_j}{dt}\Big\}\\
\nonumber &=& \sum_{i\in\mathcal{J}}\nabla V(x_i)\cdot v(t,x_i)+\dfrac12 \sum_{\substack{i,j\in\mathcal{J}\\x_j\neq x_i}}\nabla W(x_i-x_j)\cdot v(t,x_i)-\dfrac12 \sum_{\substack{i,j\in\mathcal{J}\\x_j\neq x_i}}\nabla W(x_i-x_j)\cdot v(t,x_j)\\
\nonumber &=& \sum_{i\in\mathcal{J}}\nabla V(x_i)\cdot v(t,x_i)+\dfrac12 \sum_{\substack{i,j\in\mathcal{J}\\x_j\neq x_i}}\nabla W(x_i-x_j)\cdot v(t,x_i)+\dfrac12 \sum_{\substack{i,j\in\mathcal{J}\\x_j\neq x_i}}\nabla W(x_j-x_i)\cdot v(t,x_j)\\
\nonumber &=& \sum_{i\in\mathcal{J}}\nabla V(x_i)\cdot v(t,x_i)+ \sum_{\substack{i,j\in\mathcal{J}\\x_j\neq x_i}}\nabla W(x_i-x_j)\cdot v(t,x_i)\\
\nonumber &=& \sum_{i\in\mathcal{J}}v(t,x_i)\cdot \Big \{ \nabla V(x_i) + \sum_{\substack{j\in\mathcal{J}\\x_j\neq x_i}}\nabla W(x_i-x_j) \Big\}\\
\nonumber &=& \sum_{i\in\mathcal{J}}v(t,x_i)\cdot \Big \{ \nabla V(x_i) + (\nabla W\star\mu)(x_i) \Big\}\\
\nonumber &=& \int_{\Omega}|v(t,x)|^2 d\mu(t,x)\\
&\geqslant& 0.\label{entropy ineq discrete measure}
\end{eqnarray}
The seventh equality is due to (\ref{def convolution measure}). By (\ref{entropy ineq discrete measure}) we have derived an equivalent statement as in Theorem \ref{theorem entropy ineq one component}.\\
\\
Following similar lines of argument, an entropy inequality as in Theorem \ref{theorem entropy ineq multi-component} can be derived also if we allow subpopulations with both discrete and absolutely continuous mass measures. We will, however, not go into further details in this direction.

\section{Derivation of the time-discrete model}\label{section derivation time-discrete}
Inspired by \cite{Piccoli2010, PiccoliTosinMeasTh}, we derive in this section a discrete-in-time counterpart of the weak formulation presented in (\ref{Weak Form}) and Definition \ref{def weak solution}.\\
For this aim, we introduce a strictly increasing sequence of discrete points in time $\{t_n\}_{n\in\mathcal{N}}\subset[0,T]$. Here, $\mathcal{N}:=\{0,1,\ldots,N_T\}$ is the index set, with $N_T\in\mathbb{N}$. The set of points in time is chosen such that $t_0=0$ and $t_{N_T}=T$.\\
We define $\Delta t_n:=t_{n+1}-t_n$. Since $\{t_n\}_{n\in\mathcal{N}}$ is a strictly increasing sequence, $\Delta t_n>0$ holds for all $n$.\\
To emphasize that we are working in a time-discrete setting, we write $\mu_n^{\alpha}(\cdot)$ as the time-discrete equivalent of $\mu^{\alpha}(t_n,\cdot):\mathcal{B}(\Omega)\rightarrow\mathbb{R}^+$ for all $n\in\mathcal{N}$ and all $\alpha\in\{1,2,\ldots,\nu\}$. Analogously, we write $v_n^{\alpha}(\cdot)$ for $v^{\alpha}(t_n,\cdot):\Omega\rightarrow\mathbb{R}^d$.\\
\\
For arbitrary $n\in\mathcal{N}$, integration in time of (\ref{Weak Form}) over the interval $(t_n,t_{n+1})$ yields for all $\alpha\in\{1,2,\ldots,\nu\}$ and all $\psi^{\alpha}\in C^1_0(\bar{\Omega})$:
\begin{equation}\label{integration weak form over time interval discretization}
\int_{\Omega}\psi^{\alpha}(x)d\mu^{\alpha}(t_{n+1},x)-\int_{\Omega}\psi^{\alpha}(x)d\mu^{\alpha}(t_n,x)= \int_{t_n}^{t_{n+1}}\int_{\Omega}v^{\alpha}(t,x)\cdot\nabla\psi^{\alpha}(x)d\mu^{\alpha}(t,x)dt.
\end{equation}
Assuming all necessary regularity, we can expand the right-hand term in a Taylor series around $t_n$. For the sake of brevity we define $A(t):=\int_{t_n}^{t}\int_{\Omega}v^{\alpha}(\tilde{t},x)\cdot\nabla\psi^{\alpha}(x)d\mu^{\alpha}(\tilde{t},x)d\tilde{t}$. Note that $A(t_{n+1})$ is the right-hand side of (\ref{integration weak form over time interval discretization}). Now
\begin{equation*}
A(t_{n+1})= A(t_n)+\Delta t_n \dfrac{dA}{dt}\Big|_{t=t_n} + \mathcal{O}\bigl(\Delta t_n^2\bigr).
\end{equation*}
Note that $A(t_n)=0$ and $\dfrac{d}{dt}A(t)=\int_{\Omega}v^{\alpha}(t,x)\cdot\nabla\psi^{\alpha}(x)d\mu^{\alpha}(t,x)$. By writing $\mu_n^{\alpha}(\cdot)$ instead of $\mu^{\alpha}(t_n,\cdot)$, and $v_n^{\alpha}(\cdot)$ instead of $v^{\alpha}(t_n,\cdot)$ (above we already announced to do so), (\ref{integration weak form over time interval discretization}) transforms into
\begin{equation*}
\int_{\Omega}\psi^{\alpha}(x)d\mu_{n+1}^{\alpha}(x)-\int_{\Omega}\psi^{\alpha}(x)d\mu_n^{\alpha}(x)= \Delta t_n\int_{\Omega}v_n^{\alpha}(x)\cdot\nabla\psi^{\alpha}(x)d\mu_n^{\alpha}(x)+\mathcal{O}\bigl(\Delta t_n^2\bigr).
\end{equation*}
This can also be written as
\begin{equation}\label{Taylor approx evolution time discrete}
\int_{\Omega}\psi^{\alpha}(x)d\mu_{n+1}^{\alpha}(x)= \int_{\Omega}\Bigl(\psi^{\alpha}(x)+\Delta t_n v_n^{\alpha}(x)\cdot\nabla\psi^{\alpha}(x)\Bigr)d\mu_n^{\alpha}(x)+\mathcal{O}\bigl(\Delta t_n^2\bigr).
\end{equation}
If $v_n^{\alpha}$ is `well-behaved' we expand
\begin{equation}\label{Taylor expansion psi, v well-behaved}
\psi^{\alpha}\bigl(x+\Delta t_n v_n^{\alpha}(x)\bigr)=\psi^{\alpha}(x)+\Delta t_n v_n^{\alpha}(x)\cdot\nabla\psi^{\alpha}(x)+\mathcal{O}\bigl(\Delta t_n^2\bigr),
\end{equation}
or
\begin{equation}\label{Taylor expansion psi, v well-behaved, reordered}
\psi^{\alpha}(x)+\Delta t_n v_n^{\alpha}(x)\cdot\nabla\psi^{\alpha}(x)=\psi^{\alpha}\bigl(x+\Delta t_n v_n^{\alpha}(x)\bigr)+\mathcal{O}\bigr(\Delta t_n^2\bigr).
\end{equation}
By `well-behaved' we mean that $v_n^{\alpha}$ is (at least) $\mu_n^{\alpha}$-uniformly bounded. That is, for fixed $n\in\mathcal{N}$ there exists a non-negative constant $M_n$ such that $|v_n^{\alpha}(x)|<M_n$ for $\mu^{\alpha}_n$-almost every $x$, by which $\mu^{\alpha}_n$-almost everywhere: $\Delta t_n v_n^{\alpha}(x)=\mathcal{O}(\Delta t_n)$.\\
We substitute this in (\ref{Taylor approx evolution time discrete}) to obtain
\begin{eqnarray}
\nonumber \int_{\Omega}\psi^{\alpha}(x)d\mu_{n+1}^{\alpha}(x) &=& \int_{\Omega}\Bigl(\psi^{\alpha}\bigl(x+\Delta t_n v_n^{\alpha}(x)\bigr)+\mathcal{O}\bigr(\Delta t_n^2\bigr)\Bigr)d\mu_n^{\alpha}(x)+\mathcal{O}\bigl(\Delta t_n^2\bigr)\\
&=& \int_{\Omega}\psi^{\alpha}\bigl(x+\Delta t_n v_n^{\alpha}(x)\bigr)d\mu_n^{\alpha}(x)+\mathcal{O}\bigl(\Delta t_n^2\bigr),
\end{eqnarray}
where we have taken the $\mathcal{O}\bigr(\Delta t_n^2\bigr)$-terms outside the integral, and used that $\mu^{\alpha}_n(\Omega)$ is finite, since $\mu^{\alpha}(t_n,\cdot)$ is a finite measure for all choices of $n\in\mathcal{N}$.\\
\\
We neglect the $\mathcal{O}\bigr(\Delta t_n^2\bigr)$-part and obtain
\begin{equation}\label{integral evolution time discrete}
\int_{\Omega}\psi^{\alpha}(x)d\mu_{n+1}^{\alpha}(x) \approx  \int_{\Omega}\psi^{\alpha}\bigl(\chi^{\alpha}_n(x)\bigr)d\mu_n^{\alpha}(x).
\end{equation}
Although (\ref{integral evolution time discrete}) is an approximation, we treat it as an equality from now on. Whenever we refer to (\ref{integral evolution time discrete}), we thus mean the equality rather than the approximation. In (\ref{integral evolution time discrete}) we have moreover used the definition:
\begin{definition}[One-step motion mapping]\label{def one-step motion mapping}
The \textit{one-step motion mapping} $\chi_n^{\alpha}$ is defined by
\begin{equation}\label{one-step motion mapping}
\chi^{\alpha}_n(x):= x + \Delta t_n v_n^{\alpha}(x),
\end{equation}
for all $n\in\{0,1,\ldots,N_T-1\}$. For simplicity, it will henceforth just be called \textit{motion mapping}. It provides the position at time step $n+1$ of the point located in $x$ at time step $n$.
\end{definition}

\begin{assumption}[Properties of the motion mappings]\label{assumption time-discrete motion mapping invertible and Borel}
For all $n\in\{0,1,\ldots,N_T-1\}$ and each $\alpha\in\{1,2,\ldots,\nu\}$ we assume that $\chi^{\alpha}_n:\Omega\rightarrow\Omega$ is a \textit{homeomorphism}. This means that:
\begin{enumerate}[(i)]
  \item $\chi^{\alpha}_n$ is invertible,
  \item $\chi^{\alpha}_n$ is continuous,
  \item $(\chi^{\alpha}_n)^{-1}$ is continuous.
\end{enumerate}
\end{assumption}

\begin{remark}
From the proof of Lemma \ref{Diffeomorphism maps Borel to Borel} we extract the statement that a continuous mapping is Borel. Regarding Assumption \ref{assumption time-discrete motion mapping invertible and Borel}, this implies that $\chi^{\alpha}_n$ and $(\chi^{\alpha}_n)^{-1}$ are Borel. Respectively, that is:
\begin{enumerate}[(i)]
  \item for all $\Omega'\in\mathcal{B}(\Omega)$ we have that $(\chi^{\alpha}_n)^{-1}(\Omega')\in\mathcal{B}(\Omega)$;
  \item for all $\Omega'\in\mathcal{B}(\Omega)$ we have that $\chi^{\alpha}_n(\Omega')\in\mathcal{B}(\Omega)$.
\end{enumerate}
Note that Assumption \ref{assumption time-discrete motion mapping invertible and Borel} is the somewhat weaker counterpart of the assumptions on the motion mappings in Section \ref{section kinematics mixture}.
\end{remark}

\begin{remark}\label{remark relax psi in C^1_0}
The expression in (\ref{integral evolution time discrete}) makes sense even if we do not restrict ourselves to taking only $\psi^{\alpha}\in C^1_0(\bar{\Omega})$. In the sequel we allow $\psi^{\alpha}$ to be any function that is integrable on $\Omega$ with respect to the measure $\mu^{\alpha}_{n+1}$, that is: $\psi^{\alpha}\in L^1_{\mu^{\alpha}_{n+1}}(\Omega)$.
\end{remark}
Following Remark \ref{remark relax psi in C^1_0}, we can take $\psi^{\alpha}=\mathbf{1}_{\Omega'}$ the characteristic function for any $\Omega'\in\mathcal{B}(\Omega)$.  As a result, (\ref{integral evolution time discrete}) reduces to
\begin{equation}\label{push forward measure subset}
\mu^{\alpha}_{n+1}(\Omega')=\mu^{\alpha}_n\bigl( (\chi^{\alpha}_n)^{-1}(\Omega') \bigr).
\end{equation}

\begin{definition}[Push forward]\label{def push forward}
The measure $\eta_{n+1}$ is called the push forward of the measure $\eta_{n}$ via the motion mapping $\chi^{\alpha}_n$, notation:
\begin{equation}\label{push forward measure hash notation}
\eta_{n+1}=\chi^{\alpha}_n \# \eta_n,
\end{equation}
if
\begin{equation*}
\eta_{n+1}(\Omega')=\eta_n\bigl( (\chi^{\alpha}_n)^{-1}(\Omega') \bigr),
\end{equation*}
is satisfied for all $\Omega'\in\mathcal{B}(\Omega)$.
\end{definition}
We have now derived the time-discrete version of the problem formulated in Section \ref{section weak formulation}:
\begin{definition}[Time-discrete solution]\label{def time-discrete solution}
The vector of time-discrete measures:
\begin{equation*}
\Bigl((\mu^1_n)_{n\in\mathcal{N}}, (\mu^2_n)_{n\in\mathcal{N}},\ldots, (\mu^{\nu}_n)_{n\in\mathcal{N}}\Bigr),
\end{equation*}
is called time-discrete solution, if:
\begin{enumerate}[(i)]
  \item $\mu^{\alpha}_{n=0}=\mu_0^{\alpha}$ is satisfied for each $\alpha\in\{1,2,\dots,\nu\}$, for some given set of initial measures $\mu_0^{\alpha}$ that are positive and finite,
  \item the evolution of $\mu^{\alpha}_n$ is for each $\alpha$ determined by the push forward $\mu^{\alpha}_{n+1}=\chi^{\alpha}_n \# \mu^{\alpha}_n$,
  \item for each $\alpha$, $\mu^{\alpha}_n$ is positive and finite for all $n$.
\end{enumerate}
\end{definition}
We refer to finding a time-discrete solution as Problem $(\mathcal{P})$.
\section{Solvability of Problem $(\mathcal{P})$ and properties of the solution}\label{section well-posedness time-discrete}
In this section, we prove that there exists a unique time-discrete solution in the sense of Definition \ref{def time-discrete solution}.

\subsection{Solvability}
\begin{theorem}[Global existence of time-discrete solutions]\label{Thm existence time-discrete}
Suppose that for each $n\in\mathcal{N}$ there exist constants $c_n>0$ and $C_n>0$, such that for each $\alpha\in\{1,2,\ldots,\nu\}$
\begin{equation}\label{hypothesis upper and lower bound lambda(chi inv Omega)}
c_n \lambda^d(\Omega') \leqslant \lambda^d\bigl((\chi^{\alpha}_n)^{-1}(\Omega')\bigr)\leqslant C_n \lambda^d(\Omega'), \hspace{1 cm}\text{for each }\Omega'\in\mathcal{B}(\Omega).
\end{equation}
Suppose furthermore that for each $\alpha$ the initial measure $\mu^{\alpha}_0$ is given in its refined Lebesgue decomposition (cf. Corollary \ref{Corollary refined Lebesgue decomposition})
\begin{equation*}
\mu^{\alpha}_0= \mu^{\alpha}_{\text{ac},0}+ \mu^{\alpha}_{\text{d},0}+ \mu^{\alpha}_{\text{sc},0},
\end{equation*}
where $\mu^{\alpha}_{\text{ac},0}\ll\lambda^d$, $\mu^{\alpha}_{\text{d},0}$ is discrete w.r.t. $\lambda^d$ and $\mu^{\alpha}_{\text{sc},0}$ is singular continuous w.r.t. $\lambda^d$. Assume that $\mu^{\alpha}_{\text{ac},0}$, $\mu^{\alpha}_{\text{d},0}$ and $\mu^{\alpha}_{\text{sc},0}$ are positive and finite measures.\\
\\
Then a time-discrete solution as defined in Definition \ref{def time-discrete solution} exists and it is of the form
\begin{equation}\label{unique Lebesgue decomposition time-discrete measure}
\mu^{\alpha}_n= \mu^{\alpha}_{\text{ac},n}+\mu^{\alpha}_{\text{d},n}+\mu^{\alpha}_{\text{sc},n},
\end{equation}
where $\mu^{\alpha}_{\text{ac},n}\ll\lambda^d$, $\mu^{\alpha}_{\text{d},n}$ is discrete w.r.t. $\lambda^d$, $\mu^{\alpha}_{\text{sc},n}$ is singular continuous w.r.t. $\lambda^d$, for all $n\in\mathcal{N}$, and each component is positive and finite.
\end{theorem}

\begin{proof}
The proof of Theorem \ref{Thm existence time-discrete} partly follows the lines of arguments of \cite{Piccoli2010}.\\
Let $\alpha\in\{1,2,\ldots,\nu\}$ be fixed but arbitrary. The proof goes by induction. The statement in (\ref{unique Lebesgue decomposition time-discrete measure}) is true for $n=0$ due to the given initial condition. Moreover, each of the components $\mu^{\alpha}_{\text{ac},0}$, $\mu^{\alpha}_{\text{d},0}$ and $\mu^{\alpha}_{\text{sc},0}$ is positive and finite.\\
\\
The induction hypothesis is that for some (arbitrary, fixed) $n\in\{0,1,\ldots,N_T-1\}$ the time-discrete solution exists, that it is of the form (\ref{unique Lebesgue decomposition time-discrete measure}), and that each of the three measures $\mu^{\alpha}_{\text{ac},n}$, $\mu^{\alpha}_{\text{d},n}$ and $\mu^{\alpha}_{\text{sc},n}$ in the corresponding refined Lebesgue decomposition is positive and finite.\\
\\
We now prove that if the induction hypothesis holds for this $n$, then it also holds for $n+1$.
\begin{enumerate}
  \item By the definition of the push forward operator formulated in (\ref{push forward measure subset}), for any $\Omega'\in\mathcal{B}(\Omega)$ the following holds:
        \begin{eqnarray*}
        \mu^{\alpha}_{n+1}(\Omega')&=&\mu^{\alpha}_n\bigl( (\chi^{\alpha}_n)^{-1}(\Omega') \bigr)\\
                            &=& \mu^{\alpha}_{\text{ac},n}\bigl( (\chi^{\alpha}_n)^{-1}(\Omega') \bigr) + \mu^{\alpha}_{\text{d},n}\bigl( (\chi^{\alpha}_n)^{-1}(\Omega') \bigr)
                            + \mu^{\alpha}_{\text{sc},n}\bigl( (\chi^{\alpha}_n)^{-1}(\Omega') \bigr).
        \end{eqnarray*}
      We can thus write
        \begin{equation*}
        \mu^{\alpha}_{n+1}= \chi^{\alpha}_n \# \mu^{\alpha}_{\text{ac},n} +\chi^{\alpha}_n \# \mu^{\alpha}_{\text{d},n} +\chi^{\alpha}_n \# \mu^{\alpha}_{\text{sc},n}.
        \end{equation*}
  \item We now define
        \begin{equation*}
        \mu^{\alpha}_{\text{ac},n+1}:=\chi^{\alpha}_n \# \mu^{\alpha}_{\text{ac},n}.
        \end{equation*}
      We wish to show that $\mu^{\alpha}_{\text{ac},n+1}$ is absolutely continuous w.r.t. $\lambda^d$. Let $\Omega'\in\mathcal{B}(\Omega)$ be such that $\lambda^d(\Omega')=0$. By (\ref{hypothesis upper and lower bound lambda(chi inv Omega)}):
      \begin{equation*}
      \lambda^d\bigl( (\chi^{\alpha}_n)^{-1}(\Omega') \bigr)\leqslant C_n \lambda^d(\Omega'),
      \end{equation*}
      and thus $\lambda^d(\Omega')=0$ implies $\lambda^d\bigl( (\chi^{\alpha}_n)^{-1}(\Omega') \bigr)=0$. It is part of the induction hypothesis that $\mu^{\alpha}_{\text{ac},n}\ll \lambda^d$, and thus it follows from $\lambda^d\bigl( (\chi^{\alpha}_n)^{-1}(\Omega') \bigr)=0$ that $\mu^{\alpha}_{\text{ac},n}\bigl( (\chi^{\alpha}_n)^{-1}(\Omega') \bigr)=0$. Thus
      \begin{equation*}
      \mu^{\alpha}_{\text{ac},n+1}(\Omega')=\mu^{\alpha}_{\text{ac},n}\bigl( (\chi^{\alpha}_n)^{-1}(\Omega') \bigr)=0,
      \end{equation*}
      by which we have proven that $\mu^{\alpha}_{\text{ac},n+1}\ll\lambda^d$.\\
      \\
      Since $\mu^{\alpha}_{\text{ac},n}$ is positive by the induction hypothesis, we also have that
      \begin{equation*}
      \mu^{\alpha}_{\text{ac},n+1}(\Omega')=\mu^{\alpha}_{\text{ac},n}\bigl( (\chi^{\alpha}_n)^{-1}(\Omega') \bigr)\geqslant0, \hspace{1 cm} \text{for all } \Omega'\in\mathcal{B}(\Omega),
      \end{equation*}
      and thus $\mu^{\alpha}_{\text{ac},n+1}$ is positive.\\
      Similarly, finiteness of $\mu^{\alpha}_{\text{ac},n+1}$ follows from finiteness of $\mu^{\alpha}_{\text{ac},n}$:
      \begin{equation*}
      \mu^{\alpha}_{\text{ac},n+1}(\Omega)=\mu^{\alpha}_{\text{ac},n}\bigl( (\chi^{\alpha}_n)^{-1}(\Omega) \bigr)=\mu^{\alpha}_{\text{ac},n}(\Omega)<\infty,
      \end{equation*}
      where $(\chi^{\alpha}_n)^{-1}(\Omega)=\Omega$ due to the assumed invertibility of the motion mapping.\label{Thm existence time-discrete Part Abs Cont}
  \item Similarly, define
        \begin{equation*}
        \mu^{\alpha}_{\text{d},n+1}:=\chi^{\alpha}_n \# \mu^{\alpha}_{\text{d},n}.
        \end{equation*}
      By the induction hypothesis $\mu^{\alpha}_{\text{d},n}$ is a positive, finite and discrete measure. Lemma \ref{lemma characterization discrete measure} provides that we can thus write $\mu^{\alpha}_{\text{d},n} = \sum_{i\in\mathcal{J}}a_i \delta_{x_{i}}$, for some countable index set $\mathcal{J}\subset\mathbb{N}$, a set $\{x_i\}_{i\in\mathcal{J}}\subset \Omega$ and a set of corresponding nonnegative coefficients $\{a_i\}_{i\in\mathcal{J}}\subset \mathbb{R}$, such that $\sum_{i\in\mathcal{J}}a_i<\infty$.\\
      By definition of $ \mu^{\alpha}_{\text{d},n+1}$, we have that for any $\Omega'\in\mathcal{B}(\Omega)$
        \begin{align*}
        \mu^{\alpha}_{\text{d},n+1}(\Omega')= \mu^{\alpha}_{\text{d},n}\bigl( (\chi^{\alpha}_n)^{-1}(\Omega')\bigr) &= \sum_{i\in\mathcal{J}}a_i \delta_{x_{i}}\bigl( (\chi^{\alpha}_n)^{-1}(\Omega')\bigr)\\
        &= \sum_{i\in\mathcal{J}}a_i \mathbf{1}_{x_{i}\in( (\chi^{\alpha}_n)^{-1}(\Omega'))}\\
        &= \sum_{i\in\mathcal{J}}a_i \mathbf{1}_{\chi^{\alpha}_n(x_{i})\in\Omega'}\\
        &= \sum_{i\in\mathcal{J}}a_i \delta_{\chi^{\alpha}_n(x_{i})}(\Omega').
        \end{align*}
      This proves that $\mu^{\alpha}_{\text{d},n+1}$ is a discrete measure with respect to $\lambda^d$.\\
      \\
      Positivity of $\mu^{\alpha}_{\text{d},n+1}$ follows from positivity of the Dirac measure and the fact that $a_i\geqslant 0$ for all $i\in\mathcal{J}$. Since $\chi^{\alpha}_n$ maps homeomorphically from $\Omega$ to $\Omega$, obviously $\{x_i\}_{i\in\mathcal{J}}\subset \Omega$ implies $\{\chi^{\alpha}_n(x_{i})\}_{i\in\mathcal{J}}\subset \Omega$. As a result
      \begin{equation*}
      \mu^{\alpha}_{\text{d},n+1}(\Omega)= \sum_{i\in\mathcal{J}}a_i \delta_{\chi^{\alpha}_n(x_{i})}(\Omega) = \sum_{i\in\mathcal{J}}a_i<\infty,
      \end{equation*}
      and thus $\mu^{\alpha}_{\text{d},n+1}$ is also finite.\label{Thm existence time-discrete Part Discrete}
  \item Finally, we define
        \begin{equation*}
        \mu^{\alpha}_{\text{sc},n+1}:=\chi^{\alpha}_n \# \mu^{\alpha}_{\text{sc},n}.
        \end{equation*}
        We want to prove that this measure is singular continuous w.r.t. $\lambda^d$.\\
        For any $x\in\Omega$ the definition of the push forward implies that $\mu^{\alpha}_{\text{sc},n+1}(x)= \mu^{\alpha}_{\text{sc},n}\bigl( (\chi^{\alpha}_n)^{-1}(x) \bigr)$. As $\mu^{\alpha}_{\text{sc},n}$ is by the induction hypothesis singular continuous,
        \begin{equation*}
        \mu^{\alpha}_{\text{sc},n}\bigl( (\chi^{\alpha}_n)^{-1}(x) \bigr)=0,\hspace{1 cm} \text{for any }(\chi^{\alpha}_n)^{-1}(x)\in\Omega.
        \end{equation*}
        Thus $\mu^{\alpha}_{\text{sc},n+1}(x)=0$ for all $x\in\Omega$.\\
        Let the set $B_n\in\mathcal{B}(\Omega)$ be such that $\mu^{\alpha}_{\text{sc},n}(\Omega\setminus B_n)= \lambda^d(B_n)=0$, which exists by definition of singular continuous measures (see Definition \ref{def singular cont measure}). Because
        \begin{equation*}
        \Omega\setminus B_n= \bigl(\chi^{\alpha}_n\bigr)^{-1}\Bigl(\chi^{\alpha}_n\bigl(\Omega\setminus B_n\bigr)\Bigr)
        \end{equation*}
        the following identity is true:
        \begin{eqnarray*}
        \mu^{\alpha}_{\text{sc},n}\bigl(\Omega\setminus B_n\bigr)&=& \mu^{\alpha}_{\text{sc},n}\Bigl(\bigl(\chi^{\alpha}_n\bigr)^{-1}\Bigl(\chi^{\alpha}_n\bigl(\Omega\setminus B_n\bigr)\Bigr)\Bigr)\\
        &=& \mu^{\alpha}_{\text{sc},n+1}\Bigl(\chi^{\alpha}_n\bigl(\Omega\setminus B_n\bigr)\Bigr)\\
        &=& \mu^{\alpha}_{\text{sc},n+1}\Bigl(\Omega\setminus \chi^{\alpha}_n\bigl(B_n\bigr)\Bigr).
        \end{eqnarray*}
        In the last step we used that $\chi^{\alpha}_n(\Omega)=\Omega$ (due to invertibility of $\chi^{\alpha}_n$). The bottom line is that
        \begin{equation*}
        \mu^{\alpha}_{\text{sc},n+1}\Bigl(\Omega\setminus \chi^{\alpha}_n\bigl(B_n\bigr)\Bigr) = \mu^{\alpha}_{\text{sc},n}\bigl(\Omega\setminus B_n\bigr) = 0.
        \end{equation*}
        The last equality follows from the way we have chosen $B_n$.\\
        By hypothesis of the theorem, see (\ref{hypothesis upper and lower bound lambda(chi inv Omega)}), we have that
        \begin{equation}\label{bound lambda(chi Bn)}
        \lambda^d\Bigl(\chi^{\alpha}_n\bigl(B_n\bigr)\Bigr)\leqslant \dfrac{1}{c_n}\lambda^d\Bigl(\bigl(\chi^{\alpha}_n\bigr)^{-1}\Bigl(\chi^{\alpha}_n\bigl(B_n\bigr)\Bigr)\Bigr)= \dfrac{1}{c_n}\lambda^d(B_n)=0,
        \end{equation}
        where it is important that $c_n>0$, and $\lambda^d(B_n)=0$ by definition of $B_n$. Consequently, (\ref{bound lambda(chi Bn)}) proves that $\lambda^d\Bigl(\chi^{\alpha}_n\bigl(B_n\bigr)\Bigr)=0$.\\
        \\
        Define $B_{n+1}:= \chi^{\alpha}_n\bigl(B_n\bigr)$. From the fact that $\mu^{\alpha}_{\text{sc},n+1}(x)=0$ for all $x\in\Omega$, and
        \begin{equation*}
        \mu^{\alpha}_{\text{sc},n+1}\Bigl(\Omega\setminus B_{n+1}\Bigr)=\lambda^d\Bigl(B_{n+1}\Bigr) = 0,
        \end{equation*}
        it follows that $\mu^{\alpha}_{\text{sc},n+1}$ is a singular continuous measure w.r.t. $\lambda^d$.\\
        \\
        Since $\mu^{\alpha}_{\text{sc},n}$ is positive by the induction hypothesis, we obtain
        \begin{equation*}
        \mu^{\alpha}_{\text{sc},n+1}(\Omega')=\mu^{\alpha}_{\text{sc},n}\bigl( (\chi^{\alpha}_n)^{-1}(\Omega') \bigr)\geqslant0,
        \end{equation*}
        for all $\Omega'\in\mathcal{B}(\Omega)$. Thus, $\mu^{\alpha}_{\text{sc},n+1}$ is positive. Similarly, $\mu^{\alpha}_{\text{sc},n+1}$ is finite because $\mu^{\alpha}_{\text{sc},n}$ is finite:
        \begin{equation*}
        \mu^{\alpha}_{\text{sc},n+1}(\Omega)=\mu^{\alpha}_{\text{sc},n}\bigl( (\chi^{\alpha}_n)^{-1}(\Omega) \bigr)=\mu^{\alpha}_{\text{sc},n}(\Omega)<\infty.
        \end{equation*}
        \label{Thm existence time-discrete Part Singular Continuous}
  \item We have now proven that if a time-discrete solution of the form (\ref{unique Lebesgue decomposition time-discrete measure}) exists for $n$, it also exists for $n+1$. Positivity and finiteness of $\mu^{\alpha}_{n+1}$ follow from positivity and finiteness of its three components, which we have proven above. The induction argument guarantees that (\ref{unique Lebesgue decomposition time-discrete measure}) holds for all $n\in\mathcal{N}$.
\end{enumerate}
\end{proof}

\begin{theorem}[Uniqueness of global time-discrete solutions]\label{Thm uniqueness time-discrete}
Assume the hypotheses of Theorem \ref{Thm existence time-discrete}. Then the global time-discrete solution:
\begin{equation*}
\mu^{\alpha}_n= \mu^{\alpha}_{\text{ac},n}+\mu^{\alpha}_{\text{d},n}+\mu^{\alpha}_{\text{sc},n},\hspace{1 cm}\text{for all }n\in\mathcal{N}
\end{equation*}
is unique.
\end{theorem}
\begin{proof}
Uniqueness of the time-discrete solution follows from the fact that $\mu^{\alpha}_{n+1}=\chi^{\alpha}_n \# \mu^{\alpha}_n$ defines the push forward unambiguously. To see this, assume that a unique $\mu^{\alpha}_{n}$ exists, and there are two measures $\mu^{\alpha}_{1,n+1}$ and $\mu^{\alpha}_{2,n+1}$ such that $\mu^{\alpha}_{i,n+1}=\chi^{\alpha}_n \# \mu^{\alpha}_{n}$ for each $i\in\{1,2\}$. Then we have that for any $\Omega'\in\mathcal{B}(\Omega)$ the following holds:
\begin{equation*}
\mu^{\alpha}_{1,n+1}(\Omega')=\mu^{\alpha}_n\bigl( (\chi^{\alpha}_n)^{-1}(\Omega') \bigr)=\mu^{\alpha}_{2,n+1}(\Omega'),
\end{equation*}
thus $\mu^{\alpha}_{1,n+1}\equiv\mu^{\alpha}_{2,n+1}$.\\
It now follows by induction that if the initial measure $\mu^{\alpha}_{0}$ is unique, consequently $\mu^{\alpha}_{n}$ is defined uniquely for all $n\in\mathcal{N}$. Note that the decomposition of $\mu^{\alpha}_{n}$ in its three parts is unique by Corollary \ref{Corollary refined Lebesgue decomposition}.
\end{proof}

\begin{remark}[Connections to Reference \cite{Piccoli2010}]
The results of \cite{Piccoli2010} are a special case of our setting; we already said so in our comments at the end of Section \ref{section weak formulation}.\\
Let $m_0$ be a discrete measure: $m_0:=\sum_{j=1}^N \delta_{P_j}$, and let $M_0$ be an absolutely continuous measure (with density $\rho(t,x)$). There are two ways to recover their situation. One way is to set $\nu=1$ and consider
\begin{equation*}
\mu^{1}_0= \theta m_0 + (1-\theta)M_0,
\end{equation*}
where $\theta\in[0,1]$ is a tuning parameter.\\
The second way to recover the results of \cite{Piccoli2010} is by setting $\nu=2$, and then by taking
\begin{eqnarray*}
\mu^{1}_0 &=& \theta m_0,\\
\mu^{2}_0 &=& (1-\theta)M_0,
\end{eqnarray*}
henceforth only considering the total mass measure and the barycentric velocity.
\end{remark}

\subsection{Basic properties of the time-discrete solution}
In this section we will formulate and prove a number of properties of the time-discrete solution provided by Theorems \ref{Thm existence time-discrete} and \ref{Thm uniqueness time-discrete}. These properties follow mainly from the subsequent steps done in the proofs of those theorems.

\begin{corollary}[Conservation of mass]\label{corollary conservation of mass discretized}
Assume the hypotheses of Theorem \ref{Thm existence time-discrete}. For each $n\in\mathcal{N}$ the initial mass is conserved by each of the three components of the time-discrete solution provided by Theorems \ref{Thm existence time-discrete} and \ref{Thm uniqueness time-discrete}. That is, for all $n\in\mathcal{N}$ and each $\alpha\in\{1,2,\ldots,\nu\}$:
\begin{enumerate}[(i)]
  \item $\mu^{\alpha}_{\text{ac},n}(\Omega)=\mu^{\alpha}_{\text{ac},0}(\Omega)$,
  \item $\mu^{\alpha}_{\text{d},n}(\Omega)=\mu^{\alpha}_{\text{d},0}(\Omega)$,
  \item $\mu^{\alpha}_{\text{sc},n}(\Omega)=\mu^{\alpha}_{\text{sc},0}(\Omega)$.
\end{enumerate}
As a result, also $\mu^{\alpha}_n(\Omega)=\mu^{\alpha}_0(\Omega)$.
\end{corollary}
\begin{proof}
For each $n\in\{0,1,\ldots,N_T-1\}$ and $\alpha\in\{1,2,\ldots,\nu\}$, consider the measure $\mu^{\alpha}_{\omega, n}$, where $\omega\in\{\text{ac, d, sc}\}$. By definition of $\mu^{\alpha}_{\omega, n+1}$ (see the constructive proof of Theorem \ref{Thm existence time-discrete}, Parts \ref{Thm existence time-discrete Part Abs Cont}, \ref{Thm existence time-discrete Part Discrete} and \ref{Thm existence time-discrete Part Singular Continuous}):
\begin{equation*}
\mu^{\alpha}_{\omega,n+1}:=\chi^{\alpha}_n \# \mu^{\alpha}_{\omega,n}.
\end{equation*}
For each $n$ we thus have
\begin{equation*}
\mu^{\alpha}_{\omega,n+1}(\Omega)=\mu^{\alpha}_{\omega,n}\bigl( (\chi^{\alpha}_n)^{-1}(\Omega) \bigr),
\end{equation*}
and $(\chi^{\alpha}_n)^{-1}(\Omega)=\Omega$ due to the invertibility of $\chi^{\alpha}_n$. This implies that $\mu^{\alpha}_{\omega,n+1}(\Omega)=\mu^{\alpha}_{\omega,n}(\Omega)$ for all $n\in\{0,1,\ldots,N_T\}$, and by an inductive argument: $\mu^{\alpha}_{\omega,n}(\Omega)=\mu^{\alpha}_{\omega,0}(\Omega)$ for each $n$.\\
Since $\mu^{\alpha}_n = \mu^{\alpha}_{\text{ac},n}+\mu^{\alpha}_{\text{d},n}+\mu^{\alpha}_{\text{sc},n}$ for all $n\in\mathbb{N}$, the above implies trivially that
\begin{equation*}
\mu^{\alpha}_n(\Omega)=\mu^{\alpha}_0(\Omega),\hspace{1 cm}\text{for all }n\in\mathcal{N}.
\end{equation*}
\end{proof}

\begin{corollary}
Assume the hypotheses of Theorem \ref{Thm existence time-discrete}. If $\mu^{\alpha}_{\text{ac},0}\equiv 0$, $\mu^{\alpha}_{\text{d},0}\equiv 0$, or $\mu^{\alpha}_{\text{sc},0}\equiv 0$ respectively, then for each $n$ the corresponding component of $\mu^{\alpha}_n$ vanishes. That is, for each $\alpha$, the following statements are true:
\begin{enumerate}[(i)]
  \item $\mu^{\alpha}_{\text{ac},0}\equiv 0$ implies $\mu^{\alpha}_{\text{ac},n}\equiv 0$ for all $n\in\mathcal{N}$,
  \item $\mu^{\alpha}_{\text{d},0}\equiv 0$ implies $\mu^{\alpha}_{\text{d},n}\equiv 0$ for all $n\in\mathcal{N}$,
  \item $\mu^{\alpha}_{\text{sc},0}\equiv 0$ implies $\mu^{\alpha}_{\text{sc},n}\equiv 0$ for all $n\in\mathcal{N}$.
\end{enumerate}
\end{corollary}

\begin{proof}
For any $\alpha\in\{1,2,\ldots,\nu\}$, let $\omega\in\{\text{ac, d, sc}\}$ be arbitrary. Then the corresponding part $\mu^{\alpha}_{\omega,n+1}$ in the refined Lebesgue decomposition of $\mu^{\alpha}_{n+1}$ follows uniquely from $\mu^{\alpha}_{\omega,n}$ by the push forward $\mu^{\alpha}_{\omega,n+1}=\chi^{\alpha}_n \# \mu^{\alpha}_{\omega,n}$. See Parts \ref{Thm existence time-discrete Part Abs Cont}, \ref{Thm existence time-discrete Part Discrete} and \ref{Thm existence time-discrete Part Singular Continuous} of the proof of Theorem \ref{Thm existence time-discrete}.\\
\\
Assume that $\mu^{\alpha}_{\omega,n}\equiv0$. Then for each $\Omega'\in\mathcal{B}(\Omega)$
\begin{equation*}
\mu^{\alpha}_{\omega,n+1}(\Omega')=\mu^{\alpha}_{\omega, n}\bigl( (\chi^{\alpha}_n)^{-1}(\Omega') \bigr)=0.
\end{equation*}
We thus conclude that $\mu^{\alpha}_{\omega,n+1}\equiv0$.\\
It follows by an inductive argument that $\mu^{\alpha}_{\omega,n}\equiv0$ for all $n\in\mathcal{N}$, if $\mu^{\alpha}_{\omega,0}\equiv0$ is given.
\end{proof}

\begin{remark}
If $\mu^{\alpha}_{\text{ac},0}\equiv0$, then the assumption that there is a $C_n>0$ such that
\begin{equation*}
\lambda^d\bigl((\chi^{\alpha}_n)^{-1}(\Omega')\bigr)\leqslant C_n \lambda^d(\Omega')
\end{equation*}
can be removed from the theorem. This is because, even without this assumption, the push forward of $\mu^{\alpha}_{\text{ac},n}\equiv0$ is an absolutely continuous measure (of course, more specifically $\mu^{\alpha}_{\text{ac},n+1}\equiv0$).\\
Due to similar arguments, $\mu^{\alpha}_{\text{sc},0}\equiv0$ allows us to remove the assumption that there exists a $c_n>0$ for which
\begin{equation*}
c_n \lambda^d(\Omega')\leqslant\lambda^d\bigl((\chi^{\alpha}_n)^{-1}(\Omega')\bigr).
\end{equation*}
\end{remark}

\begin{remark}
If we set $\mu^{\alpha}_{\text{ac},0}\equiv0$ and $\mu^{\alpha}_{\text{sc},0}\equiv0$, then Theorem \ref{Thm existence time-discrete} provides that $\mu^{\alpha}_n$ is discrete w.r.t. $\lambda^d$ for all $n\in\mathcal{N}$. We have shown in Lemma \ref{lemma characterization discrete measure} that any discrete measure w.r.t. $\lambda^d$ can be written as a linear combination of Dirac measures. Note however, that in $\mu^{\alpha}_n$ both the centres of the Dirac masses \textit{and} the coefficients might in principle depend on $n$. (Since the index set $\mathcal{J}$ is necessarily \textit{countable}, is suffices to allow only the positions and coefficients to be $n$-dependent.)
\end{remark}
However, due to the specific structure of the push forward, we have the following corollary of Theorem \ref{Thm existence time-discrete}:
\begin{corollary}\label{corollary sum of diracs preserves constant coeff}
If the initial measure is a linear combination of Dirac measures: $\mu^{\alpha}_0=\sum_{i\in\mathcal{J}}a_i\delta_{x_{i,0}}$ for some countable index set $\mathcal{J}$, then for each $n$ the measure is of the form $\mu^{\alpha}_n=\sum_{i\in\mathcal{J}}a_i\delta_{x_{i,n}}$, with the same (constant) coefficients $a_i$.
\end{corollary}
\begin{proof}
We extract from Part \ref{Thm existence time-discrete Part Discrete} of the proof of Theorem \ref{Thm existence time-discrete} that for any $n\in\{0,1,\ldots,N_T-1\}$: $\mu^{\alpha}_{n} = \sum_{i\in\mathcal{J}}a_i \delta_{x_{i}}$ implies $\mu^{\alpha}_{n+1}(\Omega')= \sum_{i\in\mathcal{J}}a_i \delta_{\chi^{\alpha}_n(x_{i})}(\Omega')$. Assume that $\mu^{\alpha}_0=\sum_{i\in\mathcal{J}}a_i\delta_{x_{i,0}}$ has been given. By inductive reasoning we can now deduce that $\mu^{\alpha}_{n} = \sum_{i\in\mathcal{J}}a_i \delta_{x_{i,n}}$ for all $n\in\mathcal{N}$, where
\begin{equation*}
x_{i,n}:= \bigl(\chi^{\alpha}_{n-1} \circ \chi^{\alpha}_{n-2} \circ\ldots \circ \chi^{\alpha}_0\bigr)(x_{i,0}).
\end{equation*}
\end{proof}

\begin{lemma}\label{lemma sum of diracs obey dx/dt=v}
If $\mu^{\alpha}_{n} = \sum_{i\in\mathcal{J}}a_i \delta_{x_{i,n}}$ for each $n\in\mathcal{N}$ then, for any $i\in\mathcal{J}$, the position $x_{i,n}$ satisfies a discretized version of $dx_i(t)/dt=v^{\alpha}(t,x_i(t))$, evaluated at $t=t_n$ for all $n\in\{0,1,\ldots,N_T-1\}$.
\end{lemma}
\begin{proof}
For any $i\in\mathcal{J}$, a Taylor series expansion of $x_i$ around $t=t_n$ yields (evaluated in $t=t_{n+1}$):
\begin{equation*}
x_i(t_{n+1})=x_i(t_n)+\Delta t_n\dfrac{d}{dt}x_i(t_n)+\mathcal{O}(\Delta t_n^2),
\end{equation*}
by which we find the expression
\begin{equation}\label{forward differences O(Delta t ^2)}
\dfrac{d}{dt}x_i(t_n)=\dfrac{x_i(t_{n+1})-x_i(t_n)}{\Delta t_n}+\mathcal{O}(\Delta t_n^2).
\end{equation}
The discretized version of $dx_i(t)/dt=v^{\alpha}(t,x_i(t))$ is now obtained by substitution of (\ref{forward differences O(Delta t ^2)}), and after neglecting the $\mathcal{O}(\Delta t_n^2)$-part. Write $x_{i,n}$ as the discrete-in-time equivalent of $x_i(t_n)$, and furthermore use $v_n^{\alpha}(x_{i,n})=v^{\alpha}(t_n,x_{i,n})$ like before. The discretized version becomes
\begin{equation*}
\dfrac{x_{i,n+1}-x_{i,n}}{\Delta t_n}=v_n^{\alpha}(x_{i,n}),
\end{equation*}
or
\begin{equation}\label{discretized version dx/dt=v}
x_{i,n+1}=x_{i,n}+\Delta t_nv_n^{\alpha}(x_{i,n}).
\end{equation}
We have seen in the proofs of Theorem \ref{Thm existence time-discrete} and Corollary \ref{corollary sum of diracs preserves constant coeff} that $x_{i,n+1}$ and $x_{i,n}$ are related via
\begin{equation*}
x_{i,n+1}=\chi^{\alpha}_n(x_{i,n}),
\end{equation*}
and by definition of the motion mapping (see Definition \ref{def one-step motion mapping}) this is exactly (\ref{discretized version dx/dt=v}).
\end{proof}

\begin{remark}
In Remark \ref{remark sum of diracs preserves constant coeff} we have already anticipated the statements of Corollary \ref{corollary sum of diracs preserves constant coeff} and Lemma \ref{lemma sum of diracs obey dx/dt=v}.
\end{remark}

\begin{lemma}\label{lemma diracs do not merge}
Assume that, for each $n\in\mathcal{N}$ and $\alpha\in\{1,2,\ldots,\nu\}$, the velocity field $v_n^{\alpha}$ is Lipschitz continuous, with Lipschitz constant strictly smaller than $1/\Delta t_n$. That is, there is a constant $0\leqslant K_n < 1/\Delta t_n$, such that
\begin{equation*}
|v_{n}^{\alpha}(x)-v_{n}^{\alpha}(y)|\leqslant K_{n}|x-y|, \hspace{1 cm}\text{for all }x,y\in\Omega.
\end{equation*}
For each $n\in\mathcal{N}$, the measure $\mu^{\alpha}_{\text{d},n}$ is discrete, and thus it is the linear combination of Dirac masses (Lemma \ref{lemma characterization discrete measure}). Let $\{x_{i,n}\}_{i\in\mathcal{J}}$ be the corresponding set of centres (at time-slice $n$). Now, if $\{x_{i,0}\}_{i\in\mathcal{J}}$ consists of distinct elements, then also $\{x_{i,n}\}_{i\in\mathcal{J}}$ consists of distinct elements, for each $n\in\mathcal{N}$.
\end{lemma}

\begin{proof}
The proof goes by contradiction.\\
Assume that $n\in\{1,2,\ldots,N_T\}$ is such that not all elements of $\{x_{i,n}\}_{i\in\mathcal{J}}$ are distinct, and, more specifically, let $j,k\in\mathcal{J}$ satisfy $j\neq k$ and $x_{j,n}=x_{k,n}$.\\
We know that $x_{i,n}= \bigl(\chi^{\alpha}_{n-1} \circ \chi^{\alpha}_{n-2} \circ\ldots \circ \chi^{\alpha}_0\bigr)(x_{i,0})$ for all $i\in\mathcal{J}$, as this was shown in the proof of Corollary \ref{corollary sum of diracs preserves constant coeff}.\\
By assumption of the lemma $x_{j,0}\neq x_{k,0}$. For $x_{j,n}=x_{k,n}$ to hold, there must therefore be an $\tilde{n}\in\{0,1,\ldots,n-1\}$ such that
\begin{equation*}
x_{j,\tilde{n}}\neq x_{k,\tilde{n}},
\end{equation*}
but
\begin{equation*}
\chi^{\alpha}_{\tilde{n}}(x_{j,\tilde{n}})=\chi^{\alpha}_{\tilde{n}}(x_{k,\tilde{n}}).
\end{equation*}
By definition of the push forward, the last equivalence can be written as
\begin{equation*}
x_{j,\tilde{n}}+\Delta t_{\tilde{n}}v_{\tilde{n}}^{\alpha}(x_{j,\tilde{n}})=x_{k,\tilde{n}}+\Delta t_{\tilde{n}}v_{\tilde{n}}^{\alpha}(x_{k,\tilde{n}}),
\end{equation*}
or
\begin{equation}\label{Diracs merge}
\Delta t_{\tilde{n}}\bigl(v_{\tilde{n}}^{\alpha}(x_{j,\tilde{n}})-v_{\tilde{n}}^{\alpha}(x_{k,\tilde{n}})\bigr)=x_{k,\tilde{n}}-x_{j,\tilde{n}}.
\end{equation}
By assumption of the lemma, $v_{\tilde{n}}^{\alpha}$ is Lipschitz continuous, and there is a $K_{\tilde{n}}<1/\Delta t_{\tilde{n}}$ such that
\begin{equation*}
|v_{\tilde{n}}^{\alpha}(x)-v_{\tilde{n}}^{\alpha}(y)|\leqslant K_{\tilde{n}}|x-y|.
\end{equation*}
In particular, this implies that
\begin{equation*}
|v_{\tilde{n}}^{\alpha}(x_{j,\tilde{n}})-v_{\tilde{n}}^{\alpha}(x_{k,\tilde{n}})|\leqslant K_{\tilde{n}}|x_{j,\tilde{n}}-x_{k,\tilde{n}}|< \dfrac{1}{\Delta t_{\tilde{n}}} |x_{j,\tilde{n}}-x_{k,\tilde{n}}|.
\end{equation*}
However, as a result of (\ref{Diracs merge}) we already have that
\begin{equation*}
|v_{\tilde{n}}^{\alpha}(x_{j,\tilde{n}})-v_{\tilde{n}}^{\alpha}(x_{k,\tilde{n}})|= \dfrac{1}{\Delta t_{\tilde{n}}}|x_{j,\tilde{n}}-x_{k,\tilde{n}}|,
\end{equation*}
and thus we have a contradiction. This finishes the proof.
\end{proof}

\subsection{Relaxing the conditions on the motion mapping and velocity fields}\label{section relax conditions on motion mapping}
Up to now we have made a number of assumptions with respect to the one-step motion mapping $\chi^{\alpha}_n$, and the (discretized) velocity field $v^{\alpha}_n$. We first give a short recapitulation of these assumptions here. Afterwards we indicate which of these assumptions can be relaxed, and in what way this can be done. The results of this section have a preliminary character: we expect that more is to be discovered as soon as more analysis effort is invested in this direction.\\
\\
The following restrictions on the motion mapping and velocity field were imposed so far:
\begin{enumerate}
  \item In Section \ref{section derivation time-discrete}, we assumed $v^{\alpha}_n$ to be $\mu^{\alpha}_n$-uniformly bounded. We needed this to be able to write the Taylor expansion in (\ref{Taylor expansion psi, v well-behaved}). This assumption means that, for each $n\in\mathcal{N}$ and $\alpha\in\{1,2,\ldots,\nu\}$, there exists a constant $M_n$ for which
        \begin{equation*}
        |v_n^{\alpha}(x)|<M_n, \hspace{1 cm}\text{for }\mu^{\alpha}_n\text{-almost every }x\in\Omega.
        \end{equation*}\label{enumerate assumptions 1}
  \item By Assumption \ref{assumption time-discrete motion mapping invertible and Borel}, we postulated that $\chi^{\alpha}_n:\Omega\rightarrow\Omega$ is a homeomorphism, for all $n\in\{0,1,\ldots,N_T-1\}$ and $\alpha\in\{1,2,\ldots,\nu\}$.\label{enumerate assumptions 2}
  \item In case $\mu^{\alpha}_{\text{ac}, 0}\not\equiv 0$ and $\mu^{\alpha}_{\text{sc}, 0}\not\equiv 0$, we need for all $n\in\{0,1,\ldots,N_T-1\}$ strictly positive constants $c_n$ and $C_n$, such that
        \begin{equation*}
        c_n \lambda^d(\Omega') \leqslant \lambda^d\bigl((\chi^{\alpha}_n)^{-1}(\Omega')\bigr)\leqslant C_n \lambda^d(\Omega'), \hspace{1 cm}\text{for each }\Omega'\in\mathcal{B}(\Omega),
        \end{equation*}
        to prove existence and uniqueness of the time-discrete solution (see Theorems \ref{Thm existence time-discrete} and \ref{Thm uniqueness time-discrete}).\label{enumerate assumptions 3}
  \item The Lipschitz continuity of the velocity field $v_n^{\alpha}$ is demanded in the hypothesis of Lemma \ref{lemma diracs do not merge}. More specifically, we assume that for each $n\in\{0,1,\ldots,N_T-1\}$ a constant $K_n<1/\Delta t_n$ exists, such that
      \begin{equation*}
        |v_{n}^{\alpha}(x)-v_{n}^{\alpha}(y)|\leqslant K_{n}|x-y|, \hspace{1 cm}\text{for all }x,y\in\Omega.
        \end{equation*}
        This restriction is only needed if $\mu^{\alpha}_{\text{d}, 0}\not\equiv 0$, to make sure that no two centres of Dirac masses can `merge'.\label{enumerate assumptions 4}
\end{enumerate}
Before we can relax these conditions, we introduce the following concept:
\begin{definition}[Support of a measure]
The \textit{support} of a (positive, finite) measure $\mu$ on $\mathcal{B}(\Omega)$ is defined as the set
\begin{equation*}
\supp \mu := \{x\in\Omega\,|\, \mu(\Omega')>0 \text{ for all }\Omega'\in\mathcal{B}(\Omega)\text{ open such that }x\in\Omega' \}.
\end{equation*}
\end{definition}
One can show (see \cite{Parthasarathy}, pp.~27--28) that the support is a closed set of full measure, that is
\begin{equation*}
\mu(\supp \mu)=\mu(\Omega).
\end{equation*}
As a result
\begin{equation}\label{complement support is null set}
\mu(\Omega\setminus \supp \mu)=\mu(\Omega)-\mu(\Omega\cap\supp\mu)=\mu(\Omega)-\mu(\supp\mu)=0.
\end{equation}
For any measurable set $\Omega'\subset\Omega\setminus \supp \mu$ it follows that
\begin{equation}\label{no mass in set outside support}
\mu(\Omega')\leqslant\mu(\Omega\setminus\supp\mu)=0.
\end{equation}
We have introduced the concept of a support, because it is in fact only important to consider the properties of the motion mapping $\chi^{\alpha}_n$ on $\supp \mu_n^{\alpha}$. As (\ref{no mass in set outside support}) shows, there is no mass contained in any region outside the support. For our purposes, it is irrelevant whether the image of such $\Omega'\subset\Omega\setminus\supp\mu_n^{\alpha}$ is again a subset of $\Omega$ or not. In order to have mass conservation within $\Omega$ the only thing that matters is whether $\chi_n^{\alpha}(\Omega')\subset\Omega$ for all $\Omega'\subset\supp\mu^{\alpha}_n$.\\
\\
We do not change the restriction in Part \ref{enumerate assumptions 1} of the list above. However, we remark that if we demand that $|v_n^{\alpha}(x)|<M_n$ holds for all $x\in\supp \mu_n^{\alpha}$, then it also holds for $\mu_n^{\alpha}$-a.e. $x$. Indeed, if $|v_n^{\alpha}(x)|<M_n$ for all $x\in\supp \mu_n^{\alpha}$, then
\begin{equation*}
\Omega_M:=\big\{x\in\Omega\,\big|\,M_n\leqslant|v_n^{\alpha}(x)|\big\}\subset\Omega\setminus\supp\mu_n^{\alpha},
\end{equation*}
thus, cf. (\ref{complement support is null set}),
\begin{equation*}
\mu_n^{\alpha}(\Omega_M)\leqslant\mu_n^{\alpha}(\Omega\setminus\supp\mu_n^{\alpha})=0.
\end{equation*}
In Part \ref{enumerate assumptions 2} of the list of assumptions, we pay special attention to the restriction of $\chi^{\alpha}_n$ to $\supp \mu_n^{\alpha}$, that is, $\bar{\chi}^{\alpha}_n:=\chi^{\alpha}_n\big|_{\supp \mu_n^{\alpha}}$. Moreover, we replace the assumption that the motion mapping is a homeomorphism by:
\begin{assumption}\label{assumption motion mapping including support}
For all $n\in\{0,1,\ldots,N_T-1\}$ and each $\alpha\in\{1,2,\ldots,\nu\}$ we assume that the motion mapping $\chi^{\alpha}_n:\Omega\rightarrow\mathbb{R}^d$ is such that the following properties are satisfied:
\begin{enumerate}[(i)]
  \item the range of $\bar{\chi}^{\alpha}_n:=\chi^{\alpha}_n\big|_{\supp \mu_n^{\alpha}}$ is contained in $\Omega$. That is: $\bar{\chi}^{\alpha}_n:\supp \mu_n^{\alpha}\rightarrow\Omega$;\label{assumption motion mapping including support part 1}
  \item for all $\Omega'\in\mathcal{B}(\Omega)$ we have that $(\chi^{\alpha}_n)^{-1}(\Omega')$ is measurable;
  \item for all $\Omega'\in\mathcal{B}(\Omega)$ we have that $\chi^{\alpha}_n(\Omega')$ is measurable.
\end{enumerate}
Here, $(\chi^{\alpha}_n)^{-1}$ should not be understood as the inverse mapping. The set $(\chi^{\alpha}_n)^{-1}(\Omega')$ is the pre-image of $\Omega'$. We do not assume invertibility of the motion mapping.
\end{assumption}
Note that throughout this thesis, anywhere we used the inverse mapping $(\chi^{\alpha}_n)^{-1}$, we should now read it in the pre-image sense. This is e.g. also the case in the definition of the push forward, Definition \ref{def push forward}.\\
\\
Under the new set of restrictions on the motion mapping, we still want the Theorems and Lemmas of Section \ref{section well-posedness time-discrete} to hold. We want this especially for Theorems \ref{Thm existence time-discrete} and \ref{Thm uniqueness time-discrete} (global existence and uniqueness of the time-discrete solution). To achieve this, Part \ref{enumerate assumptions 3} of the list of assumptions needs reconsideration. We make it slightly stronger:
\begin{assumption}
Suppose that for each $n\in\{0,1,\ldots,N_T-1\}$ there exist constants $c_n>0$ and $C_n>0$, such that for each $\alpha\in\{1,2,\ldots,\nu\}$
\begin{eqnarray}
c_n \lambda^d\bigl(\chi^{\alpha}_n(\Omega')\bigr) &\leqslant& \lambda^d (\Omega'),\label{stronger assumption motion mapping re. sing cont}\\
\lambda^d\bigl((\chi^{\alpha}_n)^{-1}(\Omega')\bigr) &\leqslant& C_n \lambda^d(\Omega'),
\end{eqnarray}
hold for all $\Omega'\in\mathcal{B}(\Omega)$. Again, $(\chi^{\alpha}_n)^{-1}(\Omega')$ is the pre-image of $\Omega'$.
\end{assumption}
Note that (\ref{stronger assumption motion mapping re. sing cont}) implies the left inequality of
\begin{equation*}
c_n \lambda^d(\Omega') \leqslant \lambda^d\bigl((\chi^{\alpha}_n)^{-1}(\Omega')\bigr)\leqslant C_n \lambda^d(\Omega'),
\end{equation*}
by substituting $(\chi^{\alpha}_n)^{-1}(\Omega')$ for $\Omega'$, and by using $\chi^{\alpha}_n\bigl((\chi^{\alpha}_n)^{-1}(\Omega')\bigr)=\Omega'$.\\
\\
We still demand that the velocity field is a Lipschitz continuous function if $\mu^{\alpha}_{\text{d}, 0}\not\equiv 0$ (Restriction \ref{enumerate assumptions 4}). We relax this restriction in this sense, that it is no longer required for all $x,y\in\Omega$. It suffices to demand this property only for all $x,y\in\supp\mu^{\alpha}_{\text{d},n}$. Note that $\supp\mu^{\alpha}_{\text{d},n}$ consists exactly of those points in which the Dirac measures of $\mu^{\alpha}_{\text{d},n}$ are centered. In the proof of Lemma \ref{lemma diracs do not merge} $v_n^{\alpha}$ is only evaluated in these points. The proof is therefore still valid under the proposed weaker assumption.\\
\\
The results of Section \ref{section well-posedness time-discrete} are still valid under the new assumptions proposed in this Section. We only need to modify some minor details in the proofs of Theorem \ref{Thm existence time-discrete} and Corollary \ref{corollary conservation of mass discretized}. These modifications are addressed in Appendix \ref{Appendix modification proofs}.

\begin{remark}\label{remark conditions for non-merging diracs and bounded velocity}
In Lemma \ref{lemma diracs do not merge} we have shown that the centres of the Dirac measures cannot merge if the velocity field is a Lipschitz continuous function. We acknowledge that it is equally important that two Dirac masses that belong to distinct subpopulations cannot merge. Moreover, the integrand in the social velocity terms, see (\ref{Def vsoc sum of integrals}), typically has a singularity in 0. This means that at timeslice $n$, the velocity field $v_n^{\alpha}(x)$ is unbounded if $x$ is arbitrarily close to a Dirac mass. We thus propose as an extra demand that for all $\alpha\in\{1,2,\ldots,\nu\}$ and for each $n\in\mathbb{N}$
\begin{equation*}
\supp \mu^{\alpha}_n \cap \supp \mu^{\beta}_{\text{d},n} = \emptyset, \hspace{1 cm}\text{for all }\beta\in\{1,2,\ldots,\nu\}\text{ such that }\beta\neq\alpha,
\end{equation*}
and
\begin{equation*}
\supp \mu^{\alpha}_{\text{d},n}\cap (\supp \mu^{\alpha}_{\text{ac},n}\cap\supp \mu^{\alpha}_{\text{sc},n}) = \emptyset.
\end{equation*}
If this condition is obeyed, no Dirac measures from distinct populations can be centred in the same point. Since the support of a measure is a closed set, we furthermore prevent $v_n^{\alpha}(x)$ from being unbounded for all $x\in\supp\mu^{\alpha}_n$. What happens outside $\supp\mu^{\alpha}_n$ is unimportant. We come back to this issue in Section \ref{section numerical scheme}.
\end{remark}

\section{Reformulation of the proposed velocity fields in the time-discrete setting}\label{section time-discrete vel field}
In Sections \ref{section derivation time-discrete} and \ref{section well-posedness time-discrete} we have used the velocity fields $v^{\alpha}_n$ without specifying their actual form. In this section we want to do so. It is quite self-evident that we just take the proposed form of Section \ref{section specification velocity} and determine its time-discrete counterpart.\\
For $\alpha\in\{1,2,\ldots,\nu\}$, we thus take $v^{\alpha}_n$ as the superposition of the desired and the social velocity like in (\ref{Def v decomposition}):
\begin{equation*}
v^{\alpha}_n(x):= v_{\text{des}}^{\alpha}(x)+v_{\text{soc},n}^{\alpha}(x), \hspace{1 cm} \text{for all }x\in\Omega.
\end{equation*}
For $v_{\text{soc},n}^{\alpha}$ we have the following expression:
\begin{equation*}
v_{\text{soc},n}^{\alpha}(x):= \sum_{\beta=1}^{\nu}\int_{\Omega \setminus \{x\}} f^{\alpha}_{\beta}(|y-x|)g(\theta_{xy}^{\alpha})\dfrac{y-x}{|y-x|}d\mu^{\beta}_n(y),\hspace{1 cm}\text{for all }\alpha\in\{1,2,\ldots,\nu\}.
\end{equation*}

\section{Entropy inequality for the time-discrete problem}\label{section entropy inequality disc-in-time}
In Section \ref{section derivation time-discrete} we have derived a time-discrete version of our problem. We derived an entropy inequality for the continuous-in-time problem in Section \ref{section entropy inequality cont-in-time}. Its discrete-in-time counterpart is presented in this section. We again work with absolutely continuous mass measures only, and treat (for clarity) the single-population case first.

\subsection{One population}\label{section entropy inequality disc-in-time one population}
In Section \ref{section derivation time-discrete} we derived, omitting $\mathcal{O}(\Delta t_n^2)$-terms, (\ref{integral evolution time discrete}): an equation relating the situation at time $t_n$ to the one at time $t_{n+1}$. If the time-discrete mass measures are absolutely continuous, (\ref{integral evolution time discrete}) reads
\begin{equation}\label{time discrete evolution}
\int_{\Omega}\psi(x)\rho_{n+1}(x)d\lambda^d(x) =  \int_{\Omega}\psi\bigl(\chi_n(x)\bigr)\rho_n(x)d\lambda^d(x),
\end{equation}
where we can take any $\psi\in L^1_{\mu_{n+1}}(\Omega)$ (cf. Remark \ref{remark relax psi in C^1_0}). We deduce our discrete-in-time entropy inequality here, taking (\ref{time discrete evolution}) as a starting point.\\
\\
The subscript $n$ in the sequel denotes that we consider the time-discrete version of a function at time $t_n$. Recall that $\chi_n$ denotes the one-step push forward (see Definition \ref{def one-step motion mapping}), defined by
\begin{equation*}
\chi_n(x):= x + \Delta t_n v_n(x),
\end{equation*}
Let $S_n$ be the time-discrete equivalent of the entropy $S(t_n)$. The following theorem holds:

\begin{theorem}[Discrete-in-time entropy inequality]\label{theorem time-discrete entropy ineq one component}
Assume that the time-discrete velocity $v_n$ is of the form
\begin{equation}\label{vn V plus W star rho}
v_n:=\nabla V + \nabla W\star \rho_n,
\end{equation}
where $W$ is such that $W(\xi)=W(-\xi)$ for all $\xi\in\mathbb{R}^d$. Assume moreover that the evolution of the system is governed by (\ref{time discrete evolution}), and that the entropy is defined by
\begin{equation}\label{def Sn}
S_n:=\int_{\Omega} \bigl( V(x)+\dfrac12 (W\star \rho_n)(x) \bigr) \rho_n(x)d\lambda^d(x), \hspace{1 cm}\text{for all }n\in\mathcal{N}.
\end{equation}
Then for each $n\in\{0,1,\ldots,N_T-1\}$ the following inequality holds, up to $\mathcal{O}(\Delta t_n^2)$-terms:
\begin{equation}\label{Sn+1 >= Sn one population}
S_{n+1}\geqslant S_n.
\end{equation}
\end{theorem}

\begin{proof}
Fix $n\in\{0,1,\ldots,N_T-1\}$ and take
\begin{equation*}
\psi:= V+\dfrac12 (W\star \rho_{n+1}).
\end{equation*}
By (\ref{time discrete evolution}), we derive
\begin{align}\label{Sn+1 first step}
\notag S_{n+1} &= \int_{\Omega} \bigl( V(x)+\dfrac12 (W\star \rho_{n+1})(x) \bigr) \rho_{n+1}(x)d\lambda^d(x)\\
&= \int_{\Omega} \bigl( V(\chi_n(x))+\dfrac12 (W\star \rho_{n+1})(\chi_n(x)) \bigr) \rho_{n}(x)d\lambda^d(x).
\end{align}
Note that, cf. (\ref{Taylor expansion psi, v well-behaved}) -- (\ref{Taylor expansion psi, v well-behaved, reordered}):
\begin{align}
V(\chi_n(x)) =& V(x)+ \Delta t_n v_n(x)\cdot \nabla V(x) + \mathcal{O}(\Delta t_n^2),\label{V chi x}\\
\notag (W\star \rho_{n+1})(\chi_n(x)) =& (W\star \rho_{n+1})(x) + \Delta t_n v_n(x)\cdot \nabla (W\star \rho_{n+1})(x) + \mathcal{O}(\Delta t_n^2)\\
=& (W\star \rho_{n+1})(x) + \Delta t_n v_n(x)\cdot (\nabla W\star \rho_{n+1})(x) + \mathcal{O}(\Delta t_n^2),\label{W star rho n+1 chi x}\\
\notag (W\star \rho_{n+1})(x) =& \int_{\Omega} W(x-y)\rho_{n+1}(y)d\lambda^d(y)\\
\notag =& \int_{\Omega} W(x-\chi_n(y))\rho_{n}(y)d\lambda^d(y)\\
\notag =& \int_{\Omega} W(x-y-\Delta t_n v_n(y))\rho_{n}(y)d\lambda^d(y)\\
\notag =& \int_{\Omega} W(x-y)\rho_{n}(y)d\lambda^d(y)\\
\notag &-\Delta t_n \int_{\Omega} v_n(y)\cdot \nabla W(x-y)\rho_n(y)d\lambda^d(y)+\mathcal{O}(\Delta t_n^2)\\
=& (W\star \rho_{n})(x) -\Delta t_n \int_{\Omega} v_n(y)\cdot \nabla W(x-y)\rho_n(y)d\lambda^d(y)+\mathcal{O}(\Delta t_n^2). \label{W star rho n+1 x}
\end{align}
In the second equality of (\ref{W star rho n+1 x}) we have used (\ref{time discrete evolution}) with $\psi(y)=W(x-y)$ (in this scope, $x$ is regarded as a parameter). Note that (\ref{W star rho n+1 x}) can also be written as
\begin{equation}\label{W star rho n+1 x first order}
(W\star \rho_{n+1})(x)= (W\star \rho_{n})(x) +\mathcal{O}(\Delta t_n),
\end{equation}
which we use now. Combination of (\ref{W star rho n+1 chi x}) -- (\ref{W star rho n+1 x first order}) yields
\begin{align}\label{W star rho n+1 chi x full expansion}
\notag (W\star \rho_{n+1})(\chi_n(x)) =& (W\star \rho_{n+1})(x) + \Delta t_n v_n(x)\cdot \nabla (W\star \rho_{n+1})(x) + \mathcal{O}(\Delta t_n^2)\\
\notag =& (W\star \rho_{n})(x) -\Delta t_n \int_{\Omega} v_n(y)\cdot \nabla W(x-y)\rho_n(y)d\lambda^d(y)\\
\notag & + \Delta t_n v_n(x)\cdot \nabla \Bigl( (W\star \rho_{n})(x) +\mathcal{O}(\Delta t_n) \Bigr) + \mathcal{O}(\Delta t_n^2)\\
\notag =& (W\star \rho_{n})(x) -\Delta t_n \int_{\Omega} v_n(y)\cdot \nabla W(x-y)\rho_n(y)d\lambda^d(y)\\
& + \Delta t_n v_n(x)\cdot\int_{\Omega} \nabla W(x-y) \rho_{n}(y)d\lambda^d(y) + \mathcal{O}(\Delta t_n^2).
\end{align}
Substituting (\ref{V chi x}) and (\ref{W star rho n+1 chi x full expansion}) into (\ref{Sn+1 first step}), we find
\begin{align}\label{Sn+1 second step}
\notag S_{n+1} =&  \int_{\Omega} \bigl( V(\chi_n(x))+\dfrac12 (W\star \rho_{n+1})(\chi_n(x)) \bigr) \rho_{n}(x)d\lambda^d(x)\\
\notag =& \int_{\Omega} \bigl( V(x)+ \Delta t_n v_n(x)\cdot \nabla V(x) \bigr) \rho_{n}(x)d\lambda^d(x)\\
\notag & + \dfrac12 \int_{\Omega} (W\star \rho_{n})(x)\rho_{n}(x)d\lambda^d(x)\\
\notag & - \dfrac12 \Delta t_n \int_{\Omega}\int_{\Omega} v_n(y)\cdot \nabla W(x-y) \rho_n(y)d\lambda^d(y)\rho_{n}(x)d\lambda^d(x)\\
& + \dfrac12 \Delta t_n \int_{\Omega}v_n(x)\cdot \int_{\Omega} \nabla W(x-y) \rho_{n}(y)d\lambda^d(y)\rho_{n}(x)d\lambda^d(x) + \mathcal{O}(\Delta t_n^2).
\end{align}
Since $W(\xi)=W(-\xi)$ for all $\xi\in\mathbb{R}^d$, we have that $\nabla W(\xi)=-\nabla W(-\xi)$. Using some elementary calculus, we derive
\begin{align}
\notag & - \dfrac12 \Delta t_n \int_{\Omega}\int_{\Omega} v_n(y)\cdot \nabla W(x-y) \rho_n(y)d\lambda^d(y)\rho_{n}(x)d\lambda^d(x)\\
\notag =& - \dfrac12 \Delta t_n \int_{\Omega}v_n(y)\cdot\int_{\Omega} \nabla W(x-y)\rho_{n}(x)d\lambda^d(x)\rho_n(y)d\lambda^d(y)\\
\notag =& \dfrac12 \Delta t_n \int_{\Omega}v_n(y)\cdot\int_{\Omega} \nabla W(y-x)\rho_{n}(x)d\lambda^d(x)\rho_n(y)d\lambda^d(y)\\
=& \dfrac12 \Delta t_n \int_{\Omega}v_n(x)\cdot\int_{\Omega} \nabla W(x-y)\rho_{n}(y)d\lambda^d(y)\rho_n(x)d\lambda^d(x).
\end{align}
In the last step we changed notation, replacing $x$ by $y$ and \textit{vice versa}. In this final expression we recognize the last term of (\ref{Sn+1 second step}), which can in short be written as
\begin{equation*}
\dfrac12 \Delta t_n \int_{\Omega}v_n(x)\cdot(\nabla W\star\rho_{n})(x)\rho_n(x)d\lambda^d(x).
\end{equation*}
The expression for $S_{n+1}$ in (\ref{Sn+1 second step}) now transforms into
\begin{align}\label{Sn+1 third step}
\notag S_{n+1} =& \int_{\Omega} \bigl( V(x)+ \Delta t_n v_n(x)\cdot \nabla V(x) \bigr) \rho_{n}(x)d\lambda^d(x)\\
\notag & + \dfrac12 \int_{\Omega} (W\star \rho_{n})(x)\rho_{n}(x)d\lambda^d(x)\\
\notag & + \Delta t_n \int_{\Omega}v_n(x)\cdot \int_{\Omega} \nabla W(x-y) \rho_{n}(y)d\lambda^d(y)\rho_{n}(x)d\lambda^d(x) + \mathcal{O}(\Delta t_n^2)\\
\notag =& \int_{\Omega} \bigl( V(x)+ \Delta t_n v_n(x)\cdot \nabla V(x) \bigr) \rho_{n}(x)d\lambda^d(x)\\
\notag & + \dfrac12 \int_{\Omega} (W\star \rho_{n})(x)\rho_{n}(x)d\lambda^d(x)\\
\notag & + \Delta t_n \int_{\Omega}v_n(x)\cdot (\nabla W\star\rho_{n})(x)\rho_{n}(x)d\lambda^d(x) + \mathcal{O}(\Delta t_n^2)\\
\notag =& \int_{\Omega} \bigl( V(x)+ \dfrac12(W\star \rho_{n})(x) \bigr) \rho_{n}(x)d\lambda^d(x)\\
& + \Delta t_n\int_{\Omega} v_n(x)\cdot \bigl(\nabla V(x)+(\nabla W\star\rho_{n})(x)\bigr)\rho_{n}(x)d\lambda^d(x) + \mathcal{O}(\Delta t_n^2).
\end{align}
Now we recognize the definitions of the entropy $S_n$, (\ref{def Sn}), and the velocity $v_n$, (\ref{vn V plus W star rho}), in the last equality of (\ref{Sn+1 third step}). We thus proceed:
\begin{align}
\notag S_{n+1} =& S_n + \Delta t_n\int_{\Omega} |v_n(x)|^2\rho_{n}(x)d\lambda^d(x) + \mathcal{O}(\Delta t_n^2)\\
 \geqslant& S_n +  \mathcal{O}(\Delta t_n^2).
\end{align}
Neglecting $\mathcal{O}(\Delta t_n^2)$-terms, we thus have that
\begin{equation*}
S_{n+1}\geqslant S_n,\hspace{1 cm}\text{for all }n\in\{0,1,\ldots,N_T-1\}.
\end{equation*}
\end{proof}

\begin{remark}
Note that (\ref{Sn+1 >= Sn one population}) is the discrete-in-time counterpart of (\ref{dS/dt positive}), up to $\mathcal{O}(\Delta t_n^2)$-approximation.
\end{remark}

\begin{remark}
Unlike the continuous-in-time case, no boundary terms appear in this time-discrete context. In Section \ref{section entropy inequality cont-in-time one population}, indeed these terms did appear, and we either assumed them to vanish (isolated system), or we incorporated them in the entropy inequality. However, in the proof of Theorem \ref{theorem time-discrete entropy ineq one component} no boundary terms are encountered. This is because, originally, the test functions $\psi$ were taken such that they vanish on $\partial \Omega$. This implicitly assumes that the time-discrete model is more suitable for describing an isolated system.
\end{remark}

\subsection{Multi-component crowd}
For a crowd that consists of multiple subpopulations, for each constituent of the mixture a governing equation like (\ref{time discrete evolution}) holds:
\begin{equation}\label{time discrete evolution constituent}
\int_{\Omega}\psi^{\alpha}(x)\rho^{\alpha}_{n+1}(x)d\lambda^d(x) =  \int_{\Omega}\psi^{\alpha}\bigl(\chi^{\alpha}_n(x)\bigr)\rho^{\alpha}_n(x)d\lambda^d(x),
\end{equation}
for any $\psi^{\alpha}\in L^1_{\mu^{\alpha}_{n+1}}(\Omega)$. The one-step push forward $\chi^{\alpha}_n$, is defined by
\begin{equation*}
\chi^{\alpha}_n(x):= x + \Delta t_n v^{\alpha}_n(x),
\end{equation*}
where we take the time-discrete velocities $v^{\alpha}_n$ of the form
\begin{equation}\label{def vel field time-discr entropy ineq multi-comp crowd}
v^{\alpha}_n:=\nabla V^{\alpha} + \sum_{\beta=1}^{\nu}\nabla W_{\beta}^{\alpha} \star \rho^{\beta}_n,\hspace{1 cm}\text{for all }\alpha\in\{1,2,\ldots,\nu\}.
\end{equation}
We assume that for all $\alpha, \beta\in\{1,2,\ldots,\nu\}$: $W_{\beta}^{\alpha}(\xi)=W_{\beta}^{\alpha}(-\xi)$ holds for all $\xi\in\mathbb{R}^d$, and $W_{\beta}^{\alpha}\equiv W_{\alpha}^{\beta}$.\\
Very much as in Section \ref{section entropy inequality disc-in-time one population} we define the time-discrete entropy for the multi-component crowd to be
\begin{equation}\label{def Sn total}
S_n:=\sum_{\alpha=1}^{\nu}\int_{\Omega} \Bigl[ V^{\alpha}(x)+\dfrac12 \sum_{\beta=1}^{\nu}(W_{\beta}^{\alpha}\star \rho^{\beta}_n)(x) \Bigr] \rho^{\alpha}_n(x)d\lambda^d(x), \hspace{1 cm}\text{for all }n\in\mathcal{N}.
\end{equation}

\begin{theorem}[Discrete-in-time entropy inequality]\label{theorem time-discrete entropy ineq multi-component}
Assume that $v^{\alpha}_n$ satisfies (\ref{def vel field time-discr entropy ineq multi-comp crowd}) and the accompanying conditions on $W_{\beta}^{\alpha}$.\\
Assume furthermore that the evolution of the system is governed by (\ref{time discrete evolution constituent}), and that the entropy is defined by (\ref{def Sn total}).\\
Then for each $n\in\{0,1,\ldots,N_T-1\}$ the following inequality holds, up to $\mathcal{O}(\Delta t_n^2)$-terms:
\begin{equation}\label{Sn+1 >= Sn multi-component}
S_{n+1}\geqslant S_n.
\end{equation}
\end{theorem}

\begin{proof}
Fix $n\in\{0,1,\ldots,N_T-1\}$ and take $\psi^{\alpha}:= V+\dfrac12 \sum_{\beta=1}^{\nu}(W_{\beta}^{\alpha}\star \rho^{\beta}_n)$ for each $\alpha\in\{1,2,\ldots,\nu\}$. Using (\ref{time discrete evolution constituent}), we derive that
\begin{align}\label{Sn+1 first step total}
\notag S_{n+1} &= \sum_{\alpha=1}^{\nu}\int_{\Omega} \Bigl[ V^{\alpha}(x)+\dfrac12 \sum_{\beta=1}^{\nu}(W_{\beta}^{\alpha}\star \rho^{\beta}_{n+1})(x) \Bigr] \rho^{\alpha}_{n+1}(x)d\lambda^d(x)\\
&= \sum_{\alpha=1}^{\nu}\int_{\Omega} \Bigl[ V^{\alpha}(\chi^{\alpha}_n(x))+\dfrac12 \sum_{\beta=1}^{\nu}(W_{\beta}^{\alpha}\star \rho^{\beta}_{n+1})(\chi^{\alpha}_n(x)) \Bigr] \rho^{\alpha}_n(x)d\lambda^d(x).
\end{align}
For fixed $\alpha,\beta\in\{1,2,\ldots,\nu\}$ similar expressions as in (\ref{V chi x}), (\ref{W star rho n+1 chi x}), (\ref{W star rho n+1 x}) hold. They can be derived in an identical manner. Therefore, we list here only the results:
\begin{align}
V^{\alpha}(\chi^{\alpha}_n(x)) =& V^{\alpha}(x)+ \Delta t_n v^{\alpha}_n(x)\cdot \nabla V^{\alpha}(x) + \mathcal{O}(\Delta t_n^2),\label{V chi x alpha}\\
\notag (W_{\beta}^{\alpha}\star \rho^{\beta}_{n+1})(\chi^{\alpha}_n(x)) =& (W_{\beta}^{\alpha}\star \rho^{\beta}_{n+1})(x) + \Delta t_n v^{\alpha}_n(x)\cdot \nabla( W_{\beta}^{\alpha}\star \rho^{\beta}_{n+1})(x) + \mathcal{O}(\Delta t_n^2)\\
=& (W_{\beta}^{\alpha}\star \rho^{\beta}_{n+1})(x) + \Delta t_n v^{\alpha}_n(x)\cdot (\nabla W_{\beta}^{\alpha}\star \rho^{\beta}_{n+1})(x) + \mathcal{O}(\Delta t_n^2),\label{W star rho n+1 chi x alpha}\\
(W_{\beta}^{\alpha}\star \rho^{\beta}_{n+1})(x) =& (W_{\beta}^{\alpha}\star \rho^{\beta}_{n})(x) -\Delta t_n \int_{\Omega} v^{\beta}_n(y)\cdot \nabla W_{\beta}^{\alpha}(x-y)\rho^{\beta}_n(y)d\lambda^d(y)+\mathcal{O}(\Delta t_n^2). \label{W star rho n+1 x alpha}
\end{align}
We combine (\ref{W star rho n+1 chi x alpha}) and (\ref{W star rho n+1 x alpha}), and (omitting the details) obtain the analogon of (\ref{W star rho n+1 chi x full expansion}):
\begin{align}\label{W star rho n+1 chi x alpha full expansion}
\notag (W_{\beta}^{\alpha}\star \rho^{\beta}_{n+1})(\chi^{\alpha}_n(x)) =& (W_{\beta}^{\alpha}\star \rho^{\beta}_{n})(x) -\Delta t_n \int_{\Omega} v^{\beta}_n(y)\cdot \nabla W_{\beta}^{\alpha}(x-y)\rho^{\beta}_n(y)d\lambda^d(y)\\
& + \Delta t_n v^{\alpha}_n(x)\cdot\int_{\Omega} \nabla W_{\beta}^{\alpha}(x-y) \rho^{\beta}_{n}(y)d\lambda^d(y) + \mathcal{O}(\Delta t_n^2).
\end{align}
By substitution of (\ref{V chi x alpha}) and (\ref{W star rho n+1 chi x alpha full expansion}) into (\ref{Sn+1 first step total}) we get
\begin{align}\label{Sn+1 second step total}
\notag S_{n+1} =&  \sum_{\alpha=1}^{\nu}\int_{\Omega} \Bigl[ V^{\alpha}(\chi^{\alpha}_n(x))+\dfrac12 \sum_{\beta=1}^{\nu}(W_{\beta}^{\alpha}\star \rho^{\beta}_{n+1})(\chi^{\alpha}_n(x)) \Bigr] \rho^{\alpha}_n(x)d\lambda^d(x)\\
\notag =& \sum_{\alpha=1}^{\nu}\int_{\Omega} \bigl( V^{\alpha}(x)+ \Delta t_n v^{\alpha}_n(x)\cdot \nabla V^{\alpha}(x) \bigr) \rho^{\alpha}_{n}(x)d\lambda^d(x)\\
\notag & + \dfrac12\sum_{\alpha=1}^{\nu} \int_{\Omega} \sum_{\beta=1}^{\nu}(W_{\beta}^{\alpha}\star \rho^{\beta}_{n})(x)\rho^{\alpha}_{n}(x)d\lambda^d(x)\\
\notag & - \dfrac12\Delta t_n \sum_{\alpha=1}^{\nu} \int_{\Omega}\sum_{\beta=1}^{\nu}\int_{\Omega} v^{\beta}_n(y)\cdot \nabla W_{\beta}^{\alpha}(x-y)\rho^{\beta}_n(y)d\lambda^d(y)\rho^{\alpha}_{n}(x)d\lambda^d(x)\\
& + \dfrac12 \Delta t_n \sum_{\alpha=1}^{\nu} \int_{\Omega} \sum_{\beta=1}^{\nu} v^{\alpha}_n(x)\cdot \int_{\Omega} \nabla W_{\beta}^{\alpha}(x-y) \rho^{\beta}_{n}(y)d\lambda^d(y)\rho^{\alpha}_{n}(x)d\lambda^d(x) + \mathcal{O}(\Delta t_n^2).
\end{align}
We now use that $W_{\beta}^{\alpha}(\xi)=W_{\beta}^{\alpha}(-\xi)$ for all $\xi\in\mathbb{R}^d$. This implies that $\nabla W_{\beta}^{\alpha}(\xi)=-\nabla W_{\beta}^{\alpha}(-\xi)$. We obtain
\begin{align}
\notag & - \dfrac12\Delta t_n \sum_{\alpha=1}^{\nu} \int_{\Omega}\sum_{\beta=1}^{\nu}\int_{\Omega} v^{\beta}_n(y)\cdot \nabla W_{\beta}^{\alpha}(x-y)\rho^{\beta}_n(y)d\lambda^d(y)\rho^{\alpha}_{n}(x)d\lambda^d(x)\\
\notag =& - \dfrac12\Delta t_n \sum_{\alpha=1}^{\nu} \sum_{\beta=1}^{\nu}\int_{\Omega}v^{\beta}_n(y)\cdot \int_{\Omega} \nabla W_{\beta}^{\alpha}(x-y)\rho^{\alpha}_{n}(x)d\lambda^d(x)\rho^{\beta}_n(y)d\lambda^d(y)\\
\notag =& \dfrac12\Delta t_n \sum_{\alpha=1}^{\nu} \sum_{\beta=1}^{\nu}\int_{\Omega}v^{\beta}_n(y)\cdot \int_{\Omega} \nabla W_{\beta}^{\alpha}(y-x)\rho^{\alpha}_{n}(x)d\lambda^d(x)\rho^{\beta}_n(y)d\lambda^d(y)\\
\notag =& \dfrac12\Delta t_n \sum_{\beta=1}^{\nu} \sum_{\alpha=1}^{\nu}\int_{\Omega}v^{\alpha}_n(x)\cdot \int_{\Omega} \nabla W_{\alpha}^{\beta}(x-y)\rho^{\beta}_{n}(y)d\lambda^d(y)\rho^{\alpha}_n(x)d\lambda^d(x)\\
=& \dfrac12\Delta t_n \sum_{\alpha=1}^{\nu}\int_{\Omega}\sum_{\beta=1}^{\nu} v^{\alpha}_n(x)\cdot (\nabla W_{\alpha}^{\beta}\star\rho^{\beta}_{n})(x)\rho^{\alpha}_n(x)d\lambda^d(x).
\end{align}
In the third step we changed notation: we replaced $x$ by $y$, $\alpha$ by $\beta$, and \textit{vice versa}.\\
\\
The expression for $S_{n+1}$ in (\ref{Sn+1 second step total}) transforms into
\begin{align}\label{Sn+1 third step total}
\notag S_{n+1} =& \sum_{\alpha=1}^{\nu}\int_{\Omega} \bigl( V^{\alpha}(x)+ \Delta t_n v^{\alpha}_n(x)\cdot \nabla V^{\alpha}(x) \bigr) \rho^{\alpha}_{n}(x)d\lambda^d(x)\\
\notag & + \dfrac12\sum_{\alpha=1}^{\nu} \int_{\Omega} \sum_{\beta=1}^{\nu}(W_{\beta}^{\alpha}\star \rho^{\beta}_{n})(x)\rho^{\alpha}_{n}(x)d\lambda^d(x)\\
\notag & + \Delta t_n \sum_{\alpha=1}^{\nu}\int_{\Omega}\sum_{\beta=1}^{\nu} v^{\alpha}_n(x)\cdot \Bigl(\dfrac12(\nabla W_{\beta}^{\alpha} + \nabla W_{\alpha}^{\beta}) \star\rho^{\beta}_{n}\Bigr)(x)\rho^{\alpha}_n(x)d\lambda^d(x)  + \mathcal{O}(\Delta t_n^2)\\
\notag =& \sum_{\alpha=1}^{\nu}\int_{\Omega} \Bigl[ V^{\alpha}(x)+\dfrac12 \sum_{\beta=1}^{\nu}(W_{\beta}^{\alpha}\star \rho^{\beta}_n)(x) \Bigr] \rho^{\alpha}_n(x)d\lambda^d(x)\\
\notag &+ \Delta t_n \sum_{\alpha=1}^{\nu}\int_{\Omega}v^{\alpha}_n(x)\cdot \Bigl[ \nabla V^{\alpha}(x)+\dfrac12 \sum_{\beta=1}^{\nu}\Bigl((\nabla W_{\beta}^{\alpha} + \nabla W_{\alpha}^{\beta}) \star\rho^{\beta}_{n}\Bigr)(x) \Bigr] \rho^{\alpha}_n(x)d\lambda^d(x) + \mathcal{O}(\Delta t_n^2)\\
=& S_n + \Delta t_n \sum_{\alpha=1}^{\nu}\int_{\Omega}|v^{\alpha}_n(x)|^2 \rho^{\alpha}_n(x)d\lambda^d(x) + \mathcal{O}(\Delta t_n^2),
\end{align}
where the last step is only valid under the assumption that the interactions $W^{\alpha}_{\beta}$ are symmetric (this is a hypothesis of the theorem).\\
\\
If we omit the $ \mathcal{O}(\Delta t_n^2)$-terms, we deduce from (\ref{Sn+1 third step total}) that
\begin{equation*}
S_{n+1}\geqslant S_n,\hspace{1 cm}\text{for all }n\in\{0,1,\ldots,N_T-1\},
\end{equation*}
which finishes the proof.
\end{proof}

\begin{remark}
Note that (\ref{Sn+1 >= Sn multi-component}) is the discrete-in-time counterpart of (\ref{dS/dt positive total}), again up to $\mathcal{O}(\Delta t_n^2)$-approximation.
\end{remark}

\subsection{Generalization to discrete measures}\label{section generalization entropy inequality time-discr to discrete measures}
Following the lines of Section \ref{section generalization entropy inequality to discrete measures}, we generalize the entropy inequality for the time-discrete model (cf. Theorem \ref{theorem time-discrete entropy ineq one component}) to the situation in which the discrete-in-time mass measure is of a discrete type.\\
\\
We consider a single population with discrete-in-time mass measure
\begin{equation*}
\mu_n := \sum_{i\in\mathcal{J}} \delta_{x_{i,n}}.
\end{equation*}
The push forward of the centres $x_{i,n}$ is given by (\ref{one-step motion mapping}) such that
\begin{equation}\label{evolution centres push-forward}
x_{i,n+1}:=\chi_n(x_{i,n})= x_{i,n} + \Delta t_n v_n(x_{i,n}),\hspace{1 cm}\text{for all }i\in\mathcal{J}.
\end{equation}
Here, the velocity field is the time-discrete version of (\ref{velocity gradient structure convolution with measure}):
\begin{equation*}
v_n:= \nabla V + \nabla W \star \mu_n,
\end{equation*}
where the convolution is given by
\begin{equation*}
(\nabla W \star \mu_n)(x):=\int_{\Omega\setminus\{x\}}\nabla W(x-y)d\mu_n(y),\hspace{1 cm}\text{for all }x\in\Omega.
\end{equation*}
We assume again that $W$ is symmetric:
\begin{equation*}
W(\xi)=W(-\xi),\hspace{1 cm}\text{for all }\xi\in\mathbb{R}^d.
\end{equation*}
The corresponding entropy of the system in $\Omega$ is
\begin{equation*}
S_n=\int_{\Omega} \Big\{ V(x)+\dfrac12 (W\star\mu_n)(x)\Big\}d\mu_n(x),\hspace{1 cm}\text{for all }n\in\mathcal{N}.
\end{equation*}
We again explicitly restrict ourselves to the situation that $x_{i,n}\in\Omega$ for all $i\in\mathcal{J}$ and all $n\in\mathcal{N}$; cf. Section \ref{section generalization entropy inequality to discrete measures}.\\
\\
Our aim is to derive the entropy inequality of Section \ref{section entropy inequality disc-in-time one population} for the proposed microscopic measure. That is, up to $\mathcal{O}(\Delta t_n^2)$-terms, we have that
\begin{equation*}
S_{n+1}\geqslant S_n,\hspace{1 cm}\text{for all }n\in\{0,1,\ldots,N_T-1\}.
\end{equation*}
Using the definition of $\mu_n$ (in terms of the sum of Dirac measures), the relation between $x_{i,n+1}$ and $x_{i,n}$ in (\ref{evolution centres push-forward}), we derive that
\begin{eqnarray}
\nonumber S_{n+1} &=& \sum_{i\in\mathcal{J}}\Big\{V(x_{i,n+1})+\dfrac12 \sum_{\substack{j\in\mathcal{J}\\x_{j,n+1}\neq x_{i,n+1}}}W(x_{i,n+1}-x_{j,n+1})\Big\}\\
\nonumber &=& \sum_{i\in\mathcal{J}}\Big\{V\big(x_{i,n} + \Delta t_n v_n(x_{i,n})\big)+\dfrac12 \sum_{\substack{j\in\mathcal{J}\\x_{j,n+1}\neq x_{i,n+1}}}W\big(x_{i,n} + \Delta t_n v_n(x_{i,n})-x_{j,n} - \Delta t_n v_n(x_{j,n})\big)\Big\}.\\
\label{Sn+1 discrete measure first step}
\end{eqnarray}

\begin{remark}
Note that the statement $x_{j,n+1}\neq x_{i,n+1}$ is equivalent to $x_{j,n}\neq x_{i,n}$, if $v_n$ is a Lipschitz continuous function with the Lipschitz constant strictly smaller than $1/\Delta t_n$. The implication
\begin{equation*}
\big[x_{j,n+1}\neq x_{i,n+1}\big] \Longrightarrow \big[x_{j,n}\neq x_{i,n} \big]
\end{equation*}
follows trivially from the unambiguity of the push forward, whereas the implication
\begin{equation*}
\big[x_{j,n}\neq x_{i,n}\big] \Longrightarrow \big[x_{j,n+1}\neq x_{i,n+1} \big]
\end{equation*}
is a consequence of similar arguments as in the proof of Lemma \ref{lemma diracs do not merge}.
\end{remark}

We can thus proceed on (\ref{Sn+1 discrete measure first step}):
\begin{eqnarray}
\nonumber S_{n+1} &=&  \sum_{i\in\mathcal{J}}\Big\{V\big(x_{i,n} + \Delta t_n v_n(x_{i,n})\big)+\dfrac12 \sum_{\substack{j\in\mathcal{J}\\x_{j,n}\neq x_{i,n}}}W\big(x_{i,n} + \Delta t_n v_n(x_{i,n})-x_{j,n} - \Delta t_n v_n(x_{j,n})\big)\Big\}\\
\nonumber &=& \sum_{i\in\mathcal{J}}\Big\{V(x_{i,n}) + \Delta t_n v_n(x_{i,n})\cdot \nabla V(x_{i,n}) +\mathcal{O}(\Delta t_n^2)\Big\}\\
\nonumber && +\dfrac12 \sum_{\substack{i,j\in\mathcal{J}\\x_{j,n}\neq x_{i,n}}}\Big\{W(x_{i,n}-x_{j,n}) + \Delta t_n \big(v_n(x_{i,n})-v_n(x_{j,n})\big)\cdot \nabla W(x_{i,n}-x_{j,n})+\mathcal{O}(\Delta t_n^2)\Big\}.\\
\label{Sn+1 discrete measure second step}
\end{eqnarray}
Here we have used Taylor series expansions around $x_{i,n}$ (in $V$) and around $x_{i,n}-x_{j,n}$ (in $W$), respectively. Note also that $\nabla W(x_{i,n}-x_{j,n})=-\nabla W(x_{j,n}-x_{i,n})$. We can thus write:
\begin{eqnarray}
\nonumber S_{n+1} &=& \sum_{i\in\mathcal{J}}\Big\{V(x_{i,n}) + \dfrac12 \sum_{\substack{j\in\mathcal{J}\\x_{j,n}\neq x_{i,n}}}W(x_{i,n}-x_{j,n})\Big\}\\
\nonumber && +\Delta t_n \sum_{i\in\mathcal{J}}v_n(x_{i,n})\cdot \nabla V(x_{i,n})\\
\nonumber && +\dfrac12\Delta t_n \sum_{\substack{i,j\in\mathcal{J}\\x_{j,n}\neq x_{i,n}}} v_n(x_{i,n})\cdot \nabla W(x_{i,n}-x_{j,n})\\
\nonumber && +\dfrac12\Delta t_n \sum_{\substack{i,j\in\mathcal{J}\\x_{j,n}\neq x_{i,n}}} v_n(x_{j,n})\cdot \nabla W(x_{j,n}-x_{i,n}) + \mathcal{O}(\Delta t_n^2)\\
\nonumber &=& S_n  +\Delta t_n \sum_{i\in\mathcal{J}}v_n(x_{i,n})\cdot \nabla V(x_{i,n})\\
\nonumber && +\Delta t_n \sum_{\substack{i,j\in\mathcal{J}\\x_{j,n}\neq x_{i,n}}} v_n(x_{i,n})\cdot \nabla W(x_{i,n}-x_{j,n}) + \mathcal{O}(\Delta t_n^2)\\
\nonumber &=& S_n  +\Delta t_n \sum_{i\in\mathcal{J}}v_n(x_{i,n})\cdot \Big\{\nabla V(x_{i,n})+ \sum_{\substack{j\in\mathcal{J}\\x_{j,n}\neq x_{i,n}}} \nabla W(x_{i,n}-x_{j,n})\Big\} + \mathcal{O}(\Delta t_n^2)\\
\label{Sn+1 discrete measure third step}
\end{eqnarray}
Omitting $\mathcal{O}(\Delta t_n^2)$-terms, (\ref{Sn+1 discrete measure third step}) we arrive at
\begin{eqnarray}
\nonumber S_{n+1} &=&  S_n  +\Delta t_n \sum_{i\in\mathcal{J}}v_n(x_{i,n})\cdot \Big\{\nabla V(x_{i,n})+ (\nabla W\star\mu_n)(x_{i,n})\Big\}\\
\nonumber &=& S_n +\Delta t_n \int_{\Omega}|v_n(x)|^2d\mu_n(x)\\
&\geqslant& S_n.\label{Sn+1 discrete measure fourth step}
\end{eqnarray}
The inequality in (\ref{Sn+1 discrete measure fourth step}) is an analogon of the statement of Theorem \ref{theorem time-discrete entropy ineq one component} for a discrete mass measure. We do not give details on the multi-component case.

\section{Two-scale phenomena}\label{section two-scales no sing cont}
At this stage, we want to point out which specific types of measures we have in mind to use in our model. Throughout the previous sections we have already given clues about the choices we intend to make, for example in Sections \ref{section introduction modelling approaches}, \ref{section Mass measure}, \ref{section generalization mixture} and \ref{section weak formulation}. As anticipated, we are mainly interested in two sorts of measures.\\
\\
The first option is a discrete measure (or: microscopic mass measure), as introduced already in Section \ref{section micro mass measure}. In general, a discrete measure may be the (weighted) sum of a countable yet infinite number of Dirac measures (cf. Lemma \ref{lemma characterization discrete measure}). We restrict ourselves here to a finite - but arbitrary - number of point masses (read: people). This approach enables us to trace the exact motion of a particular individual.\\
\\
Secondly, if we work with an absolutely continuous mass measure (w.r.t. the Lebesgue measure), then we are provided with a density by the Radon-Nikodym Theorem. This option (a macroscopic mass measure) has already been introduced in Section \ref{section macro mass measure}. In fact, this perspective corresponds to considering the crowd as a fluid (or a large `cloud'). We decide to do so if the number of people is large, and if we, moreover, are not interested in what happens exactly at the individual's level. Such local fluctuations have been averaged out.\\
\\
We only allow these two options, which means that from now on we explicitly exclude singular continuous measures from our area of interest. This is because it is not self-evident what interpretation we should give to such measure. This decision has been anticipated in Section \ref{section generalization mixture}.\\
\\
In Section \ref{section introduction}, we have already indicated that we are especially interested in the interplay between microscopic and macroscopic mass measures. Piccoli et al. \cite{Piccoli2010} do so by including a discrete and an absolutely continuous part in one measure. Furthermore a tuning parameter (taking values from 0 to 1) is used to enable a transition from fully microscopic to fully macroscopic. Speaking in terms of mixture theory (cf. Section \ref{section mixture theory and thermodynamics}), this means that the discrete and absolutely continuous part together describe a single constituent. They can not act independently. We already remarked in Section \ref{section weak formulation} that we consider this to be not very useful. In the numerical scheme that is described in Section \ref{section numerical scheme}, we therefore do not include this superposition of a micro and a macro part in one measure. If desired, it can however be incorporated without too much difficulty.\\
\\
Our intention is a bit different: We namely want to distinguish between subpopulations and assign to them a measure that corresponds to their `character'. More specifically, we want to use a discrete measure only if this subpopulation consists of a limited number of individuals which are of special interest. This is why we only consider discrete measures that are constituted of a \textit{finite} number of Dirac measures. Large crowds are in our framework automatically represented by an absolutely continuous mass measure. Such subpopulations consist of people that are not special (otherwise they should have been included in an independent discrete measure), and therefore it suffices to be able to observe only global behaviour. It is of no use to `do effort' to capture information from the microscopic level in a discrete measure, if we are not interested in that information any way. The use of both microscopic and macroscopic mass measures in one unified framework, makes that we work in a \textit{two-scale} setting. Our perspective is twofold: on one hand we observe from a macroscopic point of view, while on the other hand we can detect behaviour at an individual's level if we have identified this individual as `special'. 
\begin{remark}
In the future, we possibly want to come back to modelling the cloud as a particle system, consisting of a large number of individuals. We then aim at comparing the particle system's results to the results of this thesis. Note that these two approaches are incorporated in the measure-theoretical framework we are dealing with.
\end{remark}
By introducing separate mass measures, we allow the subpopulations to evolve \textbf{differently}. For example, we assign to each subpopulation its own desired velocity field. That is, for each subpopulation there is a separate goal that it wants to achieve. Moreover, it is possible to have asymmetric interactions. We have already shortly referred to what we called "predator-prey relations" in Remark \ref{remark predator-prey}. In such asymmetric situation preys are only repelled from predators, but predators are attracted to preys.\\
\\
Nature turns out to provide very clear illustrations of systems that are of a two-scale character and at the same time display asymmetric interactions of the "predator-prey" type. We point out two examples. Fascinating interaction takes place between flocks of small birds (starlings, say) and one or more larger predator birds (e.g. hawks), as is illustrated by Figure \ref{Figure flock attack}. The huge number of starlings makes it nearly impossible to distinguish between individuals in the flock. One clearly perceives a continuum-like cloud of birds; modelling this cloud by means of an absolutely continuous measure thus seems natural. As the hawk attacks the starlings, the flock reacts to the approaching enemy as if there were some macroscopic coordination. Increase and decrease of the density in the flock can be observed. It is striking to see that the group does not fall apart.\footnote{Movies of such phenomena are available on the internet, see e.g. \texttt{www.youtube.com/watch?v=b8eZJnbDHIg\&feature=related} .} Similar effects can be seen in shoals of fish being attacked. The interaction between the predator fish and the shoal triggers complex structures to appear and enables sudden transitions between macroscopic patterns. The phenomena we refer to are described e.g. in \cite{SwarmScienceNews,VicsekCollMotion}.\\
\\
Two-scale effects also occur in human crowds, when there is special interaction between a specific individual (or a limited number of them) and the rest of the crowd. These phenomena might be harder to visualize (because the individual and the group members have the same size), but this does not mean that they are not there. Phenomena of the "predator-prey" type occur, for example when a crowd is attacked by some criminal or terrorist. More `friendly' interplay is present when the special individuals have the role of tourist guides, leaders, firemen, safety guards et cetera. Our aim in Section \ref{section numerical scheme} and further is to capture these two-scale phenomena with our model.

\begin{figure}[h]
\centering
\vspace{0 cm}
\includegraphics[width=0.6\linewidth]{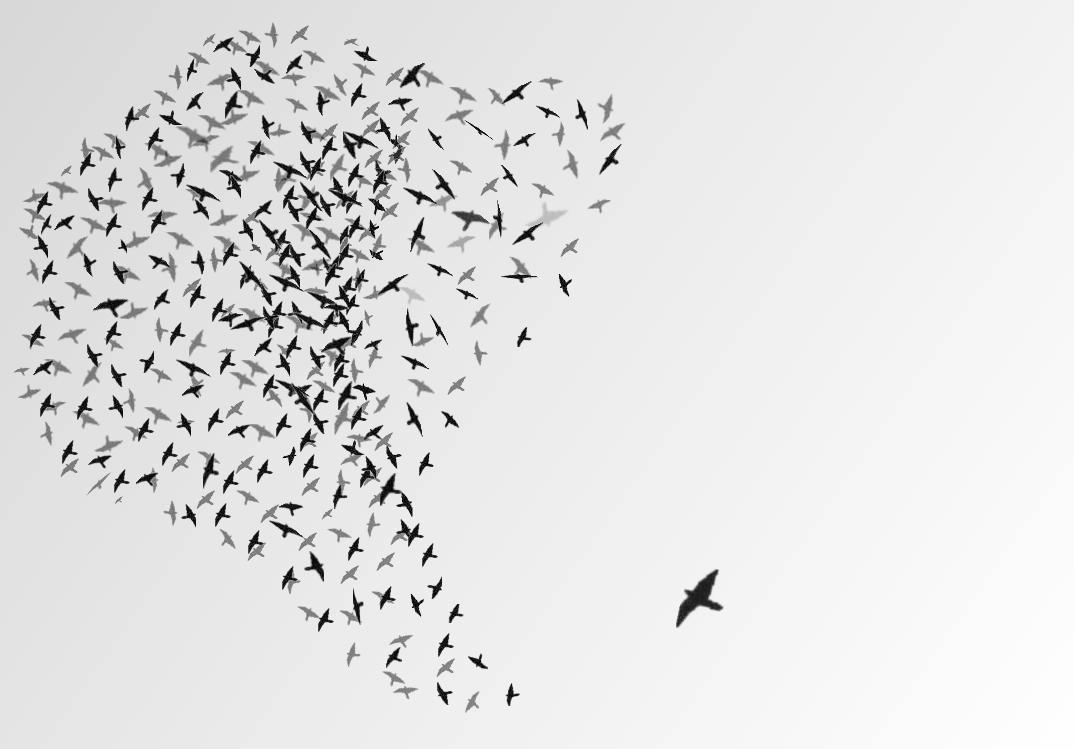}\\
\vspace{0 cm}
\caption{Illustration of two-scale, asymmetric interaction between a single predator bird and a flock of smaller individuals.}\label{Figure flock attack}
\end{figure}

\newpage
\chapter{Numerical scheme}\label{section numerical scheme}
In this section, we propose a numerical scheme for finding solutions of the time-discrete model of Section \ref{section derivation time-discrete}. See Definition \ref{def time-discrete solution}. The scheme was originally developed in \cite{Piccoli2010}.\\
\\
We consider the set $\Omega\subset \mathbb{R}^2$ such that $\Omega:=[0,\mathcal{L}]\times[0,\mathcal{W}]$, where $\mathcal{L}, \mathcal{W}\in (0,\infty)$. We restrict ourselves to the situation in which each component of the crowd is either discrete (concentrated in a finite number of points) or absolutely continuous (w.r.t. $\lambda^d$), as was motivated in Section \ref{section two-scales no sing cont}.  Note that we have seen in the proof of Theorem \ref{Thm existence time-discrete} that the push forward of a discrete measure is again discrete (and similar for the push forward of an absolutely continuous measure). The mass measure of each constituent is thus of the same type throughout time.\\

\section{Types of measures}
\subsection{Discrete measure}
Suppose that we have a discrete mass measure for constituent $\alpha$, consisting of $N^{\alpha}$ distinct Dirac distributions. We choose the particular form
\begin{equation*}
\mu^{\alpha}_n := M^{\alpha}\sum_{i=1}^{N^{\alpha}} \delta_{x^{\alpha}_{i,n}},
\end{equation*}
for the mass measure at time slice $n$. Here, $M^{\alpha}\in\mathbb{R}^+$ is a proportionality constant that makes it possible to compare a sum of Dirac measures (they measure \textit{the number of people}) to an absolutely continuous measure (that measures \textit{kilograms of people}). By  $\{x^{\alpha}_{i,n}\}_{i=1}^{N^{\alpha}} \subset \Omega$ we denote the set of (time-dependent) centres at time $t=t_n$.

\subsection{Absolutely continuous measure: spatial approximation of density}
If the constituent $\alpha$ has a corresponding density $\rho^{\alpha}_n$, then we need to approximate $\rho^{\alpha}_n$ in space. Therefore, we fix $N_{\mathcal{L}}, N_{\mathcal{W}}\in\mathbb{N}$ and we subdivide $\Omega$ into $N_{\mathcal{L}}\cdot N_{\mathcal{W}}$ cells. Define $h_{\mathcal{L}}:= \mathcal{L}/N_{\mathcal{L}}$ and $h_{\mathcal{W}}:= \mathcal{W}/N_{\mathcal{W}}$, the horizontal and vertical grid size, respectively. The size of a cell is thus $h_{\mathcal{L}}h_{\mathcal{W}}$. Each of the cells $E_{j,k}$ is given an index $(j,k)\in \mathcal{K}:=\big\{1,2,\ldots,N_{\mathcal{L}}\big\} \times\big\{1,2,\ldots,N_{\mathcal{W}}\big\}$, such that
\begin{equation*}
 E_{j,k}:= \big[(j-1)h_{\mathcal{L}}, jh_{\mathcal{L}}\big]\times \big[(k-1)h_{\mathcal{W}}, kh_{\mathcal{W}}\big].
\end{equation*}
Note that the boundaries of these cells overlap each other. This is, however, not a serious problem, since the boundaries are null sets w.r.t. $\lambda^d$. Up to a null set the cells are mutually disjoint.\\
We sketch such grid in Figure \ref{figure grid}.

\begin{figure}[ht]
\centering
\begin{tikzpicture}
    \draw[->,thick] (0,0)--(10,0) ;
    \draw[->,thick] (0,0)--(0,7);

    \draw (0,0) rectangle (5,3);
    \draw[-] (1.25,0) -- (1.25,3);
    \draw[-] (2.5,0) -- (2.5,3);
    \draw[-] (3.75,0) -- (3.75,3);

    \draw[-] (0,1) -- (5,1);
    \draw[-] (0,2) -- (5,2);

    \draw (0,5) rectangle (8.75, 6);
    \draw (7.5, 0) rectangle (8.75, 6);

    \draw[-,thick] (0,6) -- (8.75, 6);
    \draw[-,thick] (8.75,0) -- (8.75, 6);

    \draw[-] (7.5, 1) -- (8.75, 1);
    \draw[-] (7.5, 2) -- (8.75, 2);
    \draw[-] (7.5, 3) -- (8.75, 3);
    \draw[-] (7.5, 4) -- (8.75, 4);

    \draw[-] (1.25, 5) -- (1.25, 6);
    \draw[-] (2.5, 5) -- (2.5, 6);
    \draw[-] (3.75, 5) -- (3.75, 6);
    \draw[-] (5, 5) -- (5, 6);
    \draw[-] (6.25, 5) -- (6.25, 6);

    \draw[dotted] (5, 1) -- (7.5, 1);
    \draw[dotted] (5, 2) -- (7.5, 2);
    \draw[dotted] (5, 3) -- (7.5, 3);
    \draw[dotted] (6.875, 4) -- (7.5, 4);

    \draw[dotted] (1.25, 3) -- (1.25, 5);
    \draw[dotted] (2.5, 3) -- (2.5, 5);
    \draw[dotted] (3.75, 3) -- (3.75, 5);
    \draw[dotted] (5, 3) -- (5, 5);
    \draw[dotted] (6.25, 4.5) -- (6.25, 5);

    \draw[<->] (1.25,-0.25)-- node[anchor=north]{$h_{\mathcal{L}}$} (2.5,-0.25);
    \draw[<->] (-0.25,1)-- node[anchor=east]{$h_{\mathcal{W}}$} (-0.25,2);
    \draw[<->] (0,-1)-- node[anchor=north]{${\mathcal{L}}$} (8.75,-1);
    \draw[<->] (-1,0)-- node[anchor=east]{${\mathcal{W}}$} (-1,6);

    \node at (0.625,0.5){$E_{1,1}$};
    \node at (4.275,2.5){$E_{j,k}$};
    \node at (8.125,0.5){$E_{N_{\mathcal{L}},1}$};
    \node at (0.625,5.5){$E_{1,N_{\mathcal{W}}}$};
    \node at (8.125,5.5){$E_{N_{\mathcal{L}},N_{\mathcal{W}}}$};

\end{tikzpicture}
\caption{The spatial discretization of the domain by subdivision into $N_{\mathcal{L}}\cdot N_{\mathcal{W}}$ cells.}\label{figure grid}
\end{figure}
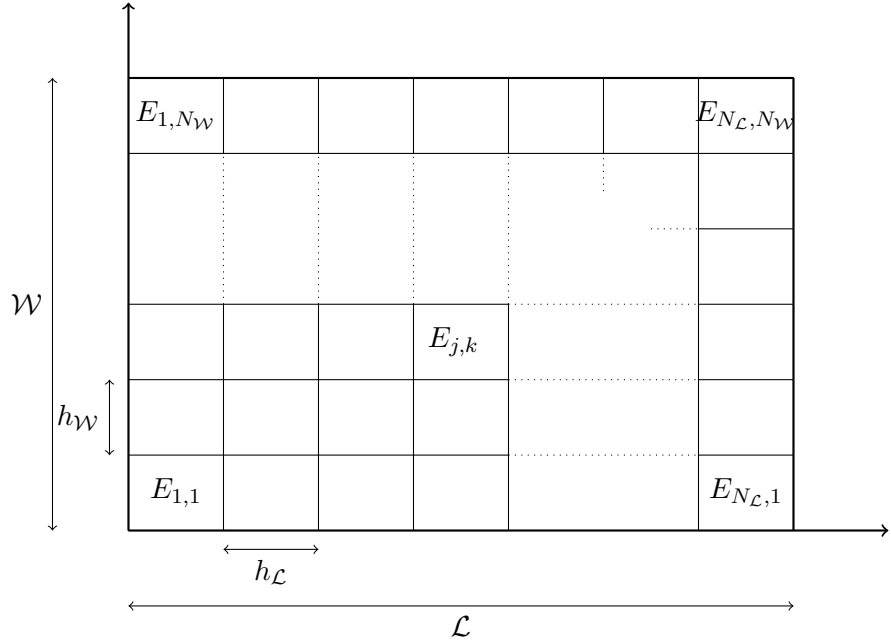

We approximate the density $\rho^{\alpha}_n$ by a piecewise constant function $\tilde{\rho}^{\alpha}_n$. This means that for each $(j,k)\in\mathcal{K}$ there is a $\rho^{\alpha}_{(j,k),n}\in\mathbb{R}^+$ such that
\begin{equation*}
\tilde{\rho}^{\alpha}_n(x) := \rho^{\alpha}_{(j,k),n},
\end{equation*}
for all $x$ in the interior of $E_{j,k}$. Since the boundaries of the cells form a null set, it is not important how we define $\tilde{\rho}^{\alpha}_n$ there. 

\section{Calculating velocities}\label{section calculating velocities}
In Section \ref{section time-discrete vel field} we have proposed time-discrete velocity fields $v^{\alpha}_n$ defined by
\begin{equation}\label{v before approx}
v^{\alpha}_n(x):= v_{\text{des}}^{\alpha}(x)+\sum_{\beta=1}^{\nu}\int_{\Omega \setminus \{x\}} f^{\alpha}_{\beta}(|y-x|)g(\theta_{xy}^{\alpha})\dfrac{y-x}{|y-x|}d\mu^{\beta}_n(y), \hspace{1 cm} \text{for all }x\in\Omega.
\end{equation}
We approximate $v^{\alpha}_n$ by $\tilde{v}^{\alpha}_n$. Where and how an approximation is needed, depends on the types of measures associated to $\alpha$ and $\beta$. We discuss this aspect in more detail in Sections \ref{subsect alpha discrete} and \ref{subsect alpha abs cont}.

\subsection{Discrete $\mu^{\alpha}_n$}\label{subsect alpha discrete}
If $\mu^{\alpha}_n$ is a discrete measure, then the evaluation of $v^{\alpha}_n$ is required in each of the points $x^{\alpha}_{i,n}\in\Omega$. The desired velocity $v_{\text{des}}^{\alpha}$ can be evaluated in these points without any difficulty. However, attention has to be paid regarding the integral terms with respect to $\mu^{\beta}_n$ that arise in the social velocity part of (\ref{v before approx}). We distinguish between the cases:
\begin{enumerate}[(i)]
  \item $\mu^{\beta}_n$ is discrete;
  \item $\mu^{\beta}_n$ is absolutely continuous.
\end{enumerate}
Firstly, let $\mu^{\beta}_n$ for each $n$ be given as
\begin{equation*}
\mu^{\beta}_n := M^{\beta}\sum_{i=1}^{N^{\beta}} \delta_{x^{\beta}_{i,n}},
\end{equation*}
for some $M^{\beta}\in\mathbb{R}^+$, $N^{\beta}\in\mathbb{N}$ and a set of centres $\{x^{\beta}_{i,n}\}_{i=1}^{N^{\beta}} \subset \Omega$. The corresponding interaction term can be calculated exactly as follows:
\begin{equation}\label{interaction integral disc disc}
\int_{\Omega \setminus \{x\}} f^{\alpha}_{\beta}(|y-x|)g(\theta_{xy}^{\alpha})\dfrac{y-x}{|y-x|}d\mu^{\beta}_n(y) = M^{\beta}\sum_{\substack{i=1\\x^{\beta}_{i,n}\neq x}}^{N^{\beta}} f^{\alpha}_{\beta}(|x^{\beta}_{i,n}-x|)g(\theta_{xx^{\beta}_{i,n}}^{\alpha})\dfrac{x^{\beta}_{i,n}-x}{|x^{\beta}_{i,n}-x|}.
\end{equation}
Any of the points $x^{\alpha}_{i,n}$ can be taken as a choice of $x$.
\begin{remark}\label{remark take vel and initial condition s.t. diracs do not coincide}
Note that the exclusion of $\{x\}$ from the domain of integration is intended to avoid interaction of a point mass with itself for the case $\alpha=\beta$. However, in the right-hand side of (\ref{interaction integral disc disc}) the contribution of interactions are also excluded if the positions of two distinct point masses coincide. That is, if we have $x^{\alpha}_{i,n}=x^{\beta}_{j,n}$ for $\alpha=\beta$ but $i\neq j$ (two distinct points from the same subpopulation), or $x^{\alpha}_{i,n}=x^{\beta}_{j,n}$ where $\alpha\neq\beta$ (two coinciding points from different subpopulations). We require our velocity field and initial conditions to be such that neither of these two situations occur; see also Remark \ref{remark conditions for non-merging diracs and bounded velocity}.
\end{remark}
Secondly, we consider the interaction integral for the case that $\mu^{\beta}_n$ is absolutely continuous w.r.t. $\lambda^d$. The approximation $\tilde{\rho}^{\beta}_n$ of the density is used to obtain
\begin{eqnarray}
\nonumber \int_{\Omega \setminus \{x\}} f^{\alpha}_{\beta}(|y-x|)g(\theta_{xy}^{\alpha})\dfrac{y-x}{|y-x|}d\mu^{\beta}_n(y) &=& \int_{\Omega} f^{\alpha}_{\beta}(|y-x|)g(\theta_{xy}^{\alpha})\dfrac{y-x}{|y-x|}\rho^{\beta}_n(y)d\lambda^d(y)\\
\nonumber &\approx& \int_{\Omega} f^{\alpha}_{\beta}(|y-x|)g(\theta_{xy}^{\alpha})\dfrac{y-x}{|y-x|}\tilde{\rho}^{\beta}_n(y)d\lambda^d(y)\\
\nonumber &=& \sum_{(j,k)\in\mathcal{K}} \rho^{\beta}_{(j,k),n} \int_{E_{j,k}} f^{\alpha}_{\beta}(|y-x|)g(\theta_{xy}^{\alpha})\dfrac{y-x}{|y-x|}d\lambda^d(y).\\
\label{approximation integral by subst piecewise const density}
\end{eqnarray}
The integrals over $E_{j,k}$ are approximated by a two-dimensional form of the trapezoid rule, using the four vertices of the rectangle (Newton-Cotes). Let these vertices be called $y_{j,k}^i$, for $i\in\{1,2,3,4\}$. For example, if $y_{j,k}^4$ is the top right vertex, then $y_{j,k}^4:=\big(jh_{\mathcal{L}}, kh_{\mathcal{W}}\big)$.\\
We obtain
\begin{equation}\label{approximation integral by using four vertices}
\int_{E_{j,k}} f^{\alpha}_{\beta}(|y-x|)g(\theta_{xy}^{\alpha})\dfrac{y-x}{|y-x|}d\lambda^d(y) \approx \dfrac{h_{\mathcal{L}}h_{\mathcal{W}}}{4}\sum_{i=1}^4 f^{\alpha}_{\beta}(|y_{j,k}^i-x|)g(\theta_{x y_{j,k}^i}^{\alpha})\dfrac{y_{j,k}^i-x}{|y_{j,k}^i-x|}.
\end{equation}
However, if $x=y_{j,k}^i$ for some $i$, we need to adapt this approximation, since typically $f^{\alpha}_{\beta}$ has a singularity at 0. We then choose to use the interaction with the midpoint of the cell instead. Let $y_{j,k}^{\text{(m)}}:= \big((j-\frac12)h_{\mathcal{L}}, (k-\frac12)h_{\mathcal{W}}\big)$ denote the midpoint of $E_{j,k}$. Consequently, we use the approximation
\begin{equation}\label{approximation integral by using midpoint}
\int_{E_{j,k}} f^{\alpha}_{\beta}(|y-x|)g(\theta_{xy}^{\alpha})\dfrac{y-x}{|y-x|}d\lambda^d(y) \approx h_{\mathcal{L}}h_{\mathcal{W}} f^{\alpha}_{\beta}(|y_{j,k}^{\text{(m)}}-x|)g(\theta_{x y_{j,k}^{\text{(m)}}}^{\alpha})\dfrac{y_{j,k}^{\text{(m)}}-x}{|y_{j,k}^{\text{(m)}}-x|},
\end{equation}
for any $x$ coinciding with a vertex of the cell $E_{j,k}$.
\begin{remark}\label{remark particle vertex coincide only num problem}
If the position $x$ coincides with a vertex of a cell, this causes a problem solely from the numerical point of view. The general scheme (trapezoid rule using the four vertices) is in that case no longer applicable for approximating the value of the integral over that cell. The concerning singularity does not cause a problem from the perspective of mathematical analysis. Namely, if we demand $f^{\alpha}_{\beta}(s)\sim 1/s$, and assume that the density is uniformly bounded, then
\begin{eqnarray}
\nonumber \Big|\int_{\Omega \setminus \{x\}} f^{\alpha}_{\beta}(|y-x|)g(\theta_{xy}^{\alpha})\dfrac{y-x}{|y-x|}d\mu^{\beta}_n(y)\Big| &\leqslant& \int_{\Omega} \Big|f^{\alpha}_{\beta}(|y-x|)g(\theta_{xy}^{\alpha})\dfrac{y-x}{|y-x|}\rho^{\beta}_n(y)\Big|d\lambda^d(y)\\
\nonumber &=& \int_{\Omega} \Big|f^{\alpha}_{\beta}(|y-x|)\Big|\,\underbrace{\Big|g(\theta_{xy}^{\alpha})\Big|}_{\leqslant 1}\,\underbrace{\Big|\dfrac{y-x}{|y-x|}\Big|}_{=1}\,\underbrace{\Big|\rho^{\beta}_n(y)\Big|}_{\leqslant \|\rho^{\beta}_n\|_{\infty}}d\lambda^d(y)\\
\nonumber &\leqslant& \|\rho^{\beta}_n\|_{\infty}\int_{\Omega} \Big|f^{\alpha}_{\beta}(|y-x|)\Big|d\lambda^d(y)\\
 &\leqslant& \|\rho^{\beta}_n\|_{\infty}\int_{B(x,R)} \Big|f^{\alpha}_{\beta}(|y-x|)\Big|d\lambda^d(y).
\end{eqnarray}
Here, $B(x,R)$ is the unit ball in $\mathbb{R}^2$ around $x$ with radius $R:=\sqrt{\mathcal{L}^2+\mathcal{W}^2}$. Due to the choice of this specific radius it is guaranteed that $\Omega\subset B(x,R)$ for each $x\in\Omega$. Let us now only take the repulsive part around $x$ into consideration. The contribution of the attraction-part (if present, cf. Section \ref{section specification velocity} and Figure \ref{Figure graphs fAR fR}) is finite any way.\\
Let $R_r$ denote the radius of repulsion. We have
\begin{eqnarray}
\nonumber \|\rho^{\beta}_n\|_{\infty}\int_{B(x,R_r)} \Big|f^{\alpha}_{\beta}(|y-x|)\Big|d\lambda^d(y) &=& 2\pi \|\rho^{\beta}_n\|_{\infty}\int_0^{R_r} \Big|f^{\alpha}_{\beta}(r)\Big|rdr\\
\nonumber &=& 2\pi \|\rho^{\beta}_n\|_{\infty}\int_0^{R_r} \Big(\dfrac{R_r}{r}-1 \Big)rdr\\
\nonumber &=& \pi R_r^2 \|\rho^{\beta}_n\|_{\infty}\\
&<& \infty.
\end{eqnarray}
We have taken the repulsive part of $f^{\alpha}_{\beta}$ in its most simple form, but yet satisfying $f^{\alpha}_{\beta}(s)\sim 1/s$.\footnote{In fact, it suffices to impose the weaker demand that $f^{\alpha}_{\beta}(s)\sim s^{-\gamma}$ for any $\gamma<2$} A multiplicative constant can of course be added without loosing finiteness of the integral.
\end{remark}
\begin{remark}
Contrary to what is stated in Remark \ref{remark take vel and initial condition s.t. diracs do not coincide}, here we do not wish to forbid the situation that a point mass coincides with the vertex of a cell. If two distinct point masses are located in the same position (the case considered in Remark \ref{remark take vel and initial condition s.t. diracs do not coincide}), this corresponds to a physically impossible situation. Although probably only rarely a point mass will be located exactly on a cell's vertex, this is \textit{not} physically undesirable, and we \textit{do} want to allow for it. For those situations, we have therefore introduced an alternative approximation (\ref{approximation integral by using midpoint}) to circumvent the occurring problems in the numerical scheme.
\end{remark}

\subsection{Absolutely continuous $\mu^{\alpha}_n$}\label{subsect alpha abs cont}
If $\mu^{\alpha}_n$ is absolutely continuous with respect to $\lambda^d$, then $v^{\alpha}_n$ is approximated by a piecewise constant function. Let $\tilde{v}^{\alpha}_n(y_{j,k}^{\text{(m)}})$ be (an approximation of) the velocity of subpopulation $\alpha$ in the midpoint of cell $E_{j,k}$. Then for each $x$ in the interior of $E_{j,k}$, the approximated velocity field $\tilde{v}^{\alpha}_n(x)$ is given by
\begin{equation*}
\tilde{v}^{\alpha}_n(x)\equiv\tilde{v}^{\alpha}_n(y_{j,k}^{\text{(m)}}).
\end{equation*}
In each midpoint $y_{j,k}^{\text{(m)}}$ evaluating the desired velocity $v_{\text{des}}^{\alpha}$ does not cause any problem. In order to deal with the integral terms in the social velocity, we again take into consideration which type of measure $\mu^{\beta}_n$ is.\\
\\
Firstly, we treat the case when $\mu^{\beta}_n$ is a discrete measure. We specifically suppose that $\mu^{\beta}_n$ has the form
\begin{equation*}
\mu^{\beta}_n := M^{\beta}\sum_{i=1}^{N^{\beta}} \delta_{x^{\beta}_{i,n}}.
\end{equation*}
If, for a given cell $E_{j,k}$, the midpoint $y_{j,k}^{\text{(m)}}$ does not coincide with any of the centre points $x^{\beta}_{i,n}$, then
\begin{equation}\label{integral evaluated in midpoint transform to sum diracs}
\int_{\Omega \setminus \big\{y_{j,k}^{\text{(m)}}\big\}} f^{\alpha}_{\beta}(|y-y_{j,k}^{\text{(m)}}|)g(\theta_{y_{j,k}^{\text{(m)}}y}^{\alpha})\dfrac{y-y_{j,k}^{\text{(m)}}}{|y-y_{j,k}^{\text{(m)}}|}d\mu^{\beta}_n(y) = M^{\beta}\sum_{i=1}^{N^{\beta}} f^{\alpha}_{\beta}(|x^{\beta}_{i,n}-y_{j,k}^{\text{(m)}}|) g(\theta_{y_{j,k}^{\text{(m)}}x^{\beta}_{i,n}}^{\alpha})\dfrac{x^{\beta}_{i,n}-y_{j,k}^{\text{(m)}}}{|x^{\beta}_{i,n}-y_{j,k}^{\text{(m)}}|}.
\end{equation}
If by chance one of the centres, say $x^{\beta}_{m,n}$, of a Dirac mass coincides with the midpoint $y_{j,k}^{\text{(m)}}$, then we have to adapt our scheme. We replace the direct interaction between these two points, by an average over the interaction between $x^{\beta}_{m,n}$ and each of the four vertices $y_{j,k}^i$ of the cell. That is, we use
\begin{equation*}
\dfrac14 \sum_{i=1}^4 f^{\alpha}_{\beta}(|x^{\beta}_{m,n}-y_{j,k}^{i}|) g(\theta_{y_{j,k}^{i}x^{\beta}_{m,n}}^{\alpha})\dfrac{x^{\beta}_{m,n}-y_{j,k}^{i}}{|x^{\beta}_{m,n}-y_{j,k}^{i}|}.
\end{equation*}
Using this expression, we obtain the following approximation:
\begin{eqnarray}
\nonumber &&\int_{\Omega \setminus \big\{y_{j,k}^{\text{(m)}}\big\}} f^{\alpha}_{\beta}(|y-y_{j,k}^{\text{(m)}}|)g(\theta_{y_{j,k}^{\text{(m)}}y}^{\alpha})\dfrac{y-y_{j,k}^{\text{(m)}}}{|y-y_{j,k}^{\text{(m)}}|}d\mu^{\beta}_n(y)\\ \nonumber &\approx& M^{\beta}\sum_{\substack{i=1\\i\neq m}}^{N^{\beta}} f^{\alpha}_{\beta}(|x^{\beta}_{i,n}-y_{j,k}^{\text{(m)}}|) g(\theta_{y_{j,k}^{\text{(m)}}x^{\beta}_{i,n}}^{\alpha})\dfrac{x^{\beta}_{i,n}-y_{j,k}^{\text{(m)}}}{|x^{\beta}_{i,n}-y_{j,k}^{\text{(m)}}|}\\
&&+\dfrac{M^{\beta}}{4} \sum_{i=1}^4 f^{\alpha}_{\beta}(|x^{\beta}_{m,n}-y_{j,k}^{i}|) g(\theta_{y_{j,k}^{i}x^{\beta}_{m,n}}^{\alpha})\dfrac{x^{\beta}_{m,n}-y_{j,k}^{i}}{|x^{\beta}_{m,n}-y_{j,k}^{i}|}.\label{approximation integral in midpoint by using average over four vertices}
\end{eqnarray}
We have merely overcome problems that might occur in particular if the spatial grid is too coarse. For a sufficiently fine grid, the condition in Remark \ref{remark conditions for non-merging diracs and bounded velocity} makes sure that there is a neighbourhood of zero density around the position of each Dirac mass. It is not important whether the velocity is properly defined in regions of zero density, or not.\\
\\
We look now to the case when $\mu^{\beta}_n$ is an absolutely continuous measure. We follow (\ref{approximation integral by subst piecewise const density}) and (\ref{approximation integral by using four vertices}), and derive
\begin{eqnarray}
\notag && \int_{\Omega \setminus \big\{y_{j,k}^{\text{(m)}}\big\}} f^{\alpha}_{\beta}(|y-y_{j,k}^{\text{(m)}}|)g(\theta_{y_{j,k}^{\text{(m)}}y}^{\alpha})\dfrac{y-y_{j,k}^{\text{(m)}}}{|y-y_{j,k}^{\text{(m)}}|}d\mu^{\beta}_n(y)\\ &\approx& \dfrac{h_{\mathcal{L}}h_{\mathcal{W}}}{4}\sum_{(j,k)\in\mathcal{K}} \rho^{\beta}_{(j,k),n} \sum_{i=1}^4 f^{\alpha}_{\beta}(|y_{j,k}^i-y_{j,k}^{\text{(m)}}|)g(\theta_{y_{j,k}^{\text{(m)}} y_{j,k}^i}^{\alpha})\dfrac{y_{j,k}^i-y_{j,k}^{\text{(m)}}}{|y_{j,k}^i-y_{j,k}^{\text{(m)}}|}.\label{interaction integral cont cont}
\end{eqnarray}
Keep in mind that the midpoint $y_{j,k}^{\text{(m)}}$ can never coincide with any of the vertices of a cell in $\Omega$.

\subsection{Summary: The approximation $\tilde{v}^{\alpha}_n$}\label{subsect summary tilde v}
We now shortly summarize the way we have defined $\tilde{v}^{\alpha}_n$ in Sections \ref{subsect alpha discrete} and \ref{subsect alpha abs cont}.
\begin{itemize}
  \item If $\mu^{\alpha}_n$ is discrete, we need the value of $\tilde{v}^{\alpha}_n$ in the corresponding centres $x^{\alpha}_{i,n}$ of the Dirac masses.\\
  Calculation of the value of $v_{\text{des}}^{\alpha}$ is straightforward. Those terms in the social velocity for which the measure $\mu^{\beta}_n$ is discrete are calculated according to (\ref{interaction integral disc disc}). If, on the other hand, the measure $\mu^{\beta}_n$ is absolutely continuous, then the approximation of (\ref{approximation integral by subst piecewise const density}) is used. Regarding the approximation of the integrals over all cells $E_{j,k}$, we distinguish between two cases:
  \begin{itemize}
    \item $x^{\alpha}_{i,n}$ does not coincide with one of the four vertices of cell $E_{j,k}$. Use the two-dimensional trapezoid rule as given by (\ref{approximation integral by using four vertices}).
    \item $x^{\alpha}_{i,n}$ coincides with one of the four vertices of cell $E_{j,k}$. Use the midpoint of the cell to approximate the integral; see (\ref{approximation integral by using midpoint}).
  \end{itemize}
  \item If $\mu^{\alpha}_n$ is absolutely continuous, the velocity in each interior point of a cell $E_{j,k}$ is approximated by the velocity in the midpoint of that cell. Evaluation of $v_{\text{des}}^{\alpha}$ in the midpoint $y_{j,k}^{\text{(m)}}$ is again straightforward. If the measure $\mu^{\beta}_n$ in the social velocity is discrete, then we distinguish between two cases:
      \begin{itemize}
        \item None of the Dirac-centres $x^{\beta}_{i,n}$ coincides with the midpoint $y_{j,k}^{\text{(m)}}$. Then use (\ref{integral evaluated in midpoint transform to sum diracs}) to evaluate the integral.
        \item One of the centres $x^{\beta}_{i,n}$ is located exactly in the midpoint $y_{j,k}^{\text{(m)}}$. Use the approximation given in (\ref{approximation integral in midpoint by using average over four vertices}).
      \end{itemize}
      If $\mu^{\beta}_n$ is absolutely continuous, then the approximation as given by (\ref{interaction integral cont cont}).
\end{itemize}

\section{Push forward of the mass measures}
The only thing that is still to be included in our numerical scheme, is the push forward of the measure $\mu^{\alpha}_n$ to obtain $\mu^{\alpha}_{n+1}$. We use $\tilde{v}^{\alpha}_n$ (see Section \ref{section calculating velocities}, in particular Section \ref{subsect summary tilde v}) to obtain a modified version of the motion mapping (cf. Definition \ref{def one-step motion mapping}):
\begin{equation*}
\tilde{\chi}^{\alpha}_n(x):= x + \Delta t_n \tilde{v}_n^{\alpha}(x).
\end{equation*}
If $\mu^{\alpha}_n$ is a discrete measure, the push forward can be found in a natural way, as was already suggested in Corollary \ref{corollary sum of diracs preserves constant coeff}. If $\mu^{\alpha}_n$ is given by
\begin{equation*}
\mu^{\alpha}_n := M^{\alpha}\sum_{i=1}^{N^{\alpha}} \delta_{x^{\alpha}_{i,n}},
\end{equation*}
then the numerical approximation of the push forward is
\begin{equation*}
\mu^{\alpha}_n := M^{\alpha}\sum_{i=1}^{N^{\alpha}} \delta_{\tilde{\chi}^{\alpha}_n(x^{\alpha}_{i,n})}.
\end{equation*}
Determining the push forward of this discrete mass measure boils down to updating the centres of the Dirac masses in the following way:
\begin{equation*}
x^{\alpha}_{i,n+1}:=\tilde{\chi}^{\alpha}_n(x^{\alpha}_{i,n})= x^{\alpha}_{i,n} + \Delta t_n \tilde{v}_n^{\alpha}(x^{\alpha}_{i,n}).
\end{equation*}
If $\mu^{\alpha}_n$ is absolutely continuous, then a little more effort is required. Since $\tilde{v}^{\alpha}_n$ is defined such that it is constant within a cell, the push forward of a cell $E_{j,k}$ is a translation by the vector $\Delta t_n \tilde{v}^{\alpha}_n(y_{j,k}^{\text{(m)}})$. Indeed we have
\begin{eqnarray*}
\tilde{\chi}^{\alpha}_n(E_{j,k})&=& \big\{\tilde{\chi}^{\alpha}_n(x)\,\big|\,x\in E_{j,k}\big\}\\
&=& \big\{x + \Delta t_n \tilde{v}_n^{\alpha}(x)\,\big|\,x\in E_{j,k}\big\}\\
&=& \big\{x + \Delta t_n \tilde{v}_n^{\alpha}(y_{j,k}^{\text{(m)}})\,\big|\,x\in E_{j,k}\big\}\\
&=& E_{j,k}+ \Delta t_n \tilde{v}_n^{\alpha}(y_{j,k}^{\text{(m)}}).
\end{eqnarray*}
The push forward of $E_{j,k}$ is indicated in Figure \ref{figure push-forward cell}.
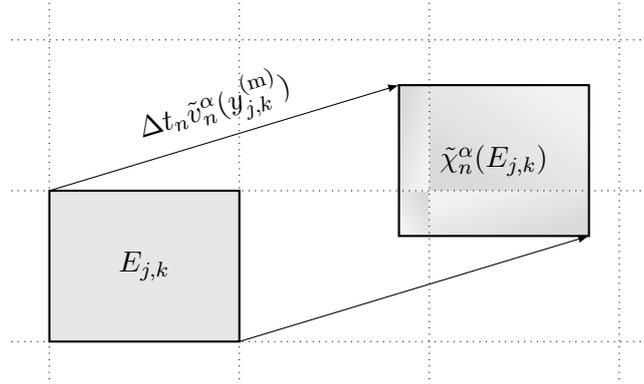
\begin{figure}[ht]
\centering
\begin{tikzpicture}
    \shade[right color=gray!10, left color=gray!30, shading angle=-45](4.6,1.4) rectangle (5,2);
    \shade[top color=gray!10, bottom color=gray!30, shading angle=-45](5,2) rectangle (7.1, 3.4);
    \shade[right color=gray!10, left color=gray!30, shading angle=45](4.6,2) rectangle (5,3.4);
    \shade[right color=gray!30, left color=gray!10, shading angle=45](5,1.4) rectangle (7.1, 2);

    \fill[gray!20] (0,0) rectangle (2.5,2);
    \draw[thick] (0,0) rectangle (2.5,2);
    \draw[dotted] (0,-0.5) -- (0,4.5);
    \draw[dotted] (2.5,-0.5) -- (2.5,4.5);
    \draw[dotted] (5,-0.5) -- (5,4.5);
    \draw[dotted] (7.5,-0.5) -- (7.5,4.5);

    \draw[dotted] (-0.5,0) -- (8,0);
    \draw[dotted] (-0.5,2) -- (8,2);
    \draw[dotted] (-0.5,4) -- (8,4);

    \draw[thick] (4.6,1.4) rectangle (7.1,3.4);

    \draw[arrows={-latex},thin] (2.5,0) -- (7.1,1.4);
    \draw[arrows={-latex},thin] (0,2) -- (4.6,3.4) node[midway, sloped, above]{$\Delta t_n \tilde{v}^{\alpha}_n(y_{j,k}^{\text{(m)}})$};

    \node at (1.25,1){$E_{j,k}$};
    \node at (5.85,2.4){$\tilde{\chi}^{\alpha}_n(E_{j,k})$};

\end{tikzpicture}
\caption{The push forward of a cell $E_{j,k}$ governed by the approximated velocity field. The `image' $\tilde{\chi}^{\alpha}_n(E_{j,k})$ clearly lies within four cells of the spatial domain.}\label{figure push-forward cell}
\end{figure}
We want to keep the spatial grid fixed in time. Also, we require our approximated density $\tilde{\rho}^{\alpha}_n$ to be constant within a cell, for every $n$. In general, the push forward of $E_{j,k}$ will not exactly coincide with a cell of our grid: cf. Figure \ref{figure push-forward cell}, where $\tilde{\chi}^{\alpha}_n(E_{j,k})$ lies in four cells. We thus propose the following update of the density:
\begin{equation}\label{update density}
\tilde{\rho}^{\alpha}_{n+1}(x)\equiv\rho^{\alpha}_{(j,k),n+1}:=\dfrac{1}{\lambda^d\big(E_{j,k}\big)}\sum_{(p,q)\in\mathcal{K}}\rho^{\alpha}_{(p,q),n}\lambda^d\big(E_{j,k}\cap \tilde{\chi}^{\alpha}_n(E_{p,q})\big),
\end{equation}
for all $x$ in the interior of $E_{j,k}$.\\
A positive contribution is given by each cell $E_{p,q}$ that is (partially) mapped into $E_{j,k}$, since only then $E_{j,k}\cap \tilde{\chi}^{\alpha}_n(E_{p,q})$ is non-empty, and thus has positive Lebesgue measure. Note that $\rho^{\alpha}_{(p,q),n}\lambda^d\big(E_{j,k}\cap \tilde{\chi}^{\alpha}_n(E_{p,q})\big)$ is exactly the mass that is transferred from $E_{p,q}$ into $E_{j,k}$. The density of $E_{j,k}$ then follows from adding the mass contributions of all cells $E_{p,q}$, and dividing by the area of $E_{j,k}$. This makes (\ref{update density}) a natural way of updating densities.\\
\\
The numerical approximation of the push forward of $\mu^{\alpha}_n$ is fully determined by the iterative scheme for its density $\tilde{\rho}^{\alpha}_n$ as given by (\ref{update density}).

\begin{remark}
Using the update of the density in (\ref{update density}), $\tilde{\rho}^{\alpha}_{n+1}$ is by definition non-negative, as it is a sum of non-negative terms. Furthermore, the total mass contained in $\Omega$ is conserved in time. We namely have that
\begin{eqnarray*}
\int_{\Omega}\tilde{\rho}^{\alpha}_{n+1}(x)d\lambda^d(x) &=& \sum_{(j,k)\in\mathcal{K}}\rho^{\alpha}_{(j,k),n+1}\lambda^d\big(E_{j,k}\big)\\
&=& \sum_{(j,k)\in\mathcal{K}}\lambda^d\big(E_{j,k}\big)\dfrac{1}{\lambda^d\big(E_{j,k}\big)} \sum_{(p,q)\in\mathcal{K}}\rho^{\alpha}_{(p,q),n}\lambda^d\big(E_{j,k}\cap \tilde{\chi}^{\alpha}_n(E_{p,q})\big)\\
&=& \sum_{(p,q)\in\mathcal{K}}\rho^{\alpha}_{(p,q),n}\sum_{(j,k)\in\mathcal{K}}\lambda^d\big(E_{j,k}\cap \tilde{\chi}^{\alpha}_n(E_{p,q})\big)\\
&=& \sum_{(p,q)\in\mathcal{K}}\rho^{\alpha}_{(p,q),n}\lambda^d\big(\tilde{\chi}^{\alpha}_n(E_{p,q})\big)\\
&=& \sum_{(p,q)\in\mathcal{K}}\rho^{\alpha}_{(p,q),n}\lambda^d\big(E_{p,q}\big)\\
&=& \int_{\Omega}\tilde{\rho}^{\alpha}_{n}(x)d\lambda^d(x).
\end{eqnarray*}
We have used in the fifth step that $\tilde{\chi}^{\alpha}_n$ is a translation, due to which $\lambda^d\big(\tilde{\chi}^{\alpha}_n(E_{p,q})\big)=\lambda^d\big(E_{p,q}\big)$.\\
Conservation of total mass in $\Omega$ is only guaranteed under the assumption that no cell (of non-zero density) is mapped outside $\Omega$; cf. Part (\ref{assumption motion mapping including support part 1}) of Assumption \ref{assumption motion mapping including support}.
\end{remark}

\newpage
\chapter{Numerical illustration: simulation results}\label{section numerical illustration}
In this section we present the results of our simulation experiments. Except for some simple test cases, we only consider two-scale situations. That is, in each experiment there are two subpopulations, one of which is discrete and the other is absolutely continuous. The component indexed $\alpha=1$ is a collection of discrete individuals. In our simulations, component 1 consists mostly of one individual only, and exceptionally of two. The component with index $\alpha=2$ is a macroscopic crowd.\\
\\
To avoid effects at the boundaries (which we have not specified so far), the domain $\Omega$ is `sufficiently large'. That is, the size of $\Omega$ is such that no mass reaches the boundary. No problems occur as long as the cells of the spatial grid on the periphery of $\Omega$ have zero density, and the positions of the individuals remain in $\Omega$.\footnote{We admit that this is a rather pragmatic solution, but for the moment it is the best we can do.}\\
\\
In the following sections, the results of the simulation are presented by giving a graphical representation of the crowd in the domain. The positions of individuals are marked by red bullets, whereas the density of the macroscopic crowd is indicated by a grey shading. The darker the colour, the higher the density; a colour bar at the right-hand side of the graph shows what shading corresponds to a certain density.\\
\\ 
Note that we essentially use the scheme of \cite{Piccoli2010, PiccoliTosinMeasTh}. The convergence and stability of the numerical solution is proven in these papers. We expect similar properties to hold for our setting. We, however, omit to give further details in this direction and focus directly on simulation results.\\
The results we present will be interpreted only on a qualitative basis. At a later stage we will try to make some of these results quantitative by recovering experimental data by S. Hoogendoorn, W. Daamen and M. Campanella (Delft University of Technology, see Section \ref{section introduction people in the field} of the Introduction) regarding experiments in a corridor.

\section{Reference setting}
For attraction-repulsion interactions, we use the following expression in the interaction integral:
\begin{eqnarray}
f^{\text{AR}}(s)&:=&\left\{
  \begin{array}{ll}
    F^{\text{AR}} \Bigl(1-\dfrac{R_r^{\text{AR}}}{s}\Bigr), & \mbox{if $0<s\leqslant R_r^{\text{AR}}$;} \\
    \\
    \dfrac{-F^{\text{AR}}}{R_r^{\text{AR}}(R_a^{\text{AR}}-R_r^{\text{AR}})} \big(s-R_r^{\text{AR}}\big)\big(s-R_a^{\text{AR}}\big), & \mbox{if $R_r^{\text{AR}}<s\leqslant R_a^{\text{AR}}$;} \\
    \\
    0, & \mbox{if $s>R_a^{\text{AR}}$.}
  \end{array}
\right.\label{fAR explicit}
\end{eqnarray}
The radii $R_r^{\text{AR}}$ and $R_a^{\text{AR}}$ are such that $0<R_r^{\text{AR}}<R_a^{\text{AR}}$. The factor $F^{\text{AR}}\in[0,\infty)$ is a positive scaling constant. Note that this is exactly the function plotted in Figure \ref{Figure graphs fAR fR} (left). Note that $f^{\text{AR}}$ is differentiable in $s=R_r^{\text{AR}}$.\\
\\
Similarly, we define for purely repulsive interaction
\begin{eqnarray}
f^{\text{R}}(s)&:=&\left\{
  \begin{array}{ll}
    F^{\text{R}} \Bigl(1-\dfrac{R_r^{\text{R}}}{s}\Bigr), & \mbox{if $0<s\leqslant R_r^{\text{R}}$;} \\
    0, & \mbox{if $s>R_r^{\text{R}}$.}
  \end{array}
\right.\label{fR explicit}
\end{eqnarray}
Again, we take $R_r^{\text{R}}>0$ and the scaling constant $F^{\text{R}}\in[0,\infty)$. This function is plotted in Figure \ref{Figure graphs fAR fR} (right).\\
\\
Unless indicated otherwise, we use a set of `standard' parameters in our simulations. As dimensions of the domain $\Omega$ we take
\begin{equation*}
\mathcal{L}= \mathcal{W}= 50.
\end{equation*}
The proportionality constant $M^{\alpha}$ is only present for $\alpha=1$ and is assigned the value
\begin{equation*}
M^{\alpha}=60,\hspace{1 cm}\text{for }\alpha=1.
\end{equation*}
The desired velocity is independent of the space variable, and only the direction is different for the two subpopulations. We take
\begin{eqnarray*}
v_{\text{des}}^{1}(x)&\equiv& -1.34\, e_1,\hspace{1 cm}\text{for all }x\in\Omega,\\
v_{\text{des}}^{2}(x)&\equiv& 1.34\, e_1,\hspace{1 cm}\text{for all }x\in\Omega.
\end{eqnarray*}
Here $e_1$ is the unit vector in the direction that corresponds to the horizontal axis in our graphical representation.\\
Regarding the interaction functions, we choose the following reference parameters:
\begin{eqnarray*}
F^{\text{AR}}&=& 0.03,\\
F^{\text{R}}&=& 0.03,\\
R_r^{\text{AR}}&=& 1.5,\\
R_a^{\text{AR}}&=& 3,\\
R_r^{\text{R}}&=& 4,\\
\sigma &=& 0.5.
\end{eqnarray*}
Attraction-repulsion interactions take place within one subpopulation, whereas repulsion-only takes place in the interactions between distinct subpopulations.\\
\\
In the sequel we indicate clearly where we deviate from these standard parameters.

\section{Two basic critical examples: one micro, one macro}\label{section numerics basic examples}
We investigate here whether the simulation results are mathematically and physically acceptable.\\
\\
Consider two individuals both having desired velocity in the direction of $e_1$. The magnitude of the desired velocities is equal and constant, but they are directed oppositely. The initial positions of the two individuals differ only in the $e_1$ direction, and they are located such that they initially want to move towards each other due to their desired velocities. Their interaction is purely of repulsive nature, and has the same parameter values for any of the two. We expect these individuals to approach each other, up to a certain distance. At this distance the desired velocity in one direction is in balance with the (oppositely directed) repulsive effect in the social velocity. The total velocity is thus zero for both individuals.\\
Let $x_i$ and $x_j$ denote the positions of the two pedestrians. For simplicity we take $\sigma=1$ and $M^i=M^j=1$. The velocity of pedestrian $i$ is given by
\begin{equation}\label{velocity two particles}
v(x_i)=v_{\text{des}}^i+f^{\text{R}}(|x_j-x_i|)\dfrac{x_j-x_i}{|x_j-x_i|},
\end{equation}
with $f^{\text{R}}$ as in (\ref{fR explicit}). The velocity is zero if
\begin{equation*}
-f^{\text{R}}(|x_j-x_i|)\dfrac{x_j-x_i}{|x_j-x_i|}=v_{\text{des}}^i.
\end{equation*}
A necessary condition is
\begin{equation}\label{nec condition zero velocity}
\Bigl|f^{\text{R}}(|x_j-x_i|)\dfrac{x_j-x_i}{|x_j-x_i|}\Bigr|=|v_{\text{des}}^i|.
\end{equation}
Note that $(x_j-x_i)/|x_j-x_i|$ is a unit vector. Considering (\ref{fR explicit}), condition (\ref{nec condition zero velocity}) reads
\begin{equation*}
F^{\text{R}} \Bigl(\dfrac{R_r^{\text{R}}}{|x_j-x_i|}-1\Bigr)=|v_{\text{des}}^i|,
\end{equation*}
if $|x_j-x_i|\leqslant R^{\text{R}}_r$. This yields that
\begin{equation}\label{predicted equil distance}
|x_j-x_i|=\dfrac{F^{\text{R}}}{|v_{\text{des}}^i|+F^{\text{R}}}R_r^{\text{R}},
\end{equation}
at the point where the velocity is zero.\\
\\
In Figure \ref{figure graph two particles approach} we show the outcome of this simulation. The individual placed initially on the left has desired velocity $v_{\text{des}}\equiv 1.34\, e_1$, while the one on the right has desired velocity $v_{\text{des}}\equiv -1.34\,e_1$. Furthermore we take $F^{\text{R}}=1$. In Figure \ref{figure graph distances}, the mutual distance is plotted against the time. Indeed the equilibrium distance is the one predicted by (\ref{predicted equil distance}).
\begin{figure}[h]
\centering
\vspace{-6 cm}
\begin{tabular}{lr}
\hspace{-1 cm}\includegraphics[width=0.6\linewidth]{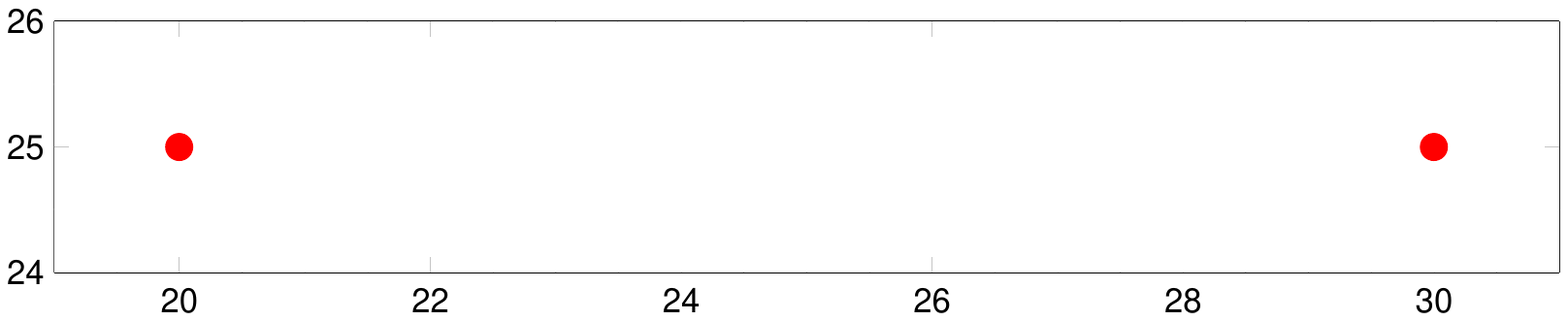}
&
\hspace{-1.5 cm}\includegraphics[width=0.6\linewidth]{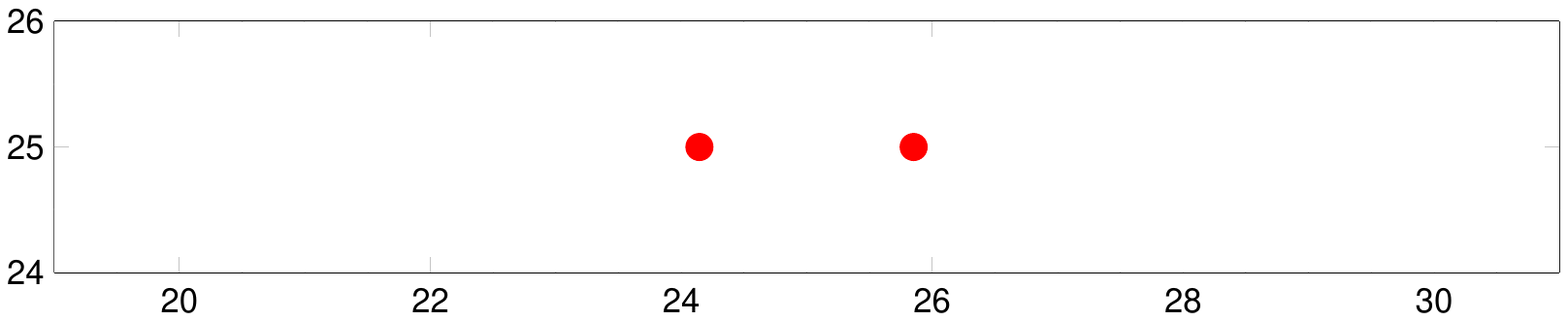}
\end{tabular}
\vspace{-5.5 cm}
\caption{Two individuals approaching each other until a certain minimal distance is reached. The images are taken at $t=0$ (left), $t=15$ (right).}\label{figure graph two particles approach}
\end{figure}

\begin{figure}[h]
\centering
\vspace{-4 cm}
\hspace{-1 cm}\includegraphics[width=0.6\linewidth]{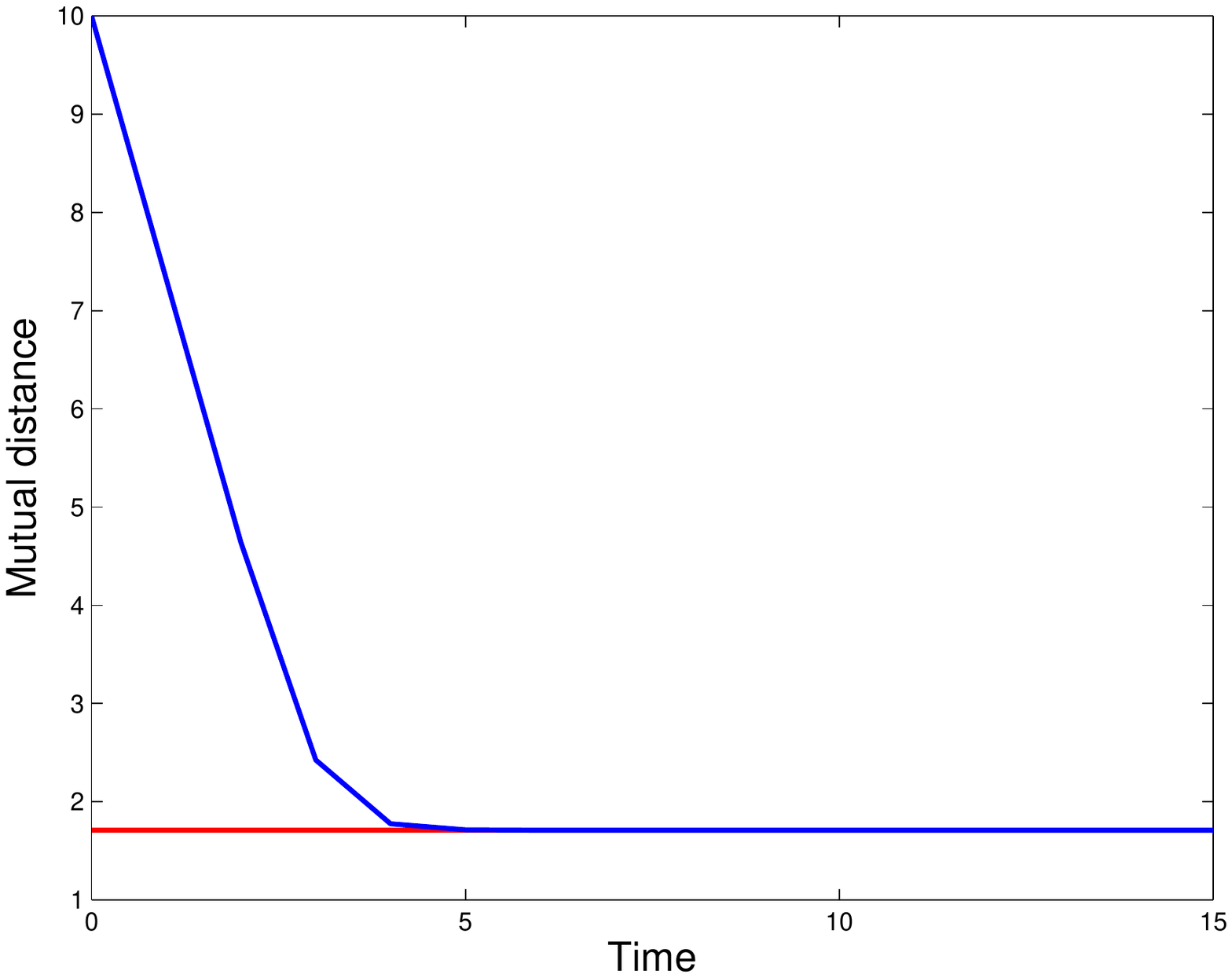}\\
\vspace{-3 cm}
\caption{The mutual distance between two approaching individuals (blue). The predicted minimal distance is indicated in red.}\label{figure graph distances}
\end{figure}

Note that the setting of Figure \ref{figure graph two particles approach} is unrealistic and thus undesirable. In everyday life these two pedestrians would namely both move a bit aside and then walk straight ahead without any constraints. In Section \ref{section introduction modelling approaches} of the Introduction, we have indicated already that a small amount of random noise can be used to avoid these deadlocks. We expect that the equilibrium configuration is instable, and that deadlocks do neither occur if we perturb the initial data (in the direction perpendicular to the desired velocities).\\
\\
In the experiment presented in Figure \ref{figure graph cloud to ball}, we observe that a macroscopic crowd tends to form a circular configuration. The crowd has no desired velocity: $v_{\text{des}}\equiv 0$, and therefore $\sigma=1$ is taken. This typical behaviour is well-known from molecular dynamics. The mutual interactions favour the minimization of the ratio circumference to area. A mathematical proof is given in \cite{Fetecau}, for the same macroscopic (continuous-in-time) equation of mass conservation. The interaction potential in \cite{Fetecau} is not the same as the one used here, but we expect that similar results can be derived along comparable lines of argument.
\begin{figure}[h]
\centering
\vspace{-4 cm}
\begin{tabular}{lr}
\hspace{-1 cm}\includegraphics[width=0.6\linewidth]{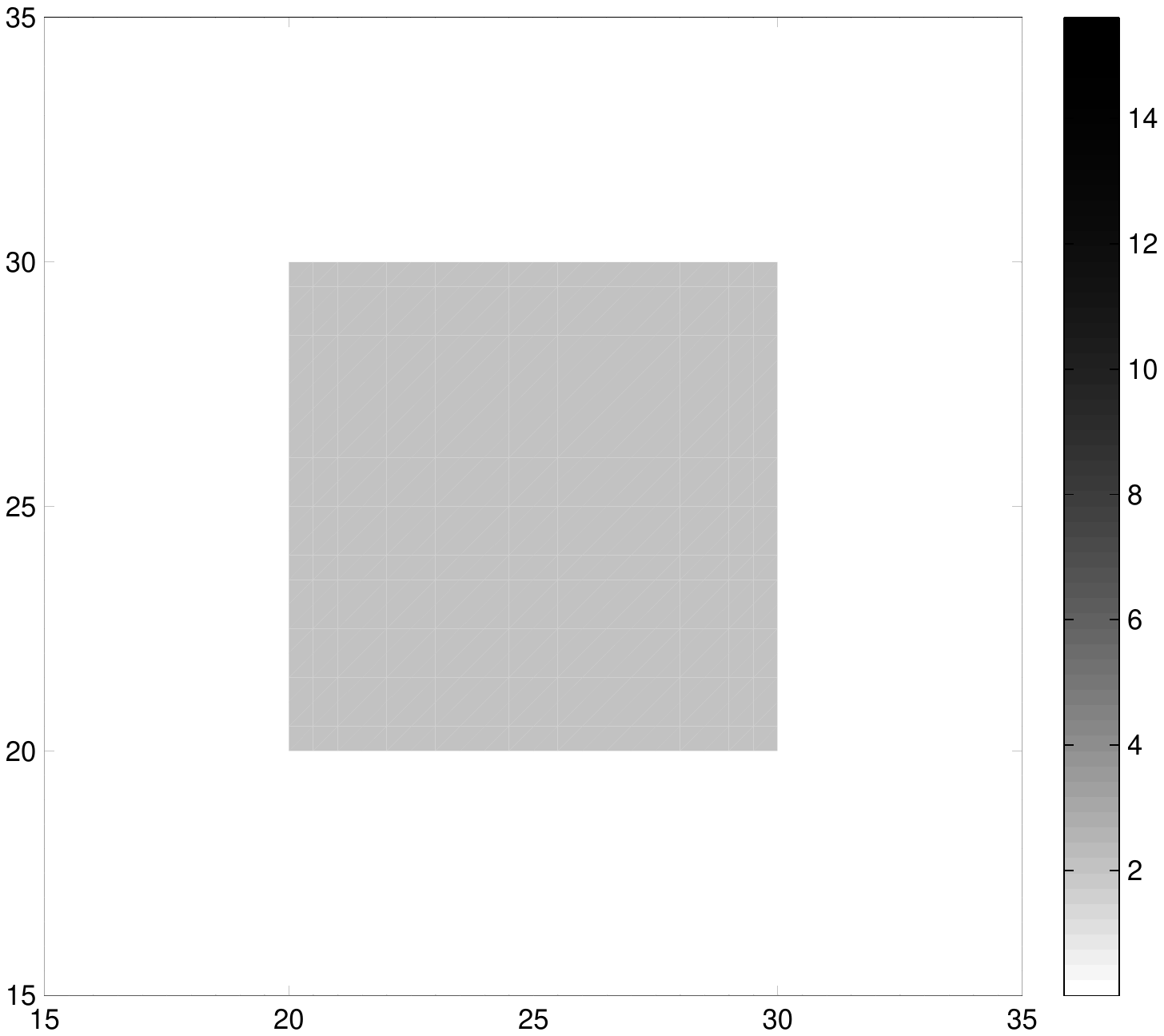}
&
\hspace{-1.5 cm}\includegraphics[width=0.6\linewidth]{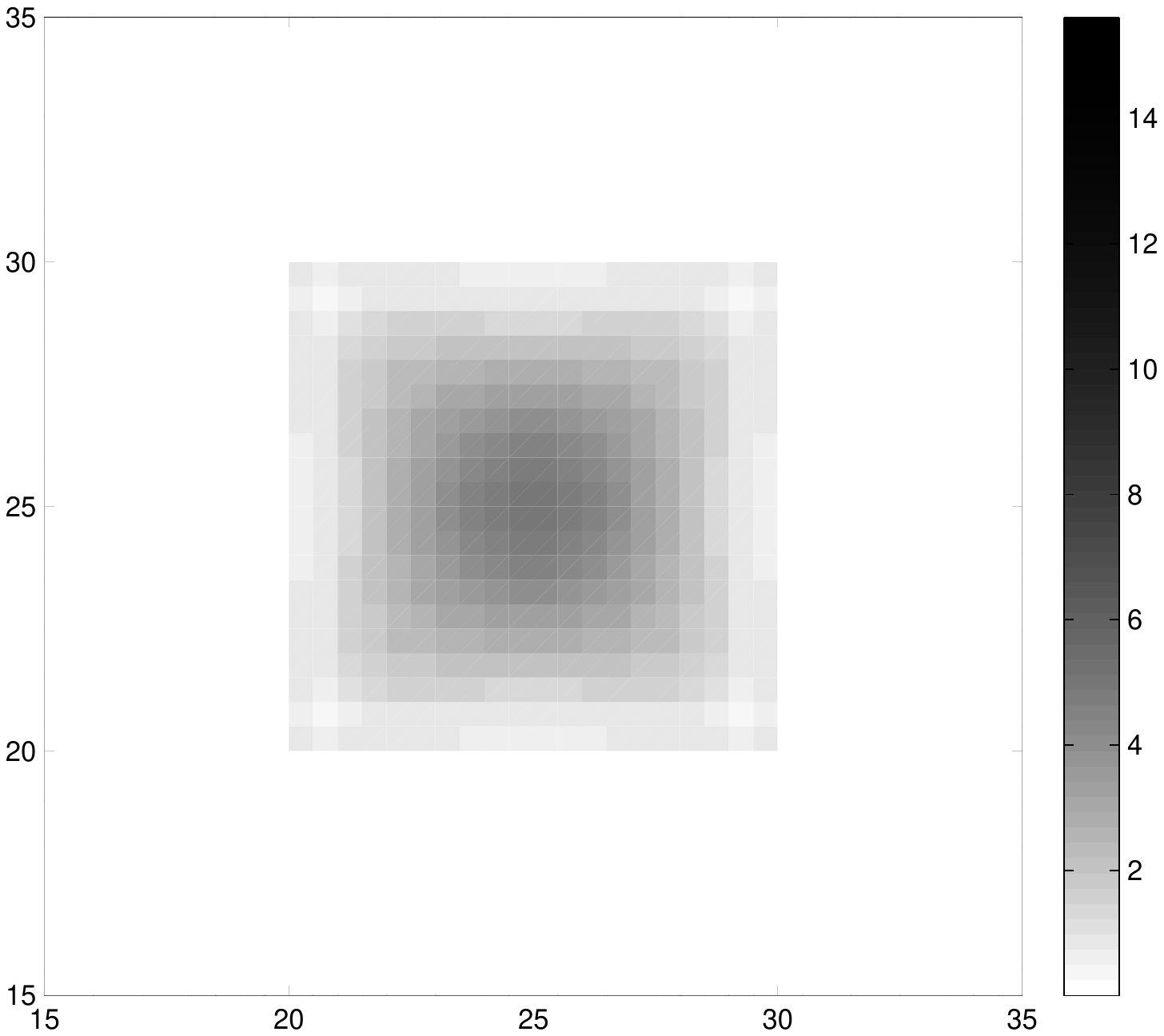}
\vspace{-7 cm}\\
\hspace{-1 cm}\includegraphics[width=0.6\linewidth]{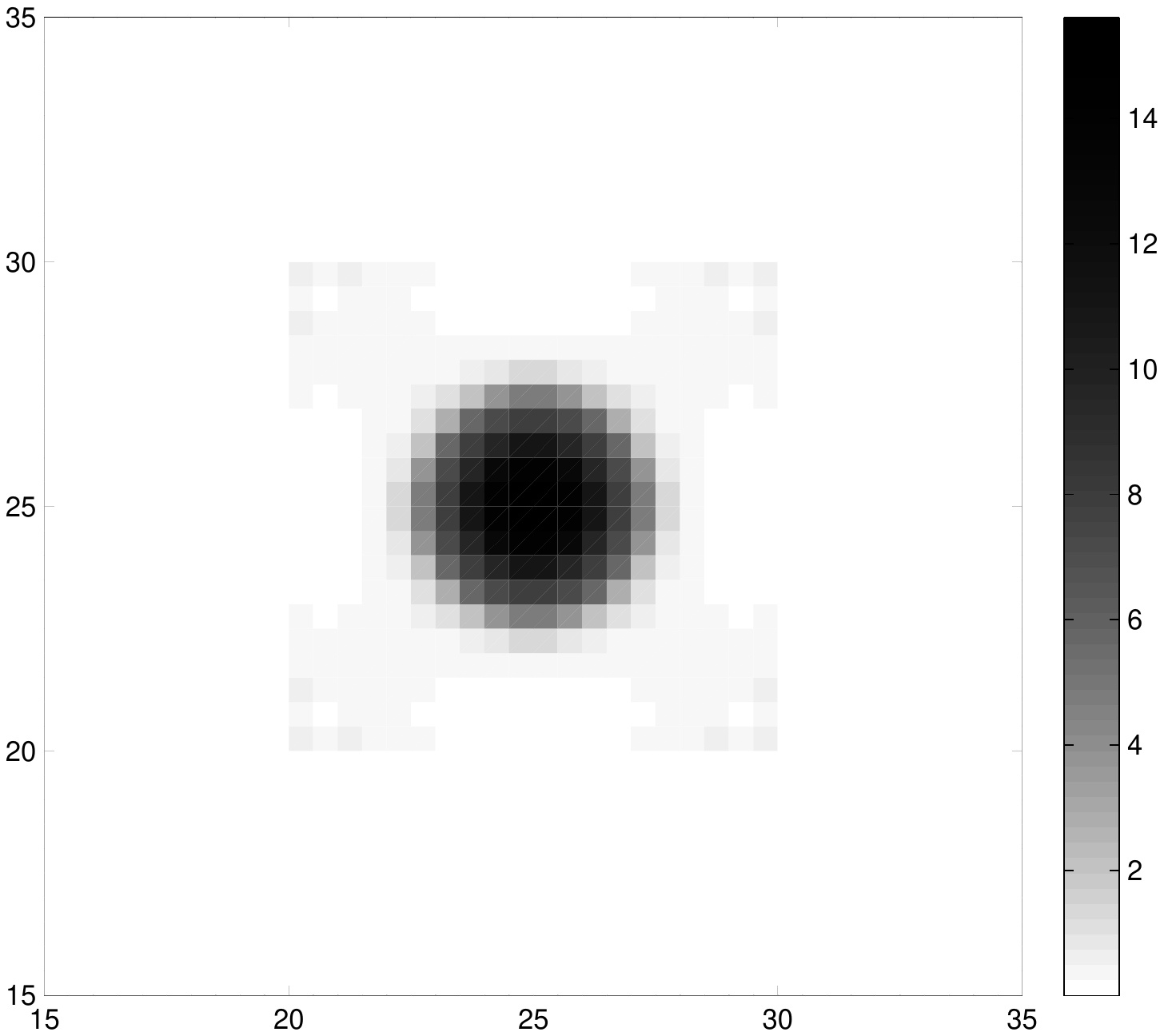} &\hspace{-3 cm}\\
\end{tabular}
\vspace{-3 cm}
\caption{The simulation of a macroscopic crowd's motion. Initially, the crowd forms a square of uniform density $\rho\equiv2$. Its configuration evolves into a circular shape. The images are taken at $t=0$ (top left), $t=30$ (top right), $t=60$ (bottom left).}\label{figure graph cloud to ball}
\end{figure}

\section{Two-scale interactions of repulsive nature}\label{section numerics two-scale interactions one ind}
In Figure \ref{figure graph head-on interaction}, we show the interaction between a macroscopic crowd and an individual that wishes to approach it, and eventually forces itself a way through. This is a situation in which we try to mimic the two-scale (`predator-prey') behaviour described in Section \ref{section two-scales no sing cont}.
\begin{figure}[h]
\centering
\vspace{-5 cm}
\begin{tabular}{lr}
\hspace{-1 cm}\includegraphics[width=0.6\linewidth]{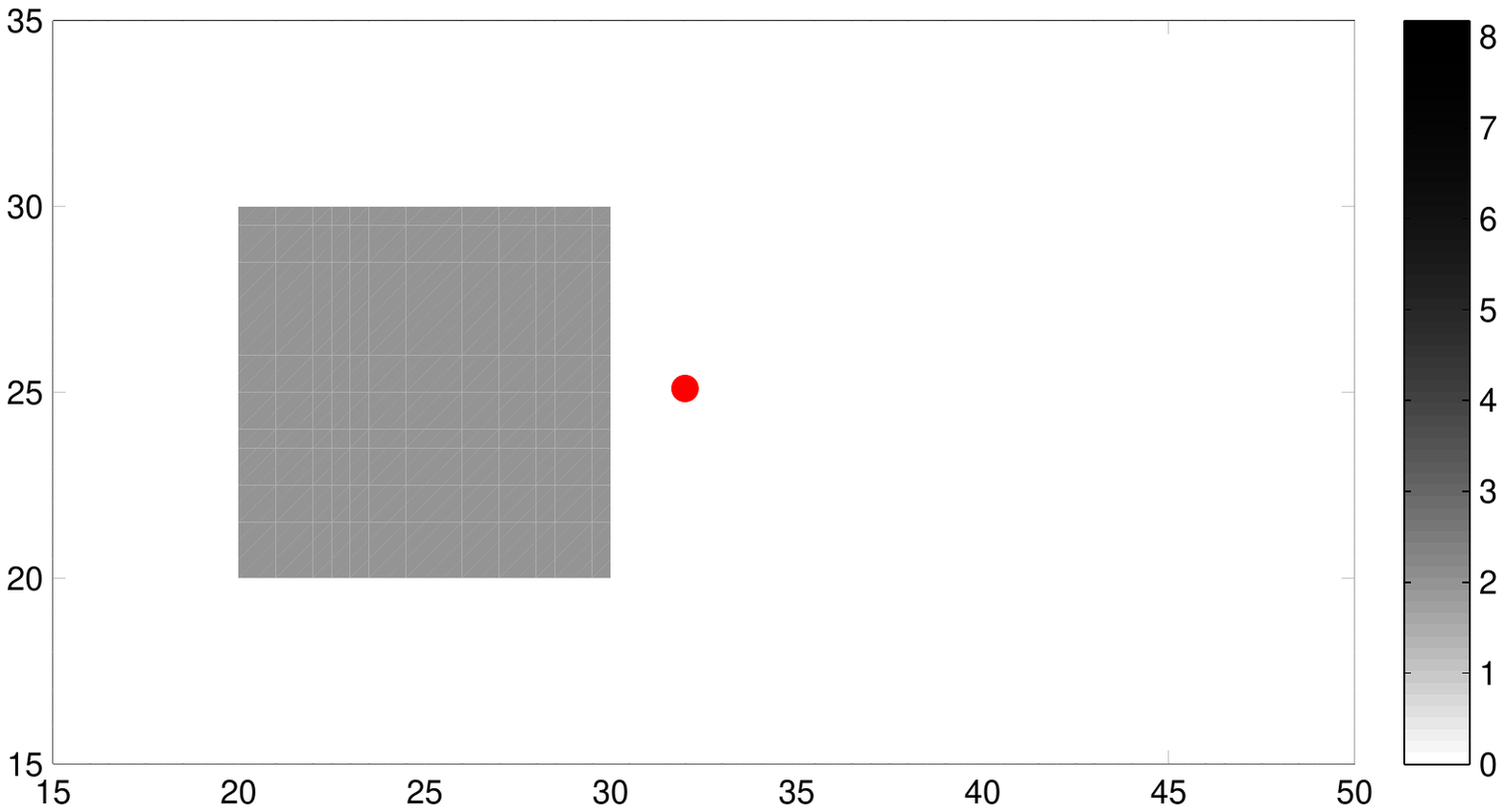}
&
\hspace{-1.5 cm}\includegraphics[width=0.6\linewidth]{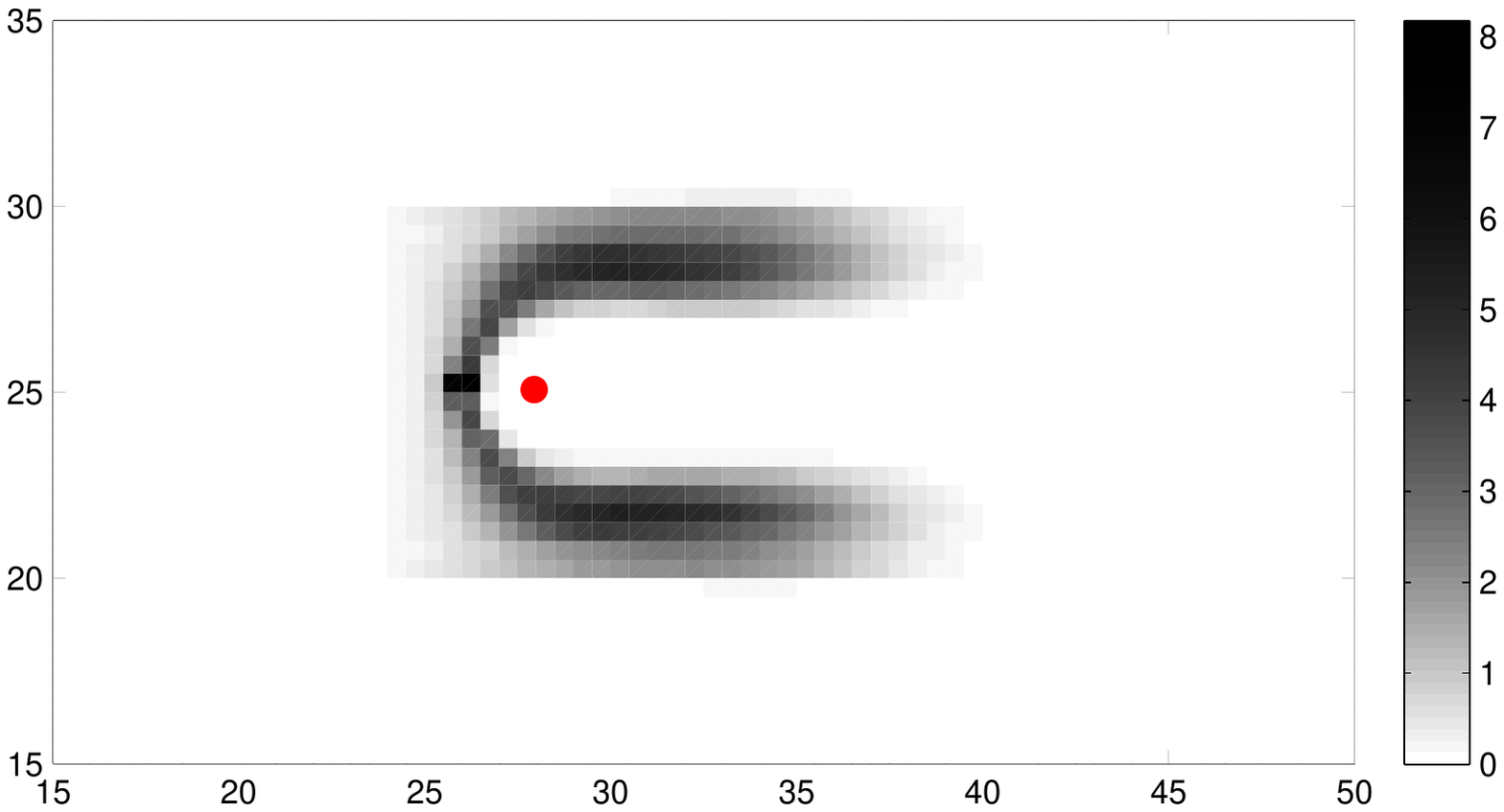}
\vspace{-9 cm}\\
\hspace{-1 cm}\includegraphics[width=0.6\linewidth]{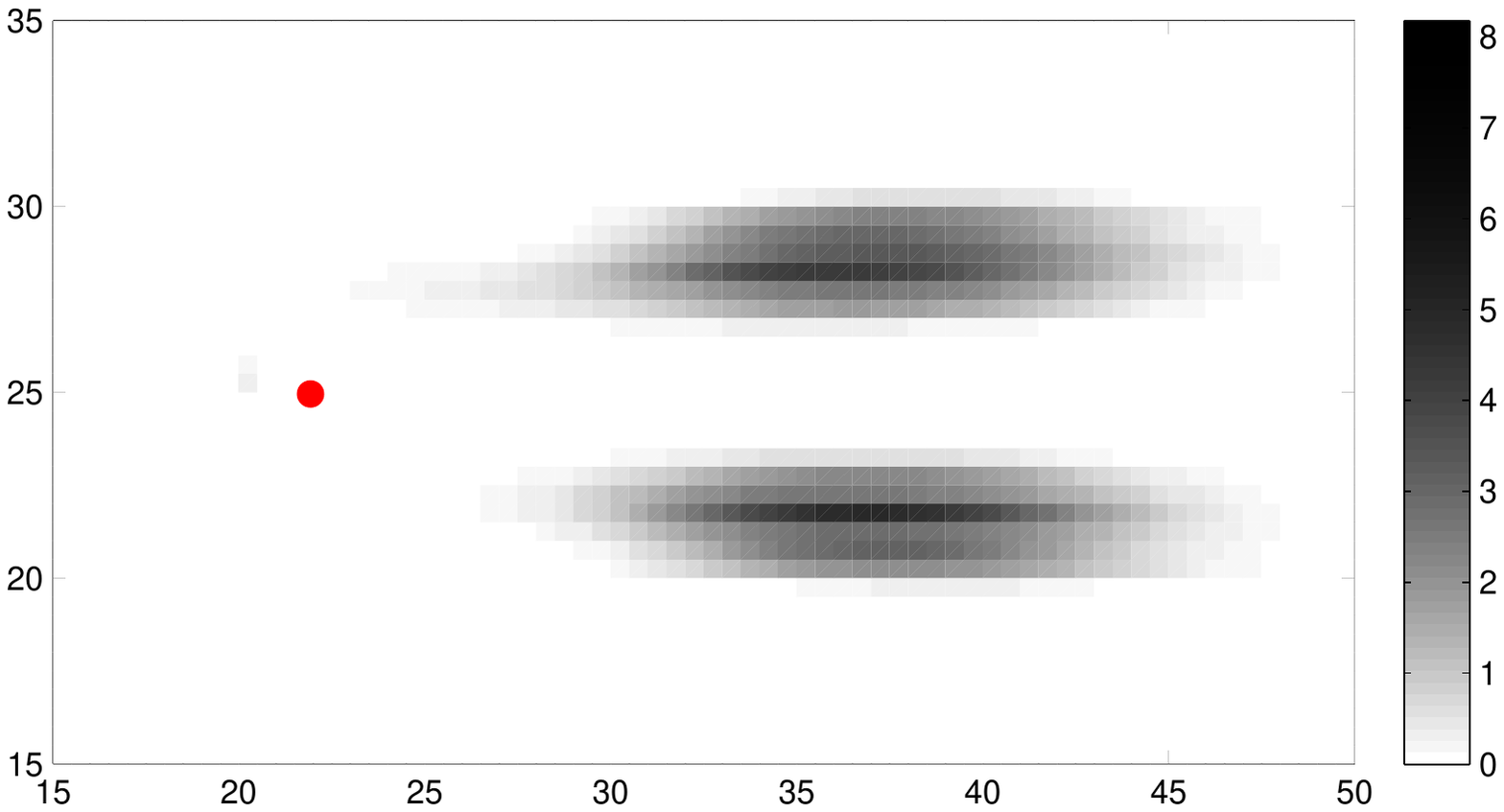} &  \hspace{-3 cm}\\
\end{tabular}
\vspace{-4.5 cm}
\caption{The interaction between a discrete individual and a macroscopic crowd. The images are taken at $t=0$ (top left), $t=5$ (top right), $t=10$ (bottom left).}\label{figure graph head-on interaction}
\end{figure}

A circular area of (nearly) zero density is formed around the individual. The size of this region typically depends on $R_r^{\text{R}}$. At a distance shorter than $R_r^{\text{R}}$ mass is driven away from the individual. If we decrease the radius $R_r^{\text{R}}$, also the `empty zone' around the individual decreases; see Figure \ref{figure graph compare head-on interaction} for this effect. We estimate the distance in front of the individual that is empty, and use the ideas of Section \ref{section numerics basic examples} to do so. Let us consider the point $x$ in the crowd that is right in front of the individual. Disregarding the effect of the rest of the macroscopic crowd, the velocity of $x$ is similar to the one given in (\ref{velocity two particles}):
\begin{equation*}
v(x)=v_{\text{des}}^2+f^{\text{R}}(|z-x|)\dfrac{z-x}{|z-x|}.
\end{equation*}
Here, the variable $z$ is used to denote the position of the individual. The desired velocity of the individual $v_{\text{des}}^1$ is coupled to $v_{\text{des}}^2$ via the relation: $v_{\text{des}}^1=-v_{\text{des}}^2$. Assume that the individual moves at this speed\footnote{This is generally not the case.} and that $x$ remains right in front of $z$. Then the mass located in $x$ is forced to move also with velocity $v_{\text{des}}^1$. This leads to the following equation:
\begin{equation}\label{velocity two particles}
v(x)=v_{\text{des}}^2+f^{\text{R}}(|z-x|)\dfrac{z-x}{|z-x|}=v_{\text{des}}^1=-v_{\text{des}}^2.
\end{equation}
Following the lines of Section \ref{section numerics basic examples}, we derive that a necessary condition is that
\begin{equation}\label{predicted radius}
|z-x|=\dfrac{MF^{\text{R}}}{2|v_{\text{des}}^2|+MF^{\text{R}}}R_r^{\text{R}}.
\end{equation}
According to this estimate the size of the empty area around the individual scales linearly with $R_r^{\text{R}}$ (note that $|z-x|\leqslant R_r^{\text{R}}$). In Figure \ref{figure graph compare head-on interaction} the configuration is given for two distinct values of $R_r^{\text{R}}$. Indeed the size of the low density zone decreases as $R_r^{\text{R}}$ decreases. The distance estimate in (\ref{predicted radius}) is indicated by a circle in blue. Right in front of the individual there is a spot of high density. This is the spot $x$ we were considering. We observe that $x$ is located just outside the blue circle, by which the estimate turns out to be quite good. We also observe that the estimate loses its value more to the sides of the individual.
\begin{figure}[h!]
\centering
\vspace{-5 cm}
\begin{tabular}{lr}
\hspace{-1 cm}\includegraphics[width=0.6\linewidth]{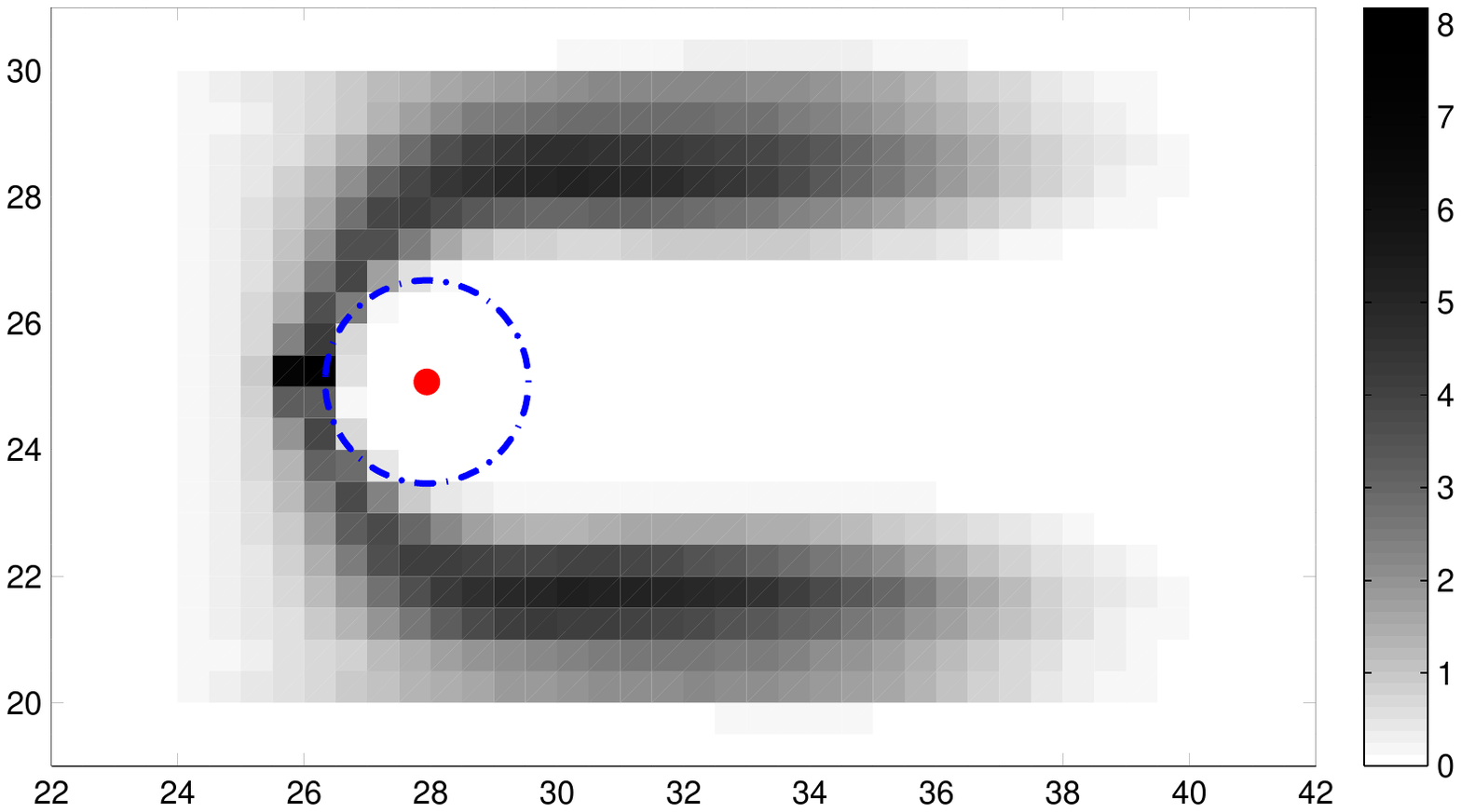}
&
\hspace{-1.5 cm}\includegraphics[width=0.6\linewidth]{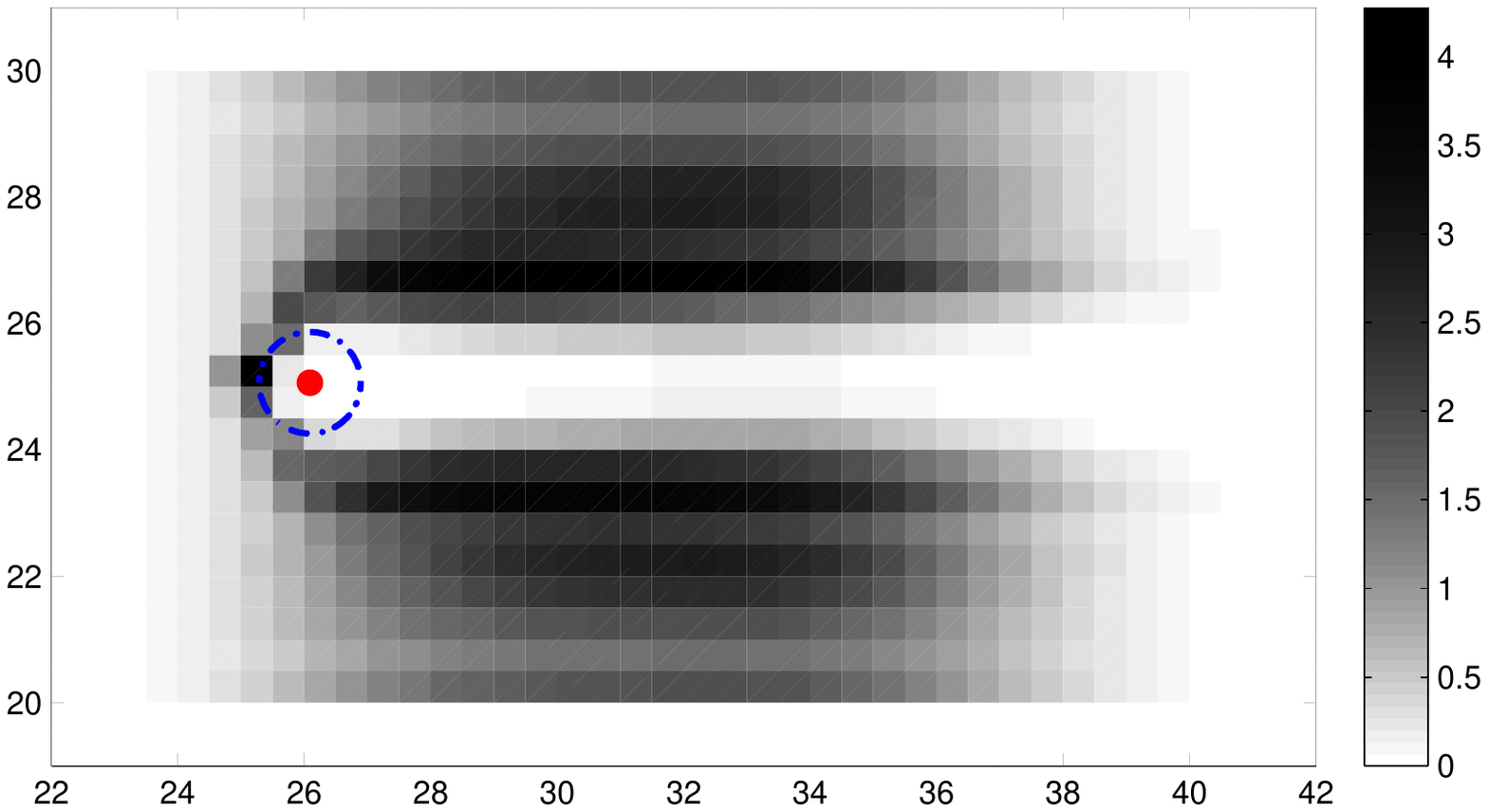}
\\
\end{tabular}
\vspace{-4.5 cm}
\caption{The size of the empty zone around the individual depends on $R_r^{\text{R}}$. On the left $R_r^{\text{R}}=4$ (this is the top right situation in Figure \ref{figure graph head-on interaction}), on the right $R_r^{\text{R}}=2$. In blue a circle is drawn with radius equal to the distance predicted in (\ref{predicted radius}). The images are both taken at $t=5$ starting from identical initial configurations.}\label{figure graph compare head-on interaction}
\end{figure}
\begin{figure}[h!]
\centering
\vspace{-5 cm}
\begin{tabular}{lr}
\hspace{-1 cm}\includegraphics[width=0.6\linewidth]{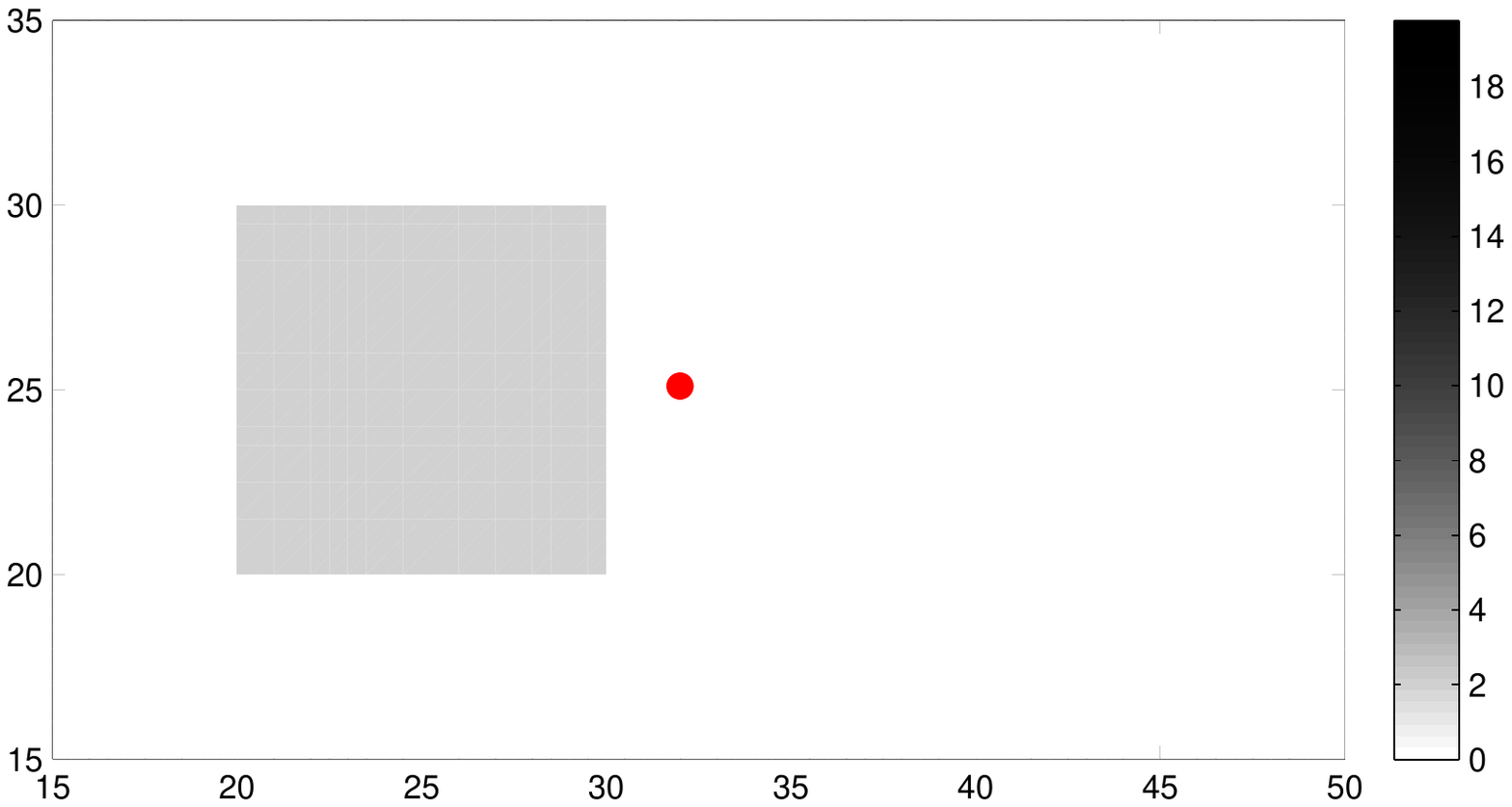}
&
\hspace{-1.5 cm}\includegraphics[width=0.6\linewidth]{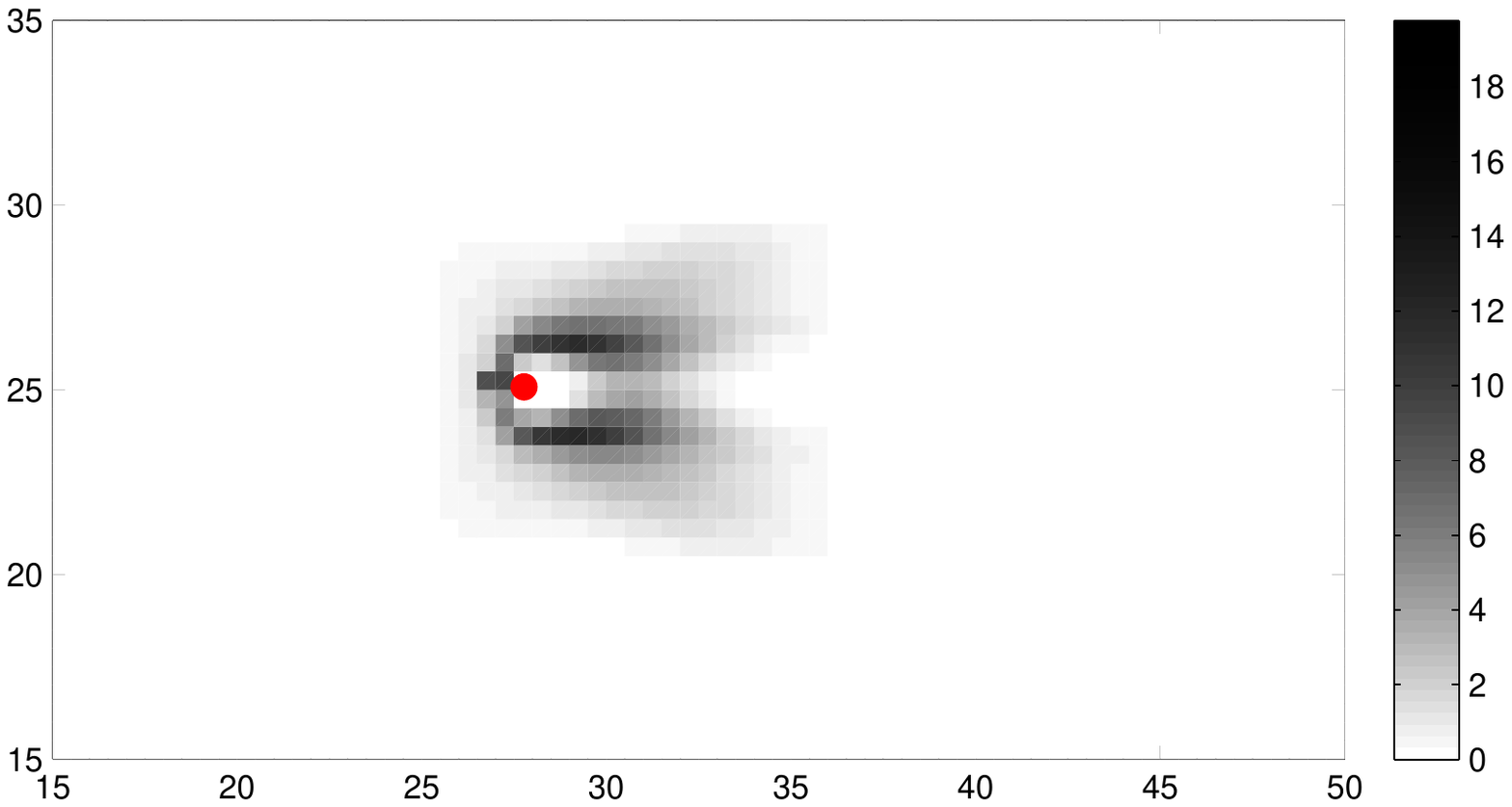}
\vspace{-9 cm}\\
\hspace{-1 cm}\includegraphics[width=0.6\linewidth]{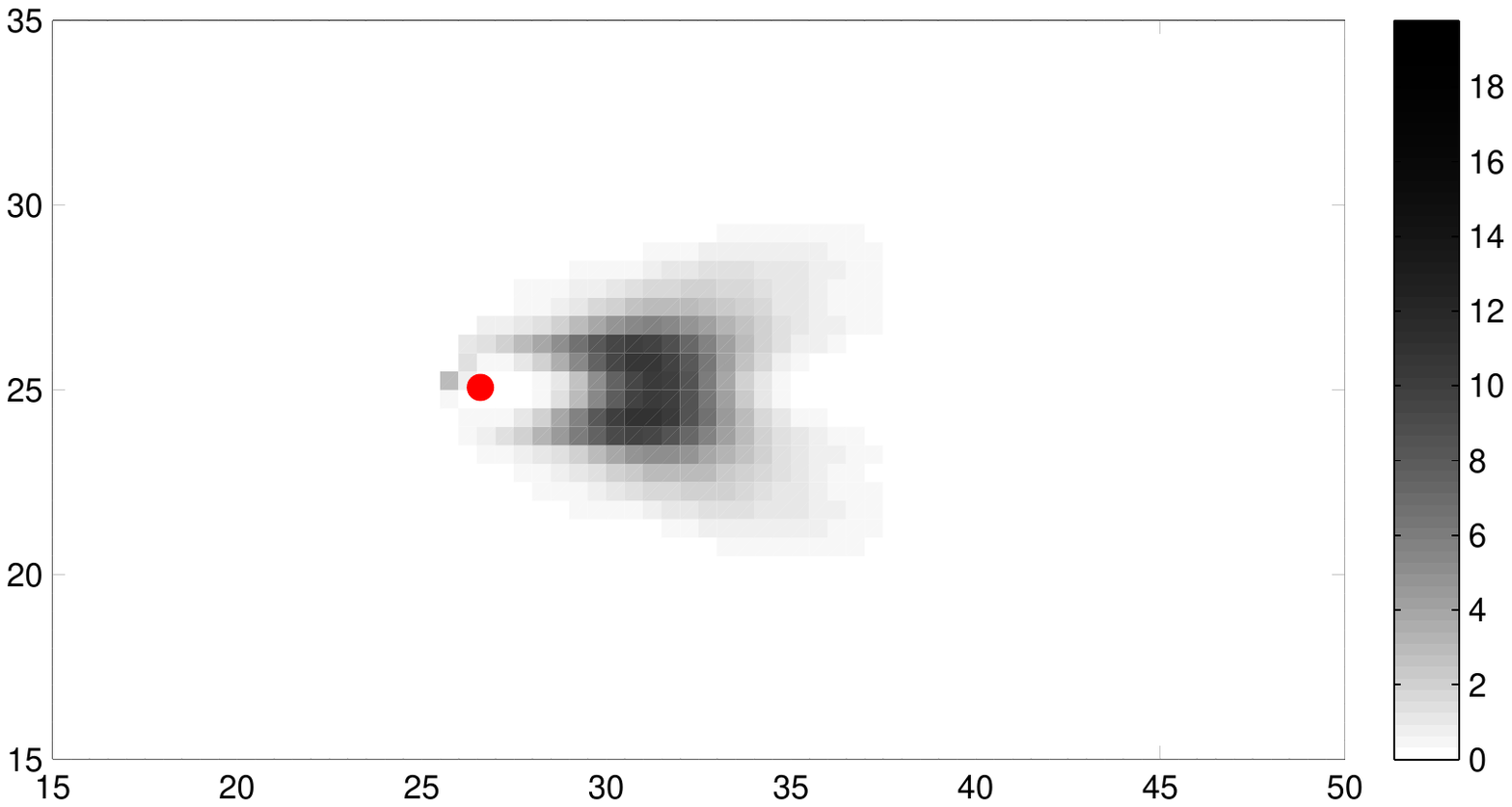} &  \hspace{-3 cm}\\
\end{tabular}
\vspace{-4.5 cm}
\caption{The interaction between a discrete individual and a macroscopic crowd. Modification of the radii of interaction makes that the crowd comes together after the individual has passed by. The images are taken at $t=0$ (top left), $t=4$ (top right), $t=5$ (bottom left).}\label{figure graph crowd together radii}
\end{figure}
After the individual passes by, we observe in Figure \ref{figure graph head-on interaction} that the crowd splits in two and that the two halves do not come together again. This effect depends on the choice of parameters. By adjusting the parameters we can achieve the crowd to `envelope' around the individual. To this end, we decrease the radius of repulsion of the interaction between the individual and the crowd. Also we increase the radius of attraction of the internal interactions in the crowd. To do so, we take $R_r^{\text{R}}=2$, $R_a^{\text{AR}}=5$. This modification makes that the macroscopic crowd does enclose the individual. The result is shown in Figure \ref{figure graph crowd together radii}.\\
We deduce from Figure \ref{figure graph compare head-on interaction} that the width of the empty zone is of the order of $2R_r^{\text{R}}$ (probably a little bit narrower, as also (\ref{predicted radius}) suggests). Let us define a critical radius $R_{\text{cr}}$ such that the width of the empty zone is $2R_{\text{cr}}$. We expect the crowd to come together again, if $R_a^{\text{AR}}\geqslant 2R_{\text{cr}}$; this is the situation in which the crowd in one half can also `feel' the part on the other side of the empty zone.

\begin{remark}
We have tried to achieve the effect of Figure \ref{figure graph crowd together radii} also by adding macroscopic mass above and below the initial square. That is, initially the crowd now has a rectangular shape twice as wide (in vertical direction in the graph). We hoped that the repulsive effects within the crowd would force the two halves to move towards each other again. However, this effect was not obtained.
\end{remark}

\begin{remark}
The desired velocity of the individual does not need to be necessarily in the direction of $e_1$. If the desired velocities of the individual and the macroscopic crowd are at an angle (not equal to $\pi$), the result does however not fundamentally differ from Figure \ref{figure graph head-on interaction}. The individual only leaves behind a diagonal trace in the crowd, instead of a horizontal one.
\end{remark}

\begin{figure}[h!]
\centering
\vspace{-4.5 cm}
\begin{tabular}{lr}
\hspace{-1 cm}\includegraphics[width=0.6\linewidth]{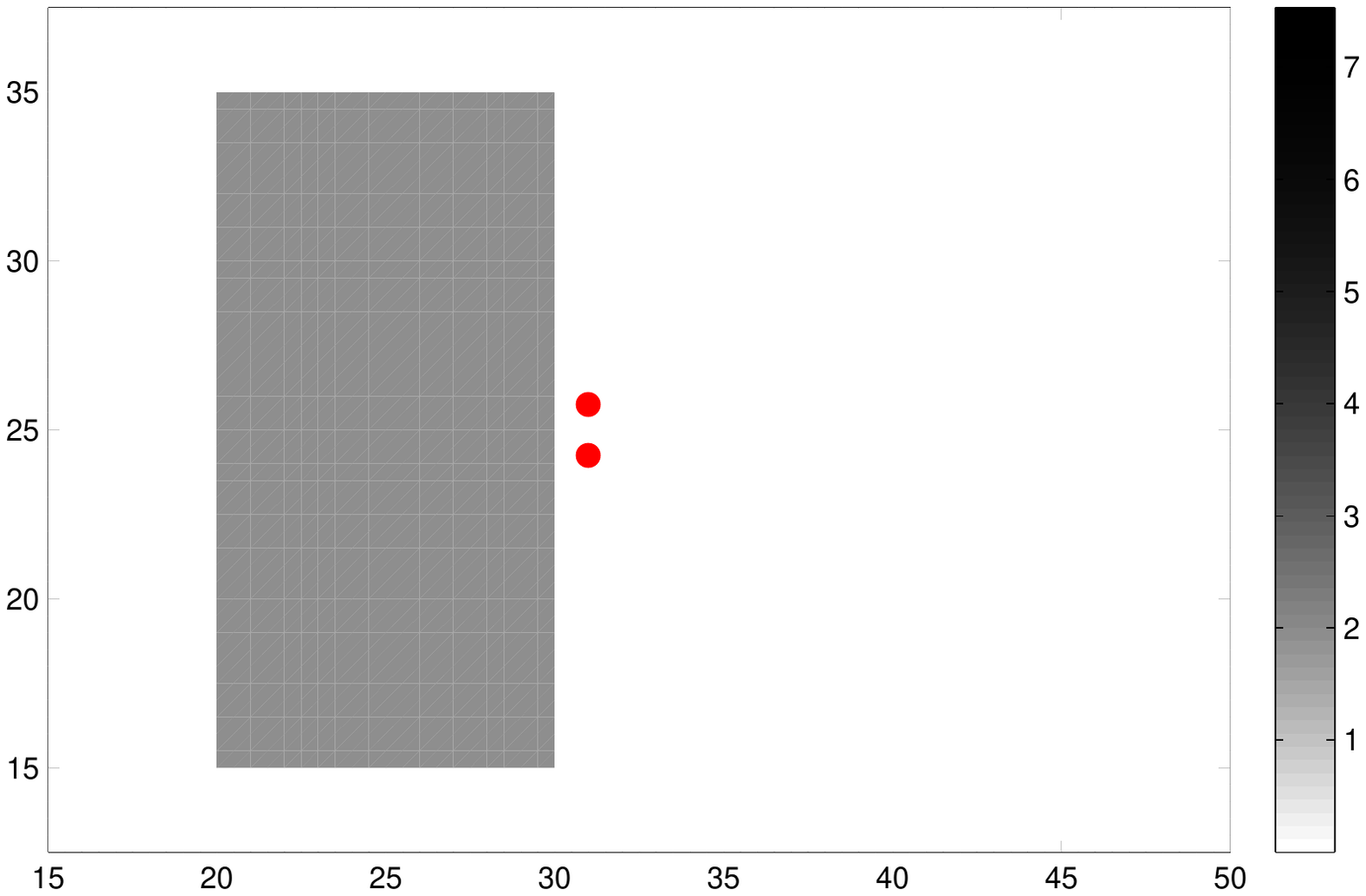}
&
\hspace{-1.5 cm}\includegraphics[width=0.6\linewidth]{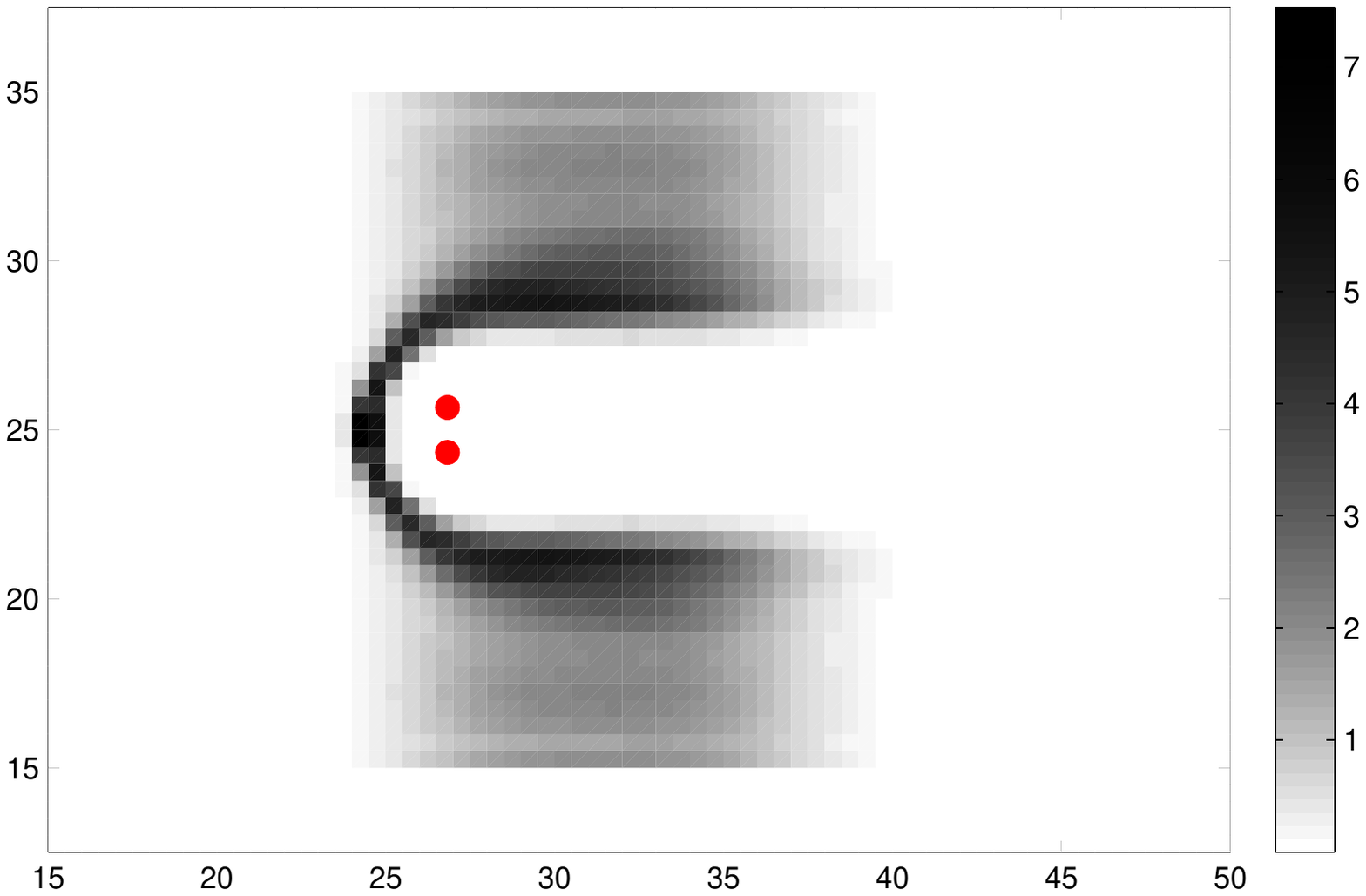}
\vspace{-8.5 cm}\\
\hspace{-1 cm}\includegraphics[width=0.6\linewidth]{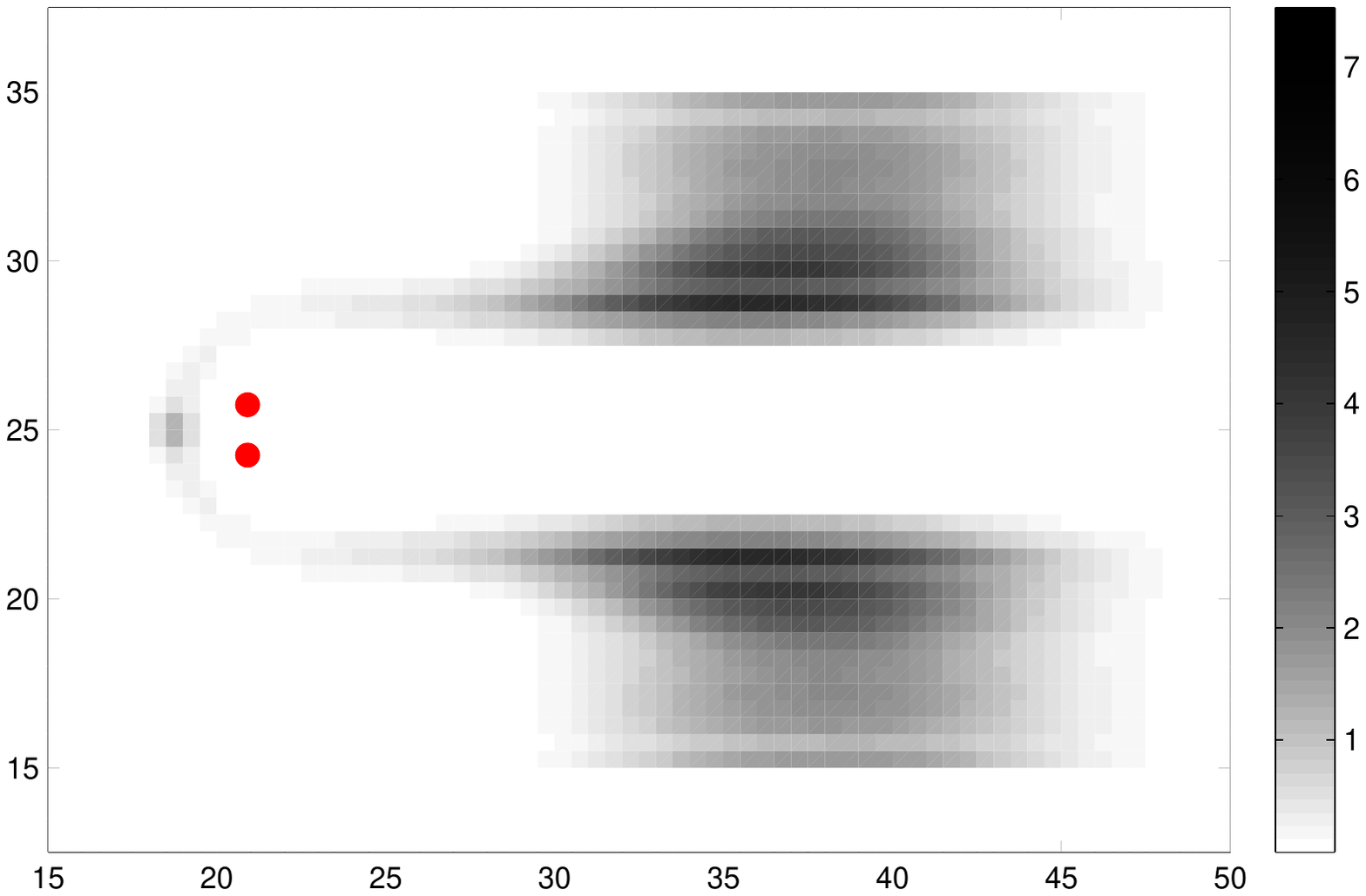} &  \hspace{-3 cm}\\
\end{tabular}
\vspace{-4 cm}
\caption{The interaction between two discrete individuals and a macroscopic crowd. The individuals are initially a distance of $R_r^{\text{AR}}$ apart. The images are taken at $t=0$ (top left), $t=5$ (top right), $t=10$ (bottom left).}\label{figure graph two particles together}
\end{figure}
\section{Interaction with two individuals}\label{section numerics two-scale interactions two ind}
Two-scale simulation experiments can involve two individuals instead of one. The outcome of such a simulation can be seen in Figure \ref{figure graph two particles together}. The two individuals initially are at a relatively short distance from one another. Their mutual distance remains small and macroscopic mass is not able to separate the two. The macroscopic crowd only moves around them. In fact, the impact of the two individuals is not fundamentally different from the impact a single individual of double mass would have.\\
\\
If we increase the initial distance between the two individuals, then the macroscopic mass does find a way to fill the space between them. It even forces the individuals more apart. This can be seen in Figure \ref{figure graph two particles further apart}. Note that it still requires some `effort' to separate the two individuals. In the top right image in Figure \ref{figure graph two particles further apart} a region of high density is present at the left-hand side of the opening between the individuals. We expect that it depends on the initial distance between the two individuals compared to $R_r^{\text{R}}$, whether macroscopic mass can pass through the opening between the two individuals or not.
\begin{figure}[h!]
\centering
\vspace{-5.1 cm}
\begin{tabular}{lr}
\hspace{-1 cm}\includegraphics[width=0.6\linewidth]{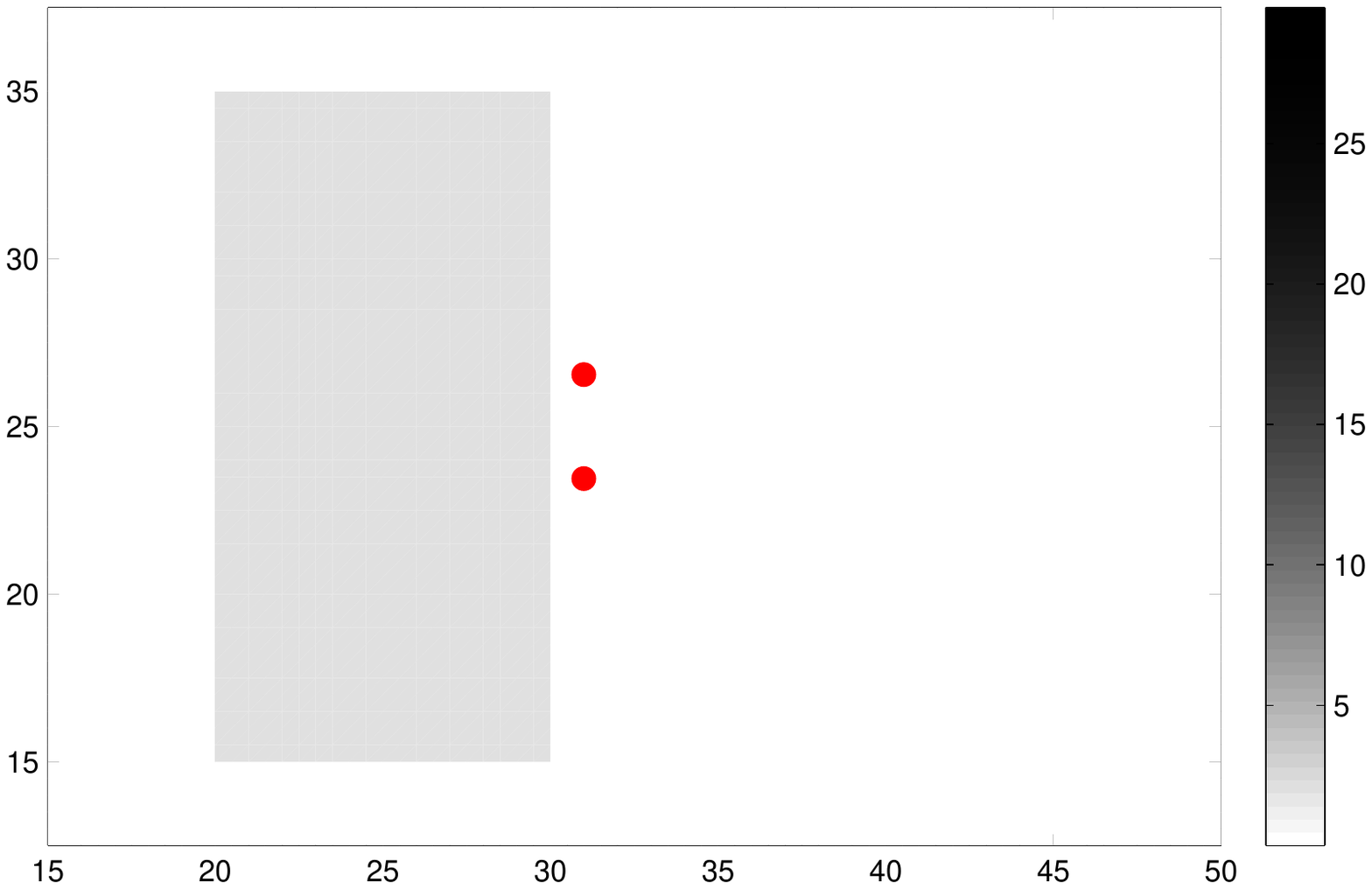}
&
\hspace{-1.5 cm}\includegraphics[width=0.6\linewidth]{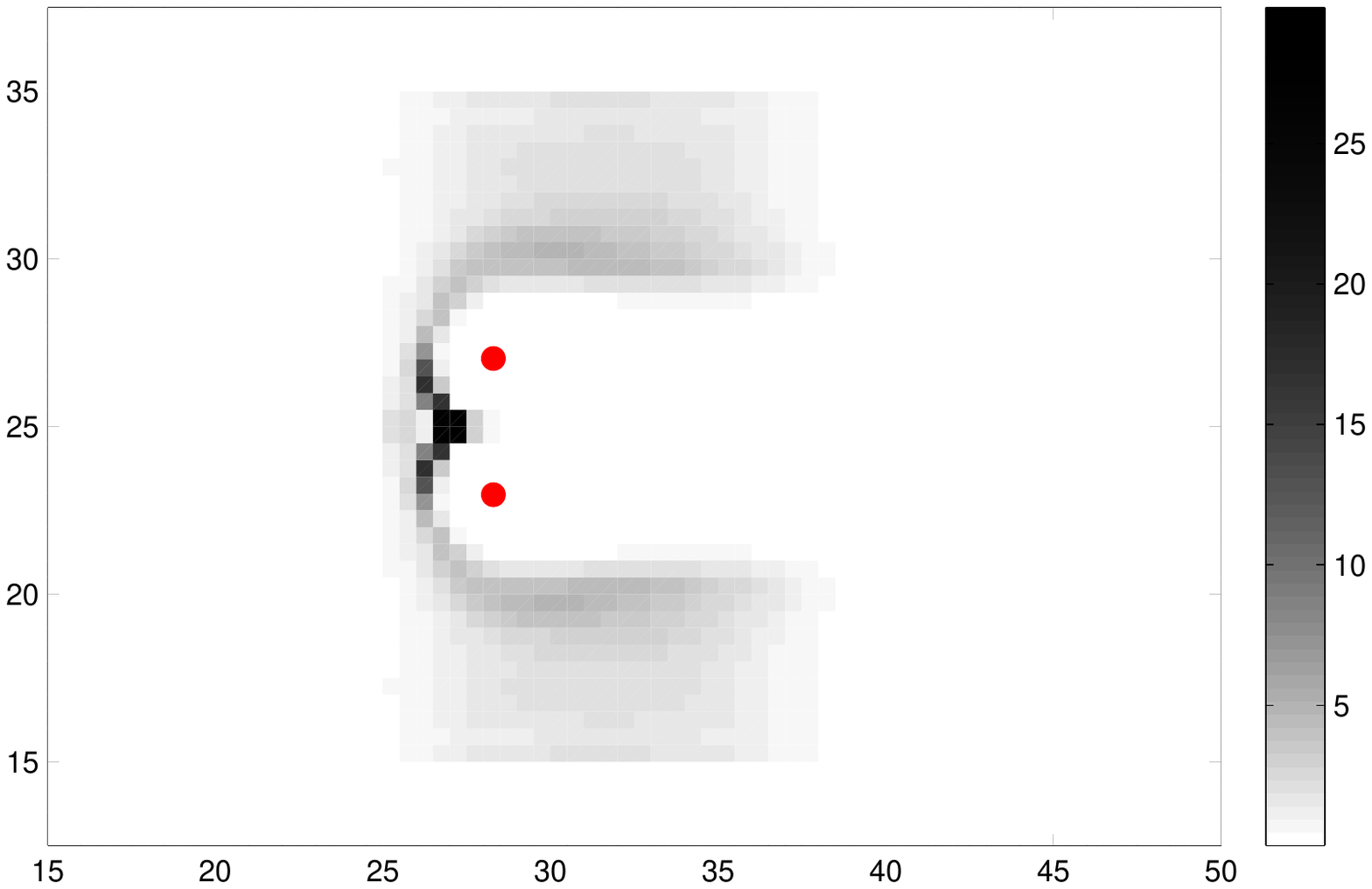}
\vspace{-8.5 cm}\\
\hspace{-1 cm}\includegraphics[width=0.6\linewidth]{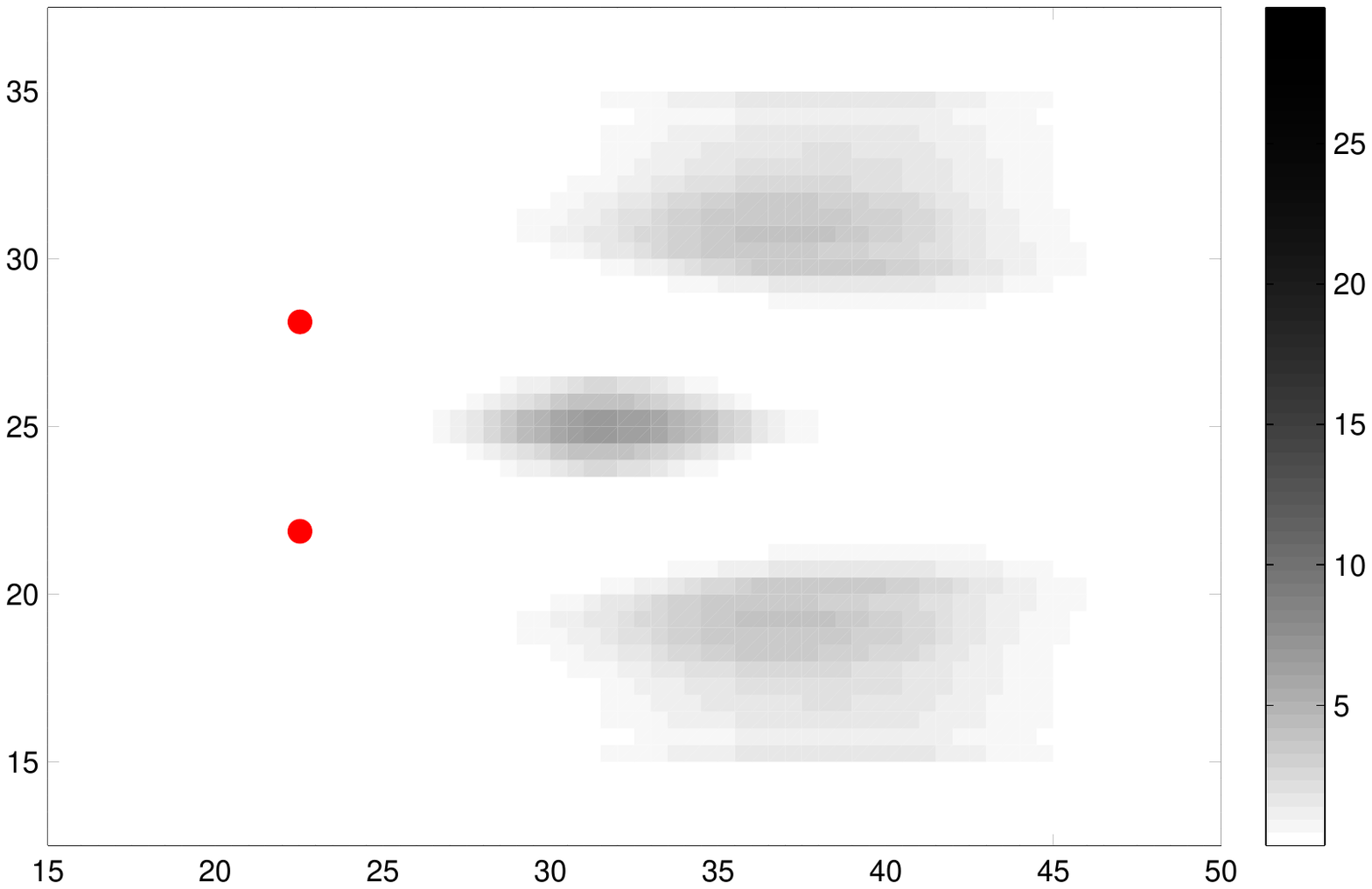} &  \hspace{-3 cm}\\
\end{tabular}
\vspace{-4 cm}
\caption{The interaction between two discrete individuals and a macroscopic crowd. Initially the distance between the individuals is slightly more than $R_a^{\text{AR}}$, and thus they do not `feel' each other. The images are taken at $t=0$ (top left), $t=5$ (top right), $t=10$ (bottom left).}\label{figure graph two particles further apart} \vspace{0.1cm}
\end{figure}

\section{Modelling leadership}\label{section numerics leader}
Can we capture `leadership' effects via our two-scale model? More precisely, can an individual be enabled to drag macroscopic mass along? We interpret such behaviour as a tendency to follow the individual, even though the macroscopic crowd in itself has no special desire to go in a certain direction. The individual then acts as a leader (or guide) for the group. To achieve this effect, in this section, we choose attraction-repulsion as the influence of the individual on the macroscopic crowd, instead of solely repulsive interactions.\\
\\
In Figure \ref{figure graph leader}, the individual is only driven by its desired velocity, which is directed to the left. The crowd does not affect its motion. Yet, the dynamics of the macroscopic crowd is influenced by the `target particle' due to attraction-repulsion. The crowd has no desired velocity, thus $\sigma=1$ is taken. All other conditions are as in the reference setting. Figure \ref{figure graph leader} shows that the individual is able to create a short slipstream of macroscopic mass, when moving through the crowd, but we can not convincingly call it a leader.
\begin{figure}[h!]
\centering
\vspace{-5 cm}
\begin{tabular}{lr}
\hspace{-1 cm}\includegraphics[width=0.6\linewidth]{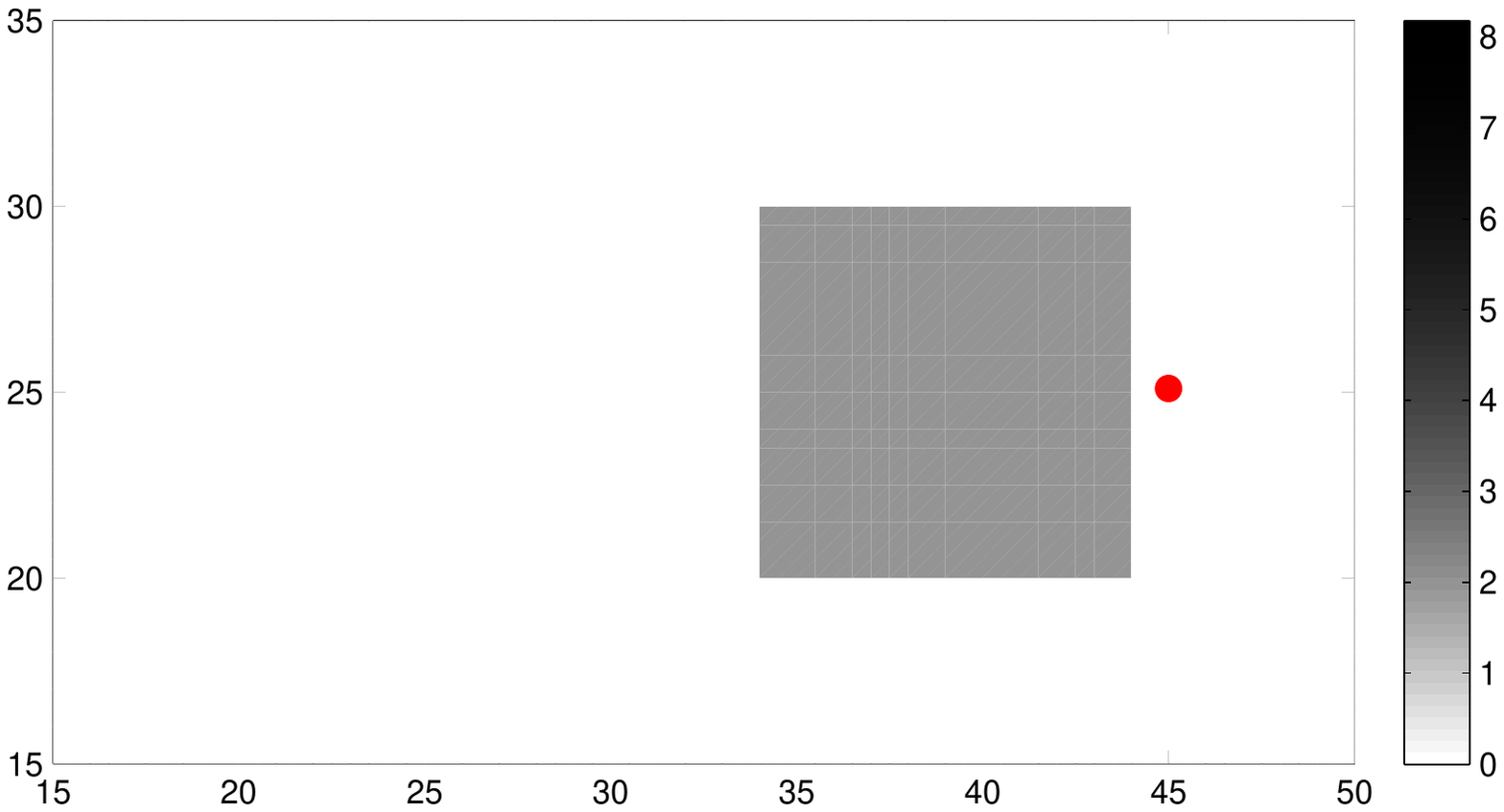}
&
\hspace{-1.5 cm}\includegraphics[width=0.6\linewidth]{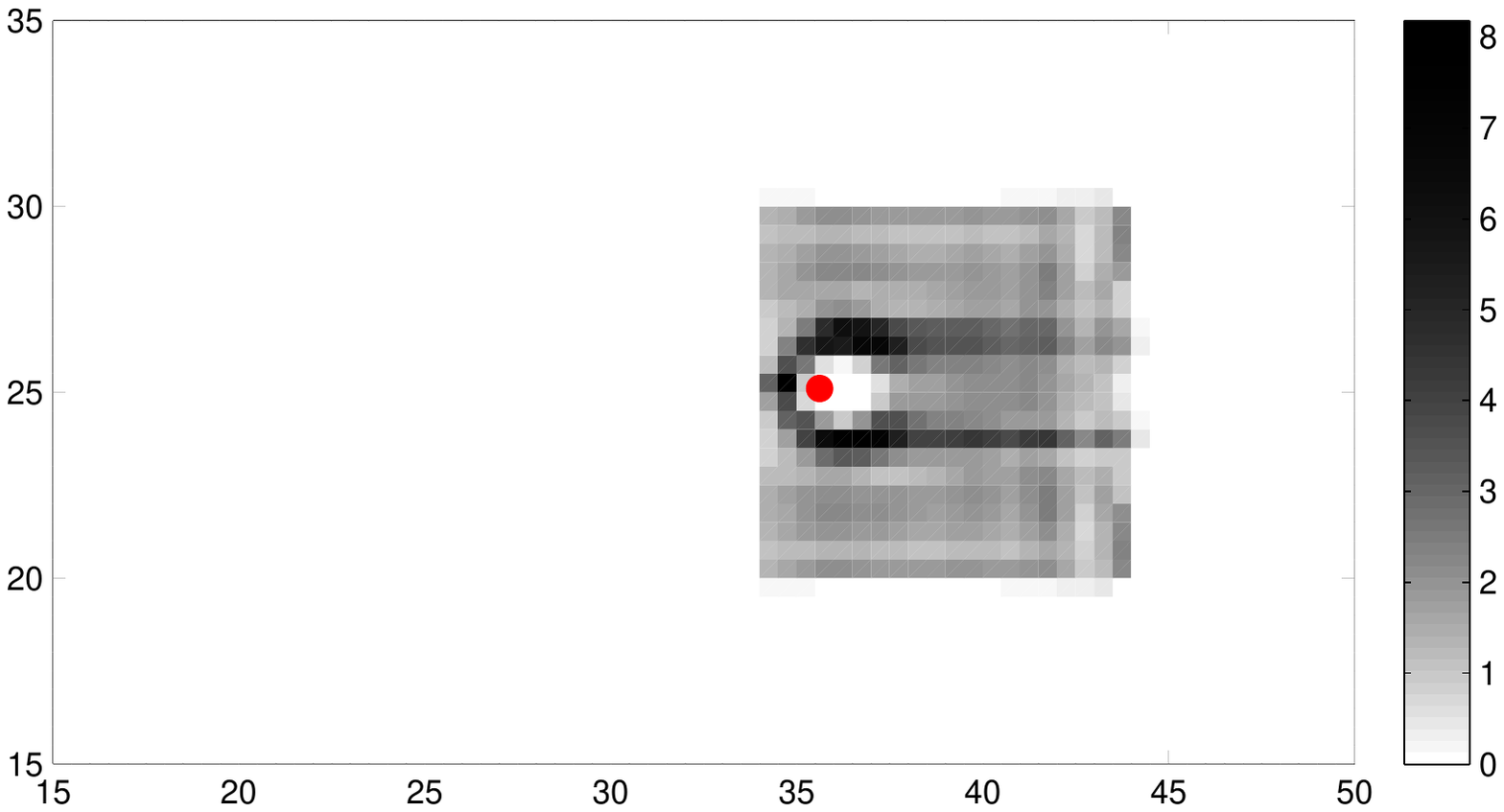}
\vspace{-9 cm}\\
\hspace{-1 cm}\includegraphics[width=0.6\linewidth]{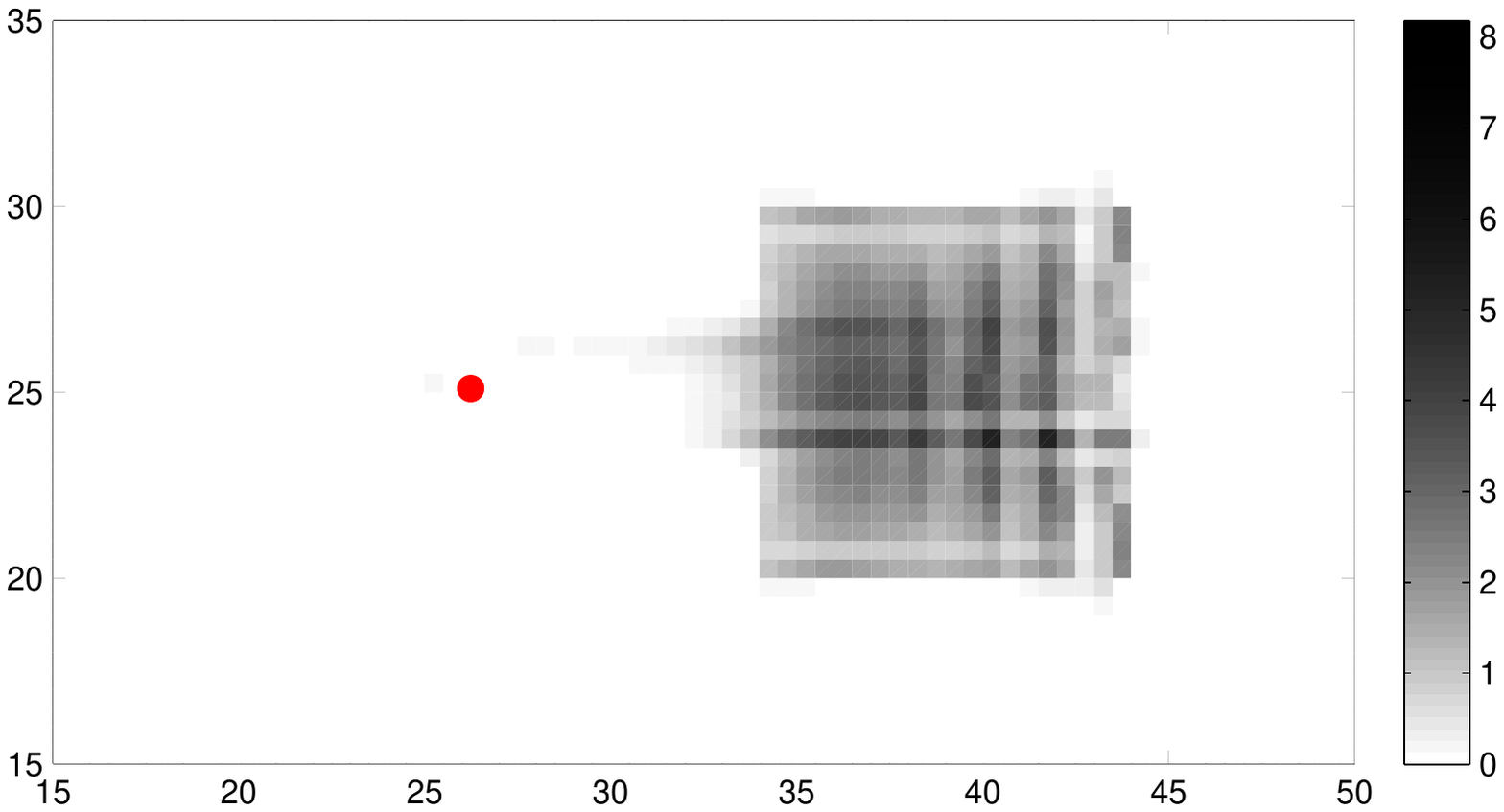} &  \hspace{-3 cm}\\
\end{tabular}
\vspace{-4.5 cm}
\caption{The interaction between a discrete individual and a macroscopic crowd. The attraction-repulsion influence of the individual on the macroscopic crowd makes that the crowd is slightly dragged along. A small amount of mass evades on the left-hand side of the crowd. The images are taken at $t=0$ (top left), $t=7$ (top right), $t=14$ (bottom left).}\label{figure graph leader}
\end{figure}
\begin{figure}[h!]
\centering
\vspace{-5.3 cm}
\begin{tabular}{lr}
\hspace{-1 cm}\includegraphics[width=0.6\linewidth]{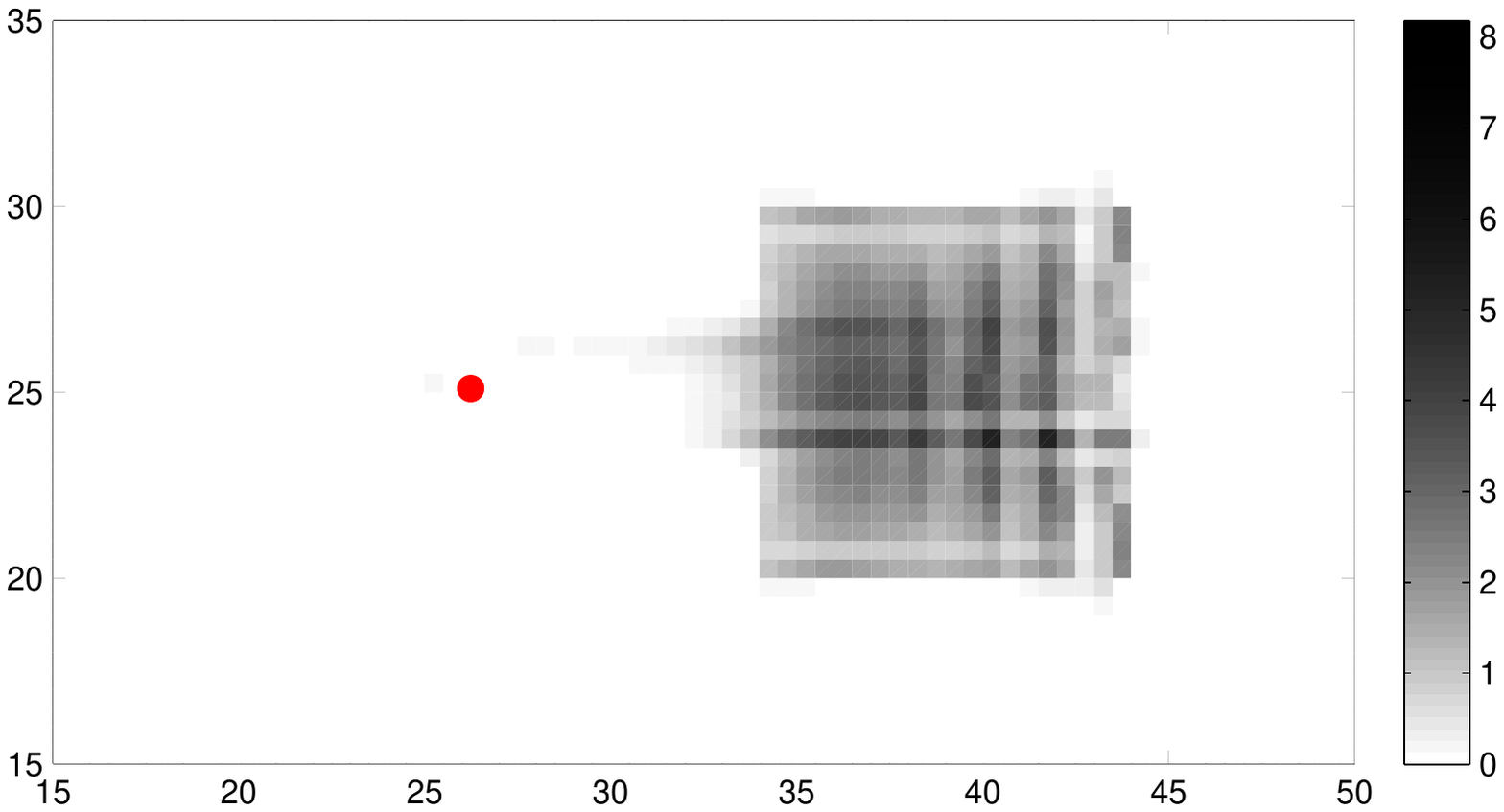}
&
\hspace{-1.5 cm}\includegraphics[width=0.6\linewidth]{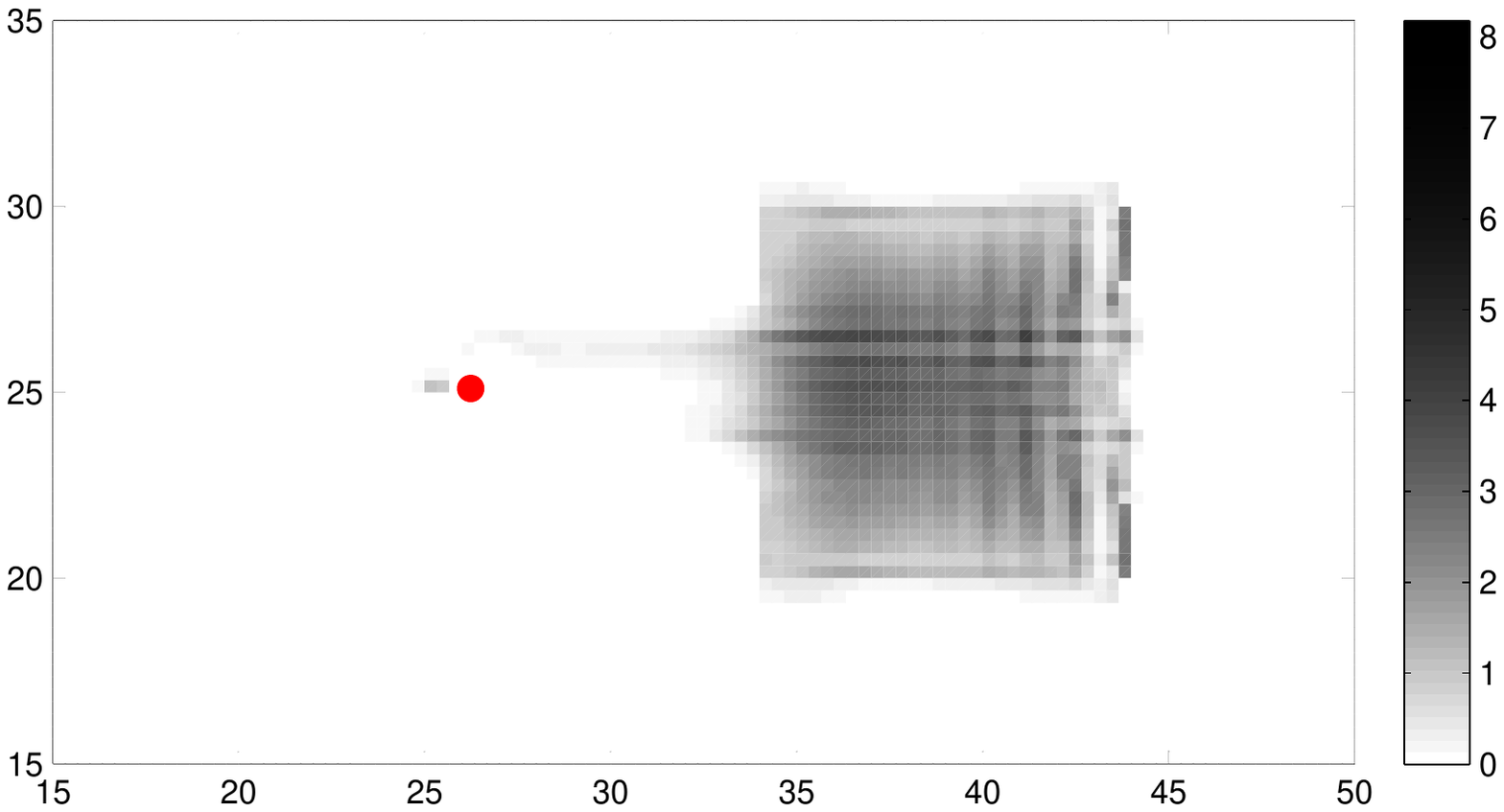}
\vspace{-9 cm}\\
\hspace{-1 cm}\includegraphics[width=0.6\linewidth]{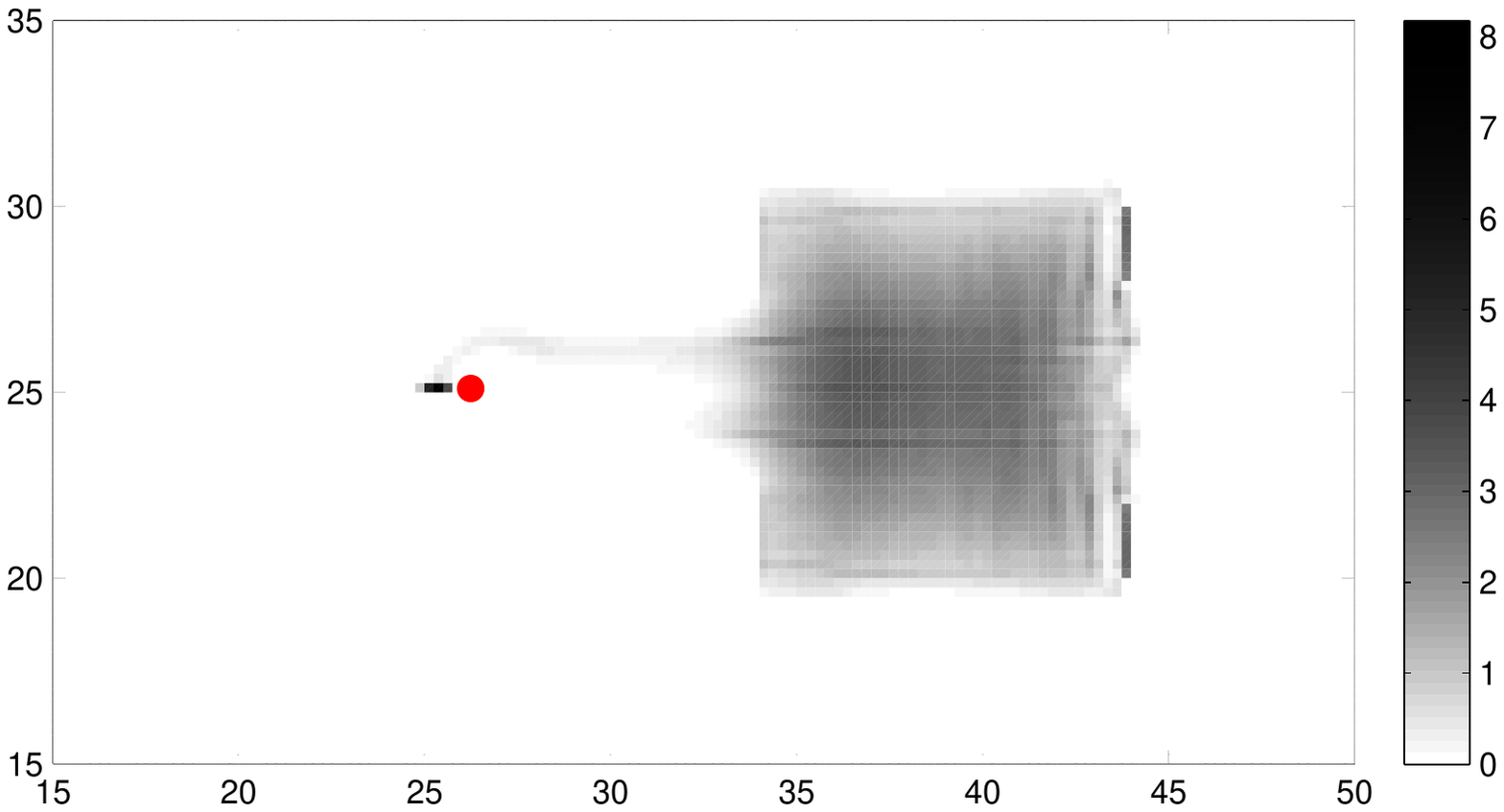} &  \hspace{-3 cm}\\
\end{tabular}
\vspace{-4.5 cm}
\caption{Decreasing mesh size for the setting of Figure \ref{figure graph leader} at time $t=14$. Top left: the original grid size of Figure \ref{figure graph leader}. Top right: the number of cells has been increased by a factor 1.5 in both horizontal and vertical direction. Bottom left: increase of the number of cells by a factor 2 compared to the top left situation. Oscillations in the density are diminished by decreasing grid size.}\label{figure graph leader finer mesh}
\end{figure}

The striking feature appearing in Figure \ref{figure graph leader} are the vertical lines, alternatingly of high and low density. These appear in the macroscopic crowd. Why does this happen? As Figure \ref{figure graph leader finer mesh} shows, these oscillations are apparently only a result of the numerical approximation. The density is smoothened if we decrease the grid size.\\
\\
We wish to enhance the effect of leadership (compared to Figure \ref{figure graph leader}) by adjusting parameters. Firstly, we increase the magnitude of the attraction-repulsion influence of the individual on the crowd. More specifically, we take $F^{\text{AR}}= 0.3$ instead of $F^{\text{AR}}= 0.03$; for interactions within the crowd $F^{\text{AR}}= 0.03$ is maintained. All other conditions are as in Figure \ref{figure graph leader}. In Figure \ref{figure graph leader attr} we see the effect of the performed parameter modification.\\
Note that within the environment of the individual, the density attains very high values, which is physically not acceptable. Let us concentrate on qualitative behaviour only. The individual successfully attracts a large part of the macroscopic crowd, at least initially. As time elapses, it loses control over the crowd. Maybe, this fact is due to his too high velocity. At first, the mass in its direct environment is distributed more or less uniformly over a circle around it; this mass shifts more and more to its rear end. Also, we observe a trace of mass left behind by the individual. Once this mass is at a distance greater than $R^{\text{AR}}_a$ from the individual, it remains near the place where it loses contact.\footnote{The scaling of the grey shading has been altered in order to be able to distinguish lower densities from zero density. Very high densities occurring in small regions dominate the scaling. The aforementioned trace of mass is therefore hardly visible, in spite of the fact the density is around $\rho\approx2$.}
\begin{figure}[h!]
\centering
\vspace{-5 cm}
\begin{tabular}{lr}
\hspace{-1 cm}\includegraphics[width=0.6\linewidth]{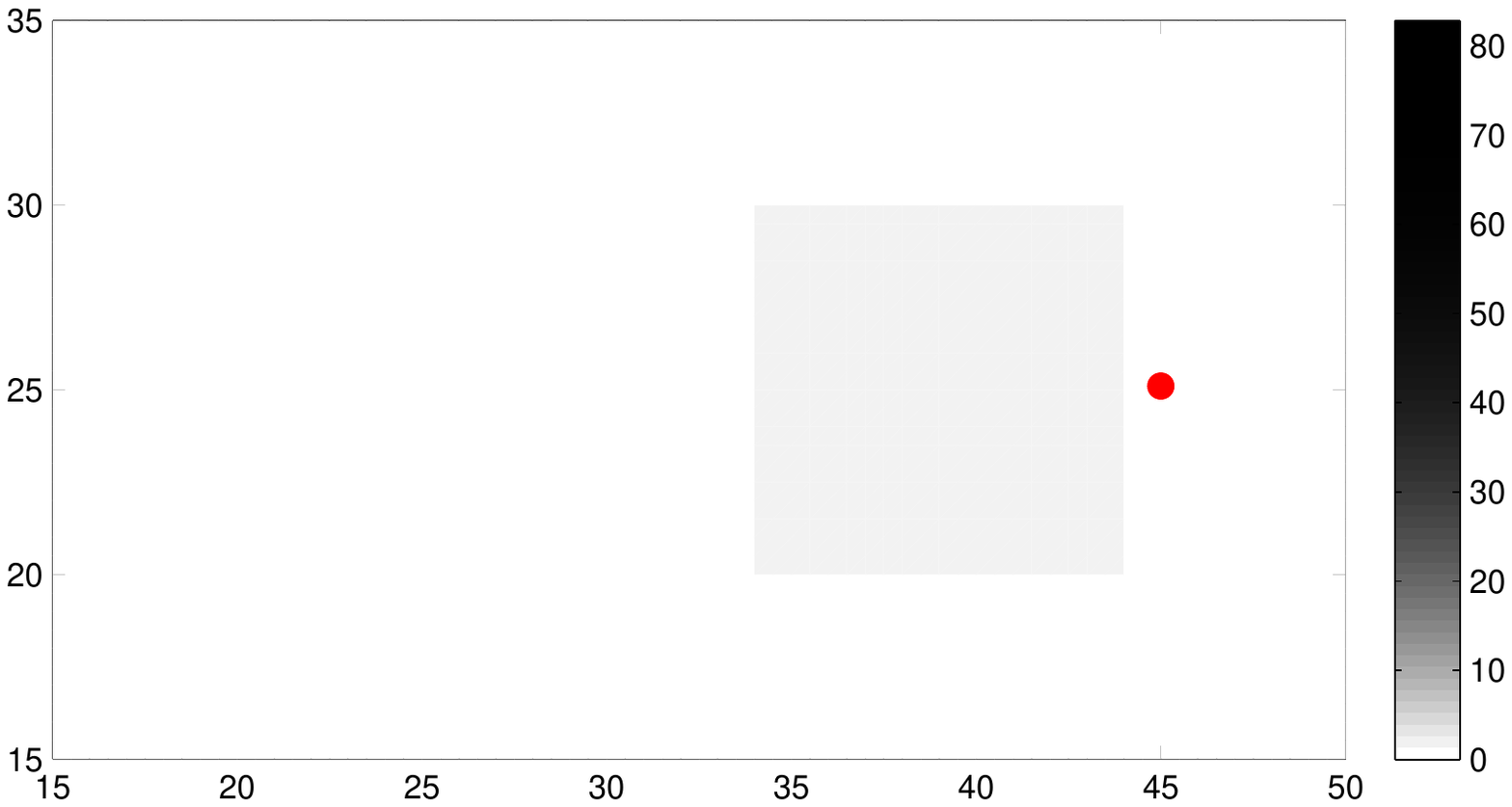}
&
\hspace{-1.5 cm}\includegraphics[width=0.6\linewidth]{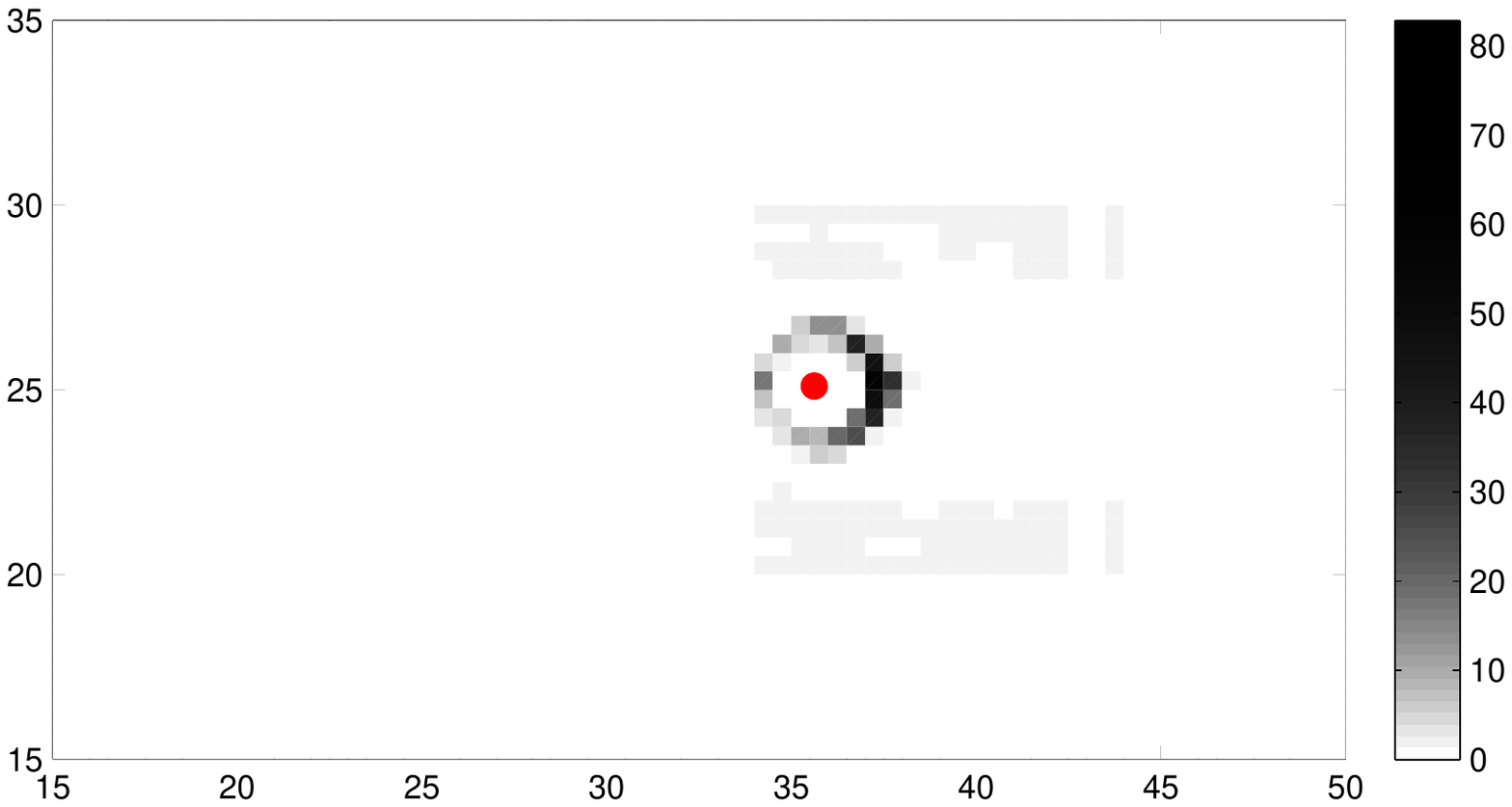}
\vspace{-9 cm}\\
\hspace{-1 cm}\includegraphics[width=0.6\linewidth]{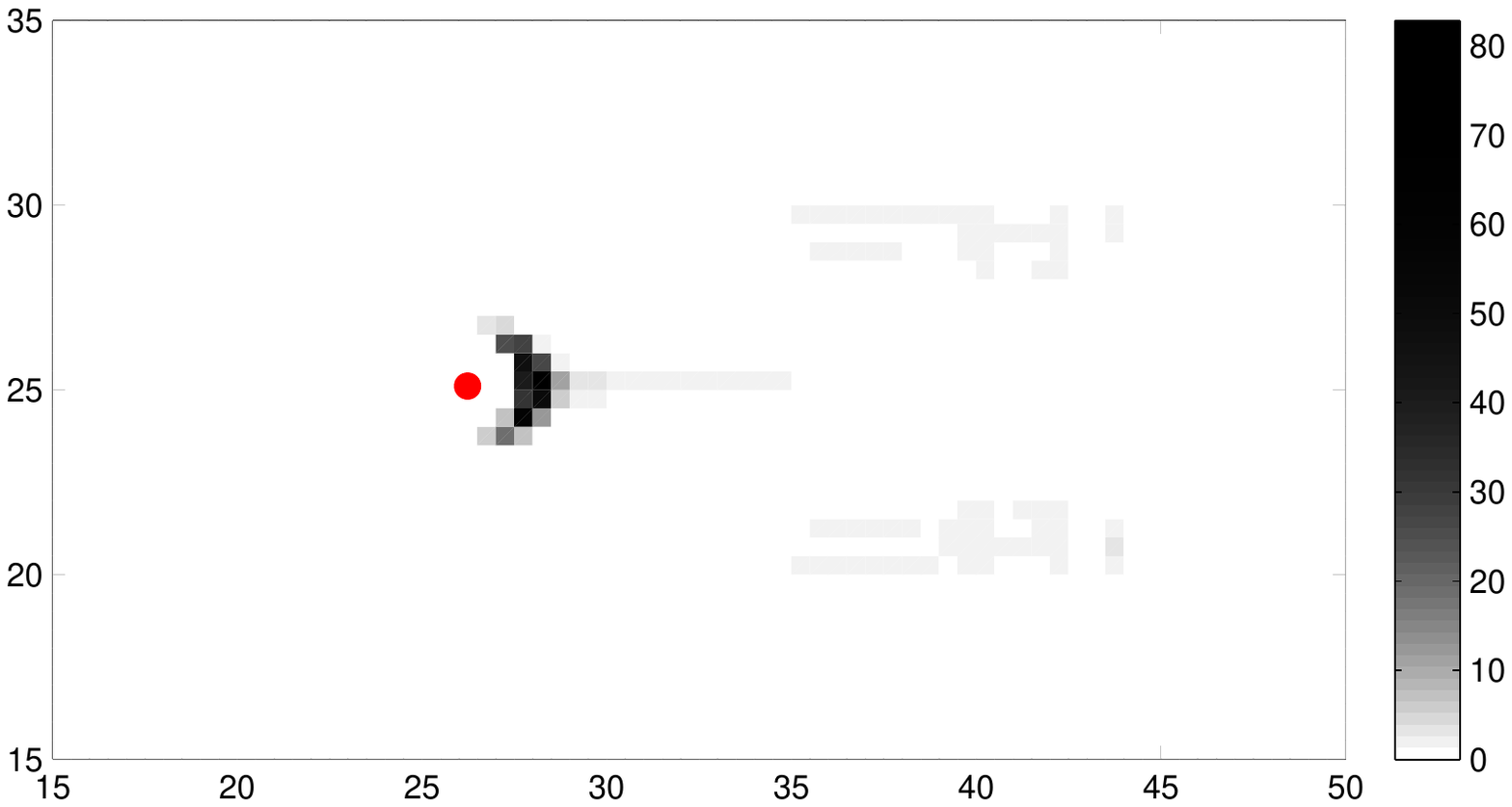} &  \hspace{-3 cm}\\
\end{tabular}
\vspace{-4.5 cm}
\caption{The interaction between a discrete individual and a macroscopic crowd. The influence of the individual on the crowd has been increased, and as a result the individual is able to drag the crowd along. The images are taken at $t=0$ (top left), $t=7$ (top right), $t=14$ (bottom left).}\label{figure graph leader attr}
\end{figure}

Effects of a similar kind can be obtained not only by modifying the interaction strength, but also by increasing the radius of attraction. In the attraction-repulsion influence of the individual on the crowd we restore the value $F^{\text{AR}}= 0.03$, and take $R^{\text{AR}}_a=6$. Indeed this enables the individual to act as a leader, cf. Figure \ref{figure graph leader radii}.\\
The macroscopic crowd is initially attracted to the individual, which is comparable to Figure \ref{figure graph leader attr}. This situation differs from Figure \ref{figure graph leader attr} in the sense that the macroscopic crowd is much less compressed. A possible explanation is twofold. The area of attraction around the leader is much larger, and moreover, the attraction towards the centre of this region is less strong.\\
Note that due to the weaker attraction-repulsion (compared to the previous case $F^{\text{AR}}=0.3$), the leader is less successful in keeping his followers with him. Finally, the macroscopic mass has moved on average a smaller distance to the left. However, a positive aspect of this setting is that the maximal density is nearly a factor 4 smaller than in Figure \ref{figure graph leader attr}.
\begin{figure}[h!]
\centering
\vspace{-5.5 cm}
\begin{tabular}{lr}
\hspace{-1 cm}\includegraphics[width=0.6\linewidth]{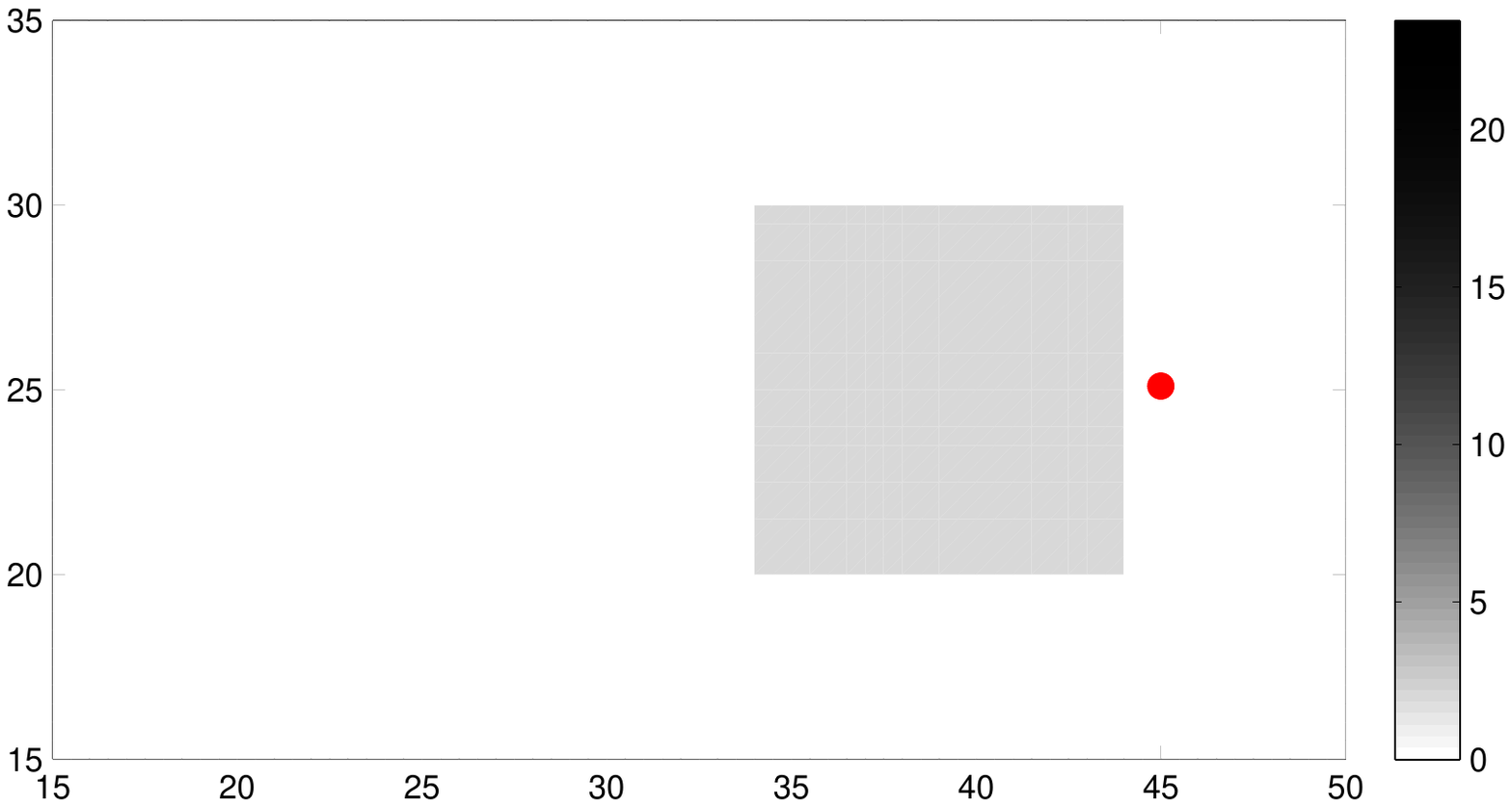}
&
\hspace{-1.5 cm}\includegraphics[width=0.6\linewidth]{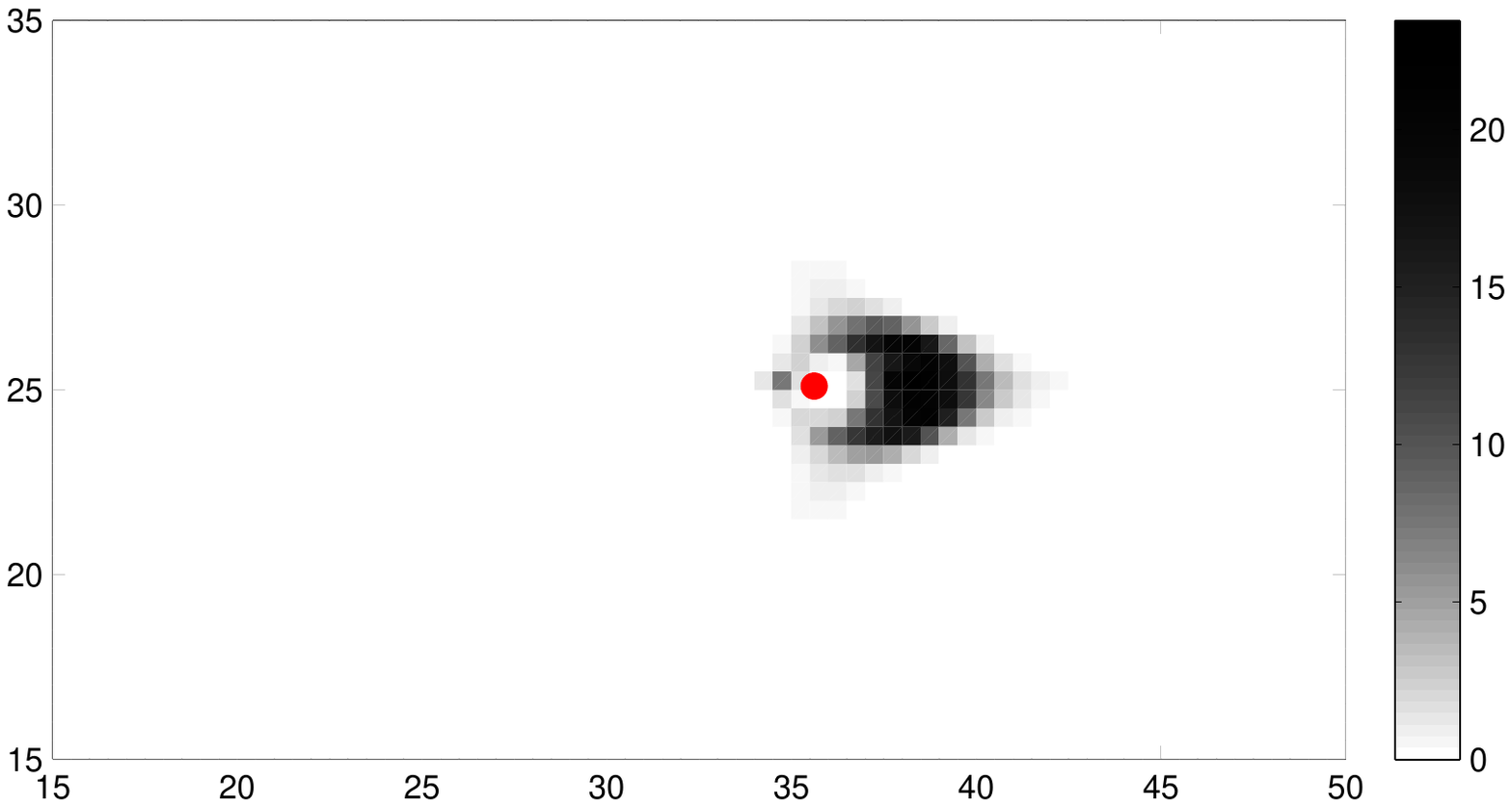}
\vspace{-9 cm}\\
\hspace{-1 cm}\includegraphics[width=0.6\linewidth]{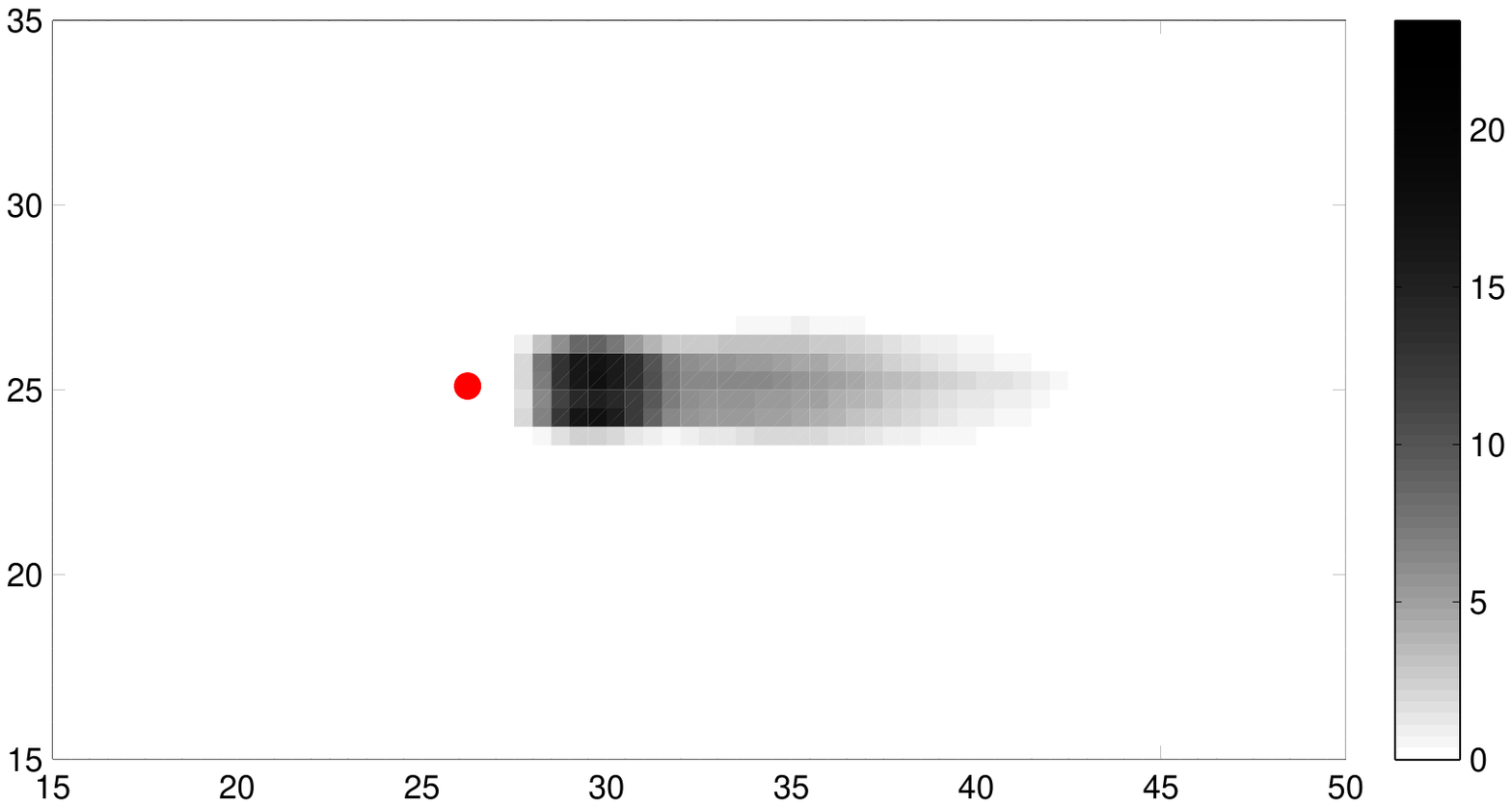} &  \hspace{-3 cm}\\
\end{tabular}
\vspace{-4.5 cm}
\caption{The interaction between a discrete individual and a macroscopic crowd. The radius of attraction has been increased, which enables the individual (temporarily) to take macroscopic crowd along. The images are taken at $t=0$ (top left), $t=7$ (top right), $t=14$ (bottom left).}\label{figure graph leader radii}
\end{figure}

\newpage
\chapter{Discussion}\label{section further work}
At this point, we wish to look back and list a few open issues with respect to modelling, analysis and simulation. We will investigate some of them in the near future.
\begin{itemize}
\item Via (mainly formal) calculations we arrived at a fully continuous-in-time model. Subsequently, we deduced a time-discrete version of this model, which we analyzed mathematically. We proved global existence and uniqueness of a time-discrete solution (in Theorems \ref{Thm existence time-discrete}--\ref{Thm uniqueness time-discrete}), and showed some properties of this solution (e.g. positivity and conservation of mass). A number of conditions on the (one-step) motion mappings and velocity fields were inevitably needed. Afterwards, we proposed a specific form of the velocity field, in which a dependence on the instantaneous configuration of the mass appears. Up to now, we have not proven that the chosen velocity field actually satisfies the imposed conditions. This is still to be done. We expect this step to be difficult, mainly due to the functional dependence of the velocity on the mass measures, and due to the fact that the dynamics of all subpopulations are coupled via their velocity fields.
\item Formulating and proving results like existence and uniqueness of solutions to the continuous-in-time problem, is our next goal. However, this will be a much more difficult task than treating the discrete-in-time model. For instance, the one-step motion mapping of the mass measures is linear in $v^{\alpha}_n$ for the time-discrete model, cf. Definition \ref{def one-step motion mapping}. This is because we linearized by using Taylor series approximations in order to obtain the time-discrete model. The motion mapping will in general no longer be linear in the velocity if we consider the continuous case. Moreover, it is all but clear whether we can even speak so easily about a \text{one-step} motion mapping and the accompanying push forward of the mass measure. These concepts are likely to translate into continuous-in-time counterparts, relating the configurations at time instances that are not an \textit{a priori} fixed time step apart.
\item We would like to enlarge the class of interaction potentials for which an entropy inequality can be derived. For the moment the potential $W_{\beta}^{\alpha}$ is only admissible if $W_{\beta}^{\alpha}(\xi)=W_{\beta}^{\alpha}(-\xi)$ for all $\xi\in\mathbb{R}^d$, and interactions are symmetric. The latter condition means that $W_{\beta}^{\alpha}\equiv W_{\alpha}^{\beta}$. See Sections \ref{section entropy inequality cont-in-time} and \ref{section entropy inequality disc-in-time}, and especially Assumption \ref{assumption gradient structure velocity multi-comp}.\\
    These restrictions on $W_{\beta}^{\alpha}$  are simply to strong for many interesting settings. An idea at least to circumvent the condition $W_{\beta}^{\alpha}(\xi)=W_{\beta}^{\alpha}(-\xi)$ is inspired by \cite{BodnarVelazquez} (especially Section 3.1 therein). We write $W_{\beta}^{\alpha}$ as the superposition of a symmetric part and a drift part:
    \begin{equation*}
    W_{\beta}^{\alpha}= W_{\beta,\text{ symm}}^{\alpha}+W_{\beta,\text{ drift}}^{\alpha},
    \end{equation*}
    where $W_{\beta,\text{ symm}}^{\alpha}(\xi)=W_{\beta,\text{ symm}}^{\alpha}(-\xi)$ is satisfied for all $\xi\in\mathbb{R}^d$. The second term $W_{\beta,\text{ drift}}^{\alpha}$ contains the deviations from symmetry. We are interested to see which restrictions on $W_{\beta,\text{ drift}}^{\alpha}$ are needed in order still to be able to derive an entropy inequality.\\
    However, some problems arise in the interpretation of such interactions. Consider the vector $\xi$ in its representation in polar coordinates. Without the condition that $W_{\beta}^{\alpha}(\xi)=W_{\beta}^{\alpha}(-\xi)$, the interaction potential $W_{\beta}^{\alpha}$ certainly depends not only on the length of $\xi$ but also on its angle (the azimuth). Since the velocity field involves $\nabla W_{\beta}^{\alpha}$, it inevitably has a non-zero azimuthal component. We should ask ourselves the question whether this is physically acceptable: are the interactions allowed to influence the velocity in a direction that is not parallel to the connection vector between two points?\\
    If we do not want to restrict ourselves to a gradient (in this case $\nabla W_{\beta}^{\alpha}$), we could replace it by a vector field of the form $\nabla W+\nabla\times\omega$ (Helmholtz decomposition). Under which conditions on the vector field $\omega$ can we still derive an entropy inequality?
\item A wider range of domains can be considered, if we know how to incorporate walls (outer boundaries) and internal obstacles in our domain. We have not treated this aspect so far, but it surely is very important to do so. In \cite{PiccoliTosin}, Piccoli proposes to project the obtained velocity on the space of admissible velocities. That is, those velocities that satisfy $v\cdot n\geqslant0$ at the boundaries. In practice, this means that outward pointing normal components of the velocity are to be set zero. Perhaps, we would like to have an alternative way, which fits better to our measure-theoretical framework. Probably one needs to define boundary measures and Cauchy fluxes, while also describing their throughput.
\item More fundamental mathematical challenges are related to \textit{discrete-to-continuum limits}. We would like to investigate in what sense we can approximate the absolutely continuous part of our mass measures by particle systems (discrete mass measures in the case of crowd dynamics). Especially, the limiting process $N\rightarrow \infty$ (where $N$ is the number of Dirac measures in the particle system) is relevant. Such approximation, in combination with the Wasserstein metric, might turn out to be useful for proving existence and uniqueness properties for the continuous-in-time model, and also to derive alternative numerical schemes for computing the density based on particle simulations.\\
    The question is which scaling one needs to obtain our macroscopic perspective in the limit. Moreover, we are interested in whether other types of scaling would lead to different limits, e.g. a \textit{mesoscopic} Boltzmann-type description (`intermediate level' between micro and macro, see e.g. \cite{Carrillo2010}) or the twofold micro-macro setting of \cite{Piccoli2010}. Perhaps there even exists a way to exploit the fact that we have multiple subpopulations. That is, we might apply different scalings to distinct subpopulations, such that e.g. for one subpopulation the limit is of Boltzmann-type, and for another the limit is a two-scale micro-macro description. What is the actual interpretation of such distinction between the two subpopulations?
\item Regarding our simulation results shown in Section \ref{section numerical illustration}: We observe that we have successfully circumvented problems that might occur at the boundary, by considering simulation runs as long as all mass is contained in the interior of $\Omega$. Such \textit{ad hoc} solution can be avoided in the future if we find appropriate tools for including boundary effects (see above). When outputting the configuration, our simulation tool is equipped to check whether the total mass is actually conserved. This turns out to be indeed the case for every instance of our simulation.
\item When examining the results of our simulation, several questions arise. All of them are the subject for further research:
    \begin{itemize}
    \item Figure \ref{figure graph cloud to ball} indicates that the equilibrium configuration for attraction-repulsion interactions is a ball. If instead, we test the same situation for purely attractive interactions, does the crowd shrink to a point? Is there a mathematical way to show that the limit is a Dirac mass? Note that in this case the right-hand inequality in (\ref{hypothesis upper and lower bound lambda(chi inv Omega)}) can no longer be satisfied.
    \item In Section \ref{section numerics two-scale interactions one ind}, we estimate the radius of the (nearly) empty area in front of the individual. It might be possible to characterize the shape and dimensions of the rest of the empty zone induced by the individual. In the same section, we show that a modification of parameters makes the crowd clog together again after the individual has gone past. What are the typical length and time scales at which this reunion takes place?
    \item The interaction between two individuals and a macroscopic crowd was simulated in Section \ref{section numerics two-scale interactions two ind}. If we zoom out, under what conditions can we deduce from the macro-crowd's behaviour that more than one individual was present (cf. Figure \ref{figure graph two particles further apart})? When do we obtain merely the same results as in the presence of one individual of double mass (cf. Figure \ref{figure graph two particles together}), and how does this effect depend on the initial mutual distance between the two individuals?
    \item In Section \ref{section numerics leader} we tried to create leaders. We observed (especially in Figures \ref{figure graph leader attr} and \ref{figure graph leader radii}) that the individual at first is able to take a part of the macroscopic crowd along, but that he loses control as time goes by. Regarding the long-time behaviour of the crowd, will eventually all mass be lost by the leader? When do leadership effects end? In Figure \ref{figure graph leader attr}, bottom left, the `tail' of mass behind the individual consists of two regions. Closest to the individual the density is relatively high, whereas after quite sudden transition the density is lower. What are the characteristics of both regions (shape, dimensions), and why is there an apparent strict separation between them?
    \item In all of the cases above: How exactly do the observed phenomena depend on the parameters? What can we say about the (in)stability of equilibrium configurations? Can we tune parameters in such a way that physically unacceptable densities do not occur? Can we identify a set of dimensionless quantities, that characterizes the dynamics in the model? Also the role of the initial conditions needs to be taken into consideration.
    \end{itemize}
\end{itemize}

\newpage
\chapter*{Acknowledgements}
\addcontentsline{toc}{chapter}{Acknowledgements}
I wish to express my appreciation to a number of people. First of all, to my supervisor Adrian Muntean for his guidance and support. He has shaped me as a mathematician, a scientist and in some sense also as a human being. "How is life?" - "Good and busy" has become a sort of mantra. Mark Peletier was and is a source of inspiration for me. Hearing him speak, convinced me again and again that there is much more to be explored about crowd dynamics.\\
\\
I would like to thank the Assessment Committee for my master's thesis; its members are: Adrian Muntean, Prof.dr.ir.~Harald van Brummelen (Dept. of Mathematics and Computer Science and Dept. of Mechanical Engineering), Dr.ir.~Huub ten Eikelder (Dept. of Biomedical Engineering), Dr.~Georg Prokert (Dept. of Mathematics and Computer Science) and Dr.ir.~Fons van de Ven (Dept. of Mathematics and Computer Science).\\
Fons van de Ven also deserves a personal word of praise for sharing his experience and participating very actively in the process leading to this thesis. Also, I am indebted to Michiel Renger~MSc for numerous discussions, often on the interpretation of certain concepts. I apologize for keeping you from your own work, Michiel!\\
The Particle Systems seminar at the Institute for Complex Molecular Systems (ICMS) has given me a better understanding of topics closely related to my problem. I acknowledge the contributions of all participants to the seminar.\\
\\
The help of Dr.ir.~Bart Markvoort (Dept. of Biomedical Engineering) to my simulation code was of vital importance. Without his program (that served as a basis for my two-scale experiments) I still would have been trying. Thank you!\\
My uncle Joep Julicher once again rendered a service to one of his relatives with his photo camera. I am most grateful for his contribution to my thesis.\\
\\
Thanks to all my fellow students, but particularly to Patrick van Meurs and Jeroen Bogers. I will miss the `Three Musketeers' of HG~8.64, and the Cup-a-Soup ritual at 12:30.\\
Furthermore, thanks are due to my friends (especially the ones studying humanities), for their neverending ridiculization of mathematics, mathematicians and science in general. Your attitude has truly provided me with the best motivation to be persistent.\\
\\
My gratitude (of course) goes to my parents, for their support that has always gone without saying, and to Janneke and Frank, Willem and Anneke for taking care of their little brother.\\
\\
I am extremely happy that the ICMS gives me the opportunity to continue my research for another four years.

\newpage
\appendix
\numberwithin{equation}{chapter}
\setcounter{theorem}{0}
\renewcommand{\thetheorem}{\Alph{chapter}.\arabic{theorem}}

\chapter{Proof of Theorem \ref{Thm Lebesgue decomp}}\label{Appendix Proof Lebesgue decomp}
This proof follows the lines indicated by \cite{EvansGariepy}, p.~42.
\begin{proof}
\begin{enumerate}
  \item Define the set $\mathcal{E}:=\{\Omega'\in \mathcal{B}(\Omega) \,\big|\, \lambda(\Omega\setminus \Omega')=0\}$. Note that $\mathcal{E}$ is non-empty since at least $\Omega\in\mathcal{E}$. Choose a sequence $\{B_k\}_{k\in\mathbb{N}^+}\subset\mathcal{E}$ such that
       \begin{equation*}
       \mu(B_k)\leqslant \inf_{\Omega'\in\mathcal{E}}\mu(\Omega')+\dfrac{1}{k},\hspace{1 cm}\text{for all }k\in\mathbb{N}^+.
       \end{equation*}
       Such sequence exists by definition of the infimum.\\
       Define $B:=\bigcap_{k=1}^{\infty}B_k$; $B$ is an element of $\mathcal{B}(\Omega)$, since any $\sigma$-algebra is closed w.r.t. countable intersections. By one of De Morgan's laws we have that $\Omega\setminus \bigl(\bigcap_{k=1}^{\infty}B_k \bigr)=\bigcup_{k=1}^{\infty}\bigl(\Omega\setminus B_k \bigr)$, for which we should note that this is not a disjoint union. Therefore we have
       \begin{equation*}
       0\leqslant\lambda(\Omega\setminus B)\leqslant \sum_{k=1}^{\infty}\lambda(\Omega\setminus B_k)=0,
       \end{equation*}
       since $\lambda(\Omega\setminus B_k)=0$ for all $k\in\mathbb{N}^+$. We have thus proven that $B\in\mathcal{E}$.\\
       In fact we even have that $\mu(B)=\inf_{\Omega'\in\mathcal{E}}\mu(\Omega')$. To see this, note that
       \begin{equation*}
       \mu(B)=\mu\Bigl(\bigcap_{k=1}^{\infty}B_k\Bigr)\leqslant \mu(B_j)\leqslant \inf_{\Omega'\in\mathcal{E}}\mu(\Omega')+\dfrac{1}{j},\hspace{1 cm}\text{for all }j\in\mathbb{N}^+.
       \end{equation*}
       Since this statement is true for any $j\in\mathbb{N}^+$, it follows that $\mu(B)\leqslant \inf_{\Omega'\in\mathcal{E}}\mu(\Omega')$. Due to the fact that $B\in\mathcal{E}$, it is clear that also $\mu(B)\geqslant \inf_{\Omega'\in\mathcal{E}}\mu(\Omega')$, thus $\mu(B)= \inf_{\Omega'\in\mathcal{E}}\mu(\Omega')$.
  \item Let $\mu_{\text{ac}}:\mathbb{B}(\Omega)\rightarrow \mathbb{R}^+$ be defined by $\mu_{\text{ac}}(\Omega')=\mu(\Omega'\cap B)$, and similarly $\mu_{\text{s}}:\mathbb{B}(\Omega)\rightarrow \mathbb{R}^+$ by $\mu_{\text{s}}(\Omega')=\mu\bigl(\Omega'\cap (\Omega\setminus B)\bigr)$. Note that it follows from this definition that $\mu=\mu_{\text{ac}}+\mu_{\text{s}}$.\\
      \\
      Let $A\in\mathcal{B}(\Omega)$, be such that $A\subset B$ and $\lambda(A)=0$. Assume that $\mu(A)>0$. Note that $B\setminus A\in\mathcal{B}(\Omega)$, and
      \begin{equation*}
      \lambda\bigl(\Omega\setminus(B\setminus A)\bigr)=\lambda\bigl((\Omega\setminus B)\cup A\bigr)=\lambda(\Omega\setminus B)+\lambda(A)=0.
      \end{equation*}
      Here we use that $\Omega\setminus B$ and $A$ are disjoint (since $A\subset B$) and that $\lambda(\Omega\setminus B)=0$ (since $B\in\mathcal{E}$). We conclude that $B\setminus A\in\mathcal{E}$.\\
      As a result of the fact that $A\subset B$, it holds that $\mu(B)=\mu(B\setminus A)+\mu(B\cap A)=\mu(B\setminus A)+\mu(A)$, and because we assumed that $\mu(A)>0$ it follows that $\mu(B\setminus A)<\mu(B)$. As $B\setminus A\in\mathcal{E}$, this contradicts $\mu(B)= \inf_{\Omega'\in \mathcal{E}} \mu(\Omega')$. Thus $\mu(A)=0$ must hold.\\
      Let $\Omega'\in \mathcal{B}(\Omega)$ now be such that $\lambda(\Omega')=0$. Then $\mu_{\text{ac}}(\Omega')=\mu(\Omega'\cap B)=0$ follows if we set $A:=\Omega'\cap B\subset B$ in the lines of arguments above. Hence, $\mu_{\text{ac}}\ll\lambda$.\\
      \\
      Furthermore $\mu_{\text{s}}(B)=\mu\bigl( B\cap(\Omega\setminus B)\bigr)=\mu(\emptyset)=0$, and, due to the fact that $B\in\mathcal{E}$, $\lambda(\Omega\setminus B)=0$. This implies $\mu_{\text{s}}\perp\lambda$.
  \item To prove uniqueness of this decomposition, assume that we have $\mu=\mu_{\text{ac}}^1+\mu_{\text{s}}^1$ and $\mu=\mu_{\text{ac}}^2+\mu_{\text{s}}^2$, where $\mu_{\text{ac}}^i\ll\lambda$ and $\mu_{\text{s}}^i\perp\lambda$ for $i\in\{1,2\}$. For all $\Omega'\in\mathcal{B}(\Omega)$ we thus have
      \begin{equation*}
      (\mu_{\text{ac}}^1-\mu_{\text{ac}}^2)(\Omega')=(\mu_{\text{s}}^2-\mu_{\text{s}}^1)(\Omega')=:\tilde{\mu}(\Omega').
      \end{equation*}
      Neither in the left-hand side, nor in the right-hand side we necessarily have a positive measure. However it can quite easily be seen that $(\mu_{\text{ac}}^1-\mu_{\text{ac}}^2)\ll\lambda$, because $\lambda(\Omega')=0$ implies $\mu_{\text{ac}}^i(\Omega')=0$ for $i=1,2$ (as $\mu_{\text{ac}}^i\ll\lambda$), and thus also $(\mu_{\text{ac}}^1-\mu_{\text{ac}}^2)(\Omega')=0$.\\
      \\
      Furthermore there exist $B_1, B_2\in\mathcal{B}(\Omega)$, such that $\mu_{\text{s}}^i(\Omega\setminus B_i)=\lambda(B_i)=0$ for all $i\in\{1,2\}$ and for all $\Omega'\in\mathcal{B}(\Omega)$. Define $\tilde{B}:=B_1\cup B_2$, and note that $\Omega\setminus \tilde{B}\subset \Omega\setminus B_i$ for each $i\in\{1,2\}$. It follows that $0\leqslant \mu_{\text{s}}^i(\Omega\setminus\tilde{B})\leqslant \mu_{\text{s}}^i(\Omega\setminus B_i)=0$, for $i\in\{1,2\}$, thus $\mu_{\text{s}}^i(\Omega\setminus\tilde{B})=0$ and thus
      \begin{equation*}
      (\mu_{\text{s}}^2-\mu_{\text{s}}^1)(\Omega\setminus\tilde{B})=0.
      \end{equation*}
      Also, $0\leqslant \lambda(\tilde{B})=\lambda(B_1)+\lambda(B_2)=0$, and thus
      \begin{equation*}
      \lambda(\tilde{B})=0,
      \end{equation*}
      from which we conclude that $(\mu_{\text{s}}^2-\mu_{\text{s}}^1)\perp\lambda$.\\
      \\
      We now have $\tilde{\mu}\ll\lambda$ and $\tilde{\mu}\perp\lambda$, and will show that this implies $\tilde{\mu}\equiv0$. Use statement (\ref{alt def sing meas Rudin}) from Lemma \ref{alternative def singular measures} to characterize singular measures, by which we have disjoint $A_1, A_2\in\mathcal{B}(\Omega)$ such that $\tilde{\mu}(\Omega')=\tilde{\mu}(\Omega'\cap A_1)$ and $\lambda(\Omega')=\lambda(\Omega'\cap A_2)$ for all $\Omega'\in\mathcal{B}(\Omega)$. For any $\Omega'\in\mathcal{B}(\Omega)$, we have $\lambda(\Omega')=\lambda(\Omega'\setminus A_2)+\lambda(\Omega'\cap A_2)= \lambda(\Omega'\setminus A_2)+\lambda(\Omega')$, which implies $\lambda(\Omega'\setminus A_2)=0$. Because $\tilde{\mu}$ is absolutely continuous w.r.t. $\lambda$, we also have $\tilde{\mu}(\Omega'\setminus A_2)=0$. Thus: $0=\tilde{\mu}(\Omega'\setminus A_2)= \tilde{\mu}\bigl((\Omega'\setminus A_2)\cap A_1\bigr)= \tilde{\mu}(\Omega'\cap A_1)=\tilde{\mu}(\Omega')$.\\
      Since $\tilde{\mu}(\Omega')=0$ for all $\Omega'\in\mathcal{B}(\Omega)$, we have that $\tilde{\mu}\equiv0$ and hence $\mu_{\text{ac}}^1\equiv\mu_{\text{ac}}^2$ and $\mu_{\text{s}}^1\equiv\mu_{\text{s}}^2$.
\end{enumerate}
\end{proof}

\newpage
\chapter{Proof of Theorem \ref{Thm decomp sing}}\label{Appendix Proof decomp sing}
The inspiration for this proof comes from \cite{Koralov}, pp.~45--46, although the theorem presented here is more general than the one stated in \cite{Koralov}.
\begin{proof}
\begin{enumerate}
\item Due to the fact that $\mu_{\text{s}}$ is singular w.r.t. $\lambda$, we are provided two disjoint sets $A_1, A_2\in\mathcal{B}(\Omega)$ such that $\mu_{\text{s}}(\Omega')=\mu_{\text{s}}(\Omega'\cap A_1)$ and $\lambda(\Omega')=\lambda(\Omega'\cap A_2)$ for all $\Omega'\in\mathcal{B}(\Omega)$. Now define the sequence of sets $\{B_n\}_{n\in\mathbb{N}^+}$ given by
\begin{eqnarray*}
B_1 &:=& \{x\in A_1 \,\big|\, \mu_{\text{s}}(x)\geqslant 1\},\\
B_n &:=& \{x\in A_1 \,\big|\, \dfrac{1}{n}\leqslant\mu_{\text{s}}(x)<\dfrac{1}{n-1}\}, \hspace{1 cm}\text{for }n\in\{2,3,\ldots\}.
\end{eqnarray*}
    Each $B_n$ contains only finitely many elements, because $\mu_{\text{s}}$ is a finite measure. It follows that $B:=\bigcup_{n=1}^{\infty}B_n$ is a countable set. Note that $B=\{x\in A_1 \,\big|\, \mu_{\text{s}}(x)>0\}$. Also note that $x\notin B$ implies $\mu_{\text{s}}(x)=0$. Indeed, if $x\in A_1\setminus B$ this is trivial. Furthermore, if $x\notin A_1$ then $\mu_{\text{s}}(x)=\mu_{\text{s}}(\{x\}\cap A_1)=\mu_{\text{s}}(\emptyset)=0$.
\item Let $\mu_{\text{d}}:\mathbb{B}(\Omega)\rightarrow \mathbb{R}^+$ be defined by $\mu_{\text{d}}(\Omega')=\mu_{\text{s}}(\Omega'\cap B)$, and similarly $\mu_{\text{sc}}:\mathbb{B}(\Omega)\rightarrow \mathbb{R}^+$ by $\mu_{\text{sc}}(\Omega')=\mu_{\text{s}}\bigl(\Omega'\cap (\Omega\setminus B)\bigr)$. Note that it follows from this definition that $\mu_{\text{s}}=\mu_{\text{d}}+\mu_{\text{sc}}$.\\
    \\
    It is readily seen that $\mu_{\text{d}}$ is discrete w.r.t. $\lambda$, because $B$ is a countable set of points in $\Omega$, for which we have
    \begin{equation*}
    \mu_{\text{d}}(\Omega\setminus B)= \mu_{\text{d}}\bigl((\Omega\setminus B)\cap B\bigr)=\mu_{\text{d}}(\emptyset)=0,
    \end{equation*}
    and (since $B\subset A_1$)
    \begin{equation*}
    \lambda(B)= \lambda(B\cap A_2)\leqslant \lambda(A_1\cap A_2)=\lambda(\emptyset)=0.
    \end{equation*}
    We now show that $\mu_{\text{sc}}$ is singular continuous w.r.t. $\lambda$. Let $x\in\Omega$ be arbitrary. $\mu_{\text{sc}}(x)=\mu_{\text{s}}\bigl(\{x\}\cap (\Omega\setminus B)\bigr)$. Thus, if $x\notin \Omega\setminus B$ then $\mu_{\text{sc}}(x)=\mu_{\text{sc}}(\emptyset)=0$. On the other hand, if $x\in\Omega\setminus B$ then $\mu_{\text{sc}}(x)=\mu_{\text{s}}(x)=0$ because $x\notin B$ (the last step has been shown above). Thus, $\mu_{\text{sc}}(x)=0$ for all $x\in \Omega$.\\
    Consider the set $A_1$. First of all, $\lambda(A_1)=\lambda(A_1\cap A_2)=\lambda(\emptyset)=0$. Moreover, $\mu_{\text{sc}}(\Omega\setminus A_1) =\mu_{\text{s}}\bigl((\Omega\setminus A_1)\cap(\Omega\setminus B)\bigr)=\mu_{\text{s}}(\Omega\setminus A_1)$, because $B\subset A_1$. Since $\mu_{\text{s}}$ is singular, we have $\mu_{\text{s}}(\Omega\setminus A_1)= \mu_{\text{s}}\bigl((\Omega\setminus A_1)\cap A_1\bigr)=\mu_{\text{s}}(\emptyset)=0$.\\
    We now have that $\mu_{\text{sc}}$ is singular continuous w.r.t. $\lambda$.
\item To prove uniqueness of this decomposition, assume that we have $\mu_{\text{s}}=\mu_{\text{d}}^1+\mu_{\text{sc}}^1$ and $\mu_{\text{s}}=\mu_{\text{d}}^2+\mu_{\text{sc}}^2$, where $\mu_{\text{d}}^i$ is discrete w.r.t. $\lambda$ and $\mu_{\text{sc}}^i$ is singular continuous w.r.t. $\lambda$ for $i\in\{1,2\}$. For all $\Omega'\in\mathcal{B}(\Omega)$ we thus have
      \begin{equation*}
      (\mu_{\text{d}}^1-\mu_{\text{d}}^2)(\Omega')=(\mu_{\text{sc}}^2-\mu_{\text{sc}}^1)(\Omega')=:\tilde{\mu}(\Omega').
      \end{equation*}
      Assume that $\tilde{A}_1$ and $\tilde{A}_2$ are the countable sets, such that $\mu_{\text{d}}^i(\Omega\setminus\tilde{A}_i)=\lambda(\tilde{A}_i)=0$ for $i\in\{1,2\}$. Define $A:=\tilde{A}_1\cup\tilde{A}_2$, which is obviously countable. We have for each $i\in\{1,2\}$, and all $\Omega\in\mathcal{B}(\Omega)$
      \begin{equation*}
      \mu_{\text{d}}^i(\Omega'\cap A)\leqslant \mu_{\text{d}}^i(\Omega')=\mu_{\text{d}}^i(\Omega'\cap A)+\mu_{\text{d}}^i(\Omega'\setminus A)\leqslant \mu_{\text{d}}^i(\Omega'\cap A)+\mu_{\text{d}}^i(\Omega\setminus \tilde{A}_i)=\mu_{\text{d}}^i(\Omega'\cap A),
      \end{equation*}
      and thus $\mu_{\text{d}}^i(\Omega')= \mu_{\text{d}}^i(\Omega'\cap A)$. This implies $\tilde{\mu}(\Omega')=\tilde{\mu}(\Omega'\cap A)$. Note that $\Omega'\cap A$ is a subset of $A$ and is thus a countable collection of points in $\Omega$, say $\{y_1,y_2,\ldots\}$. We can also write $\Omega'\cap A=\cup_{k=1}^{\infty}\{y_k\}$, which is a disjoint union. We thus have (for each $i\in\{1,2\}$): $\mu_{\text{sc}}^i(\Omega'\cap A) = \sum_{k=1}^{\infty}\mu_{\text{sc}}^i(y_k)=0$, because each term of the sum is zero by definition of singular continuous measures. As a result
      \begin{equation*}
      \tilde{\mu}(\Omega')=\tilde{\mu}(\Omega'\cap A)=\mu_{\text{sc}}^2(\Omega'\cap A)-\mu_{\text{sc}}^1(\Omega'\cap A)=0,
      \end{equation*}
      for all $\Omega'\in\mathcal{B}(\Omega)$, by which we have uniqueness of the decomposition.
\end{enumerate}
\end{proof}

\newpage
\chapter{Proof of Lemma \ref{Properties RN derivatives}}\label{Appendix Proof properpties RN deriv}
\section{Proof of Part \ref{Property three measures} of Lemma \ref{Properties RN derivatives}}
This proof is based on \cite{Halmos}, p.~133.
\begin{proof}
By assumption $\nu$ is a positive measure. Note that the Radon-Nikodym Theorem provides the existence of $d\nu/d\mu$. For simplicity of notation, let us write: $d\nu/d\mu=h$. We now first prove that the positivity of $\nu$ implies that $h\geqslant0$ almost everywhere w.r.t. $\mu$.\\
\\
The proof goes by contradiction. Assume there exists an $\tilde{\Omega}\in\mathcal{B}(\Omega)$ satisfying $\mu(\tilde{\Omega})>0$ and $h<0$ $\mu$-almost everywhere on $\tilde{\Omega}$. We introduce the following sets:
\begin{eqnarray}
\nonumber F_0 &:=& \{x\in\Omega\,\big|\,h(x)<0\},\\
\nonumber F_n &:=& \{x\in\Omega\,\big|\,h(x)\leqslant -\dfrac{1}{n}\}, \hspace{1 cm}\text{for all }n\in\mathbb{N}^+.
\end{eqnarray}
Consider the disjoint union $\tilde{\Omega}=(\tilde{\Omega}\setminus F_0) \cup (\tilde{\Omega}\cap F_0)$. From the assumption that $h<0$ $\mu$-almost everywhere in $\tilde{\Omega}$, it follows that $\mu(\tilde{\Omega}\setminus F_0)=0$. Note that since $\tilde{\Omega}$ has strictly positive measure $\mu(\tilde{\Omega})$, we get
\begin{equation*}
0<\mu(\tilde{\Omega})=\mu(\tilde{\Omega}\setminus F_0)+\mu(\tilde{\Omega}\cap F_0)=\mu(\tilde{\Omega}\cap F_0).
\end{equation*}
Furthermore, we have for all $n\in \mathbb{N}^+$
\begin{equation*}
0\leqslant \nu(\tilde{\Omega}\cap F_n)=\int_{\tilde{\Omega}\cap F_n}h d\mu \leqslant \int_{\tilde{\Omega}\cap F_n}-\dfrac{1}{n} d\mu = -\dfrac{1}{n}\mu(\tilde{\Omega}\cap F_n)\leqslant 0,
\end{equation*}
which implies that $\mu(\tilde{\Omega}\cap F_n)=0$ for any $n\in\mathbb{N}^+$.\\
\\
Note that $F_0 = \bigcup_{n=1}^{\infty}F_n$ holds.\footnote{If $\tilde{x}\in \bigcup_{n=1}^{\infty}F_n$ then there exists a specific $\tilde{n}\in\mathbb{N}^+$ such that $\tilde{x}\in F_{\tilde{n}}$. The following holds: $h(\tilde{x})\leqslant-1/\tilde{n}<0$, so $\tilde{x}\in F_0$. If $\tilde{x}\in F_0$ then $h(\tilde{x})$ is strictly negative, so $h(\tilde{x})\leqslant -1/\tilde{n}<0$ for $\tilde{n}:=\bigl\lceil 1/h(\tilde{x}) \bigr\rceil$. Thus: $\tilde{x}\in F_{\tilde{n}}\subset \bigcup_{n=1}^{\infty}F_n$.} Since $\mu$ is a positive measure, this leads to the following inequality:
\begin{equation*}
\mu(\tilde{\Omega}\cap F_0)=\mu\Biggr(\tilde{\Omega}\cap\Bigl(\bigcup_{n=1}^{\infty}F_n\Bigr)\Biggr)\leqslant \sum_{n=1}^{\infty}\mu(\tilde{\Omega}\cap F_n)=0,
\end{equation*}
where the last step follows from the fact that $\mu(\tilde{\Omega}\cap F_n)=0$ for any $n\in\mathbb{N}^+$. We now have a contradiction with the statement $\mu(\tilde{\Omega}\cap F_0)>0$. Thus, there is no such set $\tilde{\Omega}$ of positive measure, on which $h<0$ almost everywhere w.r.t. $\mu$. This implies that $h\geqslant0$ $\mu$-almost everywhere in $\Omega$.\\
\\
Now let $h$ satisfy $h(x)\geqslant0$ for all $x\in\Omega$ (note that it is possible to choose such a representant). Every nonnegative measurable function is the (pointwise) limit of an increasing sequence of nonnegative simple functions; assume that $\{h_k\}_{k=1}^{\infty}$ is such sequence. Simple functions are always measurable, and therefore the following is true:
\begin{equation*}
\lim_{k\rightarrow\infty}\int_{\Omega}h_kd\mu=\int_{\Omega}hd\mu,
\end{equation*}
see \cite{Halmos}, p.~112.\\
This limit is also valid if the domain of integration is any arbitrary $\Omega'\in\mathcal{B}(\Omega)$, as we will show now. Let $\mathbf{1}_{\Omega'}$ denote the characteristic function of $\Omega'$. The sequence $\{h_k\mathbf{1}_{\Omega'}\}$ is increasing and consists of nonnegative, measurable functions (products of two measurable functions, are again measurable). Furthermore $\mathbf{1}_{\Omega'}h$ is (by definition of $\{h_k\}$) the pointwise limit of the sequence $\{h_k\mathbf{1}_{\Omega'}\}$. Consequently (cf. \cite{Halmos}, p.~112): $\lim_{k\rightarrow\infty}\int_{\Omega}h_k\mathbf{1}_{\Omega'}d\mu=\int_{\Omega}h\mathbf{1}_{\Omega'}d\mu$, which can also be written as
\begin{equation*}
\lim_{k\rightarrow\infty}\int_{\Omega'}h_kd\mu=\int_{\Omega'}hd\mu, \hspace{1 cm}\text{for all }\Omega'\in\mathcal{B}(\Omega).
\end{equation*}
For convenience, write $d\mu/d\lambda=p$ (the existence of which is guaranteed by the Radon-Nikodym Theorem). Following the same arguments as for $h$, we can conclude that it is possible to take a representative of $p$, such that $p(x)\geqslant0$ for all $x\in\Omega$. The Radon-Nikodym Theorem furthermore ensures the measurability of $p$. The sequence $\{h_kp\}_{k=1}^{\infty}$ consists thus of measurable functions, as each element is the product of two measurable functions. Both $p$ and $h_k$ (for all $k\in\mathbb{N}^+$) are nonnegative, so the same holds for the elements of the sequence. The sequence is increasing, because $\{h_k\}$ is increasing. As a result, we can conclude that $\lim_{k\rightarrow\infty}\int_{\Omega}h_kpd\lambda=\int_{\Omega}hpd\lambda$ (cf. \cite{Halmos}, p.~112). Multiplication with $\mathbf{1}_{\Omega'}$ and some further reasoning yields, like before, that
\begin{equation*}
\lim_{k\rightarrow\infty}\int_{\Omega'}h_kpd\lambda=\int_{\Omega'}hpd\lambda, \hspace{1 cm}\text{for all }\Omega'\in\mathcal{B}(\Omega).
\end{equation*}
The functions $h_k$ are simple functions, so for each $k\in\mathbb{N}^+$ there exists a finite collection $\{A_i^{(k)}\}_{i=1}^{N^{(k)}}$ of mutually disjoint measurable subsets of $\Omega$, where $N^{(k)}\in\mathbb{N}^+$, such that
\begin{equation*}
h_k(x)=\sum_{i=1}^{N^{(k)}}\alpha_i^{(k)}\mathbf{1}_{A_i^{(k)}}(x), \hspace{1 cm}\text{for all }x\in\Omega,
\end{equation*}
for some collection of coefficients $\{\alpha_i^{(k)}\}_{i=1}^{N^{(k)}}\subset \mathbb{R}^+$. For any measurable set $A$, the following identity holds:
\begin{equation*}
\int_{\Omega'}\mathbf{1}_Ad\mu=\mu(\Omega'\cap A)=\int_{\Omega'\cap A}pd\lambda=\int_{\Omega'}\mathbf{1}_Apd\lambda, \hspace{1 cm}\text{for all }\Omega'\in\mathcal{B}(\Omega).
\end{equation*}
Therefore
\begin{eqnarray}
\nonumber \int_{\Omega'}h_kd\mu &=& \int_{\Omega'}\sum_{i=1}^{N^{(k)}}\alpha_i^{(k)}\mathbf{1}_{A_i^{(k)}}d\mu\\
\nonumber &=& \sum_{i=1}^{N^{(k)}}\alpha_i^{(k)}\int_{\Omega'}\mathbf{1}_{A_i^{(k)}}d\mu\\
\nonumber &=& \sum_{i=1}^{N^{(k)}}\alpha_i^{(k)}\int_{\Omega'}\mathbf{1}_{A_i^{(k)}}pd\lambda\\
\nonumber &=& \int_{\Omega'}\sum_{i=1}^{N^{(k)}}\alpha_i^{(k)}\mathbf{1}_{A_i^{(k)}}pd\lambda\\
\nonumber &=& \int_{\Omega'}h_kpd\lambda, \hspace{1 cm}\text{for all }\Omega'\in\mathcal{B}(\Omega)\text{ and for all }k\in\mathbb{N}^+.
\end{eqnarray}
$\int_{\Omega'}hd\mu$ and $\int_{\Omega'}hpd\lambda$ are thus limit values of the same sequence, which means they must be equal: $\nu(\Omega')=\int_{\Omega'}hd\mu=\int_{\Omega'}hpd\lambda$. From this, it follows that the integrand $hp$ is in fact the Radon-Nikodym derivative $d\nu/d\lambda$, which finishes the proof.
\end{proof}

\section{Proof of Part \ref{Property integration of g} of Lemma \ref{Properties RN derivatives}}
\begin{proof}
Define the positive part $g^+$ and the negative part $g^-$ of the function $g$ by
\begin{eqnarray}
\nonumber g^+(x)&:=& \max\{g(x),0\},\\
\nonumber g^-(x)&:=& -\min\{g(x),0\}.
\end{eqnarray}
These two functions are measurable, because $g$ is measurable (see \cite{Rudin}, p.~15).\\
\\
Assume for the moment that neither of $g^+$ and $g^-$ is identically zero.\\
\\
If we define for all $\Omega'\in\mathcal{B}(\Omega)$ the following:
\begin{equation*}
\gamma^{+/-}(\Omega'):= \int_{\Omega'}g^{+/-}d\mu,
\end{equation*}
then $\gamma^+$ and $\gamma^-$ are measures, e.g. due to \cite{Rudin} (p.~23). Note that by the superscript `$+/-$', we indicate that the statement holds in both the `$+$' case and the `$-$' case. Both $g^+$ and $g^-$ are nonnegative, so $\gamma^+$ and $\gamma^-$ are positive measures. $g^{+/-}(\Omega)>0$ follows from the assumption that $g^+$ and $g^-$ are not identically zero. Since $g\in L_{\mu}^1(\Omega)$, it follows that $\gamma^{+/-}(\Omega)<\infty$. The Radon-Nikodym Theorem can now be applied to the measures $\gamma^+$ and $\gamma^-$, and provides the existence of $d\gamma^{+/-}/d\mu$. The uniqueness of Radon-Nikodym derivatives implies that $d\gamma^{+/-}/d\mu=g^{+/-}$.\\
Now apply Part \ref{Property three measures} of Lemma \ref{Properties RN derivatives}, by setting $\nu=\gamma^+$ and $\nu=\gamma^-$ subsequently. Note that $\gamma^{+/-}\ll \mu$ is guaranteed by defining $\gamma^{+/-}$ as an integral with respect to $\mu$ (cf. \cite{Halmos}, p.~104). The result is
\begin{equation*}
\int_{\Omega'}g^{+/-}d\mu = \gamma^{+/-}(\Omega')= \int_{\Omega'}\dfrac{d\gamma^{+/-}}{d\lambda}d\lambda=
\int_{\Omega'}\dfrac{d\gamma^{+/-}}{d\mu}\dfrac{d\mu}{d\lambda}d\lambda = \int_{\Omega'}g^{+/-}\dfrac{d\mu}{d\lambda}d\lambda,
\end{equation*}
for all $\Omega'\in\mathcal{B}(\Omega)$, where Part \ref{Property three measures} has been used in the third equality.\\
\\
Note that
\begin{equation*}
\int_{\Omega'}g^{+/-}d\mu = \int_{\Omega'}g^{+/-}\dfrac{d\mu}{d\lambda}d\lambda
\end{equation*}
is automatically satisfied if $g^{+/-}\equiv0$. We can thus proceed without distinguishing between $g^{+/-}\equiv0$ and $g^{+/-}\not\equiv0$.\\
\\
Finally, we conclude that for all $\Omega'\in\mathcal{B}(\Omega)$
\begin{eqnarray}
\nonumber \int_{\Omega'}gd\mu &=& \int_{\Omega'}(g^+-g^-)d\mu\\
\nonumber &=& \int_{\Omega'}g^+d\mu-\int_{\Omega'}g^-d\mu\\
\nonumber &=& \int_{\Omega'}g^+\dfrac{d\mu}{d\lambda}d\lambda-\int_{\Omega'}g^-\dfrac{d\mu}{d\lambda}d\lambda\\
\nonumber &=& \int_{\Omega'}(g^+-g^-)\dfrac{d\mu}{d\lambda}d\lambda\\
\nonumber &=& \int_{\Omega'}g\dfrac{d\mu}{d\lambda}d\lambda,
\end{eqnarray}
by which we have the desired statement.
\end{proof}

\section{Proof of Part \ref{Property inverse} of Lemma \ref{Properties RN derivatives}}
\begin{proof}
Set $\nu=\lambda$ and apply Part \ref{Property three measures}, which states
\begin{equation*}
\dfrac{d\lambda}{d\lambda}=\dfrac{d\lambda}{d\mu}\dfrac{d\mu}{d\lambda}.
\end{equation*}
We now claim that $d\lambda/d\lambda\equiv1$, which can easily be shown. Obviously $\lambda\ll\lambda$, thus the Radon-Nikodym Theorem yields the following:
\begin{equation*}
\lambda(\Omega')=\int_{\Omega'}\dfrac{d\lambda}{d\lambda}d\lambda, \hspace{1 cm}\text{for all }\Omega'\in\mathcal{B}(\Omega).
\end{equation*}
However, also $\lambda(\Omega')=\int_{\Omega}\mathbf{1}_{\Omega'}d\lambda=\int_{\Omega'}1d\lambda$ holds. The uniqueness of Radon-Nikodym derivatives implies that $d\lambda/d\lambda\equiv1$. We thus have
\begin{equation*}
1=\dfrac{d\lambda}{d\mu}\dfrac{d\mu}{d\lambda},
\end{equation*}
and hence
\begin{equation*}
\dfrac{d\mu}{d\lambda}=\Bigl(\dfrac{d\lambda}{d\mu}\Bigr)^{-1}.
\end{equation*}
\end{proof}

\section{Proof of Part \ref{Property Cartesian product} of Lemma \ref{Properties RN derivatives}}
\begin{proof}

Let $\Omega'\in\mathcal{B}(\Omega_1\times\Omega_2)$ be arbitrary. Define the following set:
\begin{equation*}
\Omega'_1 := \{x_1\in\Omega_1 \,\big|\, (x_1,x_2)\in\Omega' \text{ for some } x_2\in\Omega_2\}.
\end{equation*}
Also define
\begin{equation*}
\Omega'_{x_1} := \{x_2\in\Omega_2 \,\big|\, (x_1,x_2)\in\Omega'\},
\end{equation*}
for any choice of $x_1\in\Omega_1$. The set $\Omega'_{x_1}$ is called the $\Omega_1$-\textit{section} of $\Omega'$ determined by $x_1$ (cf. \cite{Halmos}, p.~141). Note that $\Omega'_{x_1}$ is a subset of $\Omega_2$. In fact, $x_1$ can be regarded as a parameter here. It can be proved that every section of a measurable set is a measurable set (see \cite{Halmos}, p.~141).\\
\\
We first prove that $\nu_1\otimes\nu_2\ll \mu_1\otimes\mu_2$.\\
Assume that $\Omega'\in\mathcal{B}(\Omega_1\times\Omega_2)$ is such that $(\mu_1\otimes\mu_2)(\Omega')=0$. By definition (see \cite{Halmos}, p.~144) of the product measure $\mu_1\otimes\mu_2$ we have
\begin{equation*}
(\mu_1\otimes\mu_2)(\Omega') = \int_{\Omega'_1}\mu_2(\Omega'_{x_1})d\mu_1(x_1) = 0.
\end{equation*}
As a consequence (cf. \cite{Halmos}, p.~147)
\begin{equation*}
\mu_2(\Omega'_{x_1}) = 0, \hspace{1 cm}\text{for }\mu_1\text{-a.e. }x_1\in\Omega'_1.
\end{equation*}
Define
\begin{equation*}
\Omega_{\mu_2}^0:=\{x_1\in\Omega'_1\,\big|\,\mu_2(\Omega'_{x_1})\neq0\}.
\end{equation*}
Let $\tilde{x}_1$ be such that $\tilde{x}_1\notin \Omega_{\mu_2}^0$, then of course $\mu_2(\Omega'_{\tilde{x}_1})=0$. From the hypothesis that $\nu_2\ll\mu_2$, it follows that $\nu_2(\Omega'_{\tilde{x}_1})=0$. If we define
\begin{equation*}
\Omega_{\nu_2}^0:=\{x_1\in\Omega'_1\,\big|\,\nu_2(\Omega'_{x_1})\neq0\},
\end{equation*}
then we have thus proved that $\Omega_{\nu_2}^0 \subset \Omega_{\mu_2}^0$.\\
Since $\mu_2(\Omega'_{x_1}) = 0$ for $\mu_1$-a.e. $x_1\in\Omega'_1$, we clearly have that $\mu_1(\Omega_{\mu_2}^0)=0$. By hypothesis of the lemma $\nu_1$ is absolutely continuous w.r.t. $\mu_1$, and thus $\mu_1(\Omega_{\mu_2}^0) = 0$ implies $\nu_1(\Omega_{\mu_2}^0) = 0$. We now use $\Omega_{\nu_2}^0 \subset \Omega_{\mu_2}^0$, to conclude
\begin{equation*}
\nu_1(\Omega_{\nu_2}^0)\leqslant\nu_1(\Omega_{\mu_2}^0) = 0,
\end{equation*}
by which we find that $\nu_1(\Omega_{\nu_2}^0)= 0$. This statement can also be written as
\begin{equation*}
\nu_2(\Omega'_{x_1}) = 0, \hspace{1 cm}\text{for }\nu_1\text{-a.e. }x_1\in\Omega'_1.
\end{equation*}
It follows trivially that
\begin{equation*}
(\nu_1\otimes\nu_2)(\Omega') = \int_{\Omega'_1}\nu_2(\Omega'_{x_1})d\nu_1(x_1) = 0.
\end{equation*}
We have thus shown that $\nu_1\otimes\nu_2\ll \mu_1\otimes\mu_2$.\\
\\
We can now apply the Radon-Nikodym Theorem, which provides the existence of the unique Radon-Nikodym derivative
\begin{equation*}
\dfrac{d(\nu_1\otimes\nu_2)}{d(\mu_1\otimes\mu_2)}.
\end{equation*}
Let $\Omega'\in\mathcal{B}(\Omega_1\times\Omega_2)$ be arbitrary, and let $\mathbf{1}_{A}$ denote the characteristic function of the set $A$, where $A$ is a set in $\mathcal{B}(\Omega_1)$, $\mathcal{B}(\Omega_2)$ or $\mathcal{B}(\Omega_1\times\Omega_2)$. By hypothesis of the lemma the Radon-Nikodym derivatives $d\nu_1/d\mu_1$ and $d\nu_2/d\mu_2$ exist. We can now perform the following calculations:
\begin{eqnarray*}
(\nu_1\otimes\nu_2)(\Omega') &=& \int_{\Omega'_1}\nu_2(\Omega'_{x_1})d\nu_1(x_1)\\
&=& \int_{\Omega'_1}\Biggl(\int_{\Omega'_{x_1}}d\nu_2(x_2)\Biggr)d\nu_1(x_1)\\
&=& \int_{\Omega'_1}\Biggl(\int_{\Omega'_{x_1}}\dfrac{d\nu_2}{d\mu_2}d\mu_2(x_2)\Biggr)\dfrac{d\nu_1}{d\mu_1}d\mu_1(x_1)\\
&=& \int_{\Omega'_1}\int_{\Omega'_{x_1}}\Biggl(\dfrac{d\nu_1}{d\mu_1}\dfrac{d\nu_2}{d\mu_2}\Biggr)d\mu_2(x_2)d\mu_1(x_1)\\
&=& \int_{\Omega_1}\int_{\Omega_2}\Biggl(\dfrac{d\nu_1}{d\mu_1}\dfrac{d\nu_2}{d\mu_2}\mathbf{1}_{\Omega'_1}(x_1)\,\mathbf{1}_{\Omega'_{x_1}}(x_2)\Biggr)d\mu_2(x_2)d\mu_1(x_1)\\
&=& \int_{\Omega_1\times\Omega_2}\Biggl(\dfrac{d\nu_1}{d\mu_1}\dfrac{d\nu_2}{d\mu_2}\mathbf{1}_{\Omega'_1}(x_1)\,\mathbf{1}_{\Omega'_{x_1}}(x_2)\Biggr)d(\mu_1\otimes\mu_2)(x_1,x_2),
\end{eqnarray*}
where the last step follows from \cite{Halmos}, p.~147.\\
\\
Let $(\tilde{x}_1,\tilde{x}_2)$ be an element of $\Omega_1\times\Omega_2$. By definition of $\Omega'_{\tilde{x_1}}$, we have that $\mathbf{1}_{\Omega'_{\tilde{x}_1}}(\tilde{x}_2)=1$ if and only if $(\tilde{x}_1,\tilde{x}_2)\in\Omega'$.\\
If $(\tilde{x}_1,\tilde{x}_2)\in\Omega'$, then also $\mathbf{1}_{\Omega'_1}(\tilde{x}_1)=1$, by definition of $\Omega'_1$. We thus have
\begin{equation*}
\mathbf{1}_{\Omega'_1}(\tilde{x}_1)\,\mathbf{1}_{\Omega'_{\tilde{x}_1}}(\tilde{x}_2)=1.
\end{equation*}
On the other hand, if $(\tilde{x}_1,\tilde{x}_2)\notin\Omega'$ then $\mathbf{1}_{\Omega'_{\tilde{x}_1}}(\tilde{x}_2)=0$, and thus
\begin{equation*}
\mathbf{1}_{\Omega'_1}(\tilde{x}_1)\,\mathbf{1}_{\Omega'_{\tilde{x}_1}}(\tilde{x}_2)=0.
\end{equation*}
We conclude that for all $(x_1,x_2)\in\Omega_1\times\Omega_2$ the following identity holds:
\begin{equation*}
\mathbf{1}_{\Omega'_1}(x_1)\,\mathbf{1}_{\Omega'_{x_1}}(x_2)=\mathbf{1}_{\Omega'}(x_1,x_2).
\end{equation*}
This means that we can write
\begin{eqnarray*}
(\nu_1\otimes\nu_2)(\Omega') &=& \int_{\Omega_1\times\Omega_2}\Biggl(\dfrac{d\nu_1}{d\mu_1}\dfrac{d\nu_2}{d\mu_2}\mathbf{1}_{\Omega'}(x_1,x_2)\Biggr)d(\mu_1\otimes\mu_2)(x_1,x_2)\\
&=& \int_{\Omega'}\Biggl(\dfrac{d\nu_1}{d\mu_1}\dfrac{d\nu_2}{d\mu_2}\Biggr)d(\mu_1\otimes\mu_2).
\end{eqnarray*}
Uniqueness of the Radon-Nikodym derivative now makes sure that
\begin{equation*}
\dfrac{d(\nu_1\otimes\nu_2)}{d(\mu_1\otimes\mu_2)}=\dfrac{d\nu_1}{d\mu_1}\dfrac{d\nu_2}{d\mu_2},
\end{equation*}
which finishes the proof.
\end{proof}

\newpage
\chapter{Derivation of the entropy density for an ideal gas}\label{Appendix entropy ideal gas}
The following derivation is due to \cite{FonsvdVen}. More details can be found in \cite{Bowen1989} (pp.~16--19 and 96--102).\\
\\
We consider a gas that is compressible, non-viscous and thermal (i.e. the temperature influences its behaviour). The set of independent variables is $Q:=\{\rho, T, j\}$, where $\rho$ is the density, $T$ the absolute temperature and $j=\nabla T$, the temperature gradient. All other dependent variables are functions of $Q$. We assume the following conservation laws:
\begin{enumerate}
  \item Balance of mass:
        \begin{equation}\label{balance of mass appendix}
        \dfrac{\partial \rho}{\partial t}+\nabla\cdot(\rho v)=0.
        \end{equation}
  \item Balance of energy:
        \begin{equation}\label{balance of energy appendix}
        \rho\dfrac{D\varepsilon}{Dt}=-p\nabla\cdot v+\nabla\cdot q + \rho r.
        \end{equation}
  \item Entropy inequality (Clausius-Duhem):
        \begin{equation}\label{entropy inequality appendix}
        \rho\dfrac{D\eta}{Dt}-\nabla\cdot \Bigl(\dfrac{q}{T}\Bigr)-\dfrac{\rho r}{T}\geqslant 0.
        \end{equation}
\end{enumerate}
In the balance of energy, $\varepsilon$ is the internal energy density. Furthermore, $q$ denotes the heat flux and $r$ is the heat supply density. The entropy inequality is the local version of the Clausius-Duhem Inequality. It is precisely the inequality given by Postulate \ref{postulate local entropy inequality}, for one component. Moreover $j_{\eta}=q/T$ has been taken, as was also suggested by (1.6.9) in \cite{Bowen} (the results on p.~29 have to be reduced to one component to see this).\\
Note that we require $T>0$ for the entropy inequality to make sense.\\
\\
Since (\ref{balance of mass appendix}) can also be written as
\begin{eqnarray}
\nonumber 0 &=& \dfrac{\partial \rho}{\partial t}+\rho\nabla\cdot v + \nabla\rho\cdot v\\
&=& \dfrac{D \rho}{Dt}+\rho\nabla\cdot v,
\end{eqnarray}
it follows that $\nabla\cdot v = -\dfrac{1}{\rho}\dfrac{D\rho}{Dt}$.\\
\\
Note that
\begin{eqnarray}
\nonumber \rho \dfrac{D\varepsilon}{Dt}&=&\rho\Bigl(\dfrac{\partial \varepsilon}{\partial t}+\nabla\varepsilon\cdot v \Bigr)\\
\nonumber &=& \rho \dfrac{\partial \varepsilon}{\partial t}+\rho \nabla\varepsilon\cdot v + \varepsilon\underbrace{\Bigl( \dfrac{\partial \rho}{\partial t}+\nabla\cdot(\rho v)\Bigr)}_{=0}\\
\nonumber &=& \rho \dfrac{\partial \varepsilon}{\partial t}+ \varepsilon\dfrac{\partial \rho}{\partial t} + \nabla\varepsilon\cdot(\rho v) + \varepsilon\nabla\cdot(\rho v)\\
\nonumber &=& \dfrac{\partial}{\partial t}(\rho\varepsilon) +\nabla\cdot(\rho v\varepsilon),
\end{eqnarray}
and a similar identity holds if we take $\eta$ instead of $\varepsilon$.\\
\\
Let $F$ denote the \textit{Helmholtz free energy}, defined by $F:=\varepsilon -\eta T$. Taking the derivative $D/Dt$ in the definition of $F$ and multiplicating afterwards by $\rho$, leads to
\begin{equation}\label{rho times derivative Helmholtz free energy appendix}
\rho \dfrac{D\varepsilon}{Dt}=\rho \dfrac{DF}{Dt} + \rho \eta \dfrac{DT}{Dt}+\rho T \dfrac{D\eta}{Dt}=-p\nabla\cdot v+\nabla\cdot q + \rho r.
\end{equation}
The latter equality is due to the balance of energy (\ref{balance of energy appendix}). Equation (\ref{rho times derivative Helmholtz free energy appendix}) can also be written as
\begin{equation}\label{energy balance wrt free energy appendix}
\rho T \dfrac{D\eta}{Dt}=-\rho \dfrac{DF}{Dt} - \rho \eta \dfrac{DT}{Dt}-p\nabla\cdot v+\nabla\cdot q + \rho r.
\end{equation}
A suitable combination of the entropy inequality (\ref{entropy inequality appendix}) and (\ref{energy balance wrt free energy appendix}), $\nabla\cdot v = -\dfrac{1}{\rho}\dfrac{D\rho}{Dt}$ and of the fact that $T>0$, yields
\begin{eqnarray}
\nonumber -\rho \dfrac{DF}{Dt} - \rho \eta \dfrac{DT}{Dt}+\dfrac{p}{\rho}\dfrac{D\rho}{Dt}&=& \rho T \dfrac{D\eta}{Dt}-\nabla\cdot q - \rho r\\
\nonumber &=& \rho T \dfrac{D\eta}{Dt}-\nabla\cdot \Bigl(\dfrac{qT}{T}\Bigr) - \rho r\\
\nonumber &=& \rho T \dfrac{D\eta}{Dt}-T\nabla\cdot \Bigl(\dfrac{q}{T}\Bigr)- \dfrac{q}{T}\cdot\nabla T - \rho r\\
\nonumber &=& T\underbrace{\Bigl(\rho \dfrac{D\eta}{Dt}-\nabla\cdot \Bigl(\dfrac{q}{T}\Bigr)- \dfrac{\rho r}{T}\Bigr)}_{\geqslant0}- \dfrac{q}{T}\cdot\nabla T\\
&\geqslant& - \dfrac{q}{T}\cdot\nabla T,
\end{eqnarray}
or
\begin{equation}\label{entropy energy ineq wrt free energy appendix}
-\rho \dfrac{DF}{Dt} - \rho \eta \dfrac{DT}{Dt}+\dfrac{p}{\rho}\dfrac{D\rho}{Dt}+ \dfrac{1}{T}q\cdot\nabla T\geqslant 0.
\end{equation}
Since $F$ is a depends only on $Q=\{\rho, T, j\}$, we have
\begin{equation*}
\dfrac{DF}{Dt}=\dfrac{\partial F}{\partial \rho}\dfrac{D\rho}{Dt}+ \dfrac{\partial F}{\partial T}\dfrac{DT}{Dt}+ \dfrac{\partial F}{\partial j}\dfrac{Dj}{Dt},
\end{equation*}
by which (\ref{entropy energy ineq wrt free energy appendix}) transforms into
\begin{equation*}
\Bigl(\dfrac{p}{\rho}-\rho\dfrac{\partial F}{\partial \rho} \Bigr)\dfrac{D\rho}{Dt}-\rho\Bigl(\eta+\dfrac{\partial F}{\partial T} \Bigr)\dfrac{DT}{Dt}-\rho\dfrac{\partial F}{\partial j}\dfrac{Dj}{Dt}+\dfrac{1}{T}q\cdot\nabla T\geqslant 0.
\end{equation*}
The coefficients of $D\rho/Dt$, $DT/Dt$ and $Dj/Dt$, and $\dfrac{1}{T}q\cdot\nabla T$ are independent of $D\rho/Dt$, $DT/Dt$ and $Dj/Dt$. Furthermore, $D\rho/Dt$, $DT/Dt$ and $Dj/Dt$ can be chosen arbitrarily (pointwise), and thus the following \textit{constitutive equations} must hold
\begin{equation}\label{constitutive equations appendix}
p=\rho^2\dfrac{\partial F}{\partial \rho},\hspace{1 cm} \eta=-\dfrac{\partial F}{\partial T},\hspace{1 cm}\text{and}\hspace{1 cm} \dfrac{\partial F}{\partial j}=0,
\end{equation}
and since $T>0$ also
\begin{equation}\label{q in nabla T positive}
q\cdot \nabla T\geqslant 0.
\end{equation}
\begin{remark}
It follows eventually from (\ref{q in nabla T positive}) that $q=\lambda(\rho,T,j)\nabla T$, where $\lambda\geqslant0$. If \textit{linear thermal conduction} is assumed (that is, $q$ is linear w.r.t. $j$), then $\lambda=\lambda(\rho, T)$ and thus Fourier's Law is obtained. However, this result is not so important in the sequel.
\end{remark}
If we want a more explicit form of the constitutive equations in (\ref{constitutive equations appendix}), we are required to give also an explicit choice of $F$. It is clear from the third constitutive equation in (\ref{constitutive equations appendix}) that
\begin{equation*}
F=F(\rho, T).
\end{equation*}
\begin{definition}[Ideal gas]\label{definition ideal gas}
A gas is called \textit{ideal gas}, if
\begin{equation*}
F(\rho, T):=c(T-T_0)-cT\log \Bigl(\dfrac{T}{T_0} \Bigr)+RT\log\Bigl(\dfrac{\rho}{\rho_0}\Bigr),
\end{equation*}
where $T_0$ and $\rho_0$ are an arbitrary reference temperature and density, $c$ is the specific heat, and $R$ is the universal gas constant.
\end{definition}
In fact, $c=c_V$ is the specific heat at constant volume, which can only be a function of $T$ (cf. \cite{Zemansky}, p.~114). We have made the assumption that $c_V$ is constant. In nature, this assumption is only valid over wide temperature ranges for monoatomic gases; see \cite{Zemansky}, pp.~114--115. For the sake of clarity, we will not go into details for non-constant specific heat $c_V=c_V(T)$.\\
\\
For an ideal gas, it follows from (\ref{constitutive equations appendix}) that
\begin{equation*}
p=R\rho T,\hspace{1 cm} \text{and}\hspace{1 cm} \eta=c\log\Bigl(\dfrac{T}{T_0} \Bigr)-R\log\Bigl(\dfrac{\rho}{\rho_0}\Bigr).
\end{equation*}
In $p=R\rho T$ we recognize the ideal gas law. Furthermore, note that $\varepsilon$ can be determined explicitly using $\varepsilon=F+\eta T$. The internal energy is a function of the temperature only: $\varepsilon=\varepsilon(T)=c(T-T_0)$.

\newpage
\chapter{Modifications in the proofs of Theorem \ref{Thm existence time-discrete} and Corollary \ref{corollary conservation of mass discretized}}\label{Appendix modification proofs}
We describe the details that need to be adapted in the proofs of Theorem \ref{Thm existence time-discrete} (global existence of the time-discrete solution) and Corollary \ref{corollary conservation of mass discretized} (conservation of mass) to make them compatible with the new assumptions made in Section \ref{section relax conditions on motion mapping}.

\section{The proof of Theorem \ref{Thm existence time-discrete}}
In Part \ref{Thm existence time-discrete Part Abs Cont} of the proof of Theorem \ref{Thm existence time-discrete}, the identity $(\chi^{\alpha}_n)^{-1}(\Omega)=\Omega$ is used. This identity holds if the motion mapping is invertible. Using the pre-image, we only have that
\begin{equation*}
(\chi^{\alpha}_n)^{-1}(\Omega)\subset\Omega.
\end{equation*}
We modify the proof that the absolutely continuous part $\mu^{\alpha}_{\text{ac},n+1}$ is finite accordingly, by inserting an inequality:
\begin{equation*}
\mu^{\alpha}_{\text{ac},n+1}(\Omega)=\mu^{\alpha}_{\text{ac},n}\bigl( (\chi^{\alpha}_n)^{-1}(\Omega) \bigr)\leqslant\mu^{\alpha}_{\text{ac},n}(\Omega)<\infty.
\end{equation*}
Consider the arguments to prove that $\mu^{\alpha}_{\text{d},n+1}$ is finite in Part \ref{Thm existence time-discrete Part Discrete} of the proof of Theorem \ref{Thm existence time-discrete}. There, we used that $\{x_i\}_{i\in\mathcal{J}}\subset \Omega$ implies $\{\chi^{\alpha}_n(x_{i})\}_{i\in\mathcal{J}}\subset \Omega$ due to the fact that the motion mapping is a homeomorphism mapping from $\Omega$ to $\Omega$. Note that $\{x_i\}_{i\in\mathcal{J}}\subset \supp \mu_n^{\alpha}$ holds. In Section \ref{section relax conditions on motion mapping} we require that $\{\chi^{\alpha}_n\big|_{\supp\mu_n^{\alpha}}$ maps from $\supp \mu_n^{\alpha}$ to $\Omega$. Thus $\{\chi^{\alpha}_n(x_{i})\}_{i\in\mathcal{J}}\subset \Omega$ is still true.\\
\\
Regarding Part \ref{Thm existence time-discrete Part Singular Continuous} of the proof, note that
\begin{equation*}
\Omega\setminus B_n= \bigl(\chi^{\alpha}_n\bigr)^{-1}\Bigl(\chi^{\alpha}_n\bigl(\Omega\setminus B_n\bigr)\Bigr)
\end{equation*}
no longer holds if pre-images are used. Instead we have that
\begin{equation*}
\Omega\setminus B_n\subset \bigl(\chi^{\alpha}_n\bigr)^{-1}\Bigl(\chi^{\alpha}_n\bigl(\Omega\setminus B_n\bigr)\Bigr).
\end{equation*}
We need to prove that if
\begin{equation*}
\mu^{\alpha}_{\text{sc},n}\bigl(\Omega\setminus B_n\bigr) = 0
\end{equation*}
holds, then
\begin{equation*}
\mu^{\alpha}_{\text{sc},n+1}\Bigl(\Omega\setminus \chi^{\alpha}_n\bigl(B_n\bigr)\Bigr) = 0.
\end{equation*}
Since
\begin{equation*}
B_n\subset (\chi^{\alpha}_n)^{-1}\bigl(\chi_n^{\alpha}(B_n)\bigr),
\end{equation*}
we have that
\begin{equation*}
\Omega\setminus(\chi^{\alpha}_n)^{-1}\bigl(\chi_n^{\alpha}(B_n)\bigr)\subset\Omega\setminus B_n.
\end{equation*}
It follows that
\begin{equation*}
0=\mu^{\alpha}_{\text{sc},n}(\Omega\setminus B_n)\geqslant \mu^{\alpha}_{\text{sc},n}\Bigl(\Omega\setminus(\chi^{\alpha}_n)^{-1}\bigl(\chi_n^{\alpha}(B_n)\bigr)\Bigr),
\end{equation*}
thus
\begin{equation*}
\mu^{\alpha}_{\text{sc},n}\Bigl(\Omega\setminus(\chi^{\alpha}_n)^{-1}\bigl(\chi_n^{\alpha}(B_n)\bigr)\Bigr)=0.
\end{equation*}
Now use that
\begin{equation*}
(\chi^{\alpha}_n)^{-1}(\Omega)\subset\Omega,
\end{equation*}
and
\begin{equation*}
(\chi^{\alpha}_n)^{-1}(\Omega)\setminus (\chi^{\alpha}_n)^{-1}\bigl(\chi_n^{\alpha}(B_n)\bigr) =(\chi^{\alpha}_n)^{-1}\Bigl(\Omega\setminus \chi_n^{\alpha}(B_n)\Bigr),
\end{equation*}
to obtain
\begin{eqnarray*}
\mu^{\alpha}_{\text{sc},n}\Bigl(\Omega\setminus(\chi^{\alpha}_n)^{-1}\bigl(\chi_n^{\alpha}(B_n)\bigr)\Bigr)&\geqslant& \mu^{\alpha}_{\text{sc},n}\Bigl((\chi^{\alpha}_n)^{-1}(\Omega)\setminus(\chi^{\alpha}_n)^{-1}\bigl(\chi_n^{\alpha}(B_n)\bigr)\Bigr)\\
&=& \mu^{\alpha}_{\text{sc},n}\Bigl((\chi^{\alpha}_n)^{-1}\Bigl(\Omega\setminus \chi_n^{\alpha}(B_n)\Bigr) \Bigr)\\
&=& \mu^{\alpha}_{\text{sc},n+1}\Bigl(\Omega\setminus \chi^{\alpha}_n\bigl(B_n\bigr)\Bigr).
\end{eqnarray*}
Because
\begin{equation*}
\mu^{\alpha}_{\text{sc},n}\Bigl(\Omega\setminus(\chi^{\alpha}_n)^{-1}\bigl(\chi_n^{\alpha}(B_n)\bigr)\Bigr)=0.
\end{equation*}
we also have that
\begin{equation*}
\mu^{\alpha}_{\text{sc},n+1}\Bigl(\Omega\setminus \chi^{\alpha}_n\bigl(B_n\bigr)\Bigr)=0,
\end{equation*}
and thus we are done.\\
\\
We also need to show that $\lambda^d(B_n)=0$ implies that $\lambda^d\bigl(\chi^{\alpha}_n(B_n)\bigr)=0$. This follows directly from the newly made assumption (cf. Section \ref{section relax conditions on motion mapping}) that there is a constant $c_n>0$ such that
\begin{equation*}
c_n \lambda^d\bigl(\chi^{\alpha}_n(\Omega')\bigr) \leqslant \lambda^d (\Omega'),\hspace{1 cm}\text{for all }\Omega'\in\mathcal{B}(\Omega).
\end{equation*}
When proving finiteness of $\mu^{\alpha}_{\text{sc},n+1}$ we again can only use that $(\chi^{\alpha}_n)^{-1}(\Omega)\subset\Omega$ instead of $(\chi^{\alpha}_n)^{-1}(\Omega)=\Omega$. It follows that
\begin{equation*}
\mu^{\alpha}_{\text{sc},n+1}(\Omega)=\mu^{\alpha}_{\text{sc},n}\bigl( (\chi^{\alpha}_n)^{-1}(\Omega) \bigr)\leqslant\mu^{\alpha}_{\text{sc},n}(\Omega)<\infty.
\end{equation*}
\section{The proof of Corollary \ref{corollary conservation of mass discretized}}
We redo the proof of Corollary \ref{corollary conservation of mass discretized} completely, because it is slightly more involved.
\begin{proof}
For each $n\in\{0,1,\ldots,N_T-1\}$ and $\alpha\in\{1,2,\ldots,\nu\}$, consider the measure $\mu^{\alpha}_{\omega, n}$, where $\omega\in\{\text{ac, d, sc}\}$. By definition of $\mu^{\alpha}_{\omega, n+1}$ (see the constructive proof of Theorem \ref{Thm existence time-discrete}, Parts \ref{Thm existence time-discrete Part Abs Cont}, \ref{Thm existence time-discrete Part Discrete} and \ref{Thm existence time-discrete Part Singular Continuous}):
\begin{equation*}
\mu^{\alpha}_{\omega,n+1}:=\chi^{\alpha}_n \# \mu^{\alpha}_{\omega,n}.
\end{equation*}
For each $n$ we thus have
\begin{equation*}
\mu^{\alpha}_{\omega,n+1}(\Omega)=\mu^{\alpha}_{\omega,n}\bigl( (\chi^{\alpha}_n)^{-1}(\Omega) \bigr).
\end{equation*}
Note that $\supp \mu^{\alpha}_{\omega,n}\subset\supp\mu_n^{\alpha}\subset(\chi^{\alpha}_n)^{-1}(\Omega)$. This implies that
\begin{eqnarray*}
\mu^{\alpha}_{\omega,n+1}(\Omega) &=& \mu^{\alpha}_{\omega,n}\bigl( (\chi^{\alpha}_n)^{-1}(\Omega) \bigr)\\
&=& \mu^{\alpha}_{\omega,n}\bigl( (\chi^{\alpha}_n)^{-1}(\Omega)\cap \supp\mu_n^{\alpha} \bigr)+ \mu^{\alpha}_{\omega,n}\bigl( (\chi^{\alpha}_n)^{-1}(\Omega)\setminus \supp\mu_n^{\alpha} \bigr)\\
&=& \mu^{\alpha}_{\omega,n}\bigl(\supp\mu_n^{\alpha} \bigr)+ \mu^{\alpha}_{\omega,n}\bigl( \emptyset \bigr)\\
&=& \mu^{\alpha}_{\omega,n}\bigl(\supp\mu_n^{\alpha}\cap \supp \mu^{\alpha}_{\omega,n}\bigr)+ \underbrace{\mu^{\alpha}_{\omega,n}\bigl(\supp\mu_n^{\alpha}\setminus \supp \mu^{\alpha}_{\omega,n}\bigr)}_{=0}\\
&=& \mu^{\alpha}_{\omega,n}\bigl(\supp \mu^{\alpha}_{\omega,n}\bigr)\\
&=& \mu^{\alpha}_{\omega,n}(\Omega).
\end{eqnarray*}
By an inductive argument: $\mu^{\alpha}_{\omega,n}(\Omega)=\mu^{\alpha}_{\omega,0}(\Omega)$ for each $n$.\\
Since $\mu^{\alpha}_n = \mu^{\alpha}_{\text{ac},n}+\mu^{\alpha}_{\text{d},n}+\mu^{\alpha}_{\text{sc},n}$ for all $n\in\mathbb{N}$, the above implies trivially that
\begin{equation*}
\mu^{\alpha}_n(\Omega)=\mu^{\alpha}_0(\Omega),\hspace{1 cm}\text{for all }n\in\mathcal{N}.
\end{equation*}
\end{proof}

\newpage
\chapter[Paper `Modeling micro-macro pedestrian counterflow in\\ heterogeneous domains']{Paper `Modeling micro-macro pedestrian counterflow in heterogeneous domains'}\label{appendix paper}
The following pages contain the full content of \cite{EversMuntean}: our paper `Modeling micro-macro pedestrian counterflow in heterogeneous domains', published in \textit{Nonlinear Phenomena in Complex Systems}.
\newcommand{\setpaperheaders}{
\pagestyle{fancy}
\fancyhf{}
\fancyfoot[LE]{\thepage \hspace{0.5 cm} Modelling Crowd Dynamics}
\fancyfoot[RO]{Modelling Crowd Dynamics\hspace{0.5 cm} \thepage }
}



\section*{Supplementary material}

# Modeling micro-macro pedestrian counterflow in heterogeneous domains

Joep Evers[*]
*Department of Mathematics and Computer Science*
*Eindhoven University of Technology*
*PO Box 513, 5600 MB Eindhoven, The Netherlands*

Adrian Muntean[†]
*Centre for Analysis, Scientific Computing and Applications (CASA)*
*Department of Mathematics and Computer Science*
*Institute for Complex Molecular Systems (ICMS)*
*Eindhoven University of Technology*
*PO Box 513, 5600 MB Eindhoven, The Netherlands*
(Dated: March 11, 2011)

We present a micro-macro strategy able to describe the dynamics of crowds in heterogeneous spatial domains. Herein we focus on the example of pedestrian counterflow. The main working tools include the use of mass and porosity measures together with their transport as well as suitable application of a version of Radon-Nikodym Theorem formulated for finite measures. Finally, we illustrate numerically our microscopic model and emphasize the effects produced by an implicitly defined *social* velocity.

## I. INTRODUCTION

One of the most annoying examples of collective behavior[1] is pedestrian jams – people get clogged up together and cannot reach within the desired time the target destination. Such jams are the immediate consequence of the *simple exclusion process* [2, 3], which basically says that two individuals cannot occupy the same position $x \in \Omega \subset \mathbb{R}^d$ at the same time $t \in S :=]0,T[$, where $T \in ]0,\infty[$ is the final moment at which we are still observing our social network.

Observational data (cf. e.g. [4]) clearly indicates that such jams typically take place in certain neighborhoods of bottlenecks[5] (narrow corridors, exits, corners, inner obstacles/pillars, ...). The effect of heterogeneities[6] on the overall dynamics of the crowd is what motivates our work.

In this paper we start off with the assumption that inside a given room (e.g. a shopping mall), which we denote by $\Omega$, there are *a priori* known zones with restricted access for pedestrians (e.g. closed rooms, prohibited access areas, inner concrete structures)[7], whose union we call $\Omega_s$. Let us also assume that the remaining region, say $\Omega_p$, which is defined by $\Omega_p := \Omega - \Omega_s$, is connected. Consequently, $\Omega_p$ is accessible to pedestrians. The exits of $\Omega$ – target that each pedestrian wants to reach – are assumed to belong to the boundary of $\Omega_p$. The way we imagine the heterogeneity of $\Omega$ is sketched in Figure 1.

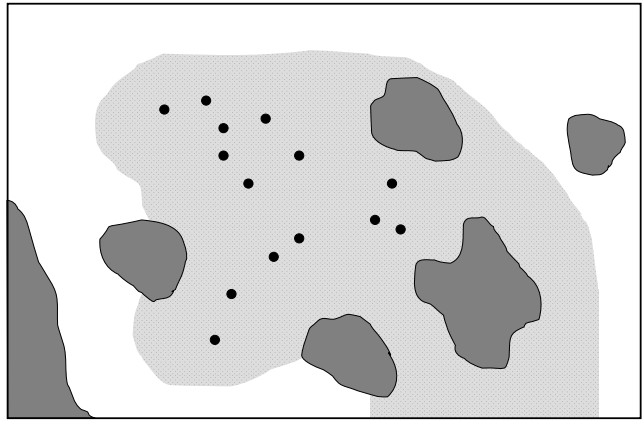

FIG. 1. Schematic representation of the heterogeneous medium $\Omega$. The little black discs represent the pedestrians, while the dark gray zones are the parts where the pedestrians cannot penetrate (i.e. subsets of $\Omega_s$). The pedestrians are considered here to be the microscopic entities, while the light grayish shadow indicates a macroscopic crowd; see Section II A for the precise distinction between *micro* and *macro* made in terms of supports of micro and macro measures.

In this framework, we choose for the following working plan: Firstly, we extend the multiscale approach developed by Piccoli et al. [8] (see also the context described in [9] and [10]) to the case of counterflow[11] of pedestrians; then we allow the pedestrian dynamics to take place in the heterogeneous domain $\Omega$, and finally, we include an implicit velocity law for the pedestrians motion. The main reason why we choose the counterflow scenario [also

[*] j.h.m.evers@student.tue.nl; Corresponding author.
[†] a.muntean@tue.nl

called bidirectional flow [12]] out of the many other well-studied crowd dynamics scenarios is at least threefold:

(i) Pedestrians counterflow is often encountered in the everyday life: at pedestrian traffic lights, or just observe next week-end, when you go shopping, the dynamics of people coming against your walking direction [especially if you are positioned inside narrow corridors].

(ii) The walkers trying to move faster by avoiding local interactions with the oncoming pedestrians facilitate the occurrence of a well-known self-organized macroscopic pattern – lane formation; see, for instance [13].

(iii) We expect the solution to microscopic models posed in narrow corridors to be computationally cheap. Consequently, extensive sensitivity analyses can be performed and the corresponding simulation results can be in principle tested against existing experimental observations [4, 14].

The presence of heterogeneities is quite natural. Pedestrians typically follow existing streets, walking paths, they trust building maps, etc. They take into account the local environment of the place where they are located. If the number of pedestrians is relatively high compared to the available walking space, then the *crowd-structure interaction* becomes of vital importance; see e.g. [15] for preliminary results in this direction.

As long term plan, we wish to understand what are the microscopic mechanisms behind the formation of lanes in heterogeneous environments. In other words, we aim at identifying links between *social force*-type microscopic models (see [16, 17], e.g.) and macroscopic models for lanes (see [13, 18], e.g.) in the presence of heterogeneities. Here we follow a measure-theoretical approach to describe the dynamics of crowds[19]. Our working strategy is very much inspired by the works by M. Böhm [20] and Piccoli et al. [8].

The paper is organized as follows: In Section II we introduce basic modeling concepts defining the mass and porosity measures needed here, as well as a coupled system of transport equations for measures. In Section III we present our concept of *social velocity*. Section IV contains the main result of our paper – the weak formulation of a micro-macro system for pedestrians moving in heterogeneous domains. We close the paper with a numerical illustration of our microscopic model (Section V) exhibiting effects induced by an implicitly defined velocity.

## II. MODELING WITH MASS MEASURES. THE POROSITY MEASURE

For basic concepts of measure theory and their interplay with modeling in materials and life science, we refer the reader, for instance, to [21] and respectively to [9, 20, 22].

### A. Mass measure

Let $\Omega \subset \mathbb{R}^d$ be a domain (read: object, body) with mass. Since we have in mind physically relevant situations only, we consider $d \in \{1, 2, 3\}$. However, most of the considerations reported here do not depend on the choice of the space dimension $d$. Let $\mu_m(\Omega')$ be defined as the mass in $\Omega' \subset \Omega$. Note that whenever we write $\Omega' \subset \Omega$, we actually mean that $\Omega'$ is such that $\Omega' \in \mathcal{B}(\Omega)$, where $\mathcal{B}(\Omega)$ the $\sigma$-algebra of the Borel subsets of $\Omega$. As a rule, we assume $\mu_m$ to be defined on all elements of $\mathcal{B}(\Omega)$.

In Sections II A 1 and II A 2, we consider two specific interpretations of this mass measure that we need to describe the behavior of pedestrians at two separated spatial scales.

#### 1. Microscopic mass measure

Suppose that $\Omega$ contains a collection of $N$ point masses (each of them of mass scaled to 1), and denote their positions by $\{p_k\}_{k=1}^N \subset \Omega$, for $N \in \mathbb{N}$. We want $\mu_m$ to be a counting measure (see Section 1.2.4 in [23], e.g.) with respect to these point masses, i.e. for all $\Omega' \in \mathcal{B}(\Omega)$:

$$\mu_m(\Omega') = \#\{p_k \in \Omega'\}. \tag{1}$$

This can be achieved by representing $\mu_m$ as the sum of Dirac measures, with their singularities located at the $p_k$, $k \in \{1, 2, \ldots, N\}$, namely:

$$\mu_m = \sum_{k=1}^{N} \delta_{p_k}. \tag{2}$$

We refer to the measure $\mu_m$ defined by (2) as *microscopic mass measure*.

#### 2. Macroscopic mass measure

Let us now consider another example of mass measure $\mu_m$. To do this, we assume that the following postulate applies to $\mu_m$:

**Postulate II.1 (Assumptions on $\mu_m$)** *(i) $\mu_m \geqslant 0$.*

*(ii) $\mu_m$ is $\sigma$-additive.*

*(iii)* $\mu_m \ll \lambda^d$, *where $\lambda^d$ is the Lebesgue-measure in $\mathbb{R}^d$.*

By Postulate II.1 (i) and (ii), we have that $\mu_m$ is a positive measure on $\Omega$, whereas (iii) implies that there is no mass present in a set that has no volume (w.r.t. $\lambda^d$). A mass measure satisfying Postulate II.1 is in this context referred to as a *macroscopic mass measure*. Radon-Nikodym Theorem[24] (see [21] for more details on this subject) guarantees the existence of a real, non-negative density $\hat{\rho} \in L^1_{\lambda^d}(\Omega)$ such that:

$$\mu_m(\Omega') = \int_{\Omega'} \hat{\rho}(x) d\lambda^d(x) \qquad \text{for all } \Omega' \in \mathcal{B}(\Omega). \quad (3)$$

Similarly, we introduce time-dependent mass measures $\mu_t$, where the time slice $t \in S$ enters as a parameter.

### B. Porosity measure

Let $\Omega \subset \mathbb{R}^d$ be a heterogeneous domain composed of two distinct regions: *free space* for pedestrian motion and a *matrix* (obstacles) such that $\Omega = \Omega_s \cup \Omega_p$ (disjoint union), where $\Omega_s$ is the matrix (solid part) of $\Omega$ and $\Omega_p$ is the free space (pores). This notation is very much inspired by the modeling of transport and chemical reactions in porous media; see [25], e.g.

Let $\mu_p(\Omega')$ be the volume of pores in $\Omega' \subset \Omega$.

**Postulate II.2 (Assumptions on $\mu_p$)** *(i) $\mu_p \geqslant 0$.*

*(ii) $\mu_p$ is $\sigma$-additive.*

*(iii) $\mu_p \ll \lambda^d$.*

By Postulate II.2 (i) and (ii), we have that $\mu_p$ is a measure on $\Omega$. We refer to $\mu_p$ as a *porosity measure* (cf. [20]). The absolute continuity statement in (iii) formulates mathematically that there cannot be a non-zero volume of pores included in a set that has zero volume (w.r.t. $\lambda^d$). Assume that $\Omega$ is such that $\lambda^d(\Omega) < \infty$. Then the Radon-Nikodym Theorem ensures the existence of a function $\phi \in L^1_+(\Omega)$ such that:

$$\mu_p(\Omega') = \int_{\Omega'} \phi d\lambda^d \quad \text{for all } \Omega' \in \mathcal{B}(\Omega). \quad (4)$$

Note that $\mu_p(\Omega')$ measures the volume of a subset of $\Omega'$ (namely of $\Omega' \cap \Omega_p$). So, we get that

$$\mu_p(\Omega') = \lambda^d(\Omega' \cap \Omega_p) \leqslant \lambda^d(\Omega') \quad \text{for all } \Omega' \in \mathcal{B}(\Omega). \quad (5)$$

We thus have $\int_{\Omega'} \phi d\lambda^d \leqslant \int_{\Omega'} d\lambda^d$, or $\int_{\Omega'} (1-\phi) d\lambda^d \geqslant 0$. Since the latter inequality holds for any choice of $\Omega'$, it follows that $\phi \leqslant 1$ almost everywhere in $\Omega$.

### C. Transport of a measure

For the sequel, we wish to restrict the presentation to the case $d=2$. For our time interval $S$ and for each $i \in \{1,2\}$, we denote the velocity field of the corresponding measure by $v^i(t,x)$ with $(t,x) \in S \times \Omega$. Let also $\mu_t^1$, and $\mu_t^2$ be two time-dependent mass measures. Note that for each choice of $i$, the dependence on $t$ of $v^i$ is comprised in the functional dependence of $v^i$ on both measures $\mu_t^1$, and $\mu_t^2$. This is clearly indicated in (9).

The fact that here we deal with two mass measures $\mu_t^1$ and $\mu_t^2$, transported with corresponding velocities $v^1$ and $v^2$, translates into:

$$\begin{cases} \dfrac{\partial \mu_t^1}{\partial t} + \nabla \cdot (\mu_t^1 v^1) = 0, \\ \dfrac{\partial \mu_t^2}{\partial t} + \nabla \cdot (\mu_t^2 v^2) = 0, \end{cases} \text{for all } (t,x) \in S \times \Omega. \quad (6)$$

These equations are accompanied by the following set of initial conditions:

$$\mu_t^i = \mu_0^i \text{ as } t=0 \text{ for } i \in \{1,2\}. \quad (7)$$

It is worth noting that (6) is the measure-theoretical counterpart of the Reynolds Theorem in continuum mechanics. To be able to interpret what a partial differential equation in terms of measures means, we give a weak formulation of (6). Essentially, for all test functions $\psi^1, \psi^2 \in C^1_0(\bar{\Omega})$ and for almost every $t \in S$, the following identity holds:

$$\frac{d}{dt} \int_\Omega \psi^i(x) d\mu_t^i(x) = \int_\Omega v^i(t,x) \cdot \nabla \psi^i(x) d\mu_t^i(x) \quad (8)$$

for all $i \in \{1,2\}$.

**Definition II.1 (Weak solution of (6))** *The pair $(\{\mu_t^1\}_{t \geqslant 0}, \{\mu_t^2\}_{t \geqslant 0})$ is called a weak solution of (6), if for all $i \in \{1,2\}$ the following properties hold:*

1. *the mappings $t \mapsto \int_\Omega \psi^i(x) d\mu_t^i(x)$ are absolutely continuous for all $\psi^i \in C^1_0(\bar{\Omega})$;*

2. $v^i \in L^2\Big(S; L^1_{\mu_t^i}(\Omega)\Big)$;

3. *Equation (8) is fulfilled.*

We refer the reader to [26] for an example where the existence of weak solutions to a similar (but easier) transport equation for measures has been rigorously shown.

### III. SOCIAL VELOCITIES

We follow very much the philosophy developed by Helbing, Vicsek and coauthors (see, e.g. [13] and references cited therein) which defends the idea that the pedestrian's motion is driven by a social force. Is worth noting that similar thoughts were given in this direction (motion

of social masses/networks) much earlier, for instance, by Spiru Haret [27] and Antonio Portuondo y Barceló [28]. Moreover, other authors (for instance, Hoogendoorn and Bovy [29]) prefer to account also for the Zipfian principle of least effort for the human behavior. We do not attempt to capture the least effort principle in this study.

### A. Specification of the velocity fields $v^i$

Until now, we have not explicitly defined the velocity fields $v^i$ ($i \in \{1,2\}$). Very much inspired by the *social force model* by Dirk Helbing et al. [16], the velocity of a pedestrian is modeled as a *desired velocity* $v^i_{\text{des}}$ perturbed by a component $v^i_{[\mu^1_t, \mu^2_t]}$. The latter component is due to the presence of other individuals, both from the pedestrian's own subpopulation and from the other subpopulation. The desired velocity is independent of the measures $\mu^1_t$ and $\mu^2_t$, and represents the velocity that an agent would have had in absence of other pedestrians.

For each $i \in \{1,2\}$, the velocity $v^i$ is defined by superposing the two velocities $v^i_{\text{des}}$ and $v^i_{[\mu^1_t, \mu^2_t]}$ as follows:

$$v^i(t,x) := v^i_{\text{des}}(x) + v^i_{[\mu^1_t, \mu^2_t]}(x), \qquad (9)$$

for all $t \in (0,T)$ and $x \in \Omega$.

For a *counterflow* scenario, the desired velocities of the two subpopulations follow opposite directions. We thus take

$$v^i_{\text{des}}(x) = v^i_{\text{des}} \in \mathbb{R}^2$$

fixed (for $i \in \{1,2\}$) and

$$v^1_{\text{des}} = -v^2_{\text{des}}.$$

The component $v^i_{[\mu^1_t, \mu^2_t]}$ models the effect of interactions with other pedestrians on the current velocity[30]. Since the interactions between members of the same subpopulation differ (in general) from the interactions between members of opposite subpopulations, we assume that $v^i_{[\mu^1_t, \mu^2_t]}$ consists of two parts:

$$v^i_{[\mu^1_t, \mu^2_t]}(x) := \int_{\Omega \setminus \{x\}} f^{\text{own}}(|y-x|) g(\alpha^i_{xy}) \frac{y-x}{|y-x|} d\mu^i_t(y)$$
$$+ \int_{\Omega \setminus \{x\}} f^{\text{opp}}(|y-x|) g(\alpha^i_{xy}) \frac{y-x}{|y-x|} d\mu^j_t(y), \qquad (10)$$

for $i \in \{1,2\}$, where $j=1$ if $i=2$ and vice versa. In (10) we have used the following:

- $f^{\text{own}}$ and $f^{\text{opp}}$ are continuous functions from $\mathbb{R}_+$ to $\mathbb{R}$, describing the effect of the mutual distance between individuals on their interaction. Compare the concept of *distance interactions* defined in [22]. $f^{\text{own}}$ incorporates the influence by members of the same subpopulation, whereas $f^{\text{opp}}$ accounts for the interaction between members of opposite subpopulations. $f^{\text{own}}$ is a composition of two effects: on the one hand individuals are repelled, since they want to avoid collisions and congestion, on the other hand they are attracted to other group mates, in order not to get separated from the group. $f^{\text{opp}}$ only contains a repulsive part, since we assume that pedestrians do not want to stick to the other subpopulation.

- $\alpha^i_{xy}$ denotes the angle between $y-x$ and $v^i_{\text{des}}(x)$: the angle under which $x$ sees $y$ if it were moving in the direction of $v^i_{\text{des}}(x)$.

- $g$ is a function from $[-\pi, \pi]$ to $[0,1]$ that encodes the fact that an individual's vision is not equal in all directions.

Regarding the specific choice of $f^{\text{own}}$, $f^{\text{opp}}$ and $g$ we are very much inspired by [16] and [8], e.g. However we do not use exactly their way of modeling pedestrians' interaction forces. We list here the following forms for the functions $f^{\text{own}}$, $f^{\text{opp}}$ and $g$ that match the given characterization:

$$f^{\text{opp}}(s) := \begin{cases} -F^{\text{opp}}\left(\frac{1}{s^2} - \frac{1}{(R^{\text{opp}}_r)^2}\right), & \text{if } s \leqslant R^{\text{opp}}_r; \\ 0, & \text{if } s > R^{\text{opp}}_r, \end{cases} \quad (11)$$

$$f^{\text{own}}(s) := \begin{cases} -F^{\text{own}}\left(\frac{1}{s} - \frac{1}{R^{\text{own}}_r}\right)\left(\frac{1}{s} - \frac{1}{R^{\text{own}}_a}\right), \\ \qquad\qquad \text{if } s \leqslant R^{\text{own}}_a; \\ 0, \qquad \text{if } s > R^{\text{own}}_a, \end{cases} \quad (12)$$

$$g(\alpha) := \sigma + (1-\sigma)\frac{1+\cos(\alpha)}{2}, \text{ for } \alpha \in [-\pi.\pi]. (13)$$

Here $F^{\text{opp}}$ and $F^{\text{own}}$ are fixed, positive constants. The constants $R^{\text{opp}}_r$, $R^{\text{own}}_r$ (*radii of repulsion*) and $R^{\text{own}}_a$ (*radius of attraction*) are fixed and should be chosen such that $0 < R^{\text{own}}_r < R^{\text{own}}_a$ and $0 < R^{\text{opp}}_r$. Furthermore, the restriction $\max\{R^{\text{own}}_a, R^{\text{opp}}_r\} \ll L$ has to be fulfilled. The interaction $f^{\text{opp}}$ is designed such that individuals "feel" repulsion (i.e. $f^{\text{opp}} < 0$) from another pedestrian if they are placed within a distance $R^{\text{opp}}_r$ from one another. The corresponding statement holds for $f^{\text{own}}$ if individuals are within the distance $R^{\text{own}}_r$. Additionally, an individual is attracted to a second individual if their mutual distance ranges between $R^{\text{own}}_r$ and $R^{\text{own}}_a$.

The function $g$ ensures that an individual experiences the strongest influence from someone straight ahead, since $g(0) = 1$ for any $\sigma \in [0,1]$. The constant $\sigma$ is a tuning parameter called *potential of anisotropy*. It determines how strongly a pedestrian is focussed on what happens in front of him, and how large the influence is of people at his sides or behind him.

In the remainder of this section, we suggest four different alternatives for the definition of $v^i_{[\mu^1_t, \mu^2_t]}$ by

indicating various special choices of distance interactions and visibility angles (conceptually similar to $\alpha_{xy}^i$) as they arise in (10). All of them boil down to including an implicit dependency of the actual velocity $v^i = v_{\text{des}}^i + v_{[\mu_t^1, \mu_t^2]}^i$. Note that this effect increases the degree of realism of the model, but on the other hand it makes the mathematical justification of the corresponding models much harder to get.

### 1. Modification of the angle $\alpha_{xy}^i$

We defined the angle $\alpha_{xy}^i$ as the angle between the vector $y - x$ and $v_{\text{des}}^i(x)$. However this is not a good definition if the pedestrian in position $x$ is not moving in the direction of $v_{\text{des}}^i(x)$ (or, in a broader sense, if the actual speed cannot be approximated sufficiently well by the desired velocity). Therefore we suggest to define $\alpha_{xy}^i = \alpha_{xy}^i(t)$ as the angle between $y - x$ and $v^i(t, x)$.

### 2. Prediction of mutual distance in (near) future

Up to now the functions $f^{\text{own}}$ and $f^{\text{opp}}$ depended on the actual distance between $x$ and $y$ at time $t$. However pedestrians are likely to anticipate on the distance they expect to have after a certain (small) time-step (say, some fixed $\Delta t \in \mathbb{R}$). In practice, this means that at a time $t \in S$ a person will modify his velocity (either in direction, or in magnitude, or both) if he foresees a collision at time $t + \Delta t \in S$.

To predict the mutual distance between $x$ and $y$ at time $t + \Delta t$, the current velocities at $x$ and $y$ are used for extrapolation. The predicted distance is: $|(y + v(y,t)\Delta t) - (x + v(x,t)\Delta t)|$. Consequently, sticking to the notation in (10), the interaction potential $f^{\text{own}}$ and $f^{\text{opp}}$ should depend on $|(y + v^i(t,y)\Delta t) - (x + v^i(t,x)\Delta t)|$ and on $|(y + v^j(t,y)\Delta t) - (x + v^i(t,x)\Delta t)|$ respectively (where $j = 1$ if $i = 2$ and vice versa).

### 3. Prediction of mutual distance within a time interval in the (near) future

The disadvantage of using $|(y + v(y,t)\Delta t) - (x + v(x,t)\Delta t)|$ is that $\Delta t$ is fixed. A pedestrian can thus only predict the distance at an *a priori* specified point in time in the future. However, people are able to anticipate also if they expect a collision to occur at a time that is not equal to $t + \Delta t$. We assume now that we are given a fixed $\Delta t_{\max} \in \mathbb{R}_+$ such that an individual can predict mutual distances by extrapolation for any time $\tau \in (t, t + \Delta t_{\max})$. Thus, $\Delta t_{\max}$ imposes a bound on how far can an individual look ahead into the future. To capture this effect, we suggest to replace $f^{\text{own}}(|y - x|)$ and $f^{\text{opp}}(|y - x|)$ by:

$$\frac{1}{\Delta t_{\max}} \int_0^{\Delta t_{\max}} f^{\text{own}}(|(y + v^i(t,y)\tau) - (x + v^i(t,x)\tau)|) d\tau, \tag{14}$$

and

$$\frac{1}{\Delta t_{\max}} \int_0^{\Delta t_{\max}} f^{\text{opp}}(|(y + v^j(t,y)\tau) - (x + v^i(t,x)\tau)|) d\tau, \tag{15}$$

respectively.

### 4. Weighted prediction

Since an individual probably attaches more value to his predictions for points in time that are nearer by than others, one additional modification comes to our mind. Let $h : [t, t + \Delta t_{\max}] \to [0, 1]$ be a weight function. Then instead of (14) and (15), we propose

$$\frac{1}{\Delta t_{\max}} \int_0^{\Delta t_{\max}} f^{\text{own}}(|(y + v^i(t,y)\tau) - (x + v^i(t,x)\tau)|) h(\tau) d\tau, \tag{16}$$

and

$$\frac{1}{\Delta t_{\max}} \int_0^{\Delta t_{\max}} f^{\text{opp}}(|(y + v^j(t,y)\tau) - (x + v^i(t,x)\tau)|) h(\tau) d\tau. \tag{17}$$

If $h$ is decreasing, then the influence of $t_1$ is larger than the influence of $t_2$, if $t_1 < t_2$ (which matches our intuition).

## B. Two-scale measures

We now consider the explicit decomposition of the measures $\mu_t^1$ and $\mu_t^2$. Let the pair $(\theta_1, \theta_2)$ be in $[0,1]^2$, and consider the following decomposition of $\mu_t^1$ and $\mu_t^2$:

$$\mu_t^i = \theta_i m_t^i + (1 - \theta_i) M_t^i, \qquad i \in \{1, 2\}. \tag{18}$$

Here, $m_t^i$ is a microscopic measure. We consider $\{p_k^i(t)\}_{k=1}^{N^i} \subset \Omega$ to be the positions at time $t$ of $N^i$ chosen pedestrians, that are members of subpopulation $i$. We want $m_t^i$ to be a counting measure with respect to these pedestrians, i.e. for all $\Omega' \in \mathcal{B}(\Omega)$:

$$m_t^i(\Omega') = \#\{p_k^i(t) \in \Omega'\}, \qquad i \in \{1, 2\}. \tag{19}$$

We thus define $m_t^i$ as the sum of Dirac masses (cf. Section II A 1), centered at the $p_k^i$, $k = 1, 2, \ldots, N^i$:

$$m_t^i = \sum_{k=1}^{N^i} \delta_{p_k^i(t)}, \qquad i \in \{1, 2\}. \tag{20}$$

$M_t^i$ is the macroscopic part of the measure, which takes into account the part of the crowd that is considered

continuous. We consequently have $M_t^i \ll \lambda^2$, since a set of zero volume cannot contain any mass. Note that we are thus in the setting of Section II A 2. Now, Radon-Nikodym Theorem guarantees the existence of a real, non-negative density $\hat{\rho}^i(t,\cdot) \in L^1_{\lambda^2}(\Omega)$ such that:

$$M_t^i(\Omega') = \int_{\Omega'} \hat{\rho}^i(t,x) d\lambda^2(x) \qquad (21)$$

for all $\Omega' \in \mathcal{B}(\Omega)$ and all $i \in \{1,2\}$.

## IV. MICRO-MACRO MODELING OF PEDESTRIANS' MOTION IN HETEROGENEOUS DOMAINS

We have already made clear that we want to model the heterogeneity of the interior of the corridor. In practice this means that pedestrians cannot enter all parts of the domain. As described in Section II B, we have a measure $\mu_p$ corresponding to the porosity of the domain (which is fixed in time). However, we note that the concept of porosity (cf. Section II B) is a macroscopic one. For this reason only the macroscopic part of the mass measure in (18) needs some modification with respect to the porosity. In this context, one should be aware of the analogy with mathematical *homogenization*. This technique distinguishes between microscopic and macroscopic scales, where we also see that some (averaged) characteristics are only defined on the macroscopic scale. For more details, the reader is referred to [25] or [31]. In $\mathbb{R}^2$, we have $\mu_p \ll \lambda^2$. Furthermore $M_t^i \ll \mu_p$ for $i \in \{1,2\}$ and a.e. $t \in S$. This is obvious, since no pedestrians can be present in a set that has no pore space (i.e. zero porosity measure). A basic property of Radon-Nikodym derivatives now gives us:

$$\frac{dM_t^i}{d\lambda^2} = \frac{dM_t^i}{d\mu_p}\frac{d\mu_p}{d\lambda^2} \qquad (22)$$

for each $i \in \{1,2\}$ and for almost every $t \in S$.

We have already defined $\hat{\rho}^i(t,\cdot) := \frac{dM_t^i}{d\lambda^2}$ and $\phi := \frac{d\mu_p}{d\lambda^2}$. If we now denote by $\rho^i(t,\cdot)$ the Radon-Nikodym derivative $\frac{dM_t^i}{d\mu_p}$, the following relation holds: $\hat{\rho}^i(t,\cdot) \equiv \rho^i(t,\cdot)\phi(\cdot)$ for all $i \in \{1,2\}$.

### A. Weak formulation for micro-macro mass measures

We now have the following measure:

$$\mu_t^i = \theta_i m_t^i + (1-\theta_i)M_t^i, \qquad i \in \{1,2\}, \qquad (23)$$

as was given in (18), where now:

$$m_t^i = \sum_{k=1}^{N^i} \delta_{p_k^i(t)}, \qquad dM_t^i(x) = \rho^i(t,x)\phi(x)d\lambda^2(x). \qquad (24)$$

This specific form of the measure will now be included in the weak formulation (8), with velocity field (9)-(10). The real positive numbers $\theta_i$ ($i \in \{1,2\}$) are intrinsic scaling parameters depending on $N^i$.

The transport equation (8) takes the following form:

$$\frac{d}{dt}\left(\theta_i \sum_{k=1}^{N^i} \psi^i(p_k^i(t)) + (1-\theta_i)\int_\Omega \psi^i(x)\rho^i(t,x)\phi(x)d\lambda^2(x)\right) =$$
$$\theta_i \sum_{k=1}^{N^i} v^i(t,p_k^i(t)) \cdot \nabla\psi^i(p_k^i(t)) + (1-\theta_i)\int_\Omega v^i(t,x) \cdot \nabla\psi^i(x)\rho^i(t,x)\phi(x)d\lambda^2(x), \qquad (25)$$

for all $i \in \{1,2\}$. Here we have used the sifting property of the Dirac distribution. In the same manner, we specify $v^i_{[\mu_t^1,\mu_t^2]}$ from (10) as

$$\begin{aligned}
v^i_{[\mu^1_t,\mu^2_t]}(x) = {} & \theta_i \sum_{\substack{k=1 \\ p^i_k(t)\neq x}}^{N^i} f^{\text{own}}(|p^i_k(t)-x|)g(\alpha^i_{xp^i_k(t)})\frac{p^i_k(t)-x}{|p^i_k(t)-x|} \\
& + (1-\theta_i)\int_\Omega f^{\text{own}}(|y-x|)g(\alpha^i_{xy})\frac{y-x}{|y-x|}\rho^i(t,y)\phi(y)d\lambda^2(y) \\
& + \theta_j \sum_{\substack{k=1 \\ p^j_k(t)\neq x}}^{N^j} f^{\text{opp}}(|p^j_k(t)-x|)g(\alpha^i_{xp^j_k(t)})\frac{p^j_k(t)-x}{|p^j_k(t)-x|} \\
& + (1-\theta_j)\int_\Omega f^{\text{opp}}(|y-x|)g(\alpha^i_{xy})\frac{y-x}{|y-x|}\rho^j(t,y)\phi(y)d\lambda^2(y),
\end{aligned} \quad (26)$$

for $i \in \{1,2\}$, and $j$ as before ($j=1$ if $i=2$ and vice versa). We have omitted the exclusion of $\{x\}$ from the domain of integration (in the macroscopic part), since $\{x\}$ is a nullset and thus negligible w.r.t. $\lambda^2$. Note that the sums may be evaluated in any point $x \in \Omega$ (not necessarily $x = p^i_k(t)$ for some $i$ and $k$); the integral parts may also be evaluated in all $x$, including $x = p^i_k(t)$ for some $i$ and $k$.

## V. NUMERICAL ILLUSTRATION

We wish to illustrate now the microscale description of a counterflow scenario (i.e. for $\theta_1 = \theta_2 = 1$) by presenting plots of the configuration of all individuals situated in a given corridor at specific moments in time.

We consider a specific instance in which there are in total 40 individuals (20 in each subpopulation). The dimensions of the corridor are $d=4$ and $L=20$. The velocity is taken as defined in (10)-(13). Furthermore, the following model parameters are used: $v^1_{\text{des}} = 1.34e_1$, $v^2_{\text{des}} = -1.34e_1$, $F^{\text{opp}} = 0.3$, $F^{\text{own}} = 0.3$, $R^{\text{opp}}_r = 2$, $R^{\text{own}}_a = 2$, $R^{\text{own}}_r = 0.5$, $F^{\text{w}} = 0.5$, $R^{\text{w}} = 0.5$, $\sigma = 0.5$.

In Figure 2, we show the configuration in the corridor at times $t=0$, $t=7.5$, and $t=15$. The individuals of the subpopulation 1 are colored blue, while the individuals of the subpopulation 2 are colored red. Clearly, self-organization can be observed in the system: Pedestrians that desire to move in the same direction form lanes (in this case, three of them). This feature is observed and described extensively in literature, cf. e.g. [13].

Another feature, pointed out by Figure 2, is the following: Within the three already formed lanes, small clusters of people are formed. This flocking is a result of the typical choice for $f^{\text{own}}$ in (12). Members of the same subpopulation are repelled if their mutual distance is in the range $(0, R^{\text{own}}_r)$; they are attracted if their mutual distance is in the range $(R^{\text{own}}_r, R^{\text{own}}_a)$. No interaction takes place if individuals are more than a distance $R^{\text{own}}_a$ apart.

The attraction part of the interaction causes individuals that are already relatively close to get even closer, until they are at a distance $R^{\text{own}}_r$. For distances around $R^{\text{own}}_r$, there is an interplay between repulsion and attraction, eventually leading to some equilibrium in the mutual distances between neighboring individuals in one cluster. In Figure 2, we observe self-organized patterns even within clusters.

We should mention here that further simulation is needed, in particular on multiply-connected domains. Note however, that as the model is now, moving individuals are not a priori hindered to penetrate any obstacle. Moreover we wish to point out that the number of particles should be sufficiently large, in order to be statistically acceptable.

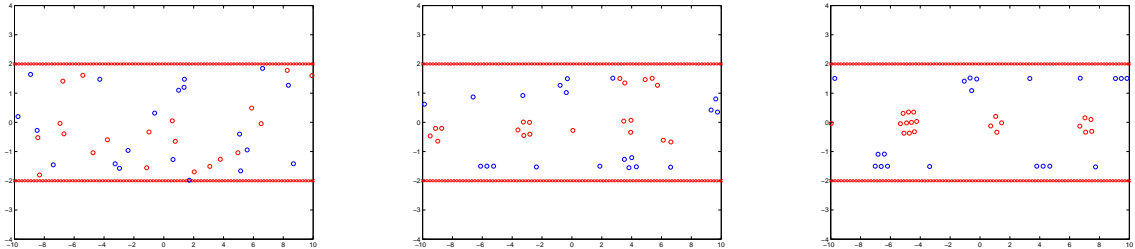

FIG. 2. The simulation of a crowd's motion in a corridor of length $L = 20$ and width $d = 4$. Each of the two sub-populations consists of 20 individuals. The images were taken at $t = 0$ (left), $t = 7.5$ (middle), $t = 15$ (right).



\setheaders{}
\newpage
\addcontentsline{toc}{chapter}{Bibliography}
\bibliographystyle{abbrv}
\bibliography{references}

\end{document}